\documentclass[10pt]{report}
\usepackage[utf8]{inputenc}

\usepackage[T1]{fontenc}
\usepackage[brazil,english]{babel}
\usepackage{graphics,graphicx}
\usepackage{geometry}
\usepackage{amsmath}
\usepackage{amssymb}
\usepackage{amsthm}
\usepackage{multicol} 
\usepackage[all]{xy} 
\usepackage{color} 
\usepackage[pdftex]{hyperref}
\usepackage[dvipsnames]{xcolor}
\usepackage{float}
\usepackage{faktor} 
\usepackage{xfrac} 
\usepackage{indentfirst}
\usepackage{graphicx}
\usepackage{mathtools}
\usepackage{comment}
\usepackage{pdfpages}
\usepackage{caption}
\usepackage{faktor}
\usepackage{geometry}
\usepackage{enumerate}
\usepackage[shortlabels]{enumitem}
\usepackage{wasysym}
\usepackage{tikz} 
\usepackage{dynkin-diagrams} 
\usepackage{epic} 
\usepackage{algorithm}
\usepackage{algorithmic}
\DeclareMathAlphabet{\mathpzc}{OT1}{pzc}{m}{it}

\usepackage{csquotes}
\usepackage[
    backend=biber,
    style=abnt-numeric,
    natbib=true,
    sorting=nyt,      
    sortcites=true    
]{biblatex}
\DeclareFieldFormat{labelnumberwidth}{[#1]}
\usepackage{ragged2e}
 \usepackage{quiver}
\usepackage{longtable}
\usepackage{titlesec}
\usepackage{setspace}
\usepackage{tocloft}

\makeatletter
\newcommand*\lofnumberline[1]{Figure~#1 -- }
\renewcommand*\l@figure[2]{%
  \begingroup
  \let\numberline\lofnumberline
  \@dottedtocline{1}{0em}{0em}{#1}{#2}%
  \endgroup
}
\makeatother

\usepackage[most]{tcolorbox}
\newtcolorbox[auto counter,number within=chapter]{codebox}[2][]{enhanced,  breakable,  colback=gray!8,  colframe=gray!20,  fonttitle=\bfseries, coltitle=black,  title=Box~\thetcbcounter: #2,  #1}

\newtheorem{teo}{Theorem}[chapter]
\newtheorem{defi}[teo]{Definition}
\newtheorem{prop}[teo]{Proposition}

\newtheorem{Ex}[teo]{Example}

\newtheorem{corollary}[teo]{Corolary}

\newcommand{\R}{\mathbb{R}}

\newcommand{\N}{\mathbb{N}}

\newcommand{\fim}{\begin{flushright} $\square$ \end{flushright}}

\usepackage{eso-pic}

\titleformat{\chapter}[block]
{\raggedright\bfseries\huge}
{\thechapter\space}          {0pt}                   
{}         

\makeatletter
\let\origps@plain\ps@plain
\newcommand{\pretextual}{%
  \pagestyle{empty}%
  \let\ps@plain\ps@empty%
}
\newcommand{\textual}{%
  \cleardoublepage%
  \let\ps@plain\origps@plain%
  \pagestyle{plain}%
  \setcounter{page}{14}%
}
\makeatother

\DeclareCaptionLabelFormat{dash}{#1~#2~--~}
\begin{document}
\pretextual

\vspace*{0.8cm}
\begin{center}
\large{\textsc{\textbf{Leticia Becher}}}

\vspace{2cm}  

\begin{minipage}[t]{\textwidth}
\begin{center}
\huge{\textbf{A Data-Constrained Framework for Marine Biogeochemistry Modeling with Applications to the Paranaguá Estuarine Complex}
}
\end{center}
\end{minipage}

\vspace{2cm}
\large{\textsc{\textbf{PhD Thesis}}}

\vspace{2cm}
\large{\textsc{\textbf{Department of Mathematics}}}

\large{\textsc{\textbf{Federal University of Paraná}}}

\vspace{2cm}
\includegraphics[width=0.4\linewidth]{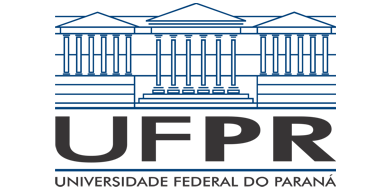}

\vspace{0.5cm}
\large{\textsc{January 2026}}
\end{center}

\newpage
\thispagestyle{empty}

\begin{flushleft}

Leticia Becher (Letícia Becher Yamashita)

\texttt{leticiabecher2017@gmail.com}

Department of Mathematics \\
Federal University of Paraná \\
Curitiba, Brazil

\vspace{1.5cm}
\begin{tabular}{p{3.5cm} p{11cm}}
\textbf{Thesis title:} & A Data-Constrained Framework for Marine Biogeochemistry Modeling with Applications to the Paranaguá Estuarine Complex \\\\
\textbf{Supervisor:} & Prof. Dr. Francisco de Melo Viríssimo \\
& LSE, United Kingdom and INCLINE-USP, Brazil \\\\
\textbf{Co-supervisors:} & Prof. Dr. Ademir Alves Ribeiro \\
& Federal University of Paraná \\\\
& Prof. Dr. Roberto Ribeiro Santos Júnior \\
& Federal University of Paraná \\\\
\textbf{Code repository:} & \url{https://github.com/leticiabecher/ThesisExperiments} \\\\
\textbf{Defense date:} & August 2025 \\\\
\textbf{Final version:} & January 2026
\end{tabular}

\end{flushleft}

\thispagestyle{empty}
\justify

\centerline{\Large\bf \textsc{Acknowledgements}}

\vspace{1.5cm}

\noindent I would like to thank my supervisors, the members of the thesis committee, and the coordinator of the Graduate Program in Mathematics, Professor Marcelo Muniz, for their support, availability, and confidence in my potential throughout the doctoral program.

\vspace{0.5cm}
\noindent In particular, I am especially grateful to Marines for her essential partnership in the conception of the initial form of this project, and to Francisco for his careful and constant supervision throughout the development of this work.

\vspace{0.5cm}
\noindent I am also deeply grateful to Rick, my husband, for his care, support, and encouragement throughout this journey.

\vspace{0.5cm}
\noindent I thank Dr. Sophy Oliver, Professor Iris Kriest, and PhD candidate Rafaela Farias do Nascimento for their support during the first iteration of this project.

\vspace{0.5cm}
\noindent I also thank Carlos Carvalho and the team of the Central Laboratory for High Performance Processing (LCPAD) at UFPR for their support with the use of the computing cluster.

\vfill
\noindent This work was carried out with the support of the Coordination for the Improvement of Higher Education Personnel – Brazil (CAPES) – Financing Code 001.

\vspace{0.5cm}
\noindent This work was partially funded by FINEP through the CT-INFRA/UFPR projects.

\newpage

\centerline{\Large\bf ABSTRACT}

\vspace{1cm}

\begin{center}
\begin{minipage}[t]{13.5cm} 
\justify
Over the last two decades, mathematical and computational models have become one of the primary tools used in the study of marine biogeochemical systems, enabling researchers to test hypotheses and to investigate scenarios of interest in the present, past, and future – previously impossible without dedicated oceanographic expeditions. These include the study of climate change, water quality prediction, nutrient flows, harmful algal bloom risk, and port and estuarine management, among others.
In general, developing models that reliably represent the environment of interest is a complex and computationally expensive task, which hinders the application of these models in specific environments, such as plumes and estuarine complexes – such as the Paranaguá Estuarine Complex (PEC) on the Paraná State coast, in Southern Brazil. Furthermore, biogeochemical models have several parameters that need to be calibrated for the region of interest, which are generally tuned empirically, based on heuristics and inferences from available data, due to the high computational cost involved in the simulations.
In this work, we propose advancing the state of the art in the development and calibration of marine biogeochemical models, within the Brazilian context, with three main contributions. First, we developed a conceptual model for the nutrient-phytoplankton dynamics of the PEC – one of the first models developed in Brazil for the region. The model is simple and computationally inexpensive, accessible to any home laptop. Next, the work proposes a systematic calibration approach for marine biogeochemical models using tracer datasets from the study regions and optimization. Finally, we present a practical application of this approach to the PEC model, where model calibration is performed using in-situ tracer data. The results show that the model, despite its simplicity, is capable of reproducing the observations provided it is properly calibrated, demonstrating the power of using optimization techniques – even in a conceptual model setting.
As future opportunities, we note that both the model and the approach are generalizable, enabling multi-parameter calibration, seasonal variation of parameters, and biochemical models coupled with higher-fidelity hydrodynamic models, with further automation potential via machine learning.


\vspace{10pt}
\noindent {\bf Keywords:} mathematical and computational modeling; marine biogeochemical modeling; estuarine dynamics; applied optimization; derivative-free optimization for least squares; parameter calibration in models; Paranaguá Estuarine Complex.

\vspace{10pt}
\noindent \textbf{Code repository:}  \url{https://github.com/leticiabecher/ThesisExperiments}

\end{minipage}
\end{center}

\cleardoublepage

\centerline{\Large\bfseries LIST OF FIGURES} 
\vspace{1cm}

\begingroup
\makeatletter
\@starttoc{lof} 
\makeatother
\endgroup

\cleardoublepage

\thispagestyle{empty}
\centerline{\Large\bf LIST OF ABBREVIATIONS, ACRONYMS AND SYMBOLS}

\vspace{1cm}
\renewcommand{\arraystretch}{1.2}
\begin{longtable}{l  p{0.7\textwidth}}
  $C_\mathrm{N}$ & Mean nitrate concentration in a given water volume (mmol m$^{-3}$). \\
  $C_\mathrm{PHY}$ & Mean phytoplankton concentration in a given water volume (mmol m$^{-3}$). \\
  $V_{\max}$ & Maximum specific growth rate of phytoplankton (d$^{-1}$). \\
  $\lambda$ & Linear mortality rate of phytoplankton (d$^{-1}$). \\
  $K$ & Monod half-saturation constant for nitrogen (mmol m$^{-3}$). \\
  $r$ & Nitrogen remineralization rate (d$^{-1}$). \\
  $I$ & Solar irradiance (kWh m$^{-3}$ d$^{-1}$). \\
  $T$ & Mean water temperature ($^o$C). \\
  $S$ & Mean water salinity (dimensionless). \\
  $\eta(C_\mathrm{N})$ & Limiting factor on phytoplankton growth due to nitrate availability, \\
                       & dimensionless, $0 \leq \eta(C_\mathrm{N}) \leq 1$. \\
  $\gamma(I)$          & Limiting factor on phytoplankton growth due to solar irradiance, \\
                       & dimensionless, $0 \leq \gamma(I) \leq 1$. \\
  $\alpha(T)$          & Limiting factor on phytoplankton growth due to water temperature, \\
                       & dimensionless, $0 \leq \alpha(T) \leq 1$. \\
  $\beta(S)$           & Limiting factor on phytoplankton growth due to water salinity, \\
                       & dimensionless, $0 \leq \beta(S) \leq 1$. \\
  $g(C_\mathrm{N},T,S,I)$ & Total limiting factor on phytoplankton growth, \\
                          & dimensionless, $0 \leq g(C_\mathrm{N},T,S,I) \leq 1$. \\
  $L_x$ & Horizontal extent of the estuarine box in the $x$-direction (m). \\
  $L_y$ & Horizontal extent of the estuarine box in the $y$-direction (m). \\
  $H$ & Mean total water depth in the model domain (m). \\
  $H_\text{up}$ & Mean depth of the upper (surface) layer in the two-box representation (m). \\
  $H_\text{low}$ & Mean depth of the lower (bottom) layer in the two-box representation (m). \\
  $S_\text{up}(t)$ & Salinity in the upper box as a function of time $t$ (dimensionless). \\
  $S_\text{low}(t)$ & Salinity in the lower box as a function of time $t$ (dimensionless). \\
  $Q_\text{ocean}(t)$ & Water volume flux exchanged with the adjacent ocean as a function of time $t$ (m$^{3}$ s$^{-1}$). \\
  $Q_\text{river}(t)$ & Freshwater inflow from rivers as a function of time $t$ (m$^{3}$ s$^{-1}$). \\
  $Q_\text{ebm}(t)$ & Net volume flux associated with the estuarine box model (EBM) forcing as a function of time $t$ (m$^{3}$ s$^{-1}$). \\
  $\text{Vol}_\text{box}$ & Water volume of a model box (m$^{3}$). \\
  $T(t)$ & Upper-box temperature, in $^o$C, as a function of time $t$. \\
  $C_\mathrm{N}^\text{river}(t)$ & Average nitrogen concentration in the river inflows to the PEC, as a function of time $t$ (mmol m$^{-3}$). \\
  $\omega$ & A point in the parameter space $\R^n$. \\
  $n$ & Dimension of the parameter vector $\omega$. \\
  $p+1$ & Number of candidate parameters in the candidate pool. \\
  $\omega^k_j$ & Candidate parameter at iteration $k$ of DFO-LS, $j = 0,1,\dots,p$. \\
  $\mathbf{\Omega}^k$ & Candidate pool in the parameter space $\R^n$ at iteration $k$, \\
                      & $\mathbf{\Omega}^k = \lbrace \omega^k_0, \omega^k_1, \dots , \omega^k_p \rbrace$. \\
  $f_\text{Misfit}(\omega)$ & Misfit function minimized by DFO-LS, $f_\text{Misfit}: \R^n \rightarrow \R_+$. \\
  $r_\text{Misfit}(\omega)$ & Misfit residual vector for $f_\text{Misfit}(\omega)$,  $r_\text{Misfit}: \R^n \rightarrow \R^m$. \\
  $m$ & Dimension of the misfit residual vector $r_\text{Misfit}(\omega)$. \\
  PEC & Paranaguá Estuarine Complex. \\
  SMS & Source-minus-sink term describing net biogeochemical sources and sinks. \\
  NP & Nutrient-phytoplankton interactions. \\
  TMM & Transport Matrix Method. \\
  KEBM & Knudsen Estuarine Box Model. \\
  ZMT & Zone of Maximum Turbidity. \\
  ODE & Ordinary Differential Equation. \\
  PDE & Partial Differential Equation. \\
  DFO & Derivative-Free Optimization. \\
  DFO-LS & Derivative-Free Optimization for Least Squares. \\
  CMA-ES & Covariance Matrix Adaptation Evolution Strategy. \\
  MOPS & Model of Oceanic Pelagic Stoichiometry. \\
\end{longtable}


\newpage

\cleardoublepage
\thispagestyle{empty}
\centerline{\Large\bfseries CONTENTS}

\vspace{1cm}

\begingroup
\makeatletter
\@starttoc{toc} 
\makeatother
\endgroup

\textual
\chapter{Introduction}

This thesis concerns the development of conceptual mathematical models in marine biogeochemistry, and the practical use and implementation of state-of-the-art derivative-free optimization techniques to calibrate the parameters in the model. Specifically, the goal of the thesis is to develop a bespoke, data-constrained conceptual model for the nutrient-phytoplankton dynamics of the Paranagu\'a Estuarine Complex (PEC), which involves deriving governing equations for the biogeochemistry, and to use in-situ observations from phytoplankton and nitrate (as the nutrient) to calibrate the model using a systematic optimization framework.

Marine biogeochemistry is an interdisciplinary branch of science that investigates how biological, chemical, and geophysical processes regulate the transformation and transport of biological and chemical components in the ocean, such as planktonic organisms and carbon. Marine biogeochemical models, in turn, translate these interactions into mathematical equations, which can then be solved (usually with the help of a computer), enabling the simulation of scenarios, hypothesis testing, and informed decision-making \cite{sarmientogruber,carbonbiogeochem,modelling,Gruber2019}. These models range from simple "box" models, such as the ones presented in this thesis, to highly complex four-dimensional models representing in detail the marine biogeochemical cycles at global scales \cite{phytoBIB,kriestoschlies2010,Kriest2012,mops,principal,onesize,TMM2007,TMMaplication,khatiwala2008,Wilson2022,Francisco2022,Francisco2024,websiteMITgcm,CoupledModelHarris,Kwiatkowski2014}. 

The applications of marine biogeochemical models are also diverse and include estimating water quality, carbon and nutrient cycles, risk of harmful algal blooms, and port management \cite{modelling,Gruber2019,carbonbiogeochem,phytoBIB,Verri2020,Davidson1999,Canuel2016}. In particular, these models are also widely used to assess the potential effects of anthropogenic climate change and pollution on the ocean ecosystem, informing public policy and management \cite{Wilson2022,Francisco2022,Francisco2024,principal,onesize,mops,kriestoschlies2010}. The latter is particularly relevant to PEC, where the anthropogenic discharge of nutrients into PEC from the surrounding urban and port environments leads to the bloom of toxic algae, which impacts the local ecosystem and compromises the quality of water \cite{marines2023,Machado1997ParanaguaBay,Martins2010,Martins2015,Mizerkowski2012,Brandini1985,Procopiak}.

Indeed, the Paranaguá Estuarine Complex has long been recognized as a system highly sensitive to nutrient enrichment, mainly because rivers, urban effluents, port activities, and hydrodynamic modifications intensify the input and retention of nitrogen- and phosphorus-rich waters \cite{Machado1997ParanaguaBay,Marone2005,Mizerkowski2012,Martins2010,Martins2015}.  Several studies in the region have documented elevated nutrient concentrations, alterations in phytoplankton structure, and recurrent episodes of harmful algal blooms \cite{Brandini1985,Procopiak,Brandini2022}. However, most of this knowledge is based on in situ measurements collected at specific times and locations. While these observations are essential, they often provide a fragmented view of a dynamic, highly variable system strongly modulated by tides, circulation, and estuarine mixing. In this context, marine biogeochemical models allow researchers to test hypotheses and to explore the consequences of different management interventions before they occur in the real system \cite{modelling,Gruber2019,carbonbiogeochem,Verri2020}. This thesis is motivated precisely by the need to fill the gap between observational studies and predictive understanding through the development of a simple conceptual model, designed to represent the local biogeochemical dynamics of the PEC in a realistic, interpretable, and computationally efficient manner \cite{kriestoschlies2010,Kriest2012,simplencomplex}. Such a framework provides a basis for future scenario analyses and can ultimately support evidence-based environmental policies and estuarine management.

Through this approach, this thesis emphasizes data-fitting and optimized modeling, seeking reproducible procedures for calibrating parameters and quantifying uncertainties under settings relevant to PEC. Before proceeding, we will briefly expand on the two main elements of this thesis: biogeochemistry and modeling.

\section*{Biogeochemistry for mathematicians}

Before exploring the concept of a biogeochemical cycle, we define tracers and compartments, using examples that are likely familiar to most readers. Some of the most well-known chemical element cycles are the water cycle, the phosphate cycle, and the carbon cycle. In any biogeochemical cycle, we refer to the elements being tracked, in their various forms and locations, as \textit{tracers}. Thus, in the examples mentioned, the tracers could be identified as H$_2$O, PO$_4$, and CO$_2$, respectively. A biogeochemical cycle typically considers a set of tracers, some of which have biological origins. For example, phytoplankton populations, which are algae and therefore realize photosynthesis, can be tracked by the amount of Chlorophyll-A present on the surface of a marine region. In this case, we could define Chlorophyll-A as the tracer, although it would be equivalent to considering the phytoplankton itself as a tracer.

On the other hand, we will now discuss \textit{compartments}. Returning to the example of the water cycle, as it passes through different states of matter, we can affirm that water also travels through different compartments, which is the definition of \textit{cycling}. In its liquid state, it is present in bodies of water, such as rivers and the ocean. In its gaseous state, it is present in the form of clouds in the Earth's atmosphere. We could then consider the ocean as one large compartment, and the atmosphere as another large compartment. Thus, a possible model for the global water cycle would involve the transport of this tracer between these two compartments.

The complexity of a biogeochemical cycle increases even further when multiple tracers and compartments are taken into account. Not only are tracers transported between compartments (geophysical part), but they also interact with one another (biochemical part). 
Due to their complex nature, it is common for mathematical models representing marine biogeochemical cycles to be obtained by coupling a biochemical conceptual model, described in diagrammatic form, with a physical transport model that represents advection and diffusion processes, described as a Dirichlet boundary problem.
The resulting mathematical model is a set of ODEs (or PDEs, depending on the approach), which rely on a set of variables and parameters. More details on the process of mathematical modeling of a marine-biogeochemical cycle are presented in \textit{Chapter \ref{chapterModelref}}. Additionally, the values to be set in each category are not static and can vary depending on the purpose for which the model is being used. This will be explored in the optimization framework steps of validation and calibration, in \textit{Chapter \ref{chapterframework}}.

\section*{Mathematical modeling for biogeochemists}

In our context, modeling refers to the mathematical representation of the interactions between elements of a biogeochemical cycle. It's virtually impossible to represent everything that occurs in an environment, just as it's unrealistic to expect results without any precision errors. Thus, creating a mathematical model involves identifying the most relevant elements of a biogeochemical cycle and quantifying the primary interactions between them.

The biochemical part of quantifying interactions between elements is generally performed based on empirical estimates obtained through laboratory experiments, while the geophysical part usually derives from a computational mathematical model that uses input data, such as ocean current flows in a few locations, to produce accurate estimates about the current flow at each point in an entire region. Even so, it is also common to use theoretical hypotheses in modeling as a way of obtaining more information from less data.

The strategy of assuming hypotheses is particularly useful when data is scarce. In some cases, data scarcity occurs due to difficult access to sampling regions. However, it is also common for the motivation behind sampling data to be using it as input in a mathematical model. In this case, the development of a mathematical model can be carried out prior to data sampling, without this factor being an impediment. In this context, in \textit{Chapter \ref{chapterPEC}}, we present an application on the biogeochemical mathematical modeling of a marine region.

Once a biogeochemical cycle has been translated into a mathematical model, this model is rewritten in the form of a computer program and processed computationally. The way in which model and data are processed by the computer is not unique and can indeed be refined. \textit{Chapter \ref{chapterCalibratingPEC}} presents an application on how to apply the framework presented in \textit{Chapter \ref{chapterframework}} for refining a marine biogeochemical model using optimization and the data available.

In the following, we provide a concise overview of the thesis's context. Some technical terms may be unfamiliar to the reader, depending on their background in research. Nevertheless, more details about each of the presented concepts will be explored in the subsequent chapters.

\section*{Motivation, Objectives, and Contributions}

Biogeochemical modeling spans from conceptual biochemical models, aimed at addressing mechanistic
questions and hypothesis testing, to high-fidelity three-dimensional configurations in which conceptual
biochemistry is coupled to complex hydrodynamics.
Typical marine biogeochemical models include dozens of parameters. Additionally, the high dimensionality of the parameter space induces correlations and non-identifiability, as different parameter combinations can fit observations equally well (equifinality), as documented across simple and 3-D calibrated models \cite{kriestoschlies2010,Kriest2012, principal, onesize}. 
Many parameters are not directly observable at the rates and scales required for operational applications. Time series may be short, sparse, and heterogeneous; laboratory estimates do not always translate to natural settings. 

Particularly, in coastal and estuarine environments, strong spatial–temporal variability and anthropogenic forcings necessitate careful scoping and process selection to ensure the model remains parsimonious and computationally tractable, producing reliable and reproducible estimates \cite{Marone2005,Martins2010}.

Due to the level of complexity in determining parameter values, it has historically been common for parameters to be inferred theoretically, based on information obtained from previous studies. However, recently an approach has been presented that estimates a set of parameters indirectly through data fitting \cite{Kriest2012}.

While manual trial-and-error tuning remains common, it is hard to reproduce, is dependent on the modeler's experience, and tends to explore only a small fraction of the parameter space - the latter being a crucial issue in the context of climate modeling \cite{Francisco2025}. On the other hand, systematic calibration specifies an objective function (model–data misfit), defines search rules (global and local), and documents choices (weights, normalizations, constraints), enabling comparability across scenarios and methodological review \cite{Kriest2012, DFOGN, DFOLSmops2022}.

When dealing with complex models, a computational approach for calibrating parameters can be costly and may exhibit numerical noise, as well as unreliable or unavailable gradients. Beyond computational expense, multiple local minima and process–resolution interactions affect fitting. These realities motivate cost-reduction strategies (for example, transport-matrix methods and accelerated spin-up) and derivative-free optimization tailored to least-squares structure \cite{TMM2007,khatiwala2008, TMMaplication, DFOGN, DFOLSmops2022}.
In this context, studying parameter calibration in a simple model provides a controlled laboratory environment for addressing targeted questions, generating mechanistic insights, and facilitating rapid experimentation at a low cost. They enable broad parameter-space exploration, structural testing, and intuition-building that can inform more complex set-ups \cite{Verri2020, Lalli1997}. In estuaries and ports -- where interest often centers on specific processes and seasonal/anthropogenic responses -- such models are particularly useful \cite{Marone2005, Martins2010, Zem2008}.

At lower computational cost, we combine heuristic global search (to locate promising initial points) with local derivative-free refinement, such as DFO-LS, a Gauss–Newton-type method for nonlinear least squares minimization \cite{DFOGN, DFOLSoriginal, DFOLSmops2022}. This arrangement facilitates studies of identifiability, noise robustness, and objective-function design, and supports post-fit sensitivity analysis \cite{Kriest2012}. In this sense, the guiding questions are: (i) Which combination of objective function, search strategy, and constraints yields reproducible and interpretable fits? (ii) Which parameters (or parameter groups) are identifiable with available data, and under what conditions? (iii) How do observational uncertainties and normalization/weighting choices affect fits and process inference? This focus aligns with prior PEC studies on circulation and biogeochemistry \cite{Marone2005, Mizerkowski2012, Souza2015, Martins2010}.

The central purpose of this thesis is to connect enthusiasts from different research areas, such as mathematics and oceanography, by providing a step-by-step guide on how to develop a systematic, data-informed parameter-calibration methodology and apply it to build a conceptual, data-constrained model of marine biogeochemical dynamics. 

The scope centers on (a) conceptual models applied to PEC; (b) multi-parameter calibration using tracer data; (c) derivative-free strategies suitable for potentially noisy objectives; and (d) cross-validation and sensitivity protocols.

The outcome is an interdisciplinary and practice-oriented study that bridges different scientific languages, making the content approachable to readers from multiple fields and pointing to several possible extensions of this work.

\newpage
\section*{Structure of this Thesis}
 This thesis is organized to move from modeling basics to optimization methods and, finally, to the application of these principles to develop data-constrained model for the PEC.

\vspace{0.5cm}
\textbf{Chapter \ref{chapterModelref}} reviews marine biogeochemical modeling at a conceptual level with illustrative experiments.

\vspace{0.5cm}
\textbf{Chapter \ref{chapterPEC}} develops a conceptual nutrient-phytoplankton model tailored to the PEC. It defines the domain, derives biogeochemical equations and simple circulation components, couples them, and implements the model numerically; results and limitations are discussed.

\vspace{0.5cm}
\textbf{Chapter \ref{chapterdfols}} introduces the optimization background used throughout the thesis, formulates parameter calibration as a (weighted) least-squares problem from tracer observations, and explores the theory and computational usage of a derivative-free strategy -- DFO-LS -- motivating its use when gradients are unreliable or simulations are noisy/expensive.

\vspace{0.5cm}
\textbf{Chapter \ref{chapterframework}} presents a systematic calibration workflow for conceptual models. It details the construction and minimization of a quadratic regression model within a trust-region framework, and demonstrates the pipeline (problem set-up, heuristic search, and local refinement) through calibration tests and case studies, with results and diagnostics.  

\vspace{0.5cm}
\textbf{Chapter \ref{chapterCalibratingPEC}} applies the calibration framework to PEC in a data-constrained setting, including a twin benchmark, a fit to one-year nitrate and phytoplankton observations, and a scenario simulation with increased riverine nitrate load.

\vspace{0.5cm}
Lastly, we summarize the main contributions of this thesis and outline potential directions for future research.

\chapter{Modeling marine  biogeochemical cycles}\label{chapterModelref}

Biogeochemical mathematical models represent the processes that govern the cycling of biological and chemical elements in an environment. They are widely used to study environmental issues such as climate change and pollution \citep{modelling,sarmientogruber,carbonbiogeochem}, serving as a useful tool for gaining greater insights into the human impact on nature and, thus, facilitating the conscious planning of these interactions. Since the 1960s, the advent of computers has enabled the development of the first computational biogeochemical models. Given the limited knowledge and computing power of the time, these were relatively simple conceptual models that simulated only a few processes \citep{modelling}. Over time, both scientific understanding and computational capacity increased, enabling the development of more complex biogeochemical models. In this chapter, we will review the foundations of biogeochemical modeling, exemplified by simple conceptual models. This basic understanding is fundamental and generalizable for the study and development of more complex models, which may be of interest to the reader in the future \citep{kriestoschlies2010, Kriest2012}.

In the following section, we will introduce some key elements mentioned in the introduction that are essential for creating a conceptual biogeochemical model.

\section{Tracers and parameters}

For context, let us imagine a situation where a swimming pool presents a leak. A technician is called, and he pours an amount of colored dye into the pool. Tracking the path taken by the dye allows the technician to estimate the location of the leak. Now, imagine someone collects a bucket of water from this pool after the technician has added the dye. Analyzing the color of the collected water allows us to infer the average concentration of the dye present in the water at the time of collection, at least in the part of the pool where the water was collected.    

In the context of marine biogeochemical modeling, a tracer is a chemical or biological component that exhibits behavior similar to that of a dye added to a swimming pool. The tracer dynamics is modeled using a geophysical model, which typically considers advection and diffusion, in addition to the environmental geography. Furthermore, the concentration of a tracer in a given region can be quantified through various sampling or identification methods (see Box \ref{box:GatheringData}). In addition to the behavior described, tracers also interact with each other. Mathematically, a computational marine-biogeochemical model predicts or estimates tracer concentrations in given marine regions over time.

\setlength{\fboxsep}{8pt}   
\setlength{\fboxrule}{0.4pt}

\vspace{0.6cm}
\begin{codebox}[label=box:GatheringData]{\textbf{Gathering tracer data: \\ How are phytoplankton concentrations estimated in marine regions?}}

Chlorophyll-A is the most abundant photosynthetic pigment in phytoplankton and is useful as an indirect indicator of these microorganisms' biomass \cite{Parsons1984}. Different sampling and analysis techniques are applied to estimate its concentration in ocean waters:
\begin{itemize}
\item In-situ collections:
A CTD (Conductivity, Temperature, and Depth) system is an oceanographic instrument used to measure the electrical conductivity, temperature, and depth of seawater throughout the water column. A Niskin bottle is a PVC or acrylic cylinder with end caps that remain open during descent and are closed remotely (by a mechanical messenger or electrical system) at the desired depth, ensuring that other water layers do not contaminate a water sample. Traditionally, researchers use oceanographic cylinders (such as Niskin bottles) attached to CTD systems to collect water samples at different depths. In the laboratory, Chlorophyll-A is extracted with solvents (usually acetone or methanol) and quantified by fluorometry or spectrophotometry (which are precise methods, but require time and careful handling \cite{UNESCO1994}).
\item Fluorometers and optical sensors:
Submersible equipment allows real-time measurements of fluorescence associated with Chlorophyll-A. These sensors are widely used in vertical profiles and in continuous monitoring systems on board ships, buoys, or gliders, providing high-resolution spatial and temporal data \cite{UNESCO1994}.
\item Remote sensing:
Satellites such as MODIS and Sentinel-3 estimate surface Chlorophyll-A concentrations based on light reflectance in the ocean, allowing for mapping of large ocean areas. Although less accurate locally, this is a crucial tool for collecting data in areas of difficult access \cite{IOCCG2014}.
\end{itemize}
Each of these methods has its own advantages and limitations: while laboratory analyses are more accurate, sensors and satellites expand the range of observations. Together, they determine the baseline data on phytoplankton concentrations in different marine regions.
\end{codebox}

\vspace{0.5cm}
On the other hand, parameters are model features that usually cannot be measured directly. They are inferred values based on models and theoretical estimates that may be constant, depend on the tracer's features, or depend on interactions between tracers. For example, the phytoplankton growth rate is a parameter that varies among species and environmental conditions; however, in most models, it is represented by an average laboratory-derived estimate rather than a directly measurable field value.

A conceptual model is a diagram that represents the main components of the biogeochemical system and how they interact with each other \citep{ modelling, sarmientogruber, Gruber2019, phytoBIB}. Mathematically, it can be translated into a set of differential equations representing changes in tracers' concentrations depending on time and interactions between them \citep{modelling,sarmientogruber, Gruber2019}.

The complexity of a conceptual model implies a necessary level of complexity in the mathematical model obtained from it. In this sense, conceptual models that consider a simplified domain, with few processes and parameters, often can be translated to simple mathematical models, with low computational costs \citep{kriestoschlies2010, principal,onesize, simplencomplex, mops}. To illustrate a simple model, we present the following example.

\vspace{0.5cm}

\begin{Ex}[Conceptual model for nitrate-phytoplankton dynamics]\label{exonebox01}

Consider a well-mixed, homogeneous oceanic region represented by a single box for the entire water column. The model tracks two tracers over time: nitrate and phytoplankton. Phytoplankton grow by assimilating nitrate during the photosynthesis process. When phytoplankton die, a fraction of their organic nitrogen is remineralized to dissolved inorganic nitrate within the box. In this example, we assume that $70\%$ of dead phytoplankton is remineralized locally (parameter $r=0.7$), while the remaining $30\%$ is lost from the modeled region (for example, export/sinking). A conceptual model for this situation is presented in Figure \ref{figConceptualModel}.

\vspace{0.7cm}
\begin{figure}[htpb!]
\caption{Diagram representing a conceptual biochemical NP-model within an aquatic environment.} 
\begin{center}
\tikzset{every picture/.style={line width=0.75pt}} 

\begin{tikzpicture}[x=0.75pt,y=0.75pt,yscale=-1,xscale=1]

\draw  [fill={rgb, 255:red, 110; green, 190; blue, 14 }  ,fill opacity=1 ] (64,47.47) .. controls (64,40.03) and (70.03,34) .. (77.47,34) -- (161.63,34) .. controls (169.07,34) and (175.1,40.03) .. (175.1,47.47) -- (175.1,87.9) .. controls (175.1,95.34) and (169.07,101.37) .. (161.63,101.37) -- (77.47,101.37) .. controls (70.03,101.37) and (64,95.34) .. (64,87.9) -- cycle ;

\draw  [fill={rgb, 255:red, 117; green, 175; blue, 239 }  ,fill opacity=1 ] (300,47.47) .. controls (300,40.03) and (306.03,34) .. (313.47,34) -- (397.63,34) .. controls (405.07,34) and (411.1,40.03) .. (411.1,47.47) -- (411.1,87.9) .. controls (411.1,95.34) and (405.07,101.37) .. (397.63,101.37) -- (313.47,101.37) .. controls (306.03,101.37) and (300,95.34) .. (300,87.9) -- cycle ;

\draw  [fill={rgb, 255:red, 248; green, 231; blue, 28 }  ,fill opacity=1 ] (108.1,130.9) -- (115.61,130.9) -- (115.61,107.95) -- (130.63,107.95) -- (130.63,130.9) -- (138.14,130.9) -- (123.12,146.19) -- cycle ;
\draw  [fill={rgb, 255:red, 248; green, 231; blue, 28 }  ,fill opacity=1 ] (206.1,54.05) -- (206.1,61.56) -- (288.1,61.56) -- (288.1,76.58) -- (206.1,76.58) -- (206.1,84.09) -- (184,69.07) -- cycle ;
\draw   (69,160.52) .. controls (69,155.81) and (72.81,152) .. (77.52,152) -- (167.48,152) .. controls (172.19,152) and (176,155.81) .. (176,160.52) -- (176,186.08) .. controls (176,190.79) and (172.19,194.6) .. (167.48,194.6) -- (77.52,194.6) .. controls (72.81,194.6) and (69,190.79) .. (69,186.08) -- cycle ;
\draw  [fill={rgb, 255:red, 248; green, 231; blue, 28 }  ,fill opacity=1 ] (109.63,226.84) -- (117.14,226.84) -- (117.14,203.9) -- (132.16,203.9) -- (132.16,226.84) -- (139.67,226.84) -- (124.65,242.14) -- cycle ;
\draw  [fill={rgb, 255:red, 248; green, 231; blue, 28 }  ,fill opacity=1 ] (188,170) -- (350.3,170) -- (350.3,133) -- (341,133) -- (357,108) -- (373,133) -- (363.7,133) -- (363.7,183.4) -- (188,183.4) -- cycle ;

\draw (72,59) node [anchor=north west][inner sep=0.75pt]   [align=left] {Phytoplankton};
\draw (335,59) node [anchor=north west][inner sep=0.75pt]   [align=left] {Nitrate};
\draw (191,16) node [anchor=north west][inner sep=0.75pt]   [align=left] {\begin{minipage}[lt]{67.35pt}\setlength\topsep{0pt}
Phytoplankton
\begin{center}
growth
\end{center}

\end{minipage}};
\draw (10,108) node [anchor=north west][inner sep=0.75pt]   [align=left] {\begin{minipage}[lt]{67.35pt}\setlength\topsep{0pt}
\begin{center}
Phytoplankton\\mortality
\end{center}

\end{minipage}};
\draw (74,157) node [anchor=north west][inner sep=0.75pt]   [align=left] {\begin{minipage}[lt]{67.9pt}\setlength\topsep{0pt}
\begin{center}
Dissolved \\organic matter
\end{center}

\end{minipage}};
\draw (134.16,206.9) node [anchor=north west][inner sep=0.75pt]   [align=left] {\begin{minipage}[lt]{36.16pt}\setlength\topsep{0pt}
\begin{center}
Sinking
\end{center}

\end{minipage}};
\draw (220,189) node [anchor=north west][inner sep=0.75pt]   [align=left] {\begin{minipage}[lt]{78.67pt}\setlength\topsep{0pt}
\begin{center}
Remineralization
\end{center}

\end{minipage}};

\end{tikzpicture}    
\end{center}
\label{figConceptualModel}

\smallskip
The mean concentrations of the tracers in the environment are represented by colored boxes, namely nitrate and phytoplankton. The arrows represent interactions between model components that alter the concentration of each tracer. Arrows pointing toward a box indicate processes that increase the concentration of that component, while arrows leaving a box indicate processes that decrease it. Dissolved organic matter is not represented as a tracer but as a conceptual component that explains a process modeled implicitly. 

\smallskip
\textbf{Source:} the author (2026).
\end{figure}
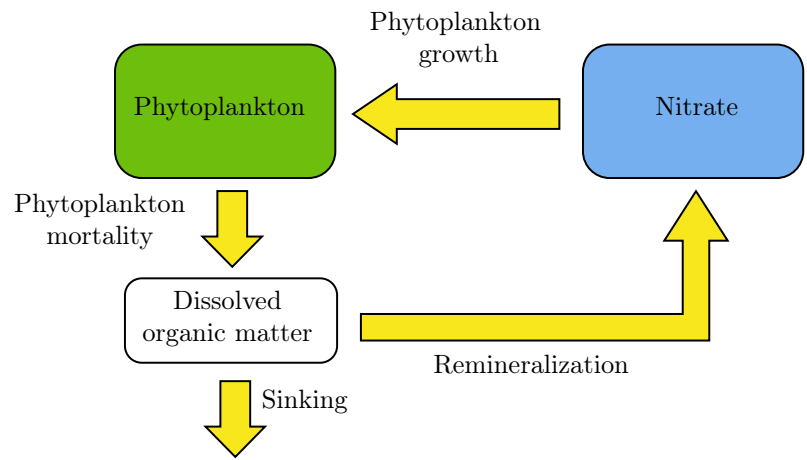
\end{Ex}

\newpage
In Example \ref{exonebox01}, the conceptual model presented is treated as a closed aquarium containing a set of interacting tracers, so the physical processes are dealt implicitly: all elements within the modeled domain are assumed to interact instantaneously. Due to this simplification, the geophysical framework can be omitted from the model representation, which can then be summarized as a diagram of biochemical interactions between tracers. The system represented in Example \ref{exonebox01} corresponds to a simple NP (Nutrient-Phytoplankton) model
(see Box~\ref{box:NPnNPZ} for a broader discussion of NP and NPZ models).

\vspace{0.5cm}
\begin{codebox}[label=box:NPnNPZ]{\textbf{NP and NPZ models: \\ Simplified Representations of Marine Ecosystem Dynamics}}

The NP (Nutrient-Phytoplankton) and NPZ (Nutrient-Phytoplankton-Zooplankton) models are among the most widely used simplified representations of marine biogeochemical interactions. Nutrients are chemical substances that living organisms need to grow and survive. In the ocean, nutrients feed phytoplankton - microscopic algae that float near the surface and use sunlight to produce energy through photosynthesis. Although tiny, phytoplankton are the base of the marine food web and play a crucial role in producing oxygen and absorbing CO$_2$ from the atmosphere.

An NP model is a system comprising two tracers: a main nutrient (N), which, for example, can be set as nitrate or phosphorus, depending on data availability, and phytoplankton (P), a simplified representation of an entire population of algae containing different species. Phytoplankton consume the nutrients in a process known as primary production, which depends on light, temperature, and the availability of nutrients. The loss of phytoplankton - through respiration, mortality, or aggregation - returns nutrients to the inorganic pool. This minimal structure captures the basic feedback between nutrient limitation and phytoplankton growth.

On the other hand, an NPZ model additionally considers zooplankton (Z) as a tracer, explicitly representing the grazing and predation of phytoplankton. Phytoplankton serve as food for zooplankton, whose metabolic losses (excretion, mortality, and sloppy feeding) regenerate nutrients and particulate organic matter. By considering zooplankton as a tracer, nonlinear trophic interactions, such as predator-prey cycles, are introduced.

These models strike a balance between simplicity and ecological realism. While an NP model helps explore nutrient limitation and productivity under equilibrium conditions, an NPZ model provides a more dynamic and detailed view of ecosystem regulation and energy transfer within the planktonic food web \cite{sarmientogruber}.
\end{codebox}

\vspace{0.5cm}
In the following section, we explore details about the mathematical representation obtained from the diagrammatic form of a conceptual model.

\section{Source Minus Sink equations}

In biogeochemical modeling, a \textit{source} refers to any process that increases the concentration of a tracer, whereas a \textit{sink} refers to any process that decreases it. For example, in the diagram depicted in Figure \ref{figConceptualModel}, each arrow entering the box representing nitrate tracer concentrations represents a source of this tracer, while each arrow exiting this box represents a sink. The source-minus-sink (SMS) equation of a tracer quantifies the aftermath concentrations of the tracer as the source and sink processes occur. Below, we discuss how to obtain SMS equations for the conceptual model presented in diagrammatic form on Example \ref{exonebox01}.

\vspace{0.5cm}
\begin{Ex}\label{exSMSmodel}
The conceptual model represented in Figure \ref{figConceptualModel} can be mathematically translated by the following ODE equations:
\begin{align}
\dfrac{d C_\mathrm{N}}{d t} = -V_{\text{max}} \cdot \dfrac{C_\mathrm{N}}{K + C_\mathrm{N}} \ C_\mathrm{PHY}   + r \cdot \lambda \cdot C_\mathrm{PHY} \label{eqN01}\\
\dfrac{d C_\mathrm{PHY}}{d t} = V_{\text{max}} \cdot \dfrac{C_\mathrm{N}}{K + C_\mathrm{N}} \ C_\mathrm{PHY} - \lambda \cdot C_\mathrm{PHY} \label{eqP01}    
\end{align}

In the system of equations  (\ref{eqN01})-(\ref{eqP01}), the nitrate concentration decreases due to phytoplankton uptake, $V_{\max}\tfrac{C_\mathrm{N}}{K+C_\mathrm{N}}\ C_\mathrm{PHY}$, and increases through remineralization of a fraction $r$ of phytoplankton mortality, that is, $r\ \lambda \ C_\mathrm{PHY}$. Phytoplankton increases via Monod-limited growth, $V_{\max}\dfrac{C_\mathrm{N}}{K+C_\mathrm{N}} \ C_\mathrm{PHY}$, and decreases due to mortality, $\lambda \ C_\mathrm{PHY}$. In these equations, we considered the following notations:

\vspace{0.3cm}
\textbf{Tracers:}
\begin{itemize}
\item $C_\mathrm{N}$: nitrate concentration (mmol m$^{-3}$);
\item $C_\mathrm{PHY}$: phytoplankton nitrogen concentration (mmol m$^{-3}$);
\end{itemize}

\vspace{0.2cm}
\textbf{Parameters:}
\begin{itemize}
\item $V_{\text{max}}$: maximum specific uptake/growth rate (day$^{-1}$);
\item $K$: half-saturation constant for nitrate uptake (mmol m$^{-3}$);
\item $\lambda$: phytoplankton mortality rate (day$^{-1}$);
\item $r$: remineralization fraction returning mortality to nitrate (dimensionless, $0\leq r \leq 1$).
\end{itemize}

\end{Ex}

\vspace{0.5cm}
The right-hand sides of Eqs.~\eqref{eqN01}-\eqref{eqP01} are written in the \emph{source-minus-sink} (SMS) form: each tracer’s net tendency equals the sum of processes that increase it (sources) minus the sum of processes that decrease it (sinks).

For the nitrate concentration, $C_{\mathrm{N}}$, we have
$$\frac{d C_{\mathrm{N}}}{dt}
=\underbrace{r \ \lambda \ C_{\mathrm{PHY}}}_{\text{source: remineralization}}
\;-\;
\underbrace{V_{\max} \ \frac{C_{\mathrm{N}}}{K+C_{\mathrm{N}}} \ C_{\mathrm{PHY}}}_{\text{sink: uptake by phytoplankton}} ~.$$
Thus, nitrate increases when remineralization ($r \ \lambda \ C_{\mathrm{PHY}}$) exceeds biological uptake, and decreases otherwise.

For the phytoplankton concentration, $C_{\mathrm{PHY}}$, we have
$$\frac{d C_{\mathrm{PHY}}}{dt}
=\underbrace{V_{\max} \ \frac{C_{\mathrm{N}}}{K+C_{\mathrm{N}}} \ C_{\mathrm{PHY}}}_{\text{source: Monod-limited growth}}
\;-\;
\underbrace{\lambda \ C_{\mathrm{PHY}}}_{\text{sink: mortality}} ~. $$
Hence, phytoplankton biomass increases when the specific growth rate $V_{\max}\tfrac{C_{\mathrm{N}}}{K+C_{\mathrm{N}}}$ exceeds the mortality rate $\lambda$; if mortality dominates, biomass declines. In both equations, each term has units of mmol m$^{-3}$ day$^{-1}$, so positive SMS implies tracer accumulation and negative SMS implies tracer loss.
We can obtain the prediction for the concentrations of the tracers from the model by solving a Cauchy problem, as illustrated in the following example.

\begin{Ex}\label{exonebox02}
Recalling the mathematical model presented in Example \ref{exSMSmodel}, we now attribute initial values to the tracer concentrations and treat it as a Cauchy problem: $C_\mathrm{N}(0)=80  \ \text{mmol \ m}^{-3}$, $C_\mathrm{PHY}(0)=50  \ \text{mmol \ m}^{-3}$.
If we additionally set the parameter values mentioned above, for example:
 $V_{\text{max}}= 1.4 \ \text{day}^{-1}, ~ K= 0.1 \ \text{mmol \ m}^{-3}; $ $r= 0.7 , \lambda= 0.05  \ \text{day}^{-1}$, the ODE system (\ref{eqN01})-(\ref{eqP01}) can then be solved numerically to simulate the evolution of the concentrations of nitrate and phytoplankton inside the box over time (Figure \ref{1caixaphyNO3grafico}). For more details on the computational experiments, see the \textit{Appendix}.


\begin{figure}[htpb!]
\caption{Evolution of nitrate (upper plot) and phytoplankton (lower plot) concentrations over time for Example~\ref{exonebox02}.}
\begin{center}
\includegraphics[scale = 0.8]{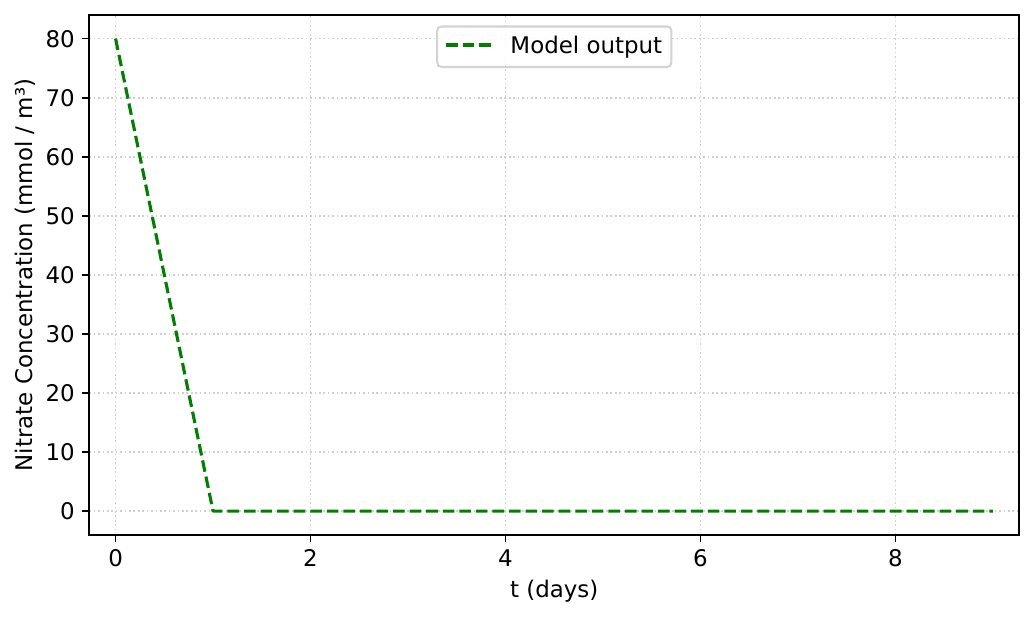}

\vspace{0.3cm}

\includegraphics[scale = 0.8]{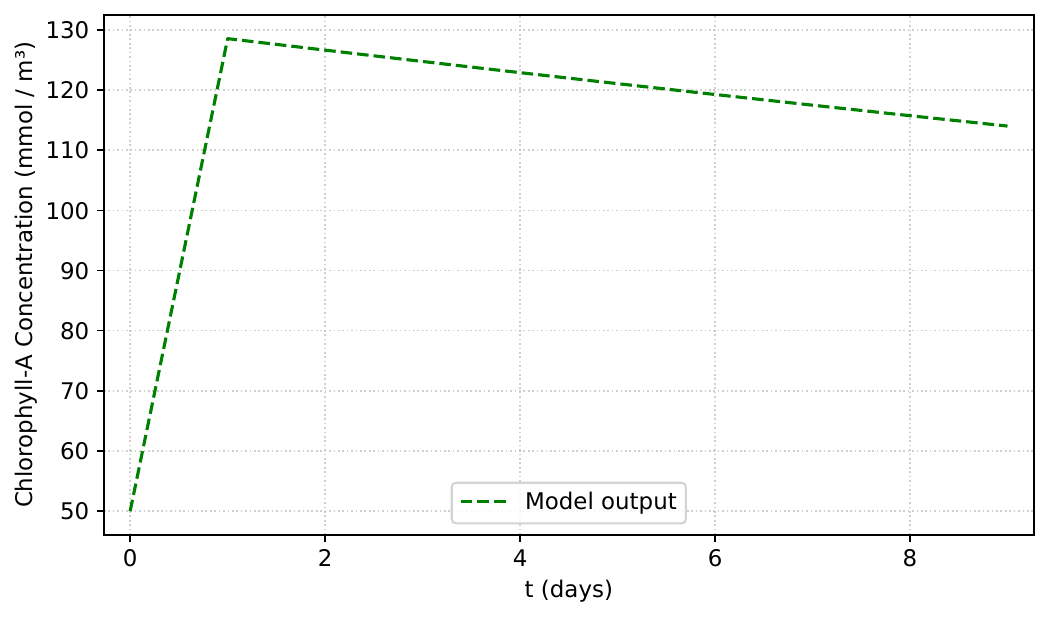}    
\end{center}
\label{1caixaphyNO3grafico}

\smallskip
\textbf{Source:} the author (2026).
\end{figure}    
\end{Ex}

\section{~Increasing ~model ~complexity: how to ~couple a ~biochemical ~conceptual ~model to a ~geophysical box ~model}

The level of complexity required for a biogeochemical model depends on the research question and the available data. While simple models are useful for exploring basic concepts and relationships, making predictions for relatively short time periods, and as a teaching tool, more complex models are necessary for simulating complex biogeochemical cycles, making long-term predictions, and studying the impacts of environmental change on ecosystem processes. A basic increase in the complexity of the model presented in Example \ref{exonebox01} involves adding physical processes while considering the domain divided into two boxes, rather than a single homogeneous box. The idea behind doing so is to achieve greater accuracy in describing the distinct features of different regions and how these features interact with one another.

A box model of a marine region is a representation obtained by either directly discretizing or conceptually depicting the region as a finite set of three-dimensional boxes, where each box is modeled as being uniformly mixed, meaning that properties such as temperature, salinity, and tracer concentration are assumed to be uniform throughout the box volume. These simplified assumptions are considered for the practical implementation of the model, where the exchange between boxes represents the circulation and mixing processes that occur in the real environment. In a biogeochemical box model, the tracers within each box have their own sources and sinks, which can arise from both biochemical interactions and geophysical circulation.

Following the classical derivation of the tracer conservation equation for a fixed control volume \cite[s.~19--23]{sarmientogruber}, we represent transport as the sum of advection and diffusion and group all in-box biogeochemical and external processes into a source–minus–sink (SMS) biochemical term. In a box-discretized aquatic model, advection corresponds to fluxes across box faces driven by the resolved velocity field, whereas diffusion represents down-gradient fluxes parameterized by an effective diffusivity. The SMS term accounts for local production and loss (for example, biological uptake, remineralization, external inputs), ensuring mass conservation at the box scale. For the discrete/matrix treatment of tracer transport in such models, see also \cite{TMM2007}.

In general, once we have an aquatic model where the space is discretized into boxes,  advection and diffusion are the two fundamental processes that describe how substances move and spread in the fluid contained in each box. Advection is the process by which a substance is transported by the flow of a fluid, as occurs when the current of a river carries materials along its course (Figure \ref{adveixox}). Advection is influenced by the speed and direction of fluid flow. Diffusion is the process by which a substance spreads from a region of high concentration to a region of low concentration, due to the random movement of particles (Figure \ref{adveixox}). Even in the absence of macroscopic fluid flow, the individual particles of the substance are always in thermal motion, resulting in their gradual dispersion. Diffusion is influenced by the concentration of the substance, the diffusion coefficient (which determines how quickly dispersion occurs), and the concentration gradient. In summary, advection refers to the transport of the substance by fluid flow, while diffusion refers to the dispersion of the substance due to the random motion of particles. The tracer conservation equation for the volume of fluid contained in a fixed box of the discretized domain has the form:
\begin{equation}\label{generaltransporteq}    
\dfrac{\partial C}{\partial t} = \dfrac{\partial C}{\partial t} \bigg\rvert_{advection}+\dfrac{\partial C}{\partial t} \bigg\rvert_{diffusion}+SMS(C) ~,
\end{equation}
where $C$ is the concentration of such tracer, the first two terms of the sum are the variation of $C$ due to advection and diffusion, and the last term is a source-minus-sink function of $C$ within the box.

\vspace{0.3cm}
\begin{Ex}\label{exampleadv}
    Consider the advection scheme represented in Figure \ref{adveixox}. When the boxes' dimensions are large enough, it is common to ignore the diffusion contributions. In this illustration, the region modeled, that is, the middle box, will have both a source and a sink contribution due to advective flux.
\end{Ex}

\begin{figure}[ht]
\caption{Illustrative representation of the advection and diffusion processes.}
\begin{center}
\input{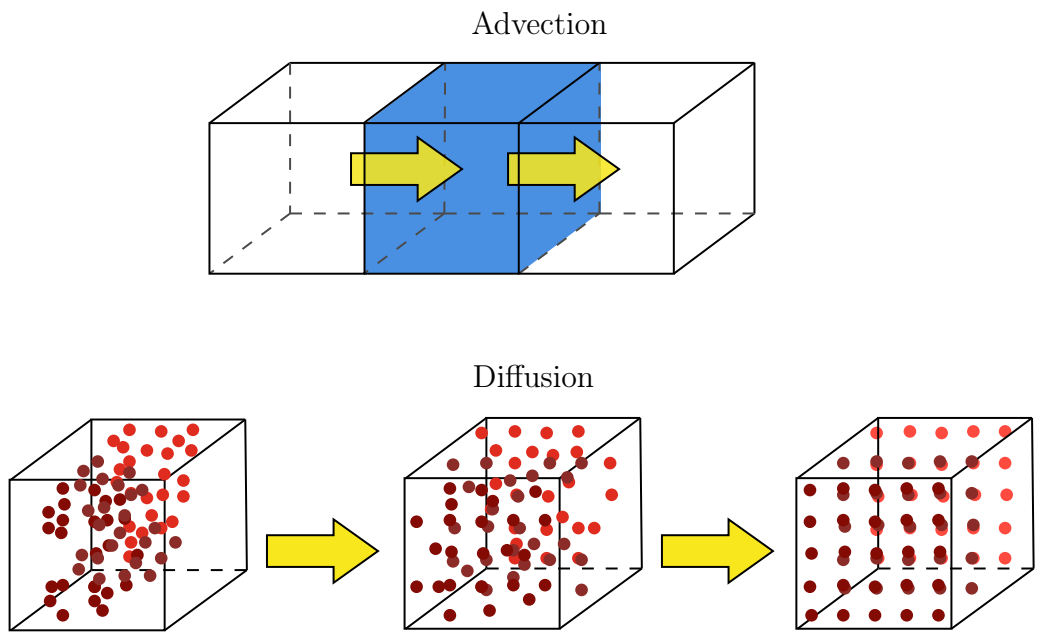}    
\end{center}     
\label{adveixox}

\smallskip
In the upper scheme, advective flux occurs in the positive x-direction; the blue region represents the modeled domain, while the arrows point in the flux direction. In the lower scheme, diffusion occurs over time; dots represent the molecules of a chemical tracer, while the arrows indicate their spatial arrangement when time passes.

\smallskip
\textbf{Source:} the author (2026).
\end{figure}

\vspace{0.2cm}
In a box model, while physical processes are incorporated within the sources and sinks of each tracer, the geographical features of the environment are captured through the boxes' configuration. For example, a complex coastal region can be represented by multiple interconnected boxes, each corresponding to different zones such as the open ocean, continental shelf, and estuarine areas. Or, when more computational processing is available, even small regions can be directly discretized into boxes.

A concrete example of a global biogeochemical box model that combines source-minus-sink formulations with physical transport is given by the Model of Oceanic Pelagic Stoichiometry (MOPS), briefly described in Box \ref{box:MOPS}.

\vspace{0.5cm}
\begin{codebox}[label=box:MOPS]{\textbf{MOPS: a global laboratory for marine biogeochemistry}}

The Model of Oceanic Pelagic Stoichiometry (MOPS) \cite{DFOLSmops2022, mops} is a marine biogeochemical model designed to represent the interactions between specific tracers in the global ocean. It represents the coupled carbon, nitrogen, phosphorus, and oxygen cycles within the ocean interior, combining physical transport with local biogeochemical transformations. \\
The dynamics of each tracer are governed by a system of ODEs representing the SMS equations of tracers, derived from both biochemical and physical processes. Each of these equations is written in the form:
$$
\dfrac{\partial C_i}{\partial t} =
\underbrace{-\ \nabla \cdot (uC_i) + \nabla \cdot (\varepsilon_i \nabla C_i)}_{\text{Transport via TMM}} +
\underbrace{\mathrm{SMS}_i(C_1, \ldots, C_\ell, x, y, z, t)}_{\text{Local biogeochemical processes}},
$$
where $ C_i $ is the concentration of tracer $ i $ (for example, nitrate, phosphate, oxygen, organic matter), $ u $ denotes the velocity field, $ \varepsilon_i $ the effective diffusivity, and $ \mathrm{SMS}_i $ the nonlinear source-minus-sink function representing local biological and chemical interactions.

In its reference configuration (MOPS-1.0) \cite{mops}, the model includes six adjustable parameters, along with several fixed ones, including maximum growth rates, half-saturation constants, and stoichiometric ratios. Once coupled to a general circulation model for the ocean dynamics, the nonlinearity and high dimensionality of the resulting discretized system make MOPS an ideal benchmark for automated calibration and optimization techniques, such as DFO-LS or CMA-ES, particularly under noisy or computationally expensive conditions.

MOPS is implemented using the use of the Transport Matrix Method (TMM) \cite{TMM2007}, which precomputes the global ocean’s advection-diffusion operator into a sparse matrix. This approach decouples the biogeochemical model from the full dynamical ocean circulation model, reducing computational cost by orders of magnitude and enabling large-scale sensitivity and parameter studies.

Overall, MOPS represents a computational laboratory for global marine biogeochemistry, combining physical realism, stoichiometric parsimony, and numerical efficiency, which is a cornerstone framework for exploring parameter calibration and large-scale carbon-nutrient dynamics. As such, it has been used in numerous applications and studies \cite{BURexperiment,keller2016,Niemeyer2019,Francisco2022, Francisco2024}.

\end{codebox}

\section{Developing a model in practice}
A realistic model can be simple or complex, depending on the goals of its representation. For example, although the values assumed for the initial concentrations and constants at Example \ref{exonebox01} were arbitrary, for this model to represent a real environment, that is, be a realistic model, it would need to be initialized with measured or estimated initial values of $C_\mathrm{N}$ and $C_\mathrm{PHY}$, and the model parameters $V{\text{max}}$, $K$, $R$, and $\lambda$ would need to be estimated based on data or literature values. This task is known as the \textbf{calibration} or \textbf{tuning} of the parameters of the model, and will be explained in detail in Chapters 3 and 4.

Developing a biogeochemical model tailored for real-world applications is generally a difficult and long process, which is complicated by the fact that the inclusion of processes in such models is not an objective process \cite{Martin2024}. In practice, such development involves several steps:
\begin{itemize}
    \item Define the system: The first step is to define the system being modeled. This includes specifying the spatial and temporal scales, the components and processes being modeled, and the specific goals of the model.

    \item Gather data: The next step is to gather data on the system being modeled. This data can include observations of environmental variables, measurements of chemical concentrations, and data on the behavior of individual processes.
    
    \item Develop conceptual model: Based on the data and knowledge of the system, a conceptual model is developed that describes the main processes and interactions between the components of the system.

    \item Formulate mathematical model: The conceptual model is then translated into a mathematical representation. This involves developing equations that describe the behavior of each process and the interactions between the components. The equations may be based on empirical relationships or more complex, physically based models.

    \item Validate the model: The model is then validated by comparing its outputs with observed data. This step is crucial for ensuring that the model accurately represents the system being modeled.

    \item Calibrate the model: After the model has been validated, it may be necessary to calibrate it to fine-tune its parameters. This involves adjusting the model's parameters to ensure that its predictions align as closely as possible with the observed data. As final step, some models might require a sensitivity analysis to ensure the model is well-behaved in the spatiotemporal domain and parameter region of interest. 

\end{itemize}

These steps are not always followed in a linear fashion and may be repeated several times as the model is refined and improved. The creation of a biogeochemical model is a complex and iterative process that requires a deep understanding of the system being modeled, as well as a strong background in mathematics and computer programming. Once the model has been validated and calibrated, it can be used to make predictions or to explore the system's behavior under various conditions. The model can also be used to evaluate the impacts of different management strategies or to gain insights into the underlying processes.

Next, we present the development of a conceptual model for the Paranaguá Estuary, based on the criteria outlined in this chapter.
\chapter{A conceptual model for the  Paranaguá Estuarine Complex}
\chaptermark{A conceptual model for the Paranaguá Estuarine Complex}
\label{chapterPEC}

In this chapter, our primary goal is to formulate a conceptual model for the dynamics between nutrients and phytoplankton in the Paranaguá Estuarine Complex (PEC). By doing so, we also provide a detailed illustration of the mathematical modeling applied to marine biogeochemistry. The initial version of the model presented in this chapter was developed in collaboration with Dr. Marines M. Wilhelm during the postgraduate course "Tópicos de Matemática II: Aspectos teóricos e matemáticos do ciclo do carbono no oceano" taught by Prof. Dr. Francisco de Melo Viríssimo, in 2021. That early formulation served as the foundation for the expanded and fully revised model presented here.


\section{The PEC environment and data availability}
We begin this section by defining the model domain. The PEC is located on the central-north coast of Paraná (25°00’S - 25°35’S; 48°15’W - 48°15’W), in southern Brazil, and its ecological and economic relevance has been reinforced by recent geochemical assessments of sedimentary organic matter \cite{marines2023}. Having an area of approximately 612km$^2$, the PEC is subdivided into two axes, the main axis being called `East-West', 56 km long, where the largest bulk port in Latin America is located, the Port of Paranaguá, and a `North-South' axis, about 30 km long, where there is interference from artisanal fishing and agriculture (Figure \ref{PEC}) \cite{Marone2005,Martins2010,Lamour2007,Martins2015}.

\begin{figure}[ht]
\caption{Model domain: Paranaguá Estuarine Complex.}
\begin{center}
\includegraphics[width=\textwidth]{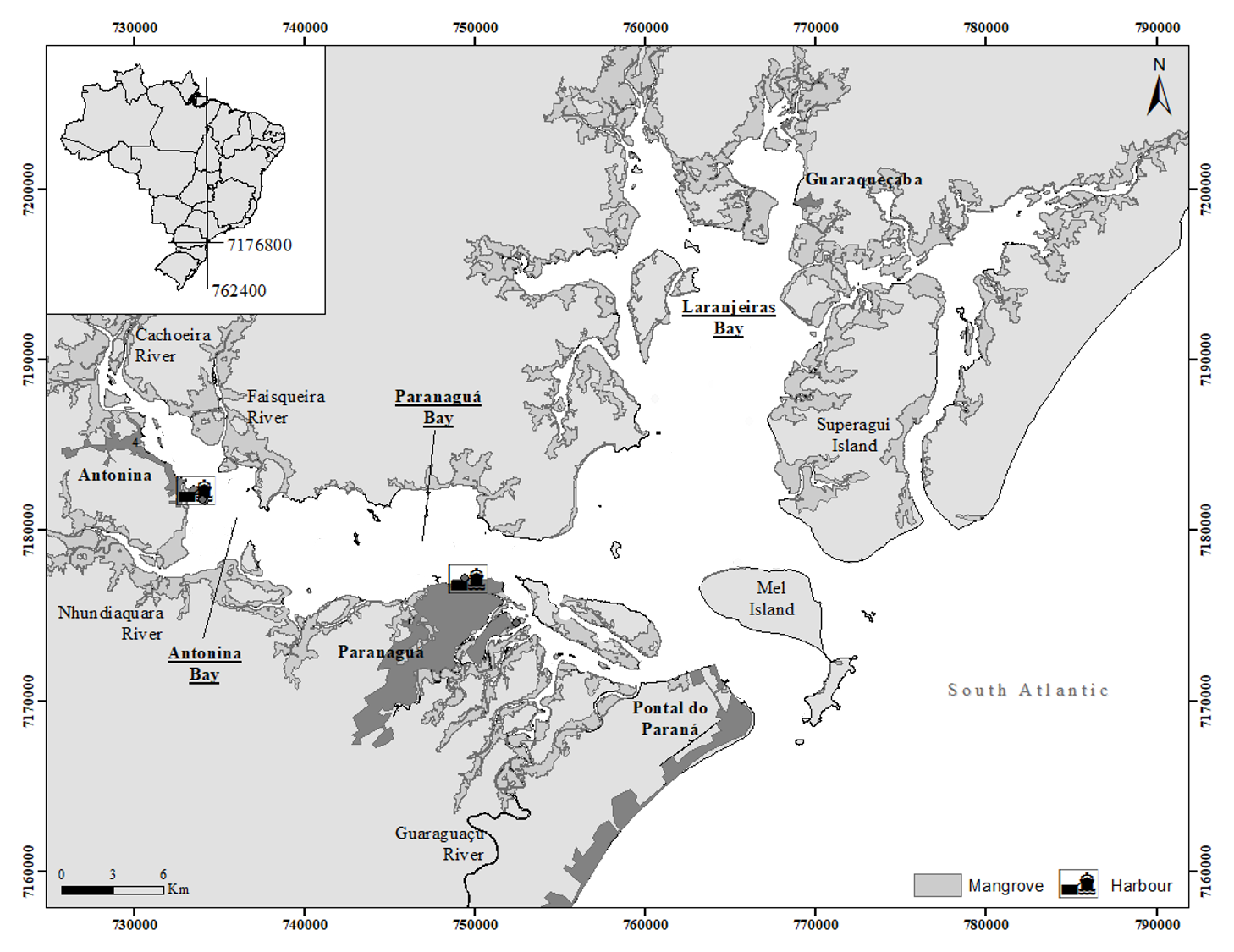}    
\end{center}
\label{PEC}

\smallskip
\textbf{Source:} \citeauthor{passos2012} (\citeyear{passos2012}).
\end{figure}

The development of a conceptual NP model for PEC was motivated by the availability of in situ and literature datasets for the model tracer concentrations and related parameters. This data availability makes the NP formulation both practical and representative of the dominant nutrient-phytoplankton dynamics observed in the system.

The adoption of a two-box circulation model for the PEC is motivated by its classification as a partially mixed estuary, where both riverine inflow and tidal exchange contribute significantly to water renewal and vertical salinity gradients \cite{Marone2005,Verri2020}. In addition, salinity data are available \cite{Machado1997ParanaguaBay}, supporting the parametrization and validation of a simple stratified-exchange framework.
In our conceptual model, the two-box representation constitutes a deliberate simplification of PEC geophysical features, capturing the essential vertical structure of the estuary (surface and bottom layers) while remaining computationally inexpensive, providing a low-cost structure for testing the optimization framework presented in the following chapters by running experiments on a personal computer.



\section{Model development}

As in the previous chapter, we start by defining the biogeochemical SMS equations for this model. As part of that, we calibrate the model ``manually'' by finding parameter values based on the literature and available data. We then present a conceptual circulation box model representing the geophysical aspects of PEC. Finally, we couple both models to obtain the final set of biogeochemical SMS equations.

The data used for the development of the PEC biogeochemical model are:
\begin{enumerate}
    \item Biochemical parameters, $K$, $V_\mathrm{max}$ and $\lambda$, here set as constants, although the last two are based on the PEC phytoplankton population data found in the literature \cite{Machado1997ParanaguaBay,Brandini2022}.
    \item Physical forcing data, as salinity, $S$, and water fluxes, $Q_\text{river}$, $Q_\text{ocean}$ and $Q_\text{ebm}$. This data was inferred from values found in the literature \cite{Machado1997ParanaguaBay,IAT2020}, based on theoretical circulation equations \cite{Verri2020}. 
    \item Initial values for the tracers' concentrations \cite{Machado1997ParanaguaBay}, $C_\mathrm{N}$ and $C_\mathrm{PHY}$, once the mathematical model is already stated as an ODE system.
\end{enumerate}

Further details on the model and the data considered are explored in the following.

\subsection{A conceptual biochemical model for the NP dynamics}

Phytoplankton populations, in general, are extremely sensitive to physical and chemical changes in the aquatic environment, resulting in rapid variations in their rates of reproduction and mortality. In PEC waters, it is known that phytoplankton is dominated by diatoms, especially during the driest season, which occurs in winter. This prevalence may be related to the high competitiveness of diatoms in nutrient-rich environments, such as estuaries. Although estuarine conditions vary significantly in terms of salinity and turbidity \cite{Procopiak}, diatoms thrive and often dominate these ecosystems due to their adaptability and efficient nutrient processing. In addition to diatoms, other phytoplankton groups can also be found in the PEC \cite{Brandini2022}, as we can see in Figure \ref{fig:phyGroups1}), but due to their occurrence in smaller proportions (Figure \ref{fig:phyGroups2}), they will not be considered in this model. The reproduction of phytoplankton in the PEC and its estimated lifespan depend on several factors, including the incidence of sunlight, salinity, water temperature, and the residence time of water in the PEC (see Box \ref{box:ResidenceTime}).

\vspace{0.5cm}
\begin{codebox}[label=box:ResidenceTime]{\textbf{Residence time of a marine environment}}
The residence time of a marine environment, such as the Paranaguá Estuarine Complex, is defined as the average time interval required for the entire volume of water contained in the estuary, as well as the materials dissolved in it, to be completely renewed through exchanges with the ocean. It is estimated that the residence time of the PEC is on the order of 3 to 3.5 days (Lana et al., 2001), a value considered short and indicative of a highly dynamic, efficiently renewed environment.
\end{codebox}

\begin{figure}[ht]
\caption{Boxplot of the main phytoplankton groups observed in the PEC waters.}
\begin{center}
\includegraphics[scale = 0.8, trim={0.8cm 0cm 0cm 0cm}, clip]{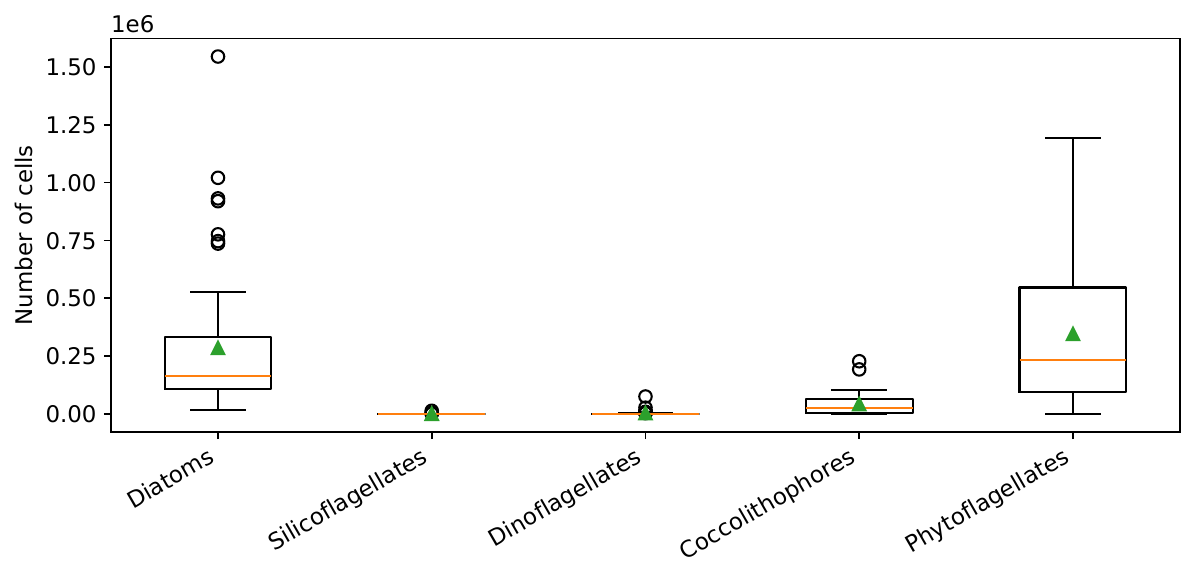}    
\end{center}
\label{fig:phyGroups1}

\smallskip
The figure shows the distribution of cell abundance (cells  m$^{-3}$) across the approximately one-year period time series for each phytoplankton group (Diatoms, Silicoflagellates, Dinoflagellates, Coccolithophores, and Phytoflagellates).
The box central line denotes the median, while the lower and upper edges correspond to the first and third quartiles.
The whiskers represent values within 1.5 × IQR (interquartile range), and points beyond these limits indicate outliers.

\smallskip
\textbf{Source:} the author (2026). 

\smallskip
\textbf{Data:} \citeauthor{Brandini2022} (\citeyear{Brandini2022}).
\end{figure}

\vspace{0.4cm}
Figure \ref{fig:phyGroups2} shows estimates on the diatom contributions for the chlorophyll-a mass in PEC.
Let $C$ be the cellular carbon on phytoplankton cells and $V$ the volume of a phytoplankton cell. For converting the total observed cells into carbon mass, we considered the equations:
\begin{enumerate}
\item For diatoms:
$$C = 0.288 / V^{0.811}$$
\item For the other phytoplankton groups:
$$C = 0.216 / V^{0.939}$$
\end{enumerate}
For each group, we calculate $V$ as the volume of a sphere with the cell's diameter. Last, we use the chlorophyll-a: Carbon ratio, $Chl: C = 0.02$, to calculate the chlorophyll-a mass per phytoplankton cell.

\begin{figure}[ht]
\caption{Estimated chlorophyll-a proportion by phytoplankton group (size-based) on PEC.}
\begin{center}
\includegraphics[scale = 0.8, trim={0.5cm 3cm 0cm 3.8cm}, clip]{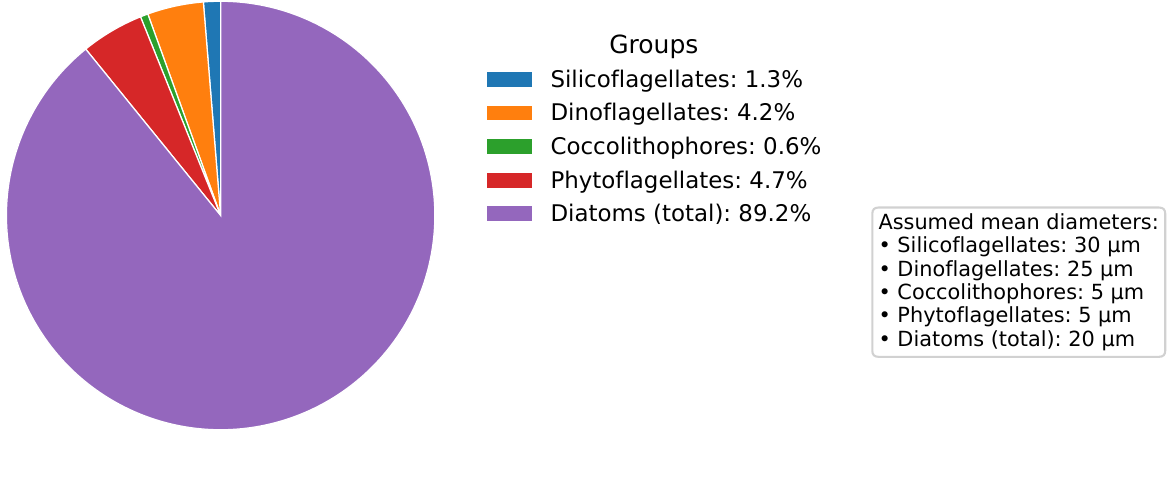}    
\end{center}
\label{fig:phyGroups2}

\smallskip
Data derived from the mean cell abundance and the volume-to-carbon conversion based on the assumed mean cell diameter for each group, assuming a constant chlorophyll-to-carbon ratio among groups. Percentages in the legend indicate each group's relative contribution.

\smallskip
\textbf{Source:} the author (2026).  

\smallskip
\textbf{Data:} \citeauthor{Machado1997ParanaguaBay} (\citeyear{Machado1997ParanaguaBay}); \citeauthor{phySize} (\citeyear{phySize}).
\end{figure}

Let us recall the NP-model diagram, presented in Figure \ref{figConceptualModel}. With the intention of carrying out numerical tests, we will consider nitrate as the nutrient tracer not only due to data availability, but also because it is the most limiting nutrient in marine environments and a strong indicator of eutrophication in coastal areas. Let $C_\mathrm{N}$ be the concentration of N and $C_\mathrm{PHY}$ the concentration of phytoplankton measured in mmol m$^{-3}$, $I$ the light intensity due to solar incidence measured in lm, $T$ the temperature in $^o$C, $S$ is the salinity (dimensionless). For describing the growth and mortality of phytoplankton depending on the nutrient availability, we follow \cite{sarmientogruber} and define the SMS equation for phytoplankton as:
\begin{equation}\label{smsPHYcapmodel}
    \dfrac{d C_\mathrm{PHY}}{d t} =  \left[ V_{\max} \cdot g(C_\mathrm{N},T,S,I) - \lambda \right] \cdot C_\mathrm{PHY} ~,
\end{equation}
The maximum specific growth rate of phytoplankton $V_{\max}$ is defined as the maximum growth of the phytoplankton population depending solely on the existing population, under ideal conditions. The function $g(C_\mathrm{N},T,S,I)$ represents the limitation of P growth due to the availability of nutrients and sunlight, variations in temperature, and salinity. We consider the mortality rate of phytoplankton as a linear function with constant rate $\lambda$. For defining the limiting function $g$, we consider a set of four limiting functions corresponding to each of the limiting elements. We have:
\begin{equation}
    g(C_\mathrm{N},T,S,I) = \eta(C_\mathrm{N}) \cdot \alpha(T) \cdot \beta(S) \cdot \gamma(I) ~, \label{limitP}
\end{equation}
with
\begin{equation*}
    \eta(C_\mathrm{N}) = \dfrac{C_\mathrm{N}}{C_\mathrm{N}+K} ~.
\end{equation*}

The Monod constant $K$ is the concentration of nitrate such that the growth of phytoplankton is equivalent to half of the maximum possible growth, $V_\mathrm{max}$. For nitrate, we have $K$ between $0.1$ and 0.3  mmol  m$^{-3}$. In this model we will also consider $V_{\max}=1.4 \ \text{day}^{-1}$, $K=0.1$ mmol  m$^{-3}$ , $\lambda=0.05 \ \text{day}^{- 1}$ \cite{sarmientogruber}.

The average water temperature varies between 23$^o$–30$^o$ in summer and 18$^o$–25$^o$ in winter \cite{Machado1997ParanaguaBay}, so the water temperature in winter is not a limiting factor on the phytoplankton growth \cite{Khan1998}. However, in summer, phytoplankton production can be limited to up to 60\% of its maximum if the temperature reaches 30$^o$. Since $T$ is the average temperature in $^o$C of the water in the upper tank during the period that the model will predict, we then define
\begin{equation}\label{smsTemp}
  \alpha(T)=\min \left\lbrace  1 ~, ~ 0.6+0.08 \cdot (30-T) \right\rbrace ~.  
\end{equation}
Equation (\ref{smsTemp}) is obtained by assuming that temperatures in the range 18$^o$C to 25$^o$C doesn't present a limiting impact on the PEC phytoplankton growth, thus $\alpha(T)=1$ for these temperatures, and by linear interpolation between temperatures 25$^o$C to 30$^o$C, considering $\alpha(25)=1$ and $\alpha(30)=0.6$ (Figure \ref{fig:alphaT}).

\vspace{0.5cm}
\begin{figure}[htpb!]
\caption{Temperature limiting impact on PEC phytoplankton growth.}
\begin{center}
\includegraphics[width=0.8\linewidth]{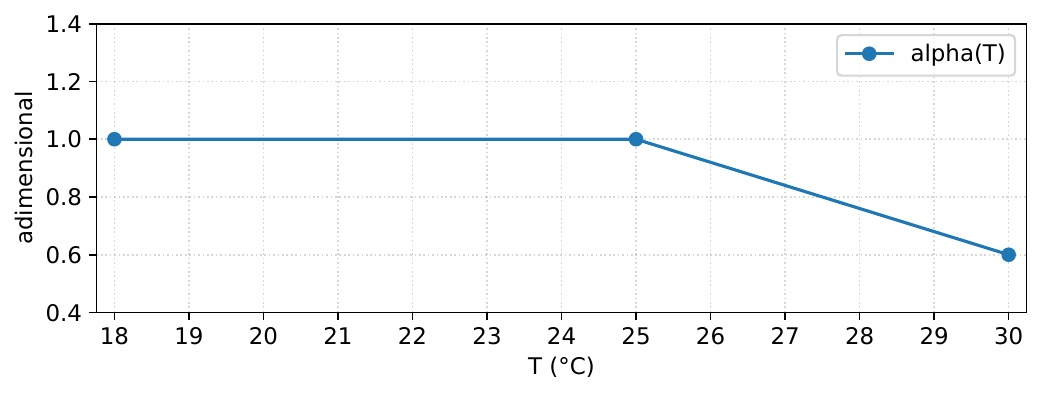}    
\end{center}
\label{fig:alphaT}

\smallskip
\textbf{Source:} the author (2026).  

\smallskip
\textbf{Data:} \citeauthor{Khan1998} (\citeyear{Khan1998}).
\end{figure}

\vspace{0.2cm}
Salinity can also be a limiting factor for phytoplankton growth. Khan et al. \cite{Khan1998} state that the ideal salinity varies from 20 to 35. As in the PEC the salinity does not exceed 33, we define
\begin{equation}\label{smsSal}
  \beta(S)=  \max \left\lbrace 0, \min \left\lbrace 1 ~, ~ 1-0.07 \cdot (20-S) \right\rbrace \right\rbrace ~.
\end{equation}

Equation (\ref{smsSal}) is obtained by assuming that salinities in the range 20 to 35 doesn't present a limiting impact on the PEC phytoplankton growth, thus $\beta(S)=1$ for these salinities, and by linear interpolation between salinities 10 to 20, considering $\beta(10)=0$ and $\beta(20)=0$. We also ensured $\beta(S) \geq 0$, which could result in errors in extreme scenarios, as $S \leq 5$ (Figure \ref{fig:betaS}).

\begin{figure}[ht]
\caption{Salinity limiting impact on PEC phytoplankton growth.}
\begin{center}
\includegraphics[width=0.8\linewidth]{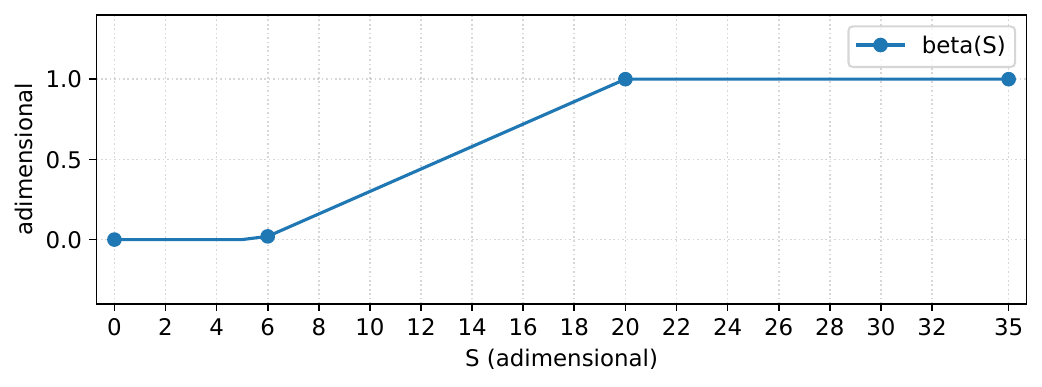}    
\end{center}
\label{fig:betaS}

\smallskip
\textbf{Source:} the author (2026).  

\smallskip
\textbf{Data:} \citeauthor{Khan1998} (\citeyear{Khan1998}).
\end{figure}

The sunlight intensity reaching the PEC surface varies from the minimum value $2.9 \ kWh \ m^{-2} \ d^{-1}$ in winter to the maximum value $5.6 \ kWh \ m^{-2} d^{-1 }$ in summer \cite{LABREN}. Therefore, the incidence of sunlight is not a limiting factor for the growth of phytoplankton in the estuary. We consider
$$\gamma(I)=1 ~.$$

In relation to phytoplankton, the SMS nitrate equation is written as:
\begin{equation}\label{smsNcapmodel}
    \dfrac{d C_\mathrm{N}}{d t}= \left[ - V_{\max} \cdot g(C_\mathrm{N},T,S,I) + r \cdot \lambda \right] \cdot C_\mathrm{PHY} ~,
\end{equation}

The percentage remineralization of dissolved organic matter into nitrate is represented by the parameter $r$ and set 0.7 \cite{sarmientogruber}. All the parameter values used in the modeling of the PEC are presented in Table \ref{ParametrosEstuario}.

\vspace{0.4cm}
\begin{table}[htpb!]
\caption{Parameter values set for the PEC conceptual model developed in this chapter.}
\begin{center}
    \begin{tabular}{l|c|c|l}
    \textbf{Parameter} & \textbf{Value} & \textbf{Unit} & \textbf{Description}\\
    \hline 
    $V_{\max}$ & 1.4 & day$^{-1}$ & \small{Maximum specific growth rate of phytoplankton}\\
    $\lambda$ & 0.05 & day$^{-1}$ & \small{Linear mortality rate of phytoplankton} \\
    $K$ & 0.1 &  mmol  m$^{-3}$ & Monod constant for nitrogen \\
    $r$ & 0.7 & day$^{-1}$ & Nitrogen remineralization rate 
    \end{tabular}
\end{center}
\label{ParametrosEstuario}

\smallskip
\textbf{Source:} the author (2026).

\smallskip
\textbf{Data:} \citeauthor{Khan1998} (\citeyear{Khan1998}); \citeauthor{sarmientogruber} (\citeyear{sarmientogruber}).
\end{table}

\vspace{0.2cm}
Sediment core studies indicate that the PEC functions as an efficient sediment trap, where the estuarine turbidity maximum promotes the retention of fine, organic-rich particles \cite{marines2023}. These conditions support the assumptions in our model regarding light limitation, stratification, and the high remineralization rates applied to phytoplankton biomass.

The set of ODE equations $\left\lbrace \dfrac{d C_\mathrm{N}}{d t}, \dfrac{d C_\mathrm{PHY}}{d t}\right\rbrace$, described on equations (\ref{smsNcapmodel}) and (\ref{smsPHYcapmodel}), represents the biochemical SMS equations for our model. Since we want these equations to depend only on the time $t$, the function $g(C_\mathrm{N},T,S,I)$ needs to be described in terms of $t$ only. As $\gamma(I)$ is taken as constant and $C_\mathrm{N}$ is computed when integrating the ODE system, it lasts to create functions $T(t)$ and $S(t)$ based on the interpolation of sampled data found in the literature or theoretical assumptions. Both of these choices are possible. Here we consider the first of these approaches, which is better described in the following subsection, once it is intrinsically related to the physical model portraying the PEC environment.

\subsection{A conceptual box-model representing geophysical interactions}
We will now present a simple circulation model for the PEC waters. Considering that the PEC is a partially mixed estuary, we will use a three-dimensional two-box model for water circulation based on Verri et al. \cite{Verri2020}, represented in Figure \ref{figCirculation}.

\vspace{0.7cm}
\begin{figure}[htpb!]
\caption{Circulation model for the PEC.}
\begin{center}
\tikzset{every picture/.style={line width=0.75pt}} 

\begin{tikzpicture}[x=0.75pt,y=0.75pt,yscale=-1,xscale=1]

\draw  [fill={rgb, 255:red, 183; green, 214; blue, 255 }  ,fill opacity=1 ] (95,35) -- (441.5,35) -- (441.5,126) -- (95,126) -- cycle ;
\draw  [fill={rgb, 255:red, 74; green, 144; blue, 226 }  ,fill opacity=1 ] (95,126) -- (441.5,126) -- (441.5,217) -- (95,217) -- cycle ;
\draw  [fill={rgb, 255:red, 248; green, 231; blue, 28 }  ,fill opacity=1 ] (26,69) -- (68,69) -- (68,59) -- (96,79) -- (68,99) -- (68,89) -- (26,89) -- cycle ;
\draw  [fill={rgb, 255:red, 248; green, 231; blue, 28 }  ,fill opacity=1 ] (442,69) -- (484,69) -- (484,59) -- (512,79) -- (484,99) -- (484,89) -- (442,89) -- cycle ;
\draw  [fill={rgb, 255:red, 248; green, 231; blue, 28 }  ,fill opacity=1 ] (513,179) -- (471,179) -- (471,189) -- (443,169) -- (471,149) -- (471,159) -- (513,159) -- cycle ;
\draw  [fill={rgb, 255:red, 208; green, 2; blue, 27 }  ,fill opacity=1 ] (355.5,181) -- (188.25,181) .. controls (173.06,181) and (160.75,168.69) .. (160.75,153.5) -- (160.75,153.5) -- (153.5,153.5) -- (167.25,126) -- (181,153.5) -- (173.75,153.5) -- (173.75,153.5) .. controls (173.75,161.51) and (180.24,168) .. (188.25,168) -- (355.5,168) -- cycle ;
\draw  [fill={rgb, 255:red, 208; green, 2; blue, 27 }  ,fill opacity=1 ] (162,124.5) -- (162,100.25) .. controls (162,83.96) and (175.21,70.75) .. (191.5,70.75) -- (325.5,70.75) -- (325.5,62) -- (356.5,77) -- (325.5,92) -- (325.5,83.25) -- (191.5,83.25) .. controls (182.11,83.25) and (174.5,90.86) .. (174.5,100.25) -- (174.5,124.5) -- cycle ;

\draw    (59.83,251.83) -- (60.17,209.28) ;
\draw [shift={(60.19,207.28)}, rotate = 90.46] [color={rgb, 255:red, 0; green, 0; blue, 0 }  ][line width=0.75]    (10.93,-3.29) .. controls (6.95,-1.4) and (3.31,-0.3) .. (0,0) .. controls (3.31,0.3) and (6.95,1.4) .. (10.93,3.29)   ;
\draw    (59.83,251.83) -- (100.19,252.26) ;
\draw [shift={(102.19,252.28)}, rotate = 180.61] [color={rgb, 255:red, 0; green, 0; blue, 0 }  ][line width=0.75]    (10.93,-3.29) .. controls (6.95,-1.4) and (3.31,-0.3) .. (0,0) .. controls (3.31,0.3) and (6.95,1.4) .. (10.93,3.29)   ;

\draw (27,33) node [anchor=north west][inner sep=0.75pt]   [align=left] {River};
\draw (460,113) node [anchor=north west][inner sep=0.75pt]   [align=left] {Ocean};
\draw (447,72) node [anchor=north west][inner sep=0.75pt]  [font=\small] [align=left] {$\displaystyle Q_{ebm}$};
\draw (28,72) node [anchor=north west][inner sep=0.75pt]  [font=\small] [align=left] {$\displaystyle Q_{river}$};
\draw (473,162) node [anchor=north west][inner sep=0.75pt]  [font=\small] [align=left] {$\displaystyle Q_{ocean}$};
\draw (104.19,255.28) node [anchor=north west][inner sep=0.75pt]   [align=left] {$\displaystyle u$};
\draw (42.19,199.45) node [anchor=north west][inner sep=0.75pt]   [align=left] {$\displaystyle w$};

\end{tikzpicture}    
\end{center}
\label{figCirculation}

\smallskip
The blue rectangles represent a side view of the PEC aquatic environment, with the estuarine head on the left and the mouth on the right, defining the model domain. The arrows point in the direction of the advective flux. Yellow arrows indicate interactions with rivers and the ocean, whereas red arrows represent water circulation within the estuary complex.

\smallskip
\textbf{Source:} the author (2026).

\smallskip
\textbf{Data:} \citeauthor{Verri2020} (\citeyear{Verri2020}).
\end{figure}
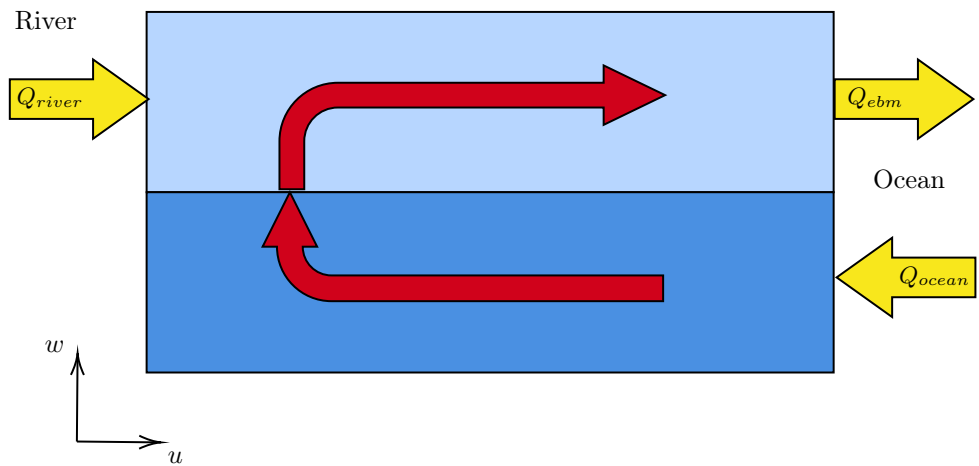

\paragraph{Boxes dimensions}

The surface area of the PEC is estimated as $601 \times 10^6 \text{m}^2$, its volume is estimated as $2 \times 10^9 \text{m}^3$, and the length of its three tidal inlets combined is estimated as $5.7 \times 10^3 \text{m}$ \cite{Cattani2009}. In our model, we represent the PEC environment with two stacked tridimensional boxes whose dimensions are represented in Figure \ref{figDimensionsEstuario}. The motivations for the choice of the boxes' dimensions as well as the dimensions themselves, are described in the following:

\begin{itemize}
    \item[$L_y$:] Approximate sum of the lengths of the tidal inlets of the PEC, that is, $L_y = 5.7 \times 10^3$ m.
    \item[$L_x$:] Approximate PEC surface area divided by $L_y$, that is,  $L_x = 105.4 \times 10^3$ m. This choice comes from the fact that the depth itself does not directly interfere with the amount of nutrients in the water since the maximum depth reached in the estuary is less than 35m, but it directly interferes with the incidence of sunlight, which affects the growth of phytoplankton. Due to the features of the internal circulation of the PEC, phytoplankton will be modeled only in the upper box, which is another motivation to consider that the surface of the entire estuary is contained in the upper box.
    \item [$H$:] The depth of the PEC waters is not homogeneous, with an average between 3 m and 5.4 m, reaching more than 15 m in some regions. For simplicity, we consider that the two boxes have the same height, that is, $H_\text{up}=H_\text{low}=H/2$, where $H_\text{up}$ and $H_\text{low}$ denote the heights of the upper and lower boxes, respectively. We define $H$ as the approximate volume of water in the  divided by its surface area, which is $H=3.4$ m.
\end{itemize} 

\vspace{0.7cm}
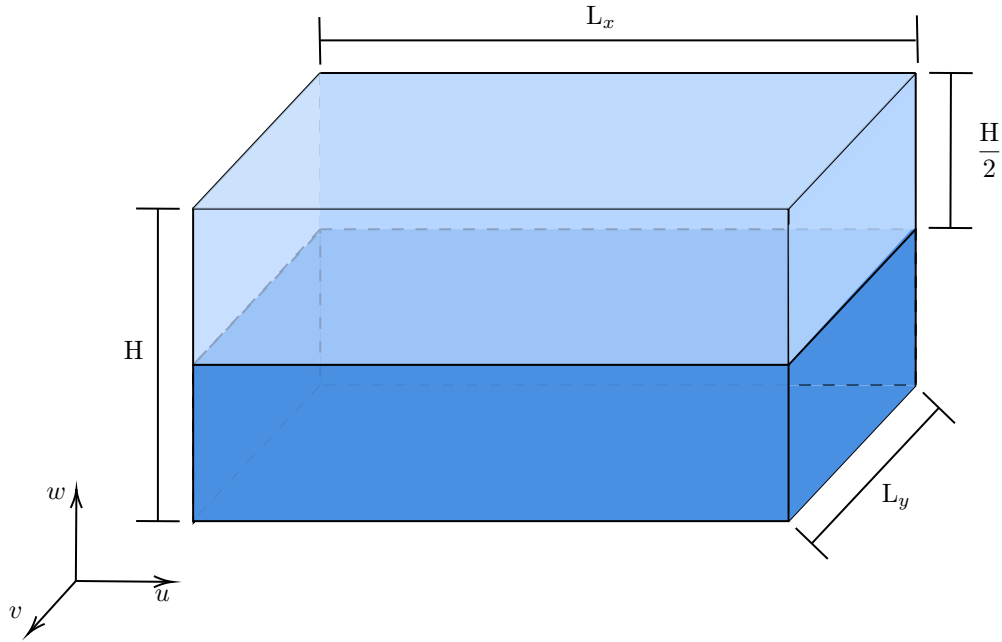
\begin{figure}[htpb!]
\caption{Discretization of the PEC domain into two boxes.} 
\begin{center}
\tikzset{every picture/.style={line width=0.75pt}} 

\begin{tikzpicture}[x=0.65pt,y=0.65pt,yscale=-1,xscale=1]

\draw  [dash pattern={on 4.5pt off 4.5pt}]  (189,95) -- (189.29,186.75) ;
\draw  [color={rgb, 255:red, 0; green, 0; blue, 0 }  ,draw opacity=1 ][fill={rgb, 255:red, 74; green, 144; blue, 226 }  ,fill opacity=1 ][dash pattern={on 4.5pt off 4.5pt}] (115,265) -- (189.29,186.75) -- (189.57,278.49) -- (115.29,356.75) -- cycle ;
\draw  [draw opacity=0][fill={rgb, 255:red, 183; green, 214; blue, 255 }  ,fill opacity=1 ][dash pattern={on 4.5pt off 4.5pt}] (189,95) -- (535.5,95) -- (535.5,186) -- (189,186) -- cycle ;
\draw  [fill={rgb, 255:red, 74; green, 144; blue, 226 }  ,fill opacity=1 ][dash pattern={on 4.5pt off 4.5pt}] (189,186) -- (535.5,186) -- (535.5,277) -- (189,277) -- cycle ;
\draw  [dash pattern={on 4.5pt off 4.5pt}]  (189,277) -- (115,356) ;
\draw    (535.5,277) -- (461.5,356) ;
\draw  [dash pattern={on 4.5pt off 4.5pt}]  (189,186) -- (115,265) ;
\draw    (535.5,186) -- (461.5,265) ;
\draw    (535.5,95) -- (461.5,174) ;
\draw    (189,95) -- (115,174) ;
\draw  [draw opacity=0][fill={rgb, 255:red, 74; green, 144; blue, 226 }  ,fill opacity=0 ] (189,186) -- (535.5,186) -- (461.5,265) -- (115,265) -- cycle ;
\draw  [draw opacity=0][fill={rgb, 255:red, 74; green, 144; blue, 226 }  ,fill opacity=1 ] (189,277) -- (535.5,277) -- (461.5,356) -- (115,356) -- cycle ;
\draw  [fill={rgb, 255:red, 183; green, 214; blue, 255 }  ,fill opacity=0.76 ] (115,174) -- (461.5,174) -- (461.5,265) -- (115,265) -- cycle ;
\draw  [fill={rgb, 255:red, 74; green, 144; blue, 226 }  ,fill opacity=0.91 ] (115,265) -- (461.5,265) -- (461.5,356) -- (115,356) -- cycle ;

\draw  [draw opacity=0][fill={rgb, 255:red, 192; green, 219; blue, 252 }  ,fill opacity=0.82 ] (189,95) -- (535.5,95) -- (461.5,174) -- (115,174) -- cycle ;
\draw    (535.5,95) -- (189,95) ;
\draw    (535.5,95) -- (535.5,277) ;
\draw    (549,290) -- (475,369) ;
\draw    (484.33,378) -- (465.67,360) ;
\draw    (558.33,299) -- (539.67,281) ;

\draw    (189.03,90.12) -- (188.45,62.43) ;
\draw    (536.75,89.12) -- (536.17,61.44) ;
\draw    (94.5,174) -- (94.5,356) ;
\draw    (81.83,173.92) -- (107.17,174.08) ;
\draw    (81.83,355.92) -- (107.17,356.08) ;

\draw    (189,76) -- (535.5,76) ;
\draw    (556,95) -- (555.71,185.25) ;
\draw    (543.33,94.92) -- (568.67,95.08) ;
\draw    (543.04,185.18) -- (568.39,185.33) ;
\draw  [draw opacity=0][fill={rgb, 255:red, 184; green, 216; blue, 255 }  ,fill opacity=0.87 ] (461.59,183.62) -- (534.37,104.18) -- (534.68,184.33) -- (461.89,263.76) -- cycle ;

\draw    (46.75,390.85) -- (20.19,420.23) ;
\draw [shift={(18.85,421.72)}, rotate = 312.11] [color={rgb, 255:red, 0; green, 0; blue, 0 }  ][line width=0.75]    (10.93,-3.29) .. controls (6.95,-1.4) and (3.31,-0.3) .. (0,0) .. controls (3.31,0.3) and (6.95,1.4) .. (10.93,3.29)   ;
\draw    (46.75,390.85) -- (47.08,339.37) ;
\draw [shift={(47.09,337.37)}, rotate = 90.36] [color={rgb, 255:red, 0; green, 0; blue, 0 }  ][line width=0.75]    (10.93,-3.29) .. controls (6.95,-1.4) and (3.31,-0.3) .. (0,0) .. controls (3.31,0.3) and (6.95,1.4) .. (10.93,3.29)   ;
\draw    (46.75,390.85) -- (101.09,391.35) ;
\draw [shift={(103.09,391.37)}, rotate = 180.53] [color={rgb, 255:red, 0; green, 0; blue, 0 }  ][line width=0.75]    (10.93,-3.29) .. controls (6.95,-1.4) and (3.31,-0.3) .. (0,0) .. controls (3.31,0.3) and (6.95,1.4) .. (10.93,3.29)   ;

\draw (73,249) node [anchor=north west][inner sep=0.75pt]   [align=left] {H$ $};
\draw (342,54) node [anchor=north west][inner sep=0.75pt]   [align=left] {L$\displaystyle _{x}$};
\draw (569,122) node [anchor=north west][inner sep=0.75pt]   [align=left] {$\displaystyle \dfrac{\text{H}}{2}$$ $};
\draw (514,332.5) node [anchor=north west][inner sep=0.75pt]   [align=left] {L$\displaystyle _{y}$};
\draw (91,394) node [anchor=north west][inner sep=0.75pt]   [align=left] {$\displaystyle u$};
\draw (6,404) node [anchor=north west][inner sep=0.75pt]   [align=left] {$\displaystyle v$};
\draw (28,335) node [anchor=north west][inner sep=0.75pt]   [align=left] {$\displaystyle w$};

\end{tikzpicture}    
\end{center}
\label{figDimensionsEstuario}

\smallskip
$L_x$ is the horizontal dimension representing the path from the PEC head to its mouth, $L_y$ is the horizontal dimension representing the width of PEC margins and $H$ is the vertical dimension representing PEC depth. The domain is subdivided into two regions with equal volume, each with dimensions $L_x$, $L_y$ and $\dfrac{H}{2}$.

\smallskip
\textbf{Source:} the author (2026).

\smallskip
\textbf{Data:} \citeauthor{Cattani2009} (\citeyear{Cattani2009}).
\end{figure}

\paragraph{Volumetric flow and salinity}

For the approximate calculation of water volume flow and salinity, we will consider the Knudsen model (KEBM), reviewed by Verri et al. \cite{Verri2020}. We will calculate an approximation of horizontal averages for flow and salinity, where we disregard heat fluxes, precipitation, wind and water temperature. We will also consider the average over the tidal cycle, approximating the lunar day by the solar day with the average tide during it, which is consistent with partially mixed estuaries like the PEC. KEBM considers the following water flows entering the estuary: $Q_\text{river}$ represents the flow coming from rivers entering the estuary at the head of the estuary; in the upper box, $Q_\text{ocean}$ represents the flow coming from the ocean entering the estuary through the lower box and $Q_\text{ebm}$ represents the water flow leaving the estuary through the estuary tidal inlets, in the upper box (Figure \ref{figCirculation}). 

All these flows are measured in $\text{m}^3 \ \text{day}^{-1}$ and, as by hypothesis, the volume of water in the estuary is constant, it is worth the relation (\ref{fluxo}). Furthermore, we assume that the salinity $S_\text{low}$ of the lower box is equal to that of the ocean salinity, and the salinity in the upper box, denoted $S_\text{up}$, is calculated as a solution to the equation (\ref{fluxosal}). We consider as input data for this model:
\begin{enumerate}
\item $S_\text{low}(t)$ and $S_\text{up}(t)$ the mean salinity in the lower and upper boxes of the model \cite{Machado1997ParanaguaBay} (Figure \ref{fig:dataSal}), respectively, at time instant $t$;
\item The flow $Q_\text{river}$ is derived from the PEC residence time, which is approximately 3.49 days \cite{Lana2001};
\item Estimates for the flows $Q_\text{ocean}$ and $Q_\text{ebm}$ obtained from the equation (\ref{fluxo}) and available data on current velocities \cite{Cattani2009} (Figure \ref{fig:dataFluxos}).
\end{enumerate}

\begin{equation}\label{fluxo}
    Q_\text{ebm}(t)=Q_\text{river}(t)+Q_\text{ocean}(t) ~,
\end{equation}
\begin{equation}\label{fluxosal}
    S_\text{up}(t) \ Q_\text{ebm}(t)= S_\text{low}(t) \ Q_\text{ocean}(t) ~.
\end{equation}

Estimates for nitrate concentration in rivers based on weighted average data from \cite{Mizerkowski2012} (Figure \ref{fig:dataNriverFlux}) and estimates for water flows are shown in Figure \ref{fig:dataFluxos}, with the rainy season corresponding to the summer and the dry season corresponding to the winter.  Apart from the nitrate input from ocean water, we consider the hypothesis that rivers are not responsible for the totality of the remaining inorganic nitrate input in the PEC. This way, we estimate the concentration of nitrate coming from the atmosphere, human activity, and other sources, including rivers, as proportional to the mass of nitrate coming from the rivers, based on calibration results, presented in \textit{Chapter \ref{chapterCalibratingPEC}}. Recent geochemical analyses of sedimentary organic matter in the PEC \cite{marines2023} show that the system is strongly dominated by terrigenous inputs, with marine-derived material becoming more relevant only near the estuarine mouth and in the northern sector of the complex. This pattern is consistent with the structure of the nutrient sources considered in our model, where riverine nitrogen supply plays a central role relative to oceanic inputs.

As an aside, Box \ref{box:SalEstimate} briefly describes an alternative approach for estimating salinity when no observational data are available.

\vspace{0.5cm}
\begin{codebox}[label=box:SalEstimate]{\textbf{Data extrapolation: \\ How to simulate salinity data, based on theoretical assumptions}}

While we used a slightly different technique to obtain interpolated values for mean tracer concentrations and salinity for each box in our model, one of the first possibilities considered was simply linear interpolation of the available data.

In an even previous approach, we considered the complete absence of salinity data. In this scenario, it is possible to obtain average salinity data for each box in our PEC model from the circulation equations (\ref{fluxo}) and (\ref{fluxosal}).

Once the circulation model is well-defined, we can also generate data on tracer transport from the lower to the upper box in the model by considering a similar assumption: that the tracer concentration in the sea near the PEC, ideally at the mouth of the PEC, is known. Such tracer concentration value is then set as the mean tracer concentration in the lower box.
\end{codebox}

\subsection{A conceptual biogeochemical model for PEC}
We consider some additional assumptions to create a model that couples the SMS equations and the circulation model. We consider the tracers' supply from the lower box to the upper box to occur exclusively through advection proportional to the transport of salinity\footnote{We do this in order to reduce the number of differential equations in the model, thus keeping it computationally inexpensive to be explored on a personal computer.}. We also consider a remineralization rate of phytoplankton that dies equal to 70\%, and the other possible factor for reducing its concentration in the upper box is its export to the ocean, since in our model, there is no advective flow that transports it to the lower box. We will also make use of the following notations:
\begin{itemize}
\item $\text{Vol}_\text{box}$ : The volume of the upper box (which, in this case, is equal to the volume of the lower box), calculated as $\text{Vol}_\text{box}=L_x \cdot L_y \cdot H/2 = 10^9 \ \text{m}^3$.
\item $T$ : The upper box temperature (Figure \ref{fig:dataTemp}), in $^o$C.
\item $C_\mathrm{N}^\text{river}$ : Average nitrate concentration in the Nhundiaquara river, used for estimate the riverine and pluvial water drained to the PEC (Figure \ref{fig:dataNriverFlux}), in mmol \ m$^{-3}$.
\item $C_\mathrm{N}^\text{low}(t)$ : The lower box average nitrate concentration (Figure \ref{fig:dataNitrate}), in mmol \ m$^{-3}$.
\end{itemize}

In Figure \ref{fig:dataNitrate}, note that the behavior of the blue plot (Upper box) is what we want to reproduce from our model for the nitrate tracer concentration throughout the year, while the other data (Lower box) is considered as data input for the model.

We also considered the entry of phytoplankton into the upper tank via advective flow from the lower box. This factor is based on the available data (Figure \ref{fig:dataPhy}) and may be a cause for the seasonal variation in the composition of the phytoplankton population in the PEC. In Figure \ref{fig:dataPhy}, also note that the behavior of the blue plot (Upper box) is what we want to reproduce from our model for the nitrate tracer concentration throughout the year, while the other data (Lower box) is considered as data input for the model. 

\vspace{0.6cm}
\begin{figure}[htpb!]
\caption{PEC daily data estimates on salinity over a one-year period.}
\begin{center}
\includegraphics[scale=0.8]{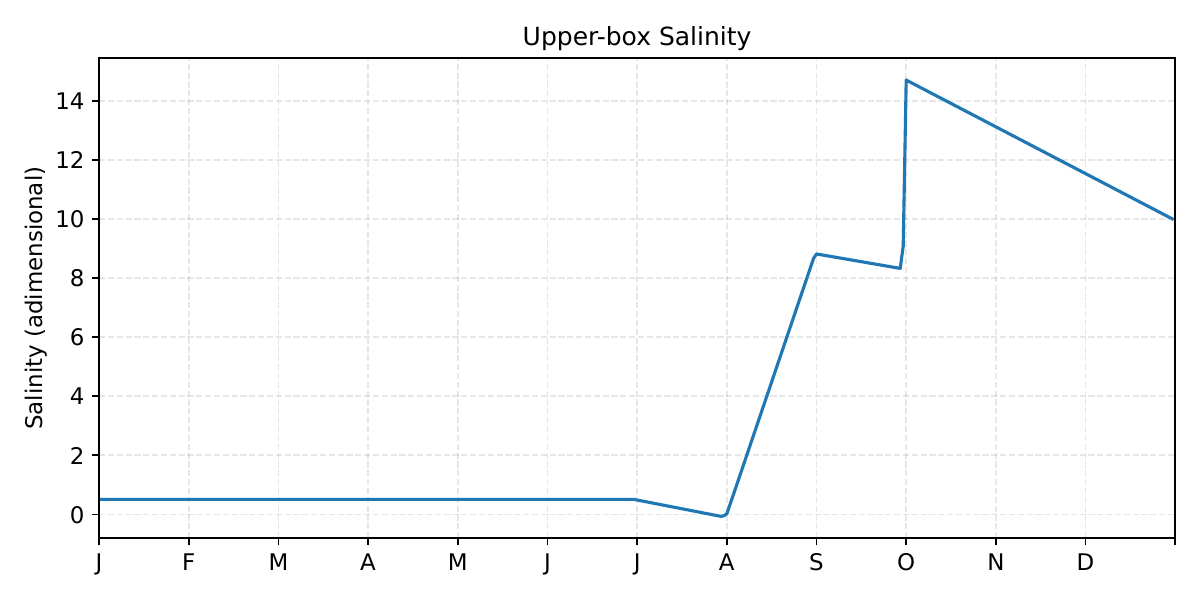}

\vspace{0.3cm}
\includegraphics[scale=0.8]{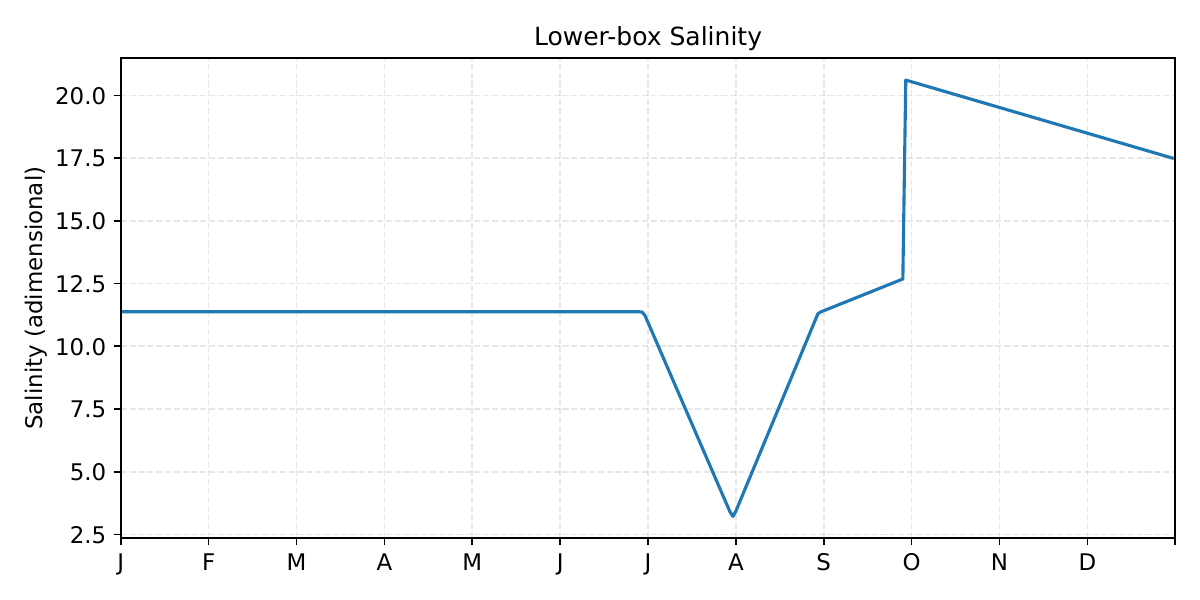}
\end{center}
\label{fig:dataSal}

\smallskip
\textbf{Source:} the author (2026).

\smallskip
\textbf{Data:} \citeauthor{Machado1997ParanaguaBay} (\citeyear{Machado1997ParanaguaBay}).
\end{figure}

\newpage
\begin{figure}[htpb!]
\caption{PEC daily data estimates on circulation fluxes over a one-year period.}

\vspace{-0.2cm}
\begin{center}
\includegraphics[scale=0.8]{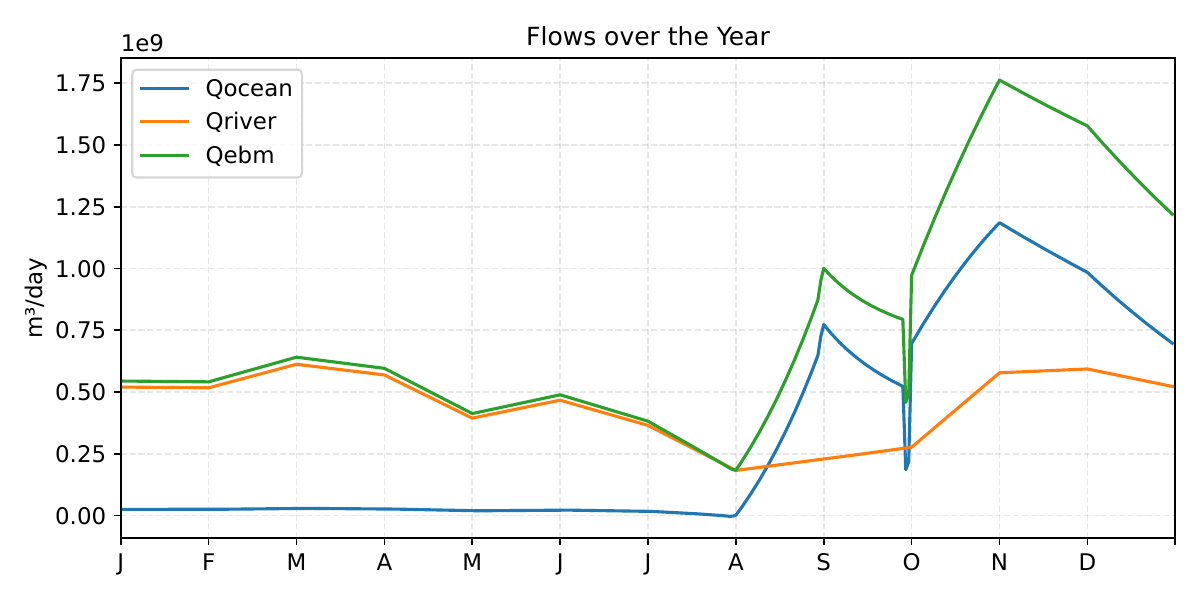}    
\end{center}
\label{fig:dataFluxos}

\smallskip
\textbf{Source:} the author (2026).

\smallskip
\textbf{Data:} IAT (\citeyear{IAT2020}).
\end{figure}

\vspace{0.4cm}
\begin{figure}[htpb!]
\caption{PEC daily data estimates on the water temperature of the upper region of PEC over a one-year period.}

\vspace{-0.4cm}
\begin{center}
\includegraphics[scale=0.8]{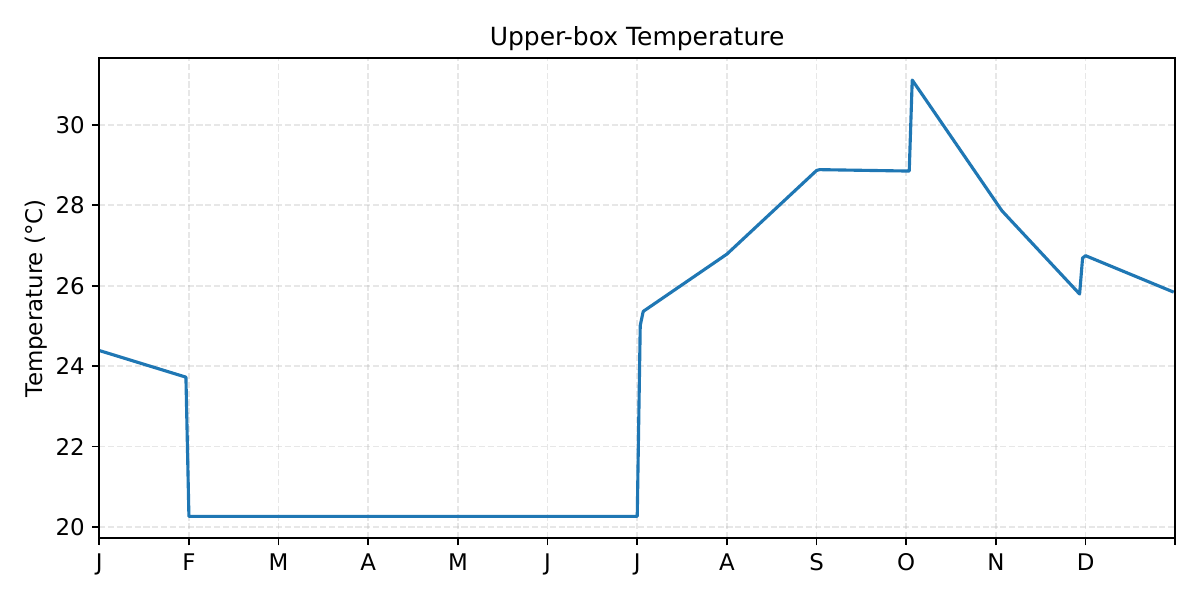}    
\end{center}
\label{fig:dataTemp}

\smallskip
\textbf{Source:} the author (2026).

\smallskip
\textbf{Data:} \citeauthor{Machado1997ParanaguaBay} (\citeyear{Machado1997ParanaguaBay}).
\end{figure}

\newpage
\begin{figure}[htpb!]
\caption{PEC daily data estimates on the Nhundiaquara river nitrate concentration over a one-year period.}

\vspace{-0.4cm}
\begin{center}
\includegraphics[scale=0.8]{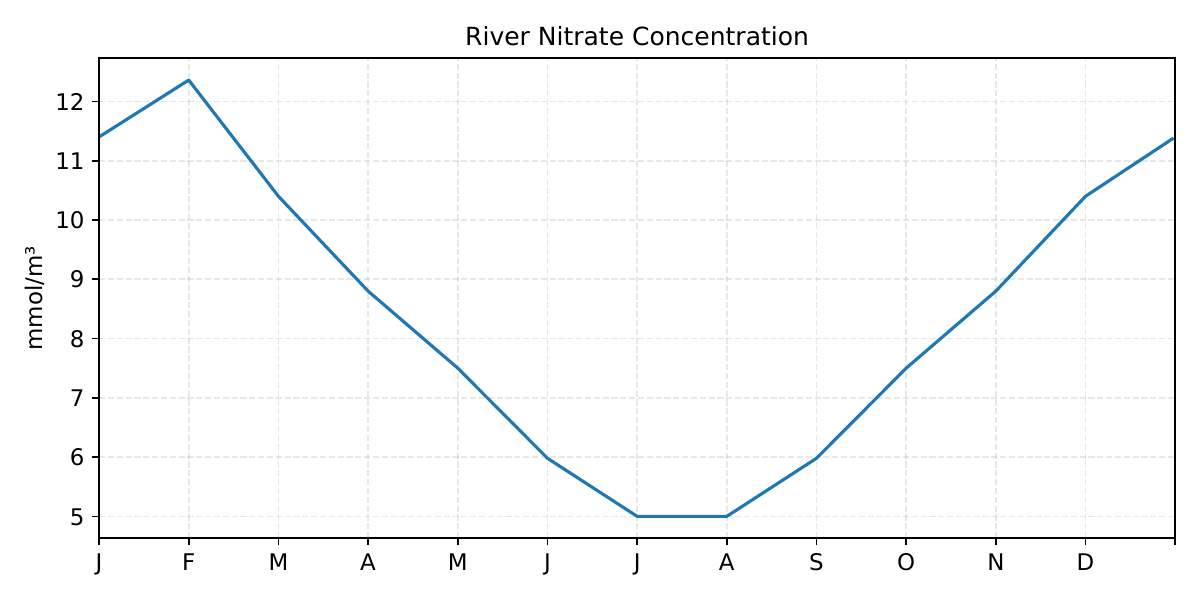}
\end{center}
\label{fig:dataNriverFlux}

\smallskip
\textbf{Source:} the author (2026).

\smallskip
\textbf{Data:} IAT (\citeyear{IAT2020}).
\end{figure}

\begin{figure}[htpb!]
\caption{PEC daily data estimates on nitrate concentrations over a one-year period.}

\vspace{-0.2cm}
\begin{center}
\includegraphics[scale=0.8]{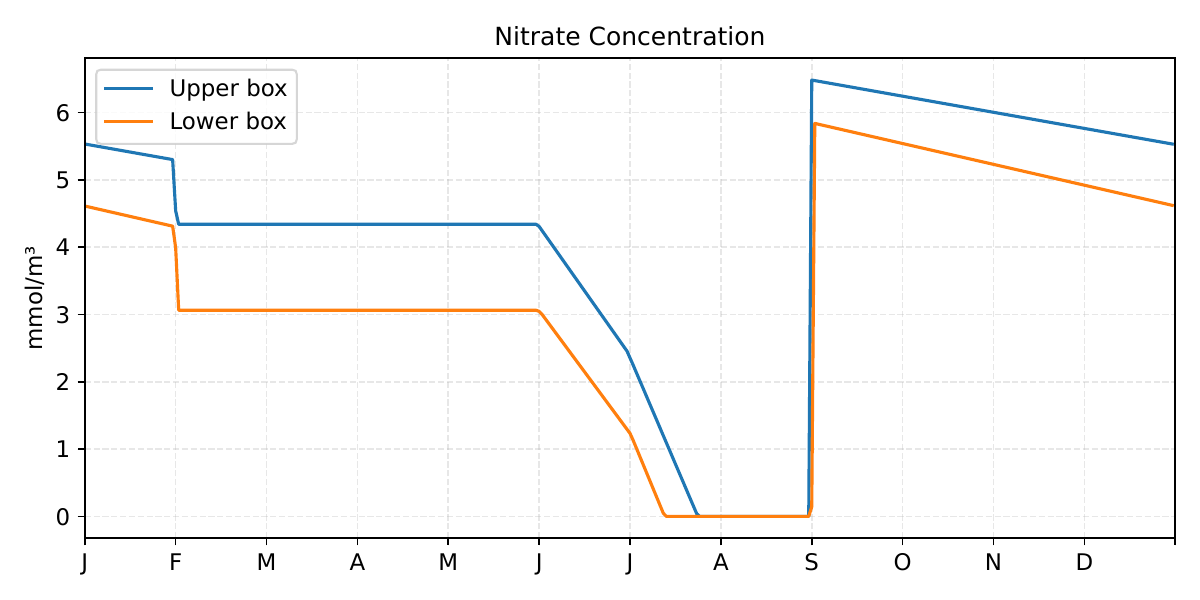}    
\end{center}
\label{fig:dataNitrate}

\smallskip
\textbf{Source:} the author (2026).

\smallskip
\textbf{Data:} \citeauthor{Machado1997ParanaguaBay} (\citeyear{Machado1997ParanaguaBay}).
\end{figure}

\begin{figure}[htpb!]
\caption{PEC daily data estimates on chlorophyll-a concentrations over a one-year period.}

\vspace{-0.6cm}
\begin{center}
\includegraphics[scale=0.8]{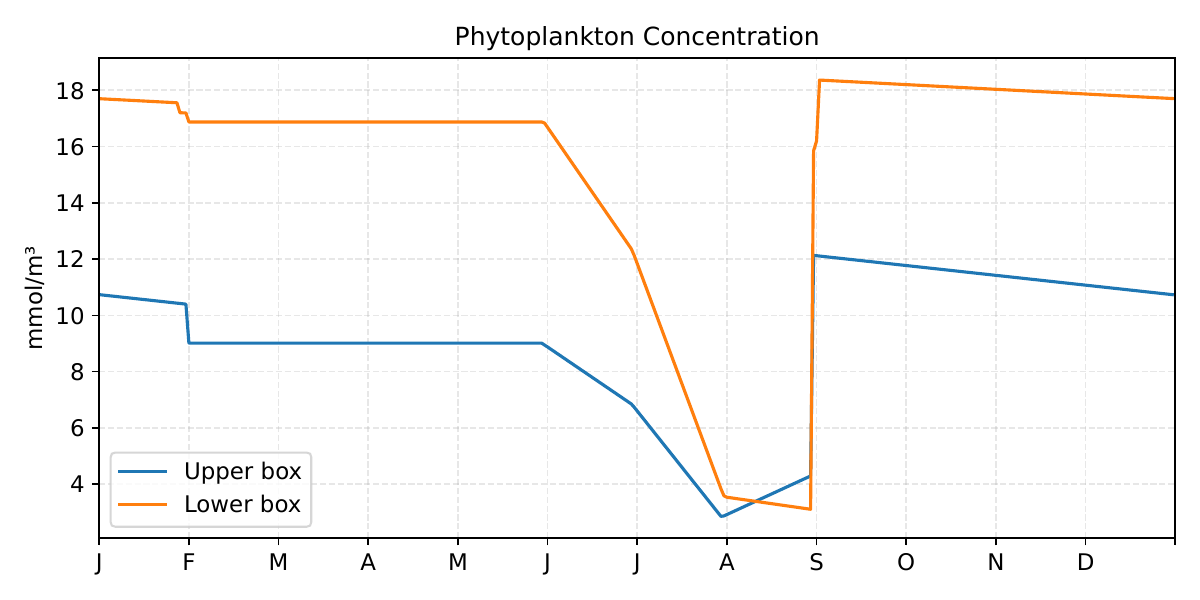}    
\end{center}
\label{fig:dataPhy}

\smallskip
\textbf{Source:} the author (2026).

\smallskip
\textbf{Data:} \citeauthor{Machado1997ParanaguaBay} (\citeyear{Machado1997ParanaguaBay}).
\end{figure}

\vspace{0.2cm}
Recalling that our initial objective was to model the NP dynamics in the PEC and considering that in our circulation model these processes only occur in the upper box, the coupling of the circulation model to the conceptual model we developed before is represented in Figure \ref{modelo}.

Considering the model schematically represented in Figure \ref{modelo}, we obtain the equations for the concentration variations of nutrient (\ref{EstuarioN}) and phytoplankton (\ref{EstuarioP}) in the upper box and consider the stationary model (\ref{estacionario}).

\begin{align}
\dfrac{d C_\mathrm{N}(t)}{dt} ~=~~ & \dfrac{0.588 \ C_\mathrm{N}^\text{river}(t) \cdot Q_\text{river}(t)}{\text{Vol}_\text{box}}+\dfrac{C_\mathrm{N}^\text{low}(t) \cdot Q_\text{ocean}(t)}{\text{Vol}_\text{box}}+ r \cdot \lambda \cdot C_\mathrm{PHY}(t) \nonumber\\ \nonumber \\
& - \alpha(T(t)) \ \beta(S_\text{up}(t)) \  \dfrac{V_{\max} \cdot C_\mathrm{PHY}(t) \ C_\mathrm{N}(t)}{C_\mathrm{N}(t)+K} -  \dfrac{C_\mathrm{N}(t) \cdot Q_\text{ebm}(t)}{\text{Vol}_\text{box}} ~.  \label{EstuarioN}   
\end{align}

\begin{align}
\dfrac{d C_\mathrm{PHY}(t)}{dt} ~=~~ & C_\mathrm{PHY}(t) \cdot \left(  \alpha(T(t)) \ \beta(S_\text{up}(t)) \ \dfrac{V_{\max} \cdot 
 C_\mathrm{N}(t)}{C_\mathrm{N}(t)+K} - \lambda - \dfrac{Q_\text{ebm}(t)}{\text{Vol}_\text{box}} \right) + \nonumber \\ \nonumber \\
 & +\dfrac{C_\mathrm{PHY}^\text{low}(t) \cdot Q_\text{ocean}(t)}{\text{Vol}_\text{box}}~. \label{EstuarioP}
\end{align}

\begin{align}
\dfrac{d C_\mathrm{N}(t)}{dt} ~=~~ & 0 ~, ~ \dfrac{d C_\mathrm{PHY}(t)}{dt} ~=~ 0 ~. \label{estacionario}
\end{align}

Before presenting the results, we briefly introduce the concept of spin-up time (Box \ref{box:spinup}), which is used to interpret model equilibrium.

\vspace{0.5cm}
\begin{codebox}[label=box:spinup]{\textbf{Spin-up time in marine biogeochemical models}}
The spin-up time of a marine biogeochemical model is the period required for the model to reach a dynamic and thermodynamic state close to equilibrium under a fixed set of observational forcings \cite{sarmientogruber}.
When modeling an estuarine environment, the spin-up time tends to be shorter, as these environments are dynamic and variations occur on short time scales.

Mathematically, the spin-up time can be viewed as an integration interval long enough for the differential equations defining the model to reach equilibrium.

Conceptually, we can say that during the spin-up time, the model "forgets" its initial conditions and begins to respond only to observations and to itself.
\end{codebox}

\vspace{0.5cm}
\begin{figure}[htpb!]
\caption{A conceptual biogeochemical NP-model for PEC.}
\begin{center}
\tikzset{every picture/.style={line width=0.75pt}} 

\begin{tikzpicture}[x=0.75pt,y=0.75pt,yscale=-1,xscale=1]

\draw  [fill={rgb, 255:red, 183; green, 214; blue, 255 }  ,fill opacity=1 ] (60,20) -- (406.5,20) -- (406.5,111) -- (60,111) -- cycle ;
\draw  [fill={rgb, 255:red, 74; green, 144; blue, 226 }  ,fill opacity=1 ] (60,111) -- (406.5,111) -- (406.5,202) -- (60,202) -- cycle ;
\draw  [fill={rgb, 255:red, 126; green, 211; blue, 33 }  ,fill opacity=0.86 ] (251,55.19) .. controls (251,50.77) and (254.58,47.19) .. (259,47.19) -- (313,47.19) .. controls (317.42,47.19) and (321,50.77) .. (321,55.19) -- (321,79.19) .. controls (321,83.61) and (317.42,87.19) .. (313,87.19) -- (259,87.19) .. controls (254.58,87.19) and (251,83.61) .. (251,79.19) -- cycle ;

\draw  [fill={rgb, 255:red, 208; green, 2; blue, 27 }  ,fill opacity=1 ] (223.1,72.84) -- (223.1,76.02) -- (251,76.02) -- (251,82.36) -- (223.1,82.36) -- (223.1,85.53) -- (211,79.19) -- cycle ;
\draw  [fill={rgb, 255:red, 189; green, 16; blue, 224 }  ,fill opacity=0.49 ] (141,55.19) .. controls (141,50.77) and (144.58,47.19) .. (149,47.19) -- (203,47.19) .. controls (207.42,47.19) and (211,50.77) .. (211,55.19) -- (211,79.19) .. controls (211,83.61) and (207.42,87.19) .. (203,87.19) -- (149,87.19) .. controls (144.58,87.19) and (141,83.61) .. (141,79.19) -- cycle ;

\draw  [fill={rgb, 255:red, 208; green, 2; blue, 27 }  ,fill opacity=1 ] (238.9,61.53) -- (238.9,58.36) -- (211,58.36) -- (211,52.02) -- (238.9,52.02) -- (238.9,48.84) -- (251,55.19) -- cycle ;
\draw  [fill={rgb, 255:red, 255; green, 255; blue, 255 }  ,fill opacity=1 ] (125.1,73.37) -- (125.1,69.62) -- (61,69.62) -- (61,62.12) -- (125.1,62.12) -- (125.1,58.37) -- (140.1,65.87) -- cycle ;
\draw  [fill={rgb, 255:red, 255; green, 255; blue, 255 }  ,fill opacity=1 ] (391.44,73.37) -- (391.44,69.87) -- (321.1,69.87) -- (321.1,62.87) -- (391.44,62.87) -- (391.44,59.37) -- (405.1,66.37) -- cycle ;
\draw  [fill={rgb, 255:red, 255; green, 255; blue, 255 }  ,fill opacity=1 ] (174,47.37) -- (174,47.37) .. controls (174,39.83) and (180.11,33.72) .. (187.65,33.72) -- (391.1,33.72) -- (391.1,29.37) -- (406.1,37.27) -- (391.1,45.17) -- (391.1,40.82) -- (187.65,40.82) .. controls (184.03,40.82) and (181.1,43.75) .. (181.1,47.37) -- (181.1,47.37) -- cycle ;
\draw  [fill={rgb, 255:red, 255; green, 255; blue, 255 }  ,fill opacity=1 ] (185.1,102.6) -- (181.59,102.6) -- (181.59,159.37) -- (174.56,159.37) -- (174.56,102.6) -- (171.05,102.6) -- (178.07,89.32) -- cycle ;
\draw  [fill={rgb, 255:red, 255; green, 255; blue, 255 }  ,fill opacity=1 ] (182.5,159) -- (244.04,159) .. controls (270.05,159) and (291.14,137.92) .. (291.14,111.91) -- (291.14,99.9) -- (295.5,99.9) -- (287.5,86.56) -- (279.5,99.9) -- (283.86,99.9) -- (283.86,111.91) .. controls (283.86,133.9) and (266.03,151.72) .. (244.04,151.72) -- (182.5,151.72) -- cycle ;
\draw    (34.83,225.25) -- (35.17,182.7) ;
\draw [shift={(35.19,180.7)}, rotate = 90.46] [color={rgb, 255:red, 0; green, 0; blue, 0 }  ][line width=0.75]    (10.93,-3.29) .. controls (6.95,-1.4) and (3.31,-0.3) .. (0,0) .. controls (3.31,0.3) and (6.95,1.4) .. (10.93,3.29)   ;
\draw    (34.83,225.25) -- (75.19,225.68) ;
\draw [shift={(77.19,225.7)}, rotate = 180.61] [color={rgb, 255:red, 0; green, 0; blue, 0 }  ][line width=0.75]    (10.93,-3.29) .. controls (6.95,-1.4) and (3.31,-0.3) .. (0,0) .. controls (3.31,0.3) and (6.95,1.4) .. (10.93,3.29)   ;

\draw (171,58.19) node [anchor=north west][inner sep=0.75pt]   [align=left] {N};
\draw (281,58.19) node [anchor=north west][inner sep=0.75pt]   [align=left] {P};
\draw (17.19,172.88) node [anchor=north west][inner sep=0.75pt]   [align=left] {$\displaystyle w$};
\draw (79.19,228.7) node [anchor=north west][inner sep=0.75pt]   [align=left] {$\displaystyle u$};
\draw (19,46) node [anchor=north west][inner sep=0.75pt]  [font=\small] [align=left] {River\\input};
\draw (416,36) node [anchor=north west][inner sep=0.75pt]  [font=\small] [align=left] {Exporting\\to the ocean};
\draw (61,121) node [anchor=north west][inner sep=0.75pt]  [font=\small] [align=left] {\begin{minipage}[lt]{77.18pt}\setlength\topsep{0pt}
\begin{center}
Internal circulation\\transfer
\end{center}

\end{minipage}};
\draw (216,61.19) node [anchor=north west][inner sep=0.75pt]  [font=\small] [align=left] {SMS};

\end{tikzpicture}    
\end{center}
\label{modelo}

\smallskip
Red arrows represent biochemical interactions, while white arrows represent geophysical circulation transfers. $N$ and $P$ denote the mean concentration of the tracers nitrate and phytoplankton in the upper box, respectively. Arrows pointing towards a tracer box represent a source of the tracer, whereas arrows pointing outward from a tracer box represent a sink of the tracer.

\smallskip
\textbf{Source:} the author (2026).
\end{figure}
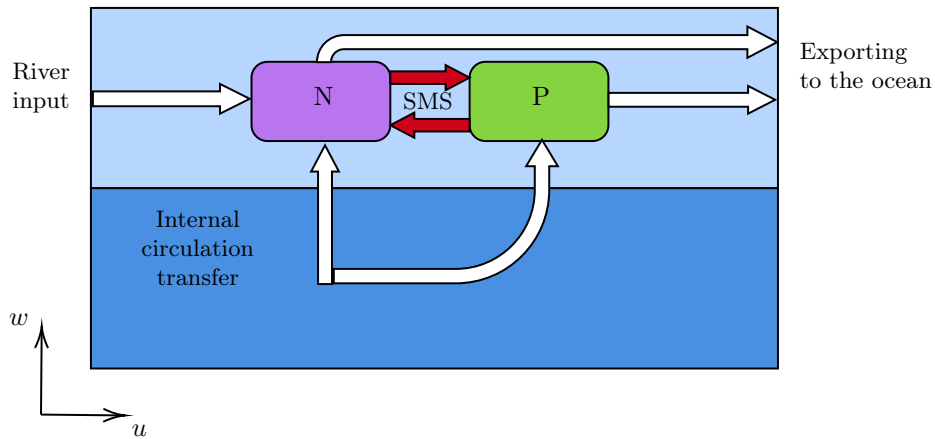

\section{Results and discussion}

This section summarizes the results obtained with the conceptual NP model developed for PEC. The simulations aim to assess the model’s internal consistency and its ability to reproduce the expected seasonal behavior of nutrients and phytoplankton in the PEC environment. Rather than providing quantitative forecasts, the results serve to evaluate how well the simplified structure captures the dominant processes and interactions described in previous field and modeling studies.

This model was implemented and run in Python with some adjustments (see the \textit{Appendix}). 

Although the model presented employs several simplifications and rough estimates, the results obtained are consistent with the available literature data. In Figure \ref{grafico1}, we present a forecast of the concentrations, in mmol m$^{-3}$, of nitrate and phytoplankton in the upper box of the estuary throughout one year. As expected, phytoplankton concentration is generally higher during the dry season than during the rainy season. Estimates for nitrate concentration were also close to the observational data.

According to Mizerkowski et al. \cite{Mizerkowski2012}, there is no consistent relation between the volume of water discharged by rivers in the PEC and the amount of nutrients. Therefore, nutrient concentrations in the rivers that drain the PEC exhibit significant variability, necessitating weekly or even daily monitoring for improved model accuracy. This would also make it possible to identify the input of nutrients transported from the coast to the estuary by rainwater and possibly find a relationship between rainwater discharge and nutrient input in the PEC.
The salinity of the estuary can vary due to precipitation, evaporation, and discharge of water and salts through rivers. Due to this change, salinity can damage cell membranes and change the density of seawater and, therefore, the buoyancy of certain organisms due to the water balance of cells \cite{Garrison2017}, so the data availability for the average salinity in PEC waters is essential for the purposes of limiting phytoplankton growth.

Low nutrient concentration is the main limiting factor for phytoplankton growth in our model, as any excess will be converted into biomass, while nitrate concentration values below 0.00037  mmol m$^{-3}$ inhibit the growth of phytoplankton \cite{sarmientogruber}. The the primary reason for the time-step in our model being counted in days, rather than fractions of a day, is that representing variations in phytoplankton growth and mortality at different daily times would necessitate a significantly greater degree of complexity. The main sources of nitrate for the estuary are the rivers that drain their respective watersheds, the ocean itself through advection and turbulent mixing -- which is not explored in our model, due to complexity increasing --  and the recycling of nearby photoautotrophic organisms \cite{sarmientogruber,carbonbiogeochem}. Furthermore, anthropogenic input is a crucial factor to consider when modeling estuarine environments (see Box \ref{box:humaninput}). 

\begin{figure}[htpb!]
\caption{NP-model: Predicted behavior for the nitrate and phytoplankton concentrations along one year on PEC.}
\begin{center}
\includegraphics[width=0.8\linewidth]{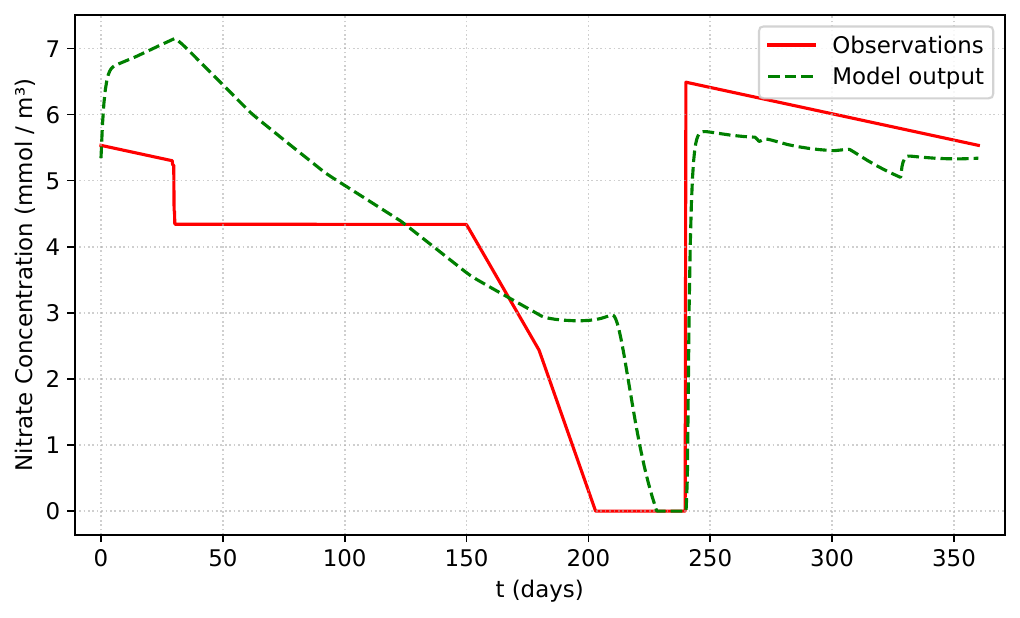}

\vspace{0.3cm}
\includegraphics[width=0.8\linewidth]{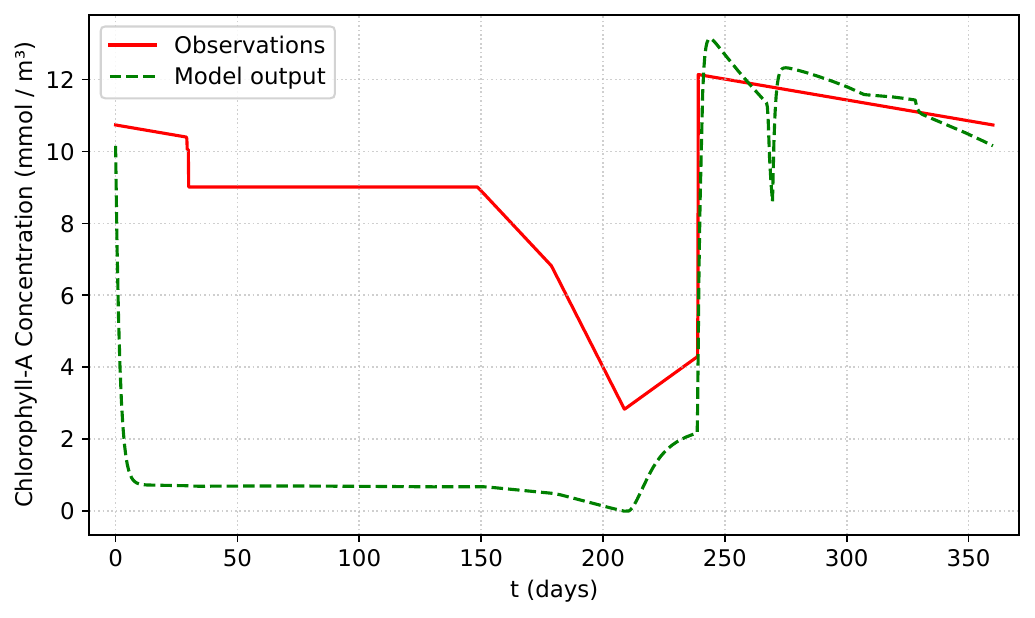}    
\end{center}
\label{grafico1}

\smallskip
Considering the model defined by equations (\ref{EstuarioN}) - (\ref{estacionario}), the settings of Table \ref{ParametrosEstuario}, and the data presented in figures \ref{fig:dataSal} - \ref{fig:dataPhy}.

\smallskip
\textbf{Source:} the author (2026).
\end{figure}

\begin{codebox}[label=box:humaninput]{\textbf{Human impacts on estuarine environments}}
Estuaries are bodies of water characterized by the interconnection between marine and river environments. They play a fundamental role in human development, as, in addition to providing natural protection for ports and industrial facilities by facilitating access to food and energy resources, estuaries also offer a range of ecological benefits, including nurseries for various marine species, water filtration, and flood protection \cite{Lana2001,Brandini1985,Procopiak}.
In this context, mathematical modeling is a powerful tool for studying the impacts of human occupation on estuaries and developing effective measures for their preservation \cite{Verri2020,modelling,Gruber2019}.
The intensification of human occupation in areas close to estuaries generates a series of environmental impacts, putting the natural balance of these ecosystems at risk. Pollution caused by the discharge of domestic and industrial effluents without adequate treatment, in addition to the flow of pesticides and other contaminants, is called \textit{anthropogenic input} and contributes to the degradation of water quality in estuaries, which can lead to eutrophication, death of marine organisms, and harmful proliferation of algae \cite{Martins2010,Martins2015,Mizerkowski2012,Perez2020,Zehr2002}.
This disturbance of the natural balance of estuaries can have serious consequences for the environment and for the communities that depend on these ecosystems for their sustenance, as it can lead to a decrease in fishing productivity and the loss of habitat for threatened species \cite{Canuel2016,Dan2020}.

Regarding the nitrogen input into PEC coming from human activity, we can outline two main types: Diffuse sources  come from surface runoff from agricultural and urban areas, from the fertilizer industry and docks in ports (Figure \ref{grafico2}), while point sources (in concentrated form) of domestic or industrial origin come from sewage that flows into water bodies without prior treatment \cite{Mizerkowski2012,Rodrigues2013}.

    Recent historical reconstructions of the PEC reveal strong anthropogenic influences, where activities such as agriculture, urban tourism, river diversion, and maintenance have reshaped the local hydrodynamics and modified its nutrient deposition patterns \cite{marines2023}, suggesting the need to consider diffuse and point sources of anthropogenic nitrate supply when developing an NP model for the PEC.
\end{codebox}


\vspace{0.5cm}
The euphotic zone in the coastal zone may be restricted to just a few centimeters of the water column, as is the case in estuaries, where the waters are turbid by the presence of sediments, particulate matter, and humic acids \cite{Giane2013}. This fact validates the design of our model, which considers the growth of phytoplankton only in the upper box. In addition to light being a limiting factor in this case, the scarcity of nutrients also plays a prominent role, regulating the development of the phytoplankton community and, consequently, the other organisms that make up the food web in that habitat \cite{Calijuri2013}.
The zone of maximum turbidity (ZMT) present in partially mixed estuaries, such as the PEC, functions as a sediment trap, as it maintains the fine particles that make up the sediment for a long time, generally composed of material particulates, nutrients, organic matter, and pollutants, as they constitute the part that adheres or is adsorbed to these substances. Furthermore, the change in the ZMT can erode the bottom when the tide is falling, and at high tide, it can deposit this material \cite{Zem2008}. This fact supports the choice of the parameter for phytoplankton remineralization rate as 70\%, even though the model does not explicitly consider flows indicating this rate.

The phytoplankton concentration is estimated through the analysis of the chlorophyll-a concentration in the water, and we use here the seasonal values presented by \cite{Brandini1985} for comparison with the results of our model. The results obtained by simulating these scenarios with our model confirmed these assumptions. Furthermore, studies indicate that a longer residence time, salinity, and anthropogenic input directly influence the increase in phytoplankton concentration, which is an indication of environmental imbalance, while frequent rains decrease salinity, nutrient concentration, and residence time of water in an estuary, which can cause a drastic reduction in the total phytoplankton concentration \cite{Barroso2016}.
Studies indicate that a longer residence time, salinity, and anthropogenic input directly influence the increase in the phytoplankton concentration, which is an indication of environmental imbalance, while frequent rains decrease the salinity, nutrient concentration, and residence time of water in an estuary, causing a drastic reduction in the local phytoplankton concentration \cite{Barroso2016}.

Mathematical modeling is a helpful tool for studying the dynamics of estuaries. However, in order to make an effective contribution, it is necessary to know the processes involved to facilitate the creation and maintenance of the model, for example, by understanding how effectively nitrate limits phytoplankton productivity, because this conclusion is based on a few studies that evaluate nitrate limitation indirectly \cite{Zehr2002}.

\begin{figure}[htpb!]
\caption{Diffuse anthropogenic supply simulation on PEC.}
\begin{center}
\includegraphics[width=0.8\linewidth]{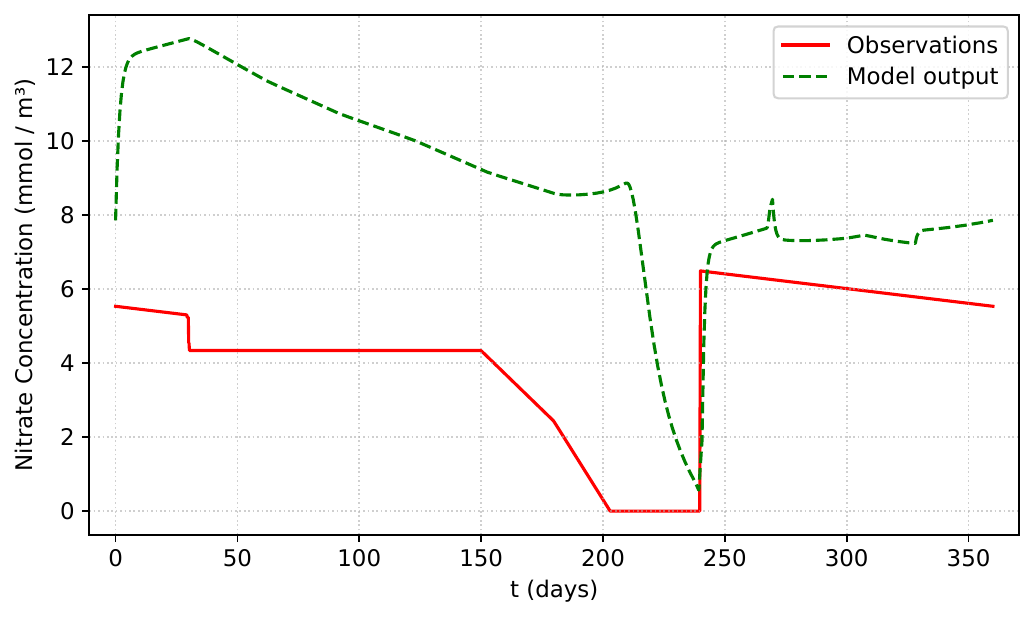}

\vspace{0.3cm}
\includegraphics[width=0.8\linewidth]{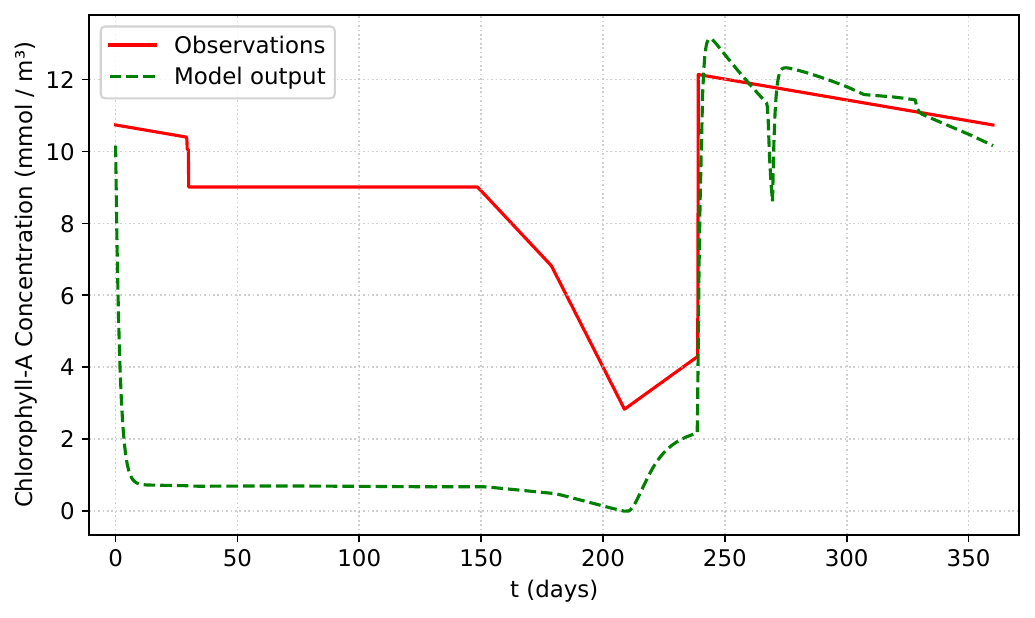}    
\end{center}
\label{grafico2}

\smallskip
The diffuse supply of nutrients directly influences a greater concentration of phytoplankton. To illustrate this situation, we increased the concentration data input from the rivers in 10  mmol m$^{-3}$, while keeping all the other model settings and equations as in Figure \ref{grafico1}. The plot shows the model's forecast during the year.

\smallskip
\textbf{Source:} the author (2026).
\end{figure}

\section{Model limitations and development opportunities}

In this modeling study, the flows were determined based on assumptions, rather than direct measurements. As our goal is to exemplify the development and application of a simple biogeochemical model when considering assumptions for data extrapolation, even with some data available on the river fluxes, there are still unpredictable factors, such as the input of rainwater and non-mapped water sources. In this context, the flow $Q_\text{ebm}$ was estimated based on the average residence time of the estuary, while the flow $Q_\text{river}$ was calculated proportionally to the river discharge with the greatest contribution to the system, including the estimated water input from rainfall. The flows $Q_\text{ebm}$, $Q_\text{river}$, and $Q_\text{ocean}$ were obtained from equations that relate flow rate and salinity, also using observed salinity data in the estuary.
In the context of modeling, incorporating data from direct measurements of flows in the estuary could increase the model's representativeness in relation to the circulation dynamics of the PEC.

A promising possibility would be to replace the simplified two-box model with a more detailed circulation model of the PEC. This more complex model could include a more refined spatial discretization and, ideally, the use of time series of observational data of the flows, which involves increased computational cost.

Another parameter estimated in this study was the riverine input of nitrate into the PEC waters. Although precisely quantifying this input is challenging, the estuary model can be adapted to generate daily estimates of nitrate input, considering the sum of both fluvial and anthropogenic contributions. The use of a more sophisticated circulation model could improve the accuracy of these estimates.

The availability of up-to-date and sufficient data for model execution, in order to reduce the need for data extrapolation, would be ideal for better output performance. The modeling of SMS functions is also an issue that can be revisited in the future.

Finally, we note that the data used for validating this model come from different sources and refer to different years. This impairs the representativeness of the model because, ideally, all data would be related to the same days of the same year. Even so, when simulating different theoretical scenarios, the model showed the expected behavior consistent with what is described in the literature. During dry periods, there was an increase in phytoplankton, which agreed with the expected estimate. Recent sedimentary reconstructions \cite{marines2023} suggest that long-term changes in the PEC are closely linked to land-use evolution, supporting the notion that simplified conceptual models can still capture the dominant local processes shaping nutrient and phytoplankton dynamics, while showing that the PEC has been well-monitored from a geochemical perspective, which is an interesting aspect to be explored further. Thus, rather than aiming for predictive accuracy, this model is intended as a computational testbed for the optimization and calibration procedures presented in \textit{Chapter \ref{chapterCalibratingPEC}}. Although more complex biogeochemical models could represent the system in greater detail, they also require higher computational resources and remain constrained by limited observational coverage.

\chapter{Model calibration and derivative-free optimization}\label{chapterdfols}

In this chapter, we present the optimization strategies used to calibrate the model parameters, which will be applied in the following chapter. Parameter calibration is treated as an adjustment process: we seek parameter values that make the model outputs match the observed data as closely as possible. 
These methods are particularly suitable for problems where the objective function is expensive to evaluate, noisy, or non-smooth, which exploration is a natural progression of the work presented here. We also discuss how to define a misfit function that properly measures the difference between model results and observations, considering normalization, regional weighting, and data uncertainty. \\
First, we review the basic concepts of optimization, including local and global minima, constrained and unconstrained problems, and how multi-objective problems can be reformulated as least-squares problems. Next, we describe the construction of the misfit function and the techniques used to minimize it. At last, we review the Derivative-Free Optimization for Least Squares (DFO-LS) algorithm \cite{principal, DFOLSmops2022}. Practical aspects are also discussed, setting the stage for the applications presented in the next chapter.

\section{Basic concepts of optimization}
Optimization is a branch of Mathematics that deals with problems of minimization or maximization of an objective under possible constraints. Such a problem must be described within a general formulation as 


\begin{align}
& \text{minimize} ~~  f( \mathbf{\omega} )\nonumber\\
& \text{subject to:} ~~   \mathbf{\omega} \in U \subset \R^n  \label{objective1constrained}
\end{align}
where $U=\left\lbrace  \mathbf{\omega} \in \R^n \ : \ g_i( \mathbf{\omega} )\leq 0 \ ,\ h_j( \mathbf{\omega} )=0\ ;\ i=1,\dots,k; j=1,\dots,\ell \right\rbrace$ is the feasible set, $f: \R^n \rightarrow \R$ is the objective function and $g:\R^n \rightarrow \R^k$, $h:\R^n \rightarrow \R^\ell$ are the functions that define the constraints.

There are two standard solutions to \eqref{objective1constrained}: local and global minimizers. We can define them as follows.

\begin{defi}[Local and global minimizers]\label{deflocalglobalmin}
    Considering the minimization problem (\ref{objective1constrained}), we say that $\Bar{ \mathbf{\omega} }_L$ is a local minimizer or a local solution for the problem if, and only if
    $$f(\Bar{ \mathbf{\omega} }_L)\leq f( \mathbf{\omega} ) , ~ \forall  \mathbf{\omega} \in B(\Bar{ \mathbf{\omega} }_L, \epsilon) \cap U ,$$
    for some $\epsilon>0$. Also, we say that $\Bar{ \mathbf{\omega} }_G$ is a global minimizer or a global solution for the problem (\ref{objective1constrained}) if, and only if
    $$f(\Bar{ \mathbf{\omega} }_G)\leq f( \mathbf{\omega} ) , ~ \forall  \mathbf{\omega} \in  U .$$
\end{defi}

For practical reasons, it is unusual to search for a global minimizer for an optimization problem. In practice, finding global minimizers is rarely guaranteed; thus, a common strategy is to explore multiple local minimizers and select the best among them.

\vspace{0.5cm}
It is also possible to use previous knowledge to estimate which optimization algorithms are faster in calculating a solution for the optimization problem. This would be done by analyzing the features of the functions $f$, $g$, and $h$ (for example, smoothness, noise, convexity) and finding an algorithm known by producing good results related to such features, but also by empirical testing on representative problems.

\vspace{0.5cm}
Another situation arises when we are dealing with \textit{multi-objective optimization}. In this case, we would want to find a local optimizer for the problem (\ref{objective1constrained}), with the addition that $f: \R^n \rightarrow \R^m$, for some $m \in \mathbb{N}$. Several optimization techniques are directly applicable to such problems, although not every optimization method is specifically designed to handle multi-objective optimization. Yet, sometimes, it is possible to convert a multi-objective optimization problem into a \textit{single-objective} optimization problem, as is the case of a least squares problem. This reformulation will be explored in the following sections.

\section{Calibrating parameters from observations}

In many applications, the forward model to be calibrated is, in practice, treated a \emph{black-box}: the optimization algorithm sees only inputs (parameters) and outputs (diagnostics), but does not have access to the analytical form of the objective function. In ocean biogeochemistry, for example, the inputs are parameter values, and the outputs are predicted tracer fields. Calibration then seeks parameters whose model outputs best match an observation data set.



The known strategies for automatic parameter calibration \cite{principal, DFOLSmops2022} involve optimizing an objective function known as the \textit{misfit function}. This function quantifies the difference between model outputs and observation data and can be seen as a measure of calibration quality. We outline a generic construction and weighting scheme suitable for global applications (and used in complex settings such as \cite{DFOLSmops2022}).

Let $\omega$ be a vector whose entries are the parameters we want to calibrate; $\mathcal{O}_{q,i,j}$ the observed concentration of the tracer $A_q$ on a subregion $R_{i,j}$, with volume $V_{i,j}$, inside a larger region $R_j$, with volume $V_j$, at a fixed moment in time; 
$m_{q,i,j}( \mathbf{\omega} )$ the model prediction for the concentration of the tracer $A_q$ on $R_{i,j}$, at the moment corresponding to the observation $\mathcal{O}_{q,i,j}$; where $ \mathbf{\omega} \in \R^{\ell_\text{parameters}}$, $q=1, \dots, \ell_\text{tracers}$, $i=1, \dots, \ell_\text{j}$, $j=1, \dots, \ell_\text{regions}$. We want to simultaneously minimize a set of functions $f_{q,i,j}: \R^{\ell_\text{parameters}} \rightarrow \R_+$, defined as:
\begin{equation}
    f_{q,i,j}( \mathbf{\omega} )= \lvert m_{q,i,j}( \mathbf{\omega} )- \mathcal{O}_{q,i,j} \rvert ~,
\end{equation}
for $ \mathbf{\omega} \in \R^{\ell_\text{parameters}}$, $q=1, \dots, \ell_\text{tracers}$, $i=1, \dots, \ell_j$, $j=1, \dots, \ell_\text{regions}$. Taking into account the fact that the parameters must lie on some defined intervals to make sense, there is also a set of constraints that must be satisfied, which we will not define yet, but call it $U \subset \R^{\ell_\text{parameters}}$. We could handle this task by using a strategy for simultaneously solving the constrained minimization problems:
\begin{align}
    & \text{minimize}  ~~ f_{q,i,j}( \mathbf{\omega} )\nonumber\\
    & \text{subject to:}  ~~  \mathbf{\omega} \in U ~, \label{minMisfit00}
\end{align}
for $ \mathbf{\omega} \in \R^{\ell_\text{parameters}}$, $q=1, \dots, \ell_\text{tracers}$, $i=1, \dots, \ell_\text{j}$, $j=1, \dots, \ell_\text{regions}$.

\vspace{0.5cm}
One can notice that the problem (\ref{minMisfit00}) may be converted into a \textit{least-squares problem}. This strategy reduces the problem's dimensionality, transforming it from a multi-objective optimization problem into a single-objective one. Considering that we can accommodate the problem constraints within the strategies to solve the optimization problem, we will now describe how to rewrite the objective function of the problem (\ref{minMisfit00}) as one corresponding to a least-squares constrained problem:
\begin{align}
    & \text{minimize}  ~ f(\omega) =  \sum_{q=1}^{\ell_\text{tracers}} r_q(\omega)^2 \nonumber\\
    & \text{subject to:}  ~ \omega \in U \subset \R^{\ell_\text{parameters}} \label{lqproblem}
\end{align}
where $r_q$ is the residual function corresponding to the tracer $A_q$, which will be better described in the following. We now discuss details on the residual formulation presented in \cite{onesize}. First, consider a tracer $A_q$, a region $R_j$ with its corresponding subregions $R_{i,j}$, and a set of corresponding observations. A first attempt to write a residual function $r_{q,j}(\omega)$ could be:

\begin{equation}
    r_{q,j}(\omega)= \sqrt{\sum_{i=1}^{\ell_j} f_{q,i,j}( \mathbf{\omega} )^2 }
\end{equation}

\vspace{0.5cm}
Mathematically, this is a valid attempt. But, numerically, it may lead to \textit{ill-conditioning} due to several factors. For instance, in the context of marine biogeochemistry models, the observations for some tracers can be typically very small in comparison to other tracers, so the minimization of $f(\omega)$ may end up not taking into account these observations. Moreover, the impact of the observations over the misfit function can be disproportional to the volume of the region sampled. To prevent both of these possible outcomes, we can consider the attribution of weights, leading to the following formulation for the residual function:

\begin{equation}
    \Tilde{r} _{q,j}(\omega)= \sqrt{\dfrac{V_j}{\sum_{j=1}^{\ell_\text{regions}} V_j}} \dfrac{\sqrt{\sum_{i=1}^{\ell_j} f_{q,i,j}( \mathbf{\omega} )^2 \frac{V_{i,j}}{V_j}}}{\sum_{i=1}^{\ell_j} \mathcal{O}_{q,i,j} \frac{V_{i,j}}{V_j}}
\end{equation}

\noindent where it is easy to notice that the weightings satisfy
$\sum_{j=1}^{\ell_\text{regions}} \dfrac{V_j}{\sum_{j=1}^{\ell_\text{regions}} V_j}=1$ ~ and ~ $\sum_{i=1}^{\ell_j} \frac{V_{i,j}}{V_j}=1$.

\vspace{0.5cm}
Thus, we may define the residual function corresponding to a tracer $A_q$ as:

\begin{equation}
r_q(\omega)=\sqrt{\sum_{j=1}^{\ell_\text{regions}} \Tilde{r} _{q,j}(\omega)^2} ~.    
\end{equation}

\vspace{0.5cm}
 During the calibration and validation steps of the biogeochemical model, the sparsity of data and noise pose significant obstacles to obtaining precise and reliable tracer observations. While the first problem can be mitigated by filtering on the optimized parameters, the latter can be handled by modifying the function $f_{q,i,j}$ to account for some noise. If we consider the modification of $f_{q,i,j}$ of the function $f_{q,i,j}^\text{Noise}:\R^{\ell_\text{parameters}} \rightarrow \R_+$ :

\begin{equation}
    f_{q,i,j}^\text{Noise}( \mathbf{\omega} )= \lvert m_{q,i,j}( \mathbf{\omega} )- \left( \mathcal{O}_{q,i,j} + \varepsilon_{q,i,j} \right) \rvert ~,
\end{equation}

\noindent we can now solve a minimization problem that accounts for uncertainties over the observations by solving the problem:
\begin{align}
    \text{minimize} & ~~ f^\text{Noise}( \mathbf{\omega} ) = \sum_{q=1}^{\ell_\text{tracers}} r_q^\text{Noise}(\omega)^2 \nonumber \\
    \text{subject to:} & ~~  \mathbf{\omega} \in U ~, \label{ProblemMisfit01}
\end{align}
where the residual function corresponding to a tracer $A_q$ is defined as:

\begin{equation}
r_q^\text{Noise}(\omega)=\sqrt{\sum_{j=1}^{\ell_\text{regions}} \Tilde{r} _{q,j}(\omega)^2} ~,     
\end{equation}

\noindent with 

\begin{equation}
    \Tilde{r}_{q,j}^\text{Noise}(\omega)= \sqrt{\dfrac{V_j}{\sum_{j=1}^{\ell_\text{regions}} V_j}} \dfrac{\sqrt{\sum_{i=1}^{\ell_j} f_{q,i,j}^\text{Noise}( \mathbf{\omega} )^2 \frac{V_{i,j}}{V_j}}}{\sum_{i=1}^{\ell_j} \mathcal{O}_{q,i,j} \frac{V_{i,j}}{V_j}} ~.
\end{equation}

\vspace{0.5cm}
In applied studies, specific index sets and data choices instantiate $f^\text{Noise}$; diverse DFO algorithms \cite{DFOLSoriginal,DFOreview1,Larson2019} have been used to minimize \eqref{ProblemMisfit01}. Due to their advantages, these methodologies have already been applied to optimize models in more complex settings, such as in \cite{Kriest2012}, where a sensitivity analysis reduced parameter dimension, and in \cite{principal, DFOLSmops2022}, where the optimization strategies Covariance Matrix Adaptation Evolution Strategy (CMA-ES) and Derivative Free Optimization for Least Squares (DFO-LS) delivered competitive fits under limited computational budgets.

\section{Performance measure for an optimization algorithm}

The performance of an optimization algorithm in model calibration depends mainly on two aspects: its computational cost and the quality of the fit it produces.
At the same time, the accuracy with which the calibrated parameters reproduce the observed data must be evaluated using problem-dependent metrics, such as the least squares error or observational parameter recovery, as discussed in the following chapter.
A reliable algorithm must therefore balance these two dimensions: obtaining a sufficiently accurate fit while keeping the computational costs compatible with the application's constraints. In well-defined calibration problems, this translates into convergence to a consistent set of parameters that performs well across all chosen metrics.\cite{onesize}.

When comparing the performance of two algorithms, several common situations must be considered. Some of these situations are directly related to the modeling and implementation of the fitting strategy, such as the model's or optimizer's ability to handle noise in the observational data and their sensitivity to parameter variations. 

Another common problem is the scarcity of observational data, especially in the context of marine biogeochemistry. In the context of global ocean modeling, collecting data in remote environments, such as the deep ocean and oceanic regions near the poles, poses a significant challenge. Even when data collection is possible, observations (such as tracer concentrations) can vary significantly in time and space, making it difficult to obtain a complete dataset. Still, even in ideal situations in the aforementioned contexts, data may not be available because they are not yet being monitored. These limitations should also be taken into account when evaluating the performance of the algorithms.

A standard comparison protocol involves running both algorithms under identical conditions and recording the real-time iteration count and error obtained. These methodologies have already been applied to systematically compare the CMA-ES and DFO-LS algorithms in the context of global ocean biogeochemical modeling.\cite{DFOLSmops2022}.

\vspace{0.5cm}
\begin{codebox}[label=box:cmaes]{\textbf{An optimization alternative: evolutionary computation}}

Evolutionary computation algorithms are optimization methods inspired by the process of natural selection. They work by iteratively generating, evaluating, and improving a population of candidate solutions to a problem, mimicking biological evolution through mechanisms such as mutation, recombination, and selection. In this context, contrasting with the DFO-LS algorithm, which is based on a deterministic trust-region approach, the Covariance Matrix Adaptation Evolution Strategy (CMA-ES) \cite{cmaes01,cmaes02} is a stochastic algorithm that relies on systematic sampling and the gradual improvement of the best-performing candidates for the parameters being optimized. Figure \ref{fig:cmaes} shows a schematic representation of the CMA-ES algorithm iteration near convergence.
\end{codebox}

\begin{figure}[h]
\caption{Iterative scheme of the CMA-ES algorithm.}
\begin{center}
\input{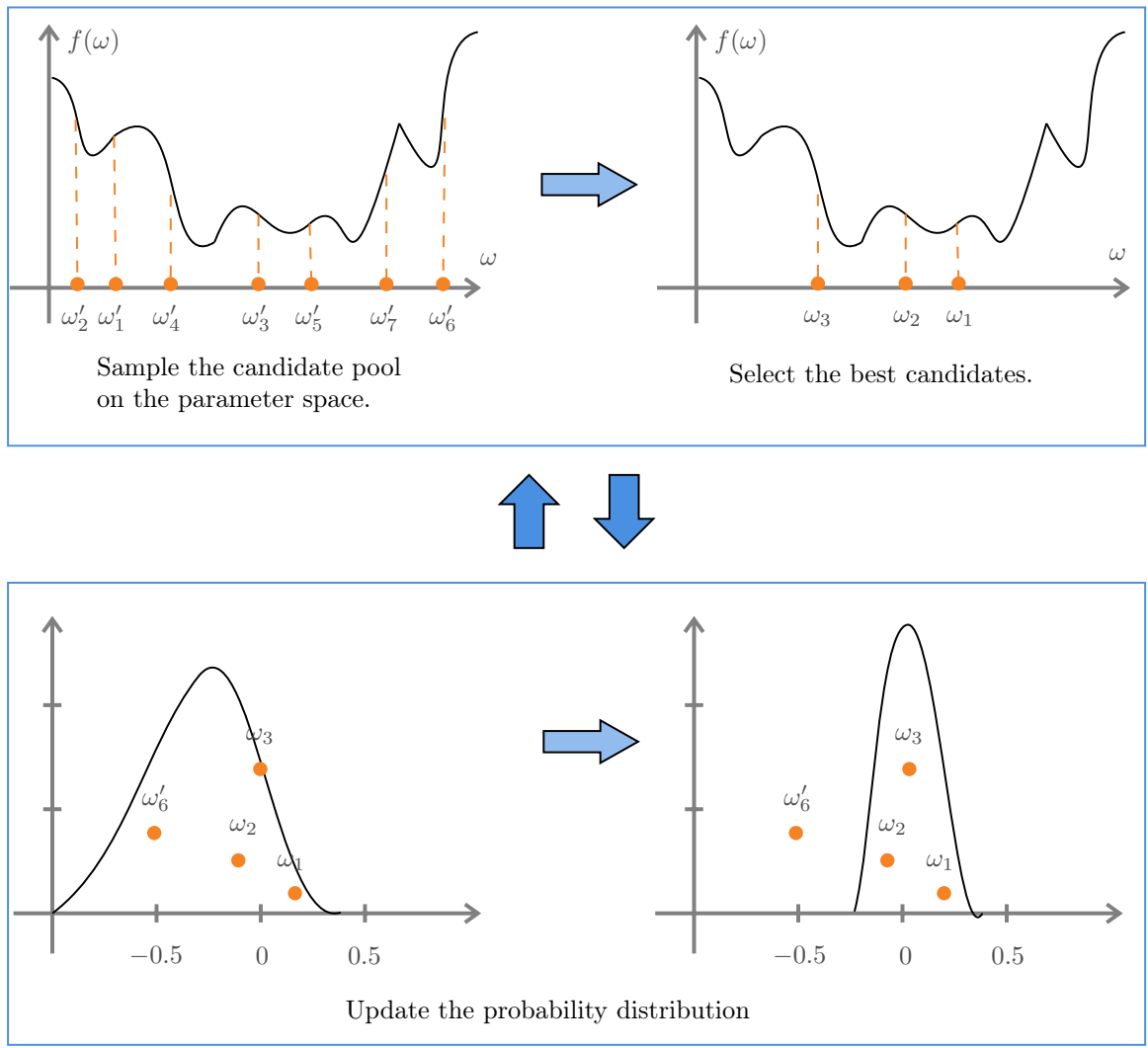}    
\end{center}
\label{fig:cmaes}

\smallskip 
\textbf{Source:} the author (2026).
\end{figure}

\section{DFO-LS: Derivative Free Optimization for Least-Squares}

An optimization strategy recently explored for parameter calibration is DFO-LS, which is a derivative-free trust-region method explicitly tailored to nonlinear least-squares objectives. It builds local surrogate models of the residual vector from function evaluations only and chooses steps by minimizing a quadratic approximation within a trust region, which can markedly reduce the number of expensive model evaluations when the objective is a sum of squared residuals \cite{DFOGN,DFOLSmops2022}. Thanks to its residual-aware design, DFO-LS is well suited to noisy, nonconvex black-box calibrations with simple constraints.

In the following, we outline the core elements of DFO-LS, as interpolation-set management, linear residual modeling, quadratic objective modeling, trust-region subproblems, acceptance/updates, and restarts. We also note that methods of this class have been applied in complex biogeochemical settings \cite{DFOLSmops2022}, illustrating their utility beyond conceptual problems.

DFO-LS is a trust-region least-squares minimization method. Compared with evolutionary strategies (see Box \ref{box:cmaes}), DFO-LS exploits the least-squares structure explicitly, often reducing the number of expensive model evaluations on noisy, nonconvex black-box problems \cite{DFOLSmops2022}. Owing to these advantages (no gradients required, robustness to noise, and explicit exploitation of the residual structure), DFO-LS has been used to calibrate increasingly complex models. Applications include global ocean biogeochemistry (for example, MOPS) \cite{DFOLSmops2022}, climate and atmospheric chemistry models \cite{shoemaker2015}, among many more \cite{Wild_2015, ALARIE2021100011}.

Consider a continuous function $\mathrm{r}(\mathbf{\omega})$, where
\begin{align}
\mathrm{r}: & \R^n \rightarrow \R^m \label{Fdef}\\
& \mathbf{\omega} \mapsto (r_1, ... , r_m)  \nonumber
\end{align}
and the constraints $g(\mathbf{\omega}) \leq 0$ , $h(\mathbf{\omega})=0$, where
\begin{align}
g: & \R^n \rightarrow \R^{m'} \label{gdef}\\
& x \mapsto (g_1(\mathbf{\omega}), ... , g_{m'}(\mathbf{\omega})) \nonumber \\
h: & \R^n \rightarrow \R^{m''} \label{hdef}\\
& x \mapsto (h_1(\mathbf{\omega}), ... , h_{m''}(\mathbf{\omega})) \nonumber   
\end{align}

The DFO-LS method is an iterative algorithm that searches for a local solution for a problem equivalent to the least-squares one for $\mathrm{r}(\mathbf{\omega})$ under the constraints $g(\mathbf{\omega})\leq 0$, $h(\mathbf{\omega})=0$. That is, considering the parameter space 
\begin{equation}\label{Sdef}
U=\{ \mathbf{\omega} \in \R^n : g(\mathbf{\omega}) \leq 0, h(\mathbf{\omega})=0 \} \ ,    
\end{equation}
DFO-LS iteratively search for a local minimizer for the problem:
\begin{equation}\label{minProblem}
\min_{\mathbf{\omega} \in U} f(\mathbf{\omega})=\sum_{i=1}^{m}[r_i(\mathbf{\omega})]^2 .    
\end{equation}

In our setting, the objective $f$ is precisely the least-squares misfit, and feasibility is addressed within the trust-region framework.

DFO-LS is a development of the earlier Derivative-Free Gauss-Newton method (DFO-GN) \cite{DFOGN}, with enhancements in geometry management and restarts. The DFO-LS outline can be seen as:
\begin{enumerate}
    \item Choose an initial point  $\mathbf{\omega}_0 \in \R^n$.
    \item Choose $n$ points around $\mathbf{\omega}_0$.
    \item Iteratively:
    \begin{enumerate}[label=\arabic*.]
        \item Build a quadratic regression model for $f(\mathbf{\omega})$ around $\mathbf{\omega}_k$, based on the n+1 points already selected.
        \item Exclude one old point and add another using a trust-region strategy.
        \item Choose the best point in the pool to be $\mathbf{\omega}_{k+1}$.
    \end{enumerate}
\end{enumerate}
These steps are illustrated in Figure \ref{dfols}. Details of each component of the DFO-LS algorithm are presented next.

DFO-LS must be given a starting location within the parameter space to use for initialization.  The strategy chosen to increase confidence that DFO-LS had found a good minimizer was to start from multiple initial points. DFO-LS can also be restarted to improve global exploration: \emph{hard restarts} reinitialize the entire interpolation set at new locations; \emph{soft restarts} shift part of the set toward geometry-improving points \cite{DFOLSmops2022}.

\subsection{Building the quadratic regression model}

From a regression standpoint, DFO-LS builds its surrogate by posing a \emph{local least-squares fit} of the residual vector. Given an iteration $k$, we sample a set of points $Y_k = \left\lbrace y_1, \dots, y_p \right\rbrace$ such that $y_t = \omega_k + s_t$ and $\lVert s_t \rVert$ is small, for $t = 1, \dots, p$, that is, we sample a set of points in a close neighborhood of $\omega_k$. Based on these evaluations, we build the equivalent of a linear regression model for the misfit function gradient (which is not defined, since it is a black-box function) in a neighborhood of $\omega_k$. This approximation is then used to construct the equivalent of a Hessian for the misfit function at $\omega_k$. Based on both of these approximations, DFO-LS builds the equivalent of a second-order Taylor approximation for the misfit function in a neighborhood of $\omega_k$. Minimizing this local quadratic surrogate model for the misfit function defines the trust-region subproblem. This regression view clarifies both efficiency (one model serves all residuals) and robustness to noise (least-squares averaging), while keeping constraint handling natural via the trust-region geometry \cite{DFOGN, DFOLSmops2022}.

\begin{figure}[htpb!]
\caption{Iterative scheme of the DFO-LS algorithm.}
\begin{center}
\input{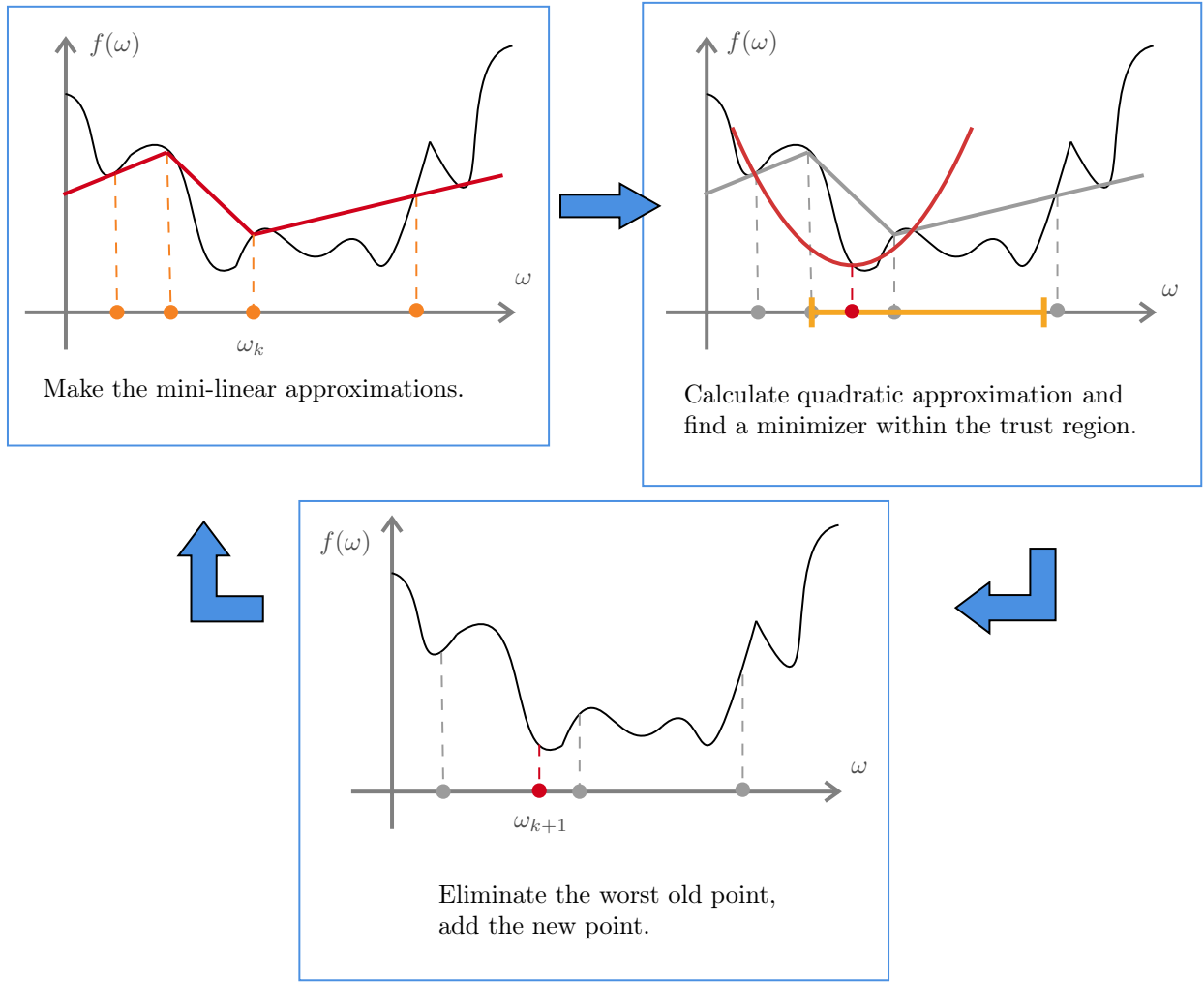}    
\end{center}
\label{dfols}

\smallskip
After evaluating the objective function at a set of points, the DFO-LS algorithm iteratively performs the sequence of steps represented by the upper-right panel, upper-left panel, and lower panel, respectively. In all panels, dots represent candidates $\omega$, the background black curve represents the plot of the black-box misfit function $f(\omega)$. In the upper-left panel, the dots represent each element of the candidate pool for the iteration $k$, where $\omega_k$ is the best fit at this iteration, while the red solid line segments represent the information used for building a quadratic local model for $f(\omega)$. In the upper-right panel, the red curve represents the quadratic local model for $f(\omega)$ in a neighborhood of $\omega_k$, while the interval marked by the yellow horizontal line represents the search bounds for its minimizer.
The highlighted red point in the lower panel represents the candidate considered the best fit at iteration $k+1$, while a removed gray point corresponds to the least fit at this iteration.

\smallskip
\textbf{Source:} the author (2026).
\end{figure}

Consider the functions $\mathrm{r}(\mathbf{\omega})$, $f(\mathbf{\omega})$, and the constraints set $S$ as described in (\ref{Fdef})-(\ref{minProblem}). The DFO-LS method chooses a trust region method to search for a good local solution\footnote{A better characterization of a good solution is given in the next section.} to the problem (\ref{minProblem}), which is equivalent to finding a solution to:
\begin{equation}\label{ls2}
    \min_{\mathbf{\omega} \in U} ~ \dfrac{1}{2} \ \lVert \mathrm{r}(\mathbf{\omega}) \rVert^2 ~ .
\end{equation}
We can notice that when $U=\R^n$, the problem (\ref{ls2}) is the least squares minimization one, which is an unconstrained minimization problem. Despite this, we will now see how DFO-LS seeks to solve a constrained problem.

Consider given $p$ points\footnote{Here $p$ is a positive integer, whose value can be set accordingly to different purposes. For example, when seeking exploration in the early iterations, it could be useful to choose $p<n$ for reducing the cost of the algorithm, while for exploitation it could be best to choose $p \geq n$. DFO-LS initialize with $p=n+1$ as default.} in the sample space $S$, that is, $\left\lbrace y_1, \dots, y_p \right\rbrace \in U$, where $p \in \N$, $p \geq n$, and a point $\mathbf{\omega} \in U$ (for simplicity of notation, in this section we are going to write just $\mathbf{\omega}$ instead of $\mathbf{\omega}^k$). First, we define, for convenience, $y_0 := \mathbf{\omega}$. We now have a set of $p+1$ points, $\left\lbrace y_0, y_1, \dots, y_p \right\rbrace \in U $.
So we are going to build a linear local model $m_{\ell}(\cdot)$ for the function $\mathrm{r}(\cdot)$ around $\mathbf{\omega}$, that is:
\begin{equation}\label{prob003}
\mathrm{r}(\mathbf{\omega}+s) \approx m_{\ell}(s) := \Bar{F}+\Bar{J} s ~,    
\end{equation}
where $\Bar{F} \in \R^m$ and $\Bar{J} \in \R^{m \times n}$. The entries of $\Bar{F}$ and $\Bar{J}$ are found by solving the following regression problem, which is equivalent to an (unconstrained) least squares problem: 
\begin{equation}\label{prob001}
    \min_{\Bar{r}, \Bar{J}} ~ \sum_{t=0}^p \ \lVert m_{\ell}(y_t- \mathbf{\omega})-\mathrm{r}(y_t) \rVert^2 ~.
\end{equation}

Consider the notation:
\begin{itemize}
    \item $\Bar{r}_i$ being the $i$-th entry of the vector $\Bar{F}$;
    \item $\Bar{J}(i,j)$ to the entry at the $i$-th row and $j$-th column of the matrix $\Bar{J}$;
    \item $\Bar{J}(i,:)$ to the $i$-th row of the matrix $\Bar{J}$;
    \item $\Bar{J}(:,j)$ to the $j$-th column of the matrix $\Bar{J}$.
\end{itemize}

\vspace{0.4cm}
\begin{prop} \label{prop01}
Problem (\ref{prob001}) is equivalent to the linear system of equations:
\begin{equation}\label{prob002}
    W \ \left[ \begin{matrix}
        \Bar{r}_i \\
        \Bar{J}(i,:)^T
    \end{matrix} \right] = \left[ \begin{matrix}
r_i(y_0) \\
\vdots \\
r_i(y_p)
\end{matrix} \right] ~,
\end{equation}
where the matrix $W \in \R^{(p+1) \times (n+1)}$ is defined as
\begin{equation}\label{defW}
    W :=  \left[ \begin{matrix}
1 & (y_0- \mathbf{\omega})^T \\
\vdots & \vdots \\
1 & (y_p- \mathbf{\omega})^T \end{matrix} \right]  ~.
\end{equation}
\end{prop}

\noindent\textbf{Proof of Proposition \ref{prop01}:}

Note that, for each $t=0, 1, \dots, p$, we have:

\begin{equation}\label{RC01}
    m_{\ell}(y_t- \mathbf{\omega}) - \mathrm{r}(y_t) ~  = ~ \Bar{F} + \Bar{J} \ (y_t- \mathbf{\omega}) - \mathrm{r}(y_t)  ~ = ~ \left[ \begin{matrix}
\Bar{F} & \Bar{J}
\end{matrix} \right] \ \left[ \begin{matrix}
1 \\ (y_t- \mathbf{\omega}) \end{matrix} \right] 
- \mathrm{r}(y_t) ~.
\end{equation}

From (\ref{RC01}), we have that the $i$-th entry of the vector $E_t:=(m_{\ell}(y_t- \mathbf{\omega}) - \mathrm{r}(y_t)) \in \R^m$, $i=1,2, \dots, m$, is given by:
\begin{equation}\label{RC02}
E_t(i) = \left[ \begin{matrix}
\Bar{r}_i & \Bar{J}(i,:)
\end{matrix} \right] \ \left[ \begin{matrix}
1 \\ (y_t- \mathbf{\omega}) \end{matrix} \right] - r_i(y_t) ~. 
\end{equation}

So, by (\ref{RC02}) we have that:
$$\lVert m_{\ell}(y_t- \mathbf{\omega}) - \mathrm{r}(y_t) \rVert^2 = \sum_{i=1}^m E_t(i)^2 ~.$$

And, therefore,
\begin{align*}
\sum_{t=0}^p \ \lVert m_{\ell}(y_t- \mathbf{\omega}) - \mathrm{r}(y_t) \rVert^2 ~& = ~\sum_{t=0}^p \sum_{i=1}^m E_t(i)^2 \\
& = ~ \sum_{i=1}^m \sum_{t=0}^p E_t(i)^2 \\
& = ~ \sum_{i=1}^m  \ \Bigg\lVert \left[ \begin{matrix}
1 & (y_0- \mathbf{\omega})^T \\
\vdots & \vdots \\
1 & (y_p- \mathbf{\omega})^T \end{matrix} \right] \ \left[ \begin{matrix}
        \Bar{F}(i) \\
        \Bar{J}(i,:)^T
    \end{matrix} \right] - \left[ \begin{matrix}
r_i(y_0) \\
\vdots \\
r_i(y_p)
\end{matrix} \right] \Bigg\rVert ^2 ~.
\end{align*}

Then, solving the problem (\ref{prob001}) by least squares method is equivalent to solve $m$ independent problems in the form:
$$\Bigg\lVert \left[ \begin{matrix}
1 & (y_0- \mathbf{\omega})^T \\
\vdots & \vdots \\
1 & (y_p- \mathbf{\omega})^T \end{matrix} \right] \ \left[ \begin{matrix}
        \Bar{r}_i \\
        \Bar{J}(i,:)^T
    \end{matrix} \right] - \left[ \begin{matrix}
r_i(y_0) \\
\vdots \\
r_i(y_p)
\end{matrix} \right] \Bigg\rVert ^2 ~,$$
$i=1, \dots, m$ by least squares method. To achieve that, we have to find solutions to the overdetermined linear systems:
\begin{equation*}
    W \ \left[ \begin{matrix}
        \Bar{r}_i \\
        \Bar{J}(i,:)^T
    \end{matrix} \right]:=  \left[ \begin{matrix}
1 & (y_0- \mathbf{\omega})^T \\
\vdots & \vdots \\
1 & (y_p- \mathbf{\omega})^T \end{matrix} \right] \ \left[ \begin{matrix}
        \Bar{r}_i \\
        \Bar{J}(i,:)^T
    \end{matrix} \right] = \left[ \begin{matrix}
r_i(y_0) \\
\vdots \\
r_i(y_p)
\end{matrix} \right] ~.
\end{equation*}
\fim

\vspace{0.4cm}
\begin{prop}\label{prop02}
The set $\{ y_1- \mathbf{\omega} , \cdots , y_p- \mathbf{\omega} \}$ spans $\R^n$ iff the matrix $W$ defined in (\ref{defW}) has full column rank.    
\end{prop}

\noindent \textbf{Proof of Proposition \ref{prop02}:}

Let the matrices $W' \in \R^{p \times n}$, $0_n \in \R^{1 \times n}$ and $1_p \in \R^{p \times 1}$ be defined as:
$$W':= \left[ \begin{matrix}
    (y_1- \mathbf{\omega})^T \\
    \vdots \\
    (y_p- \mathbf{\omega})^T
\end{matrix} \right] ~, ~ 0_n:= \left[ \begin{matrix}
   0 & \cdots & 0 
\end{matrix} \right] ~ , ~ 1_p:= \left[ \begin{matrix}
   1 \\
    \vdots \\
    1
\end{matrix} \right] ~.$$
Then we can view the $W$ as a block matrix:
$$W=\left[ \begin{matrix}
    1 ~& 0_n \\
    1_p & W'
\end{matrix} \right] = \left[ \begin{matrix}
    1 ~& 0 & \cdots & 0\\
    1_p & W'(:,1) & \cdots & W'(:,n)
\end{matrix} \right]~.$$

\noindent ($\Rightarrow$) Assuming that $\text{span} \left( \{ y_1- \mathbf{\omega} , \cdots , y_p- \mathbf{\omega} \} \right) = \R^n$, we have that $\text{rank} \left(W' \right) = \text{rank} \left(W'^T \right) = n$. Since $W' \in \R^{p \times n}$, we have that $W'$ has full column rank and so its columns are linearly independent. That implies that the set of vectors:
$$W_{col} := \left\lbrace \left[ \begin{matrix}
    1 ~ \\
    1_p
\end{matrix} \right] ~ , ~ \left[ \begin{matrix}
    0\\
    W'(:,j)
\end{matrix} \right] ~ : ~ j=1, \dots, n \right\rbrace$$
is linearly independent. That is, the columns of $W$ are linearly independent and so $W$ has full column rank. 

\noindent ($\Leftarrow$) On the other hand, assuming that $W$ has full column rank we have that $W_{col}$ is a linearly independent set. In particular, that implies that the columns of $W'$ are linearly independent, which is equivalent to  $rank(\{ y_1- \mathbf{\omega} , \cdots , y_p- \mathbf{\omega} \})=n$. Therefore $\text{span} \left( \{ y_1- \mathbf{\omega} , \cdots , y_p- \mathbf{\omega} \} \right) = \R^n$.

\fim

As a consequence of Proposition \ref{prop02}, we have the following two corollaries.

\begin{corollary}
The linear system (\ref{prob002}) is overdetermined iff the set $\{ y_1- \mathbf{\omega} , \cdots , y_p- \mathbf{\omega} \} $ spans $\R^n$.
\end{corollary}

\begin{corollary}
When $p<n$, the linear system (\ref{prob002}) has at least one exact solution.
\end{corollary}

So, in the case where $p \geq n$ it is possible that there are no exact solutions to the system (\ref{prob002}), which justifies using the least-squares method to solve it.  On the other hand, when we take $p < n$, there will always be some exact solution to this system, which can reduce the computation costs in the early iterations while providing sufficient exploration of the solution space. 

Given these results, the first question is how exactly DFO-LS builds the mini local linear approximations represented by $m_{\ell}(s)$. The answer, already given, is:
\begin{itemize}
    \item When $p<n$, we can make use of a direct solver to the linear system (\ref{prob002});
    \item and, when $p \geq n$, we use the least-squares strategy to solve it.
\end{itemize}
In the second of these cases, we solve the new problem:
\begin{equation}\label{prob002ls}
    W^T \ W \ \left[ \begin{matrix}
        \Bar{r}_i \\
        \Bar{J}(i,:)^T
    \end{matrix} \right] = W^T \ \left[ \begin{matrix}
r_i(y_0) \\
\vdots \\
r_i(y_p)
\end{matrix} \right] ~,
\end{equation}
which always has at least one solution. Once we have found $\Bar{F}$ and $\Bar{J}$, we can use them in the model $m_{\ell}(s)$ described by (\ref{prob003}), which will be our mini local linear models for $\mathrm{r}(y_t+s)$, $t=1, \dots, p$. Considering that this procedure is made at every iteration $k$ of the DFO-LS algorithm, we can see $\mathbf{\omega}$ mentioned in this section as the iterate $\mathbf{\omega}_k$ and the model $m_{\ell}(s)$ as a local model around $\mathbf{\omega}_k$, namely $m_{\ell}^k(s)$. The other points and matrices used in this step are somewhat disposable. 

Now that we obtained $m_{\ell}(s)$, it's time to create the quadratic regression model, $m_q(s) \approx \lVert \mathrm{r}(\mathbf{\omega}+s) \rVert^2$. This model is created in the most intuitive way possible:
$$m_q(s) := \lVert m_{\ell}(s) \rVert^2 ~.$$
We can notice that, as in classical quasi-Newton methods, one maintains an iterate-dependent matrix $B_k\approx\nabla^2\phi(\mathbf{\omega}_k)$ updated by secant conditions (see Box \ref{box:hessian}). 

\vspace{0.7cm}
\begin{codebox}[label=box:hessian]{\textbf{Hessian Approximation strategies in Optimization}}
    Given a function $f: \R^n \rightarrow \R^n$, computing its Hessian $\nabla^2 f(x)$ directly may be expensive, unavailable, or unstable in the presence of noise. Quasi-Newton methods offer a practical alternative to computing the Hessian matrix explicitly when solving nonlinear optimization problems, as they build an approximation that is iteratively refined as the algorithm progresses.

One of the most widely used quasi-Newton schemes is the \textbf{BFGS} (Broyden-Fletcher-Goldfarb-Shanno) update, which constructs a symmetric positive-definite matrix $B_k$ meant to approximate the true Hessian. The update combines curvature information extracted from successive iterates,
$$
s_k = x_{k+1} - x_k, 
$$
and
$$
y_k = \nabla f(x_{k+1}) - \nabla f(x_k),
$$
through the formula
$$
B_{k+1} = B_k
- \frac{B_k s_k s_k^\top B_k}{s_k^\top B_k s_k}
+ \frac{y_k y_k^\top}{y_k^\top s_k}.
$$
This construction ensures that $B_{k+1}$ satisfies the secant condition $B_{k+1} s_k = y_k$, while preserving positive definiteness whenever $y_k^\top s_k > 0$. As a result, BFGS often achieves superlinear convergence without requiring second-order derivatives.

Other quasi-Newton updates follow similar principles but differ in numerical properties and robustness. The \textbf{DFP} (Davidon-Fletcher-Powell) update is an earlier alternative with a dual structure to BFGS, while the \textbf{SR1} (Symmetric Rank-1) update allows curvature corrections that may violate positive definiteness but can provide useful approximations in problems with highly non-quadratic structure. Limited-memory variants such as \textbf{L-BFGS} (Limited-memory BFGS) store only a small number of past curvature pairs, making quasi-Newton methods applicable to large-scale problems.

A thorough introduction to these updates and their theoretical foundations can be found in \emph{Nocedal \& Wright, Numerical Optimization} \cite{nocedal2006}.
\end{codebox}

\vspace{0.7cm}
Here, $B_k=\Bar{J}_k^T\Bar{J}_k$ arises \emph{implicitly} from the regression fit of residuals and is refreshed whenever the interpolation set is updated. Thus DFO-LS behaves like a Gauss--Newton-type quasi-Newton method whose curvature matrix is built from data-driven Jacobian surrogates rather than explicit secant updates; regularization (for example, $B_k+\lambda I$) can be used if $\Bar{J}_k$ is rank-deficient.
 
 Now that we already built the quadratic local model $m_q^k(s)$ around $\mathbf{\omega}_k$, which for simplicity we will again be calling $m_q(s)$ and $\mathbf{\omega}$, respectively. At this step, we choose a new point to replace an old one in our candidate pool $\{ y_0, \dots, y_p \}$. For simplicity, we will be calling this new point $\Bar{x}$. This point will be the next iterate, that is, $\mathbf{\omega}_{k+1}:= \Bar{x}$. We find the point $\Bar{x}$ by finding a solution $\Bar{s}$ to the problem:
\begin{equation}\label{probTR}
    \min_{\lVert s \rVert \leq \Delta} m_q(s) ~,
\end{equation}
and setting $\Bar{x} := \mathbf{\omega}+\Bar{s}$, where $\Delta$ is the trust-region radius chosen for this iteration. 

The DFO-LS method steps are shown in Algorithm 1, in the \textit{Appendix}.

\subsection{The DFO-LS algorithm explained}
We now proceed to describe the DFO-LS scheme in detail.
In the following, we will consider the notation $\mathbf{\Omega}^k \subset \R^n$ to denote the set of candidates considered in iteration $k$, that is, $ \mathbf{\Omega}^k = \lbrace \omega^k_0, \omega^k_1, \dots , \omega^k_p \rbrace$. The exception is the first set, denoted $\mathbf{\Omega}^\text{Ini}$, which is not associated with any iteration $k = 0,1,2, \dots$, but has the exactly same structure as a candidate pool containing $p$ candidates.
A short version of the DFO-LS algorithm we will approach is the following pseudocode:

\begin{enumerate}
    \item Set the initial candidate pool $\mathbf{\Omega}^\text{Ini} = \left\{ \omega_0^{-1}, \dots, \omega_p^{-1} \right\}$ (the last candidate does not need to be chosen yet). \\
    Set $k=-1$.
    \begin{enumerate}[start=2, label=\arabic*.]
        \item[\textbf{While}] the stop criteria is not met \textbf{:}
        \item[] Set $k=k+1$.
        \item Choose a candidate $\Tilde{\omega}^{k}$ for replacing the last column of the previous candidate pool,  $\omega_p^{k-1}$, which will lead to a new candidate pool $\mathbf{\Omega}^k = \left\{ \omega_0^{k}, \dots, \omega_p^{k} \right\}$. \\
        Reorder the columns of $\mathbf{\Omega}^k$ to ensure property (\ref{property1}).
        \item Create a quadratic local model for $f_\text{Misfit}(\omega)$ around $\omega_0^k$.
        \item Perform a trust-region strategy using the quadratic model to find a local minimum $\bar{\omega}^k$ for $f_\text{Misfit}(\omega)$ around $\omega_0^k$. 
        \item[\textbf{If}] the resulting point is a good enough improvement \textbf{:}
            \begin{enumerate}[start=5, label=\arabic*.]
            \item Replace $\omega_p^k$ with $\bar{\omega}^k$ in the set $\mathbf{\Omega}^k$. \\
            Reorder the columns of $\mathbf{\Omega}^k$ to ensure property (\ref{property1}).
            \end{enumerate}
        
    \end{enumerate}

\end{enumerate}

In the following, we provide some details about this algorithm, which are especially useful when calibrating a biogeochemical model. The DFO-LS algorithm ensures the convergence (or else the termination or restart of the algorithm executing) of a sequence of points $\left\lbrace \omega^k \right\rbrace$ in the parameter space to a calibrated parameter $\bar{\omega}$ in the parameter space; such sequence is represented by $\left\lbrace \omega_0^k \right\rbrace$. Although for $k\geq 1$ all the points $\omega_0^k$ are generated by the algorithm, the first point $\omega_0^0$ is required for initializing the algorithm. Here, we always choose to set $\omega_0^0$ as the best parameter calibration we have at the moment, but any point in the parameter space is also a possible choice.
All the columns of $\mathbf{\Omega}^\text{Ini}$ in step 1, except for the first one, may be chosen randomly; this also applies to the choice of $\Tilde{\omega}^{k}$ in step 2. 

\vspace{0.2cm}
\subsection{Creating the quadratic model}

When creating a quadratic local model for the misfit function in step 3 of the pseudocode presented in the previous subsection, we follow the procedure described by \textit{C. Cartis et al. (2019, p.32:6)}~\cite{DFOGN}. We first create a linear local model for the residual function $r_\text{Misfit}(\omega)$ around $\omega_0^k$, which can be seen as mimicking or predicting the behavior of the first order Taylor approximation for the residual function around $\omega_0^k$ in situations where its first derivatives are available. This is a regression model using data from the columns of $\mathbf{\Omega}^k$. 
Therefore, when defining the local linear models for the residual function, we will be considering the following notation:

\vspace{0.5cm}
\begin{tabular}{l l}
$s$ : & A step (that is, a vector) in the parameter space. \\
$m_\ell^k(s)$ : & Linear local model for the residual function $r_\text{Misfit}(\omega)$ around the point $\omega_0^k$. See (\ref{linearproperty01}). \\
$m_q^k(s)$ : & Quadratic local model for the misfit function $f_\text{Misfit}(\omega)$ around the point $\omega_0^k$. See (\ref{quadproperty01}).
\end{tabular}

\vspace{0.7cm}
For the linear local model of the residual function, we have:
\begin{equation}\label{linearproperty01}
    m_\ell^k(s) \approx r_\text{Misfit}(\omega_0^k + s) ~,
\end{equation}

\noindent for every $s$ sufficiently small. We take the following notation:

\vspace{0.2cm}
\begin{tabular}{l l}
$r_k$ : & Approximation for $r_\text{Misfit}(\omega_0^k)$. \\
$J_k$ : & Approximation for the Jacobian matrix of $r_\text{Misfit}(\omega)$ at $\omega_0^k$.
\end{tabular}

\vspace{0.5cm}
Thus, in our linear local model, we will use a linear regression strategy for defining $r_k$ and $J_k$ at each iteration and set:
\begin{equation}
    m_\ell^k(s) = r_k + J_k s ~.
\end{equation}

\vspace{0.5cm}
For defining $r_k$ and $J_k$, we solve the following linear regression problem about the residual of the misfit function evaluated over the candidate pool:
\begin{equation}\label{linearReg01}
   \min_{r_k,J_k} \sum_{i=0}^\lambda \rVert 
m_\ell^k(\omega_i^k-\omega_0^k) - r_\text{Misfit}(\omega_i^k) \lVert^2 ~. 
\end{equation}

Solving problem (\ref{linearReg01}) is equivalent to finding a solution for the linear system
\begin{equation}
    W_k \left[ \begin{matrix}
        r_{k,i} \\
        J_{k,i}^T 
    \end{matrix} \right] = b_{k,i} ~, ~ \forall i=1,2, \dots, N ~,
\end{equation}

\noindent where
\begin{equation*}
    W_k = \left[ \begin{matrix}
        1 & (\omega_0^k-\omega_0^k) \\
        1 & (\omega_1^k-\omega_0^k) \\
        \vdots & \vdots \\
        1 & (\omega_p^k-\omega_0^k)
    \end{matrix} \right] ~\text{, and} ~ b_{k,i} = \left[ \begin{matrix}
        r(\omega_0^k) \\
        r(\omega_1^k) \\
        \vdots  \\
        r(\omega_p^k)
    \end{matrix} \right] ~\text{for all} ~i=1, \dots, N ~,
\end{equation*}

\begin{equation*}
    r_k=\left[ \begin{matrix}
        r_{k,1} \\
        r_{k,2} \\
        \vdots \\
        r_{k,N}
    \end{matrix}\right]_{N} ~ \text{, and } ~ J_k=\left[ \begin{matrix}
        J_{k,1} & J_{k,2} & \cdots & J_{k,N}
    \end{matrix}\right]_{n \times N} ~.
\end{equation*}

\vspace{0.5cm}
For the quadratic local model of the misfit function, we have:
\begin{equation}\label{quadproperty01}
    m_q^k(s) \approx f_\text{Misfit}(\omega_0^k + s) ~,
\end{equation}

for every $s$ sufficiently small. We take the following notation:

\vspace{0.2cm}
\begin{tabular}{l l}
$f_k$ : & Approximation for $f_\text{Misfit}(\omega_0^k)$. \\
$g_k$ : & Approximation for the gradient vector of $f_\text{Misfit}(\omega)$ at $\omega_0^k$. \\
$H_k$ : & Approximation for the Hessian matrix of $f_\text{Misfit}(\omega)$ at $\omega_0^k$.
\end{tabular}

\vspace{0.2cm}
We set our quadratic model as:
\begin{equation}
    m_q^k(s)=f_k+g_k^T s + \dfrac{1}{2} s^T H_k s ~,
\end{equation}
where we derive $f_k$, $g_k$ and $H_k$ from some operations with $r_k$ and $J_k$ obtained for the linear model $m_\ell^k(s)$. They are defined as:
\begin{equation*}
        f_k  = r_k^T r_k ~, ~ g_k  = 2 J_k^T r_k  ~\text{, and } ~ H_k  = 2 J_k^T J_k ~.
\end{equation*}

\vspace{0.2cm}
\subsection{Minimizing the quadratic model with a trust-region strategy}

At the iteration $k$, we want to find a minimizer step $s$ for the following constrained optimization problem:
\begin{align}
    \min ~ & m_q^k(s) \nonumber\\
    \text{Subject to:} ~ & s \in \Gamma^k \label{trSubproblem}
\end{align}

\noindent where the set of constraints, namely the trust region, is defined as:
\begin{equation}
    \Gamma^k = \left[ \Delta_{\ell 1}^k, \Delta_{u 1}^k \right] \times \left[ \Delta_{\ell 2}^k, \Delta_{u 2}^k \right] \times \cdots \times \left[ \Delta_{\ell n}^k, \Delta_{u n}^k \right] ~,
\end{equation}
for some $\Delta_{\ell i}^k, \Delta_{u i}^k, ~ i = 1, \dots, n$, satisfying the condition:
\begin{equation*}
    \Delta_{\ell i}^k < \Delta_{u i}^k, ~ \forall i = 1, \dots, n ~.
\end{equation*}
If we consider the notations $s = (s_1, \dots, s_n)$, note that the constraints for problem (\ref{trSubproblem}) are equivalent to:
\begin{equation}\label{searchboundstr}
    \Delta_{\ell i}^k \leq s_i < \Delta_{u i}^k, ~ \forall i = 1, \dots, n ~.
\end{equation}
The set of inequalities (\ref{searchboundstr}) characterizes each interval $\left[ \Delta_{\ell i}^k, \Delta_{u i}^k \right]$ as search bounds for $s_i$, for $i=1, 2, \dots, n$.


For examples on the basic usage of DFO-LS for parameter calibration on Python, see the \textit{Appendix}. 
The DFO-LS algorithm searches for the vector of parameters that minimizes the misfit function, leading to what we are calling the calibration of such parameters. The output obtained by DFO-LS function consists of convergence info, including the vector of parameters found by this method and the number of misfit function evaluations. All of these will become clearer with the examples presented in the next chapter.

In this chapter, we explored how to frame calibration as an optimization task, why we use a least-squares misfit, and how derivative-free methods such as DFO-LS approaches black-box, noisy, and constrained problems. Building on that, the next chapter outlines a step-by-step approach for calibrating calibrating models and test the approach in practice with few examples.

\chapter{A systematic approach for calibrating conceptual models}
\label{chapterframework}

The aim of chapter is to describe in detail how the optimization tools introduced in \textit{Chapter \ref{chapterdfols}} can be used in practice for systematically calibrating models, with a focus on conceptual models based on systems of ordinary differential equations (ODEs).

In general, the success of a calibration strategy for the parameters of a mathematical model using optimization methods depends on three main ingredients:
\begin{itemize}
    \item A sufficiently accurate mathematical model for the phenomenon of interest;
    \item An adequate framing of the fitting as an optimization problem;
    \item A proper strategy for solving the optimization problem.
\end{itemize}

To illustrate and test this approach in practice, we consider a sequence of increasingly complex test problems that resemble the structure and dimension of conceptual models such as the PEC model.
We focus on the appropriate choice of residual function, initial guess, and bounds. Then, we present a calibration case study, in which we design and analyze several numerical experiments built around a small set of representative ODE models.
In all the computational experiments presented in this section, we used the DFO-LS optimization Python solver \cite{roberts2025dfo-ls}.

Throughout this chapter, the term \emph{parameters} denotes the unknown quantities to be calibrated, \emph{observational parameters} denotes the reference values we aim to recover via calibration, \emph{model output} denotes the dataset produced by the model when setting a particular choice of parameters, and \emph{observations} as the dataset generated for reference, which is used for recovering the observational parameters.

\paragraph{Chapter notation and DFO-LS setting:}
In the following, we recall the basic DFO-LS setting in a matrix-based notation that is close to the Python implementation used later. We do not consider advanced options related to restarts or noise on the DFO-LS solver, since the basic usage and the advanced options related to bounds are sufficient for the examples presented in this chapter.

At an iteration $k$, DFO-LS builds and updates a \emph{candidate pool} of parameter vectors. We represent this pool by the matrix:
$$\mathbf{\Omega}^k = \left[ \begin{matrix}
\omega^k_0 &  \omega^k_1 & \cdots & \omega^k_p     
\end{matrix}  \right]_{n \times (p+1)} ~,$$
where each column $\omega^k_j \in \R^n$ is a candidate parameter vector (that is, a full set of values for the $n$ parameters being calibrated), and $p+1$ is a fixed number of candidates considered at each iteration $k$.

For each candidate $\omega^k_j$, we evaluate the misfit function:
$$f_\text{Misfit} : \mathbb{R}^n \to \mathbb{R}_+ ~,$$
and collect the corresponding values into the vector:
$$f_\text{Misfit}\left(\mathbf{\Omega}^k \right) = \left[ \begin{matrix}
f_\text{Misfit}(\omega^k_0) \\
f_\text{Misfit}(\omega^k_1) \\
\vdots \\
f_\text{Misfit}(\omega^k_p)     
\end{matrix}  \right]_{p+1} ~.$$

We order the candidates so that:
\begin{equation}\label{property1}
 f_\text{Misfit}(\omega^k_j) \geq f_\text{Misfit}(\omega^k_{j-1}) \geq 0 ~, ~  \forall j=1, 2, \dots, p ~.  
\end{equation}
Thus, $\omega^k_0$ is the current best candidate at iteration $k$.

The misfit function itself is defined as the squared norm of a residual misfit function $r_\text{Misfit} : \mathbb{R}^n \to \mathbb{R}^m$,
$$f_\text{Misfit}(\omega) = \|r_\text{Misfit}(\omega)\|^2 ~,$$
where $m$ is the number of scalar residual entries of the residual function (typically related to the number of data points considered during its definition).

At iteration $k$, the residual misfit function evaluated at each candidate is represented by the matrix:
$$r_\text{Misfit}\left(\mathbf{\Omega}^k \right) = \left[ \begin{matrix}
r_\text{Misfit}(\omega^k_0) &  r_\text{Misfit}(\omega^k_1) & \cdots & r_\text{Misfit}(\omega^k_p)     
\end{matrix}  \right]_{m \times (p+1)} ~,$$
whose columns are the $m$-dimensional residual vectors $r_\text{Misfit}(\omega^k_j)$.

The DFO-LS algorithm builds, at each iteration, a local quadratic model for $f_\text{Misfit}$ based on these residual evaluations and solves a trust-region subproblem to obtain the next candidate. The details of this process are discussed in \textit{Chapter \ref{chapterdfols}}; here we focus on how to construct $r_\text{Misfit}$ in a way that is meaningful for our models and data. The following example illustrates one possible way of defining $r_\text{Misfit}$ for a conceptual biogeochemical model.

\vspace{0.5cm}
\begin{Ex}\label{exNotation}
    Consider the two-box NP model for the Paranagu\'a Estuarine Complex presented in \textit{Chapter \ref{chapterPEC}} (Equations (\ref{EstuarioN})-(\ref{estacionario})), where the observations for the monthly mean concentration of two different tracers in the upper box, nitrate $C_N$ and phytoplankton $C_P$, are available for 12 different times, $t_1, t_2, \dots, t_{12}$, summing up to the total of $m=24$ observations. These observations will be denoted $N_\text{obs}(t_i)$ and $P_\text{obs}(t_i)$, respectively. Let us also denote the concentrations of nitrate and phytoplankton, respectively, predicted by our model for each time $t_i$ and depending on the choice of parameters $\omega$, as $N_\text{predicted}(\omega,t_i)$ and $P_\text{predicted}(\omega,t_i)$. In this example, we denote the entries of  the residual function as 
    $$r_\text{Misfit}(\omega) = \left(r_{1,\text{Misfit}}(\omega), \ r_{2,\text{Misfit}}(\omega), \ \dots, \ r_{m,\text{Misfit}(\omega)} \right) ~,$$
    and set
    \begin{align*}
        r_{i,\text{Misfit}}(\omega) \ = & \begin{cases}
            & \dfrac{N_\text{predicted}(\omega, t_i) - N_\text{obs}(t_i)}{24 \cdot N_\text{obs}(t_i)}, ~ \forall i=1,2, \dots, 12, \\
            & \dfrac{P_\text{predicted}(\omega, t_{i-12}) - P_\text{obs}(t_{i-12})}{24 \cdot P_\text{obs}(t_{i-12})}, ~ \forall i=13,14, \dots, 24.
        \end{cases}
    \end{align*}

Thus, the misfit function is defined as: 
\begin{equation*}
    f_\text{Misfit}(\omega) = \sum_{i=1}^{24} \left[ r_{i,\text{Misfit}}(\omega) \right]^2 = \lVert r_\text{Misfit}(\omega) \rVert^2 ~.
\end{equation*}
\end{Ex}

\vspace{0.5cm}
Example \ref{exNotation} will serve as a template for the construction of residual functions in the subsequent sections, where we will adapt it to simpler ODE models that are easier to analyze in detail.


In the following section, we perform a sequence of tests designed to clarify the role of the residual function, initial guess, and bounds in parameter calibration, starting with a very simple model. We focus on ODE models that share some key characteristics with conceptual models such as the PEC model:
\begin{itemize}
    \item The model is defined by a system of ODEs' equilibrium, which is usually periodic.
    \item The parameters are calibrated based on data from observations.
    \item Each calibrated parameter is constrained by predefined lower and upper bounds.
\end{itemize}

In what follows we briefly recall the concepts of fixed point and limit cycle, then introduce a first test problem with a scalar ODE.

A \textbf{fixed point} (also called an \textit{equilibrium point} or \textit{steady state}) of a dynamical system described by a set of ordinary differential equations (ODEs)
$$
\dot{x} = F(x),
$$
is a point \( x^* \in \mathbb{R}^n \) such that
$$
F(x^*) = 0.
$$
At this point, the system does not change in time: if the system starts at $ x^* $, it remains there for all $ t $. A stable fixed point acts as an \textit{attractor}: trajectories starting near it move toward and remain near that equilibrium,
as fixed points represent time-invariant solutions.

On the other hand, a \textbf{limit cycle} is a closed, isolated periodic orbit.
Formally, it is a trajectory $ \gamma(t)$ satisfying
$$
\gamma(t + T) = \gamma(t),
$$
for some period $T > 0$. We say that a limit cycle is \textit{stable} when trajectories starting near the cycle converge to  $ \gamma(t) $ as  $ t \to \infty $ . 

The use of observations for parameter calibration is more clearly illustrated in the following example.

\vspace{0.5cm}
\begin{Ex}[A constant equilibrium]\label{ex:const_equil}
Let $a,b>0$ be constants. Consider the linear ODE
\begin{equation}\label{dfolstest1.0}
    \dfrac{dx}{dt} = a - b x ~.
\end{equation}
The equilibrium solution is $\bar{x} = \dfrac{a}{b}$. Every solution of (\ref{dfolstest1.0}) converges to the equilibrium, meaning that, when integrating this ODE system for a sufficiently large time interval, we obtain:
$$x(t) \approx \bar{x} ~.$$

We use this simple model to illustrate the basic DFO-LS calibration scheme. We fix a value $a>0$ and an observational equilibrium $\bar{x}>0$, and we use DFO-LS to calibrate the parameter $b$. Since $\bar{x} = a/b$, we expect $b$ to converge to $a/\bar{x}$. Equivalently, we can think of the observational parameter vector as:
$$[a_\text{Obs}, b_\text{Obs}] = [2,10] ~,$$
so that $\bar{x} = a_\text{Obs}/b_\text{Obs} = 0.2$, and the calibration aims to recover $b_\text{Obs}$.

Because this model has a single state variable, the residual can be defined as a direct comparison between the model equilibrium $\bar{x}$ (considered as observational data) and one sampled value from the system integration:
$$r_\text{Misfit}(b) = x(b) - \bar{x},$$
where $x_(b)$ denotes the numerical solution obtained by integrating system (\ref{dfolstest1.0}) and sampling the last integration point. The misfit function is then
$$f_\text{Misfit}(b) = \left[r_\text{Misfit}(b)\right]^2.$$

For generating the model output to be compared with the observational data, we performed the numerical integration with:
\begin{itemize}
    \item $t_\text{max}$ large enough for the solution to be close to equilibrium,
    \item a fixed time step $\Delta t$,
    \item an initial condition for the integration, $x(0) = x_0$.
\end{itemize}
When varying $a$, $b$ and $\bar{x}$, it is important to verify that the choices of $t_\text{max}$, $\Delta t$ and $x_0$ remain appropriate. We then choose an initial guess $b_0$ for the parameter calibration and run DFO-LS.
The Python implementation and the detailed settings for this test are given in the \textit{Appendix}. The DFO-LS search for $b$ is shown in Figure~\ref{fig:teste1b}.
\end{Ex}

\vspace{0.6cm}
\begin{figure}[htpb!]
\caption{DFO-LS calibration of parameter $b$ in Example~\ref{ex:const_equil}.}
\begin{center}
\includegraphics[scale=0.7]{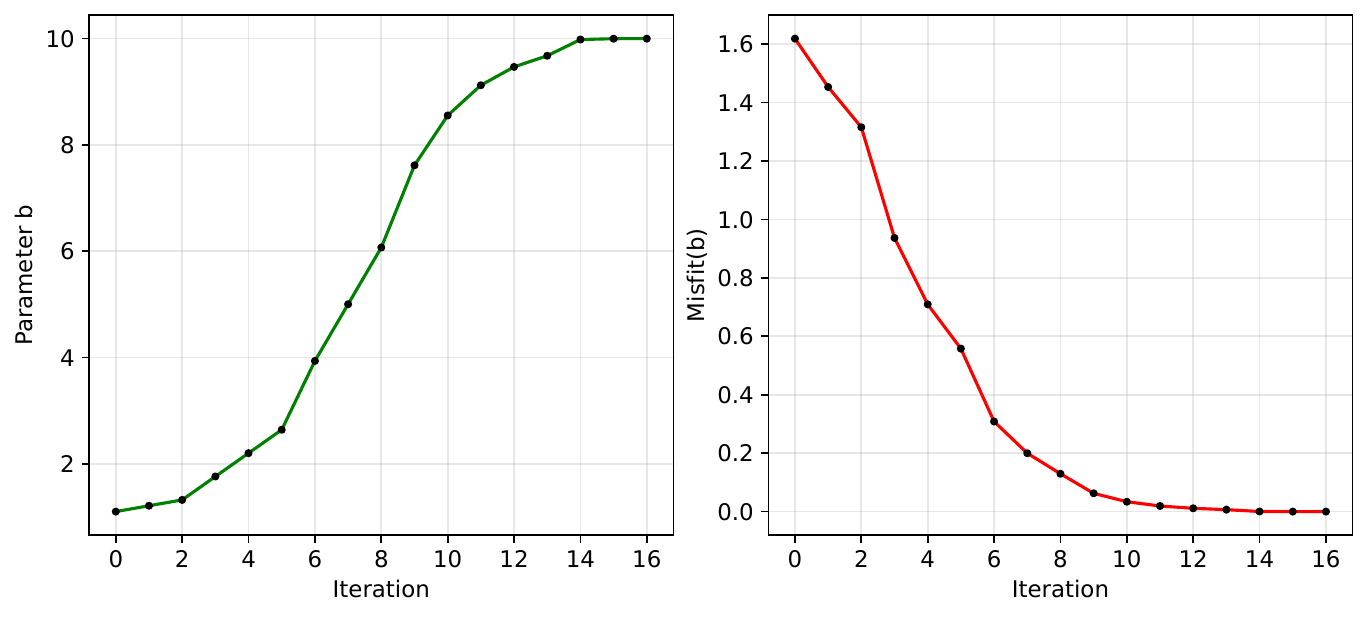}    
\end{center} 
\label{fig:teste1b}

\smallskip
In the left-hand plot, the points represent the values of $b$ obtained at each iteration. In the right-hand plot, the points represent the corresponding misfit values $f_\text{Misfit}(b)$ attained at each iteration. The observational parameter $b_\text{Obs} = 10$ is recovered in 16 iterations, with $f_\text{Misfit}(b_\text{Obs}) = 0$.

\smallskip 
\textbf{Source:} the author (2026).
\end{figure}

\vspace{0.5cm}
This example shows that even a minimal residual definition can be sufficient for parameter calibration when the dynamics are simple and the equilibrium is well-behaved. In the next section, we extend this framework to nonlinear systems with limit cycles and more parameters.

\section{Calibration case study}\label{sec:calibcasestudy}

The previous scalar example illustrates the basic mechanics of parameter calibration with DFO-LS. We now consider a case study based on ODE systems with cyclic or nonlinear behavior, designed to emulate some key features of conceptual biogeochemical models:
\begin{itemize}
    \item the presence of limit cycles (periodic equilibria),
    \item the simultaneous calibration of several parameters,
    \item possible ill-conditioning of the misfit function,
    \item the need for model-based choices of residuals, observation sets, and bounds.
\end{itemize}

The case study is organized into four classes of models:
\begin{enumerate}
    \item An ODE system with a cyclic equilibrium (Part~1): calibration of a single parameter in a two-dimensional system with a known limit cycle;
    \item An ODE system with a cyclic equilibrium (Part~2): simultaneous calibration of two parameters in a similar system;
    \item A difficult case: a sinusoidal system where different parameter values produce indistinguishable observations, illustrating limitations of systematic calibration;
    \item A nonlinear example: a two-dimensional nonlinear system (Sel'kov model) with a limit cycle and no closed-form solution.
\end{enumerate}

In each subsection we follow the same structure:
\begin{itemize}
    \item \emph{Model and observations:} definition of the ODE system and the way observational data are generated.
    \item \emph{Residual definition:} construction of $r_\text{Misfit}$ from samples and observations.
    \item \emph{Numerical experiment:} a concrete experiment (or a family of experiments) based on the same example, varying the number and distribution of observations, the length of the integration window, and/or other design parameters in order to investigate calibration performance.
\end{itemize}

In all tests, the ultimate goal is to recover the chosen observational parameters from the synthetic observational data, which was generated by the model itself.

\subsection{An ODE system with cyclic equilibrium -- Part 1}\label{subsec:ODE001}

\paragraph{Model and observations:} We first consider a two-dimensional ODE system with a periodic equilibrium (limit cycle) and a single parameter to be calibrated. Let
\begin{equation}\label{edo21label}
\begin{cases}
    x' & = -y - \left( \dfrac{(t+c)^{-2}}{e^a+(t+c)^{-1}} \right) + \ln \left( e^a+(t+c)^{-1} \right) + y_0 ~,\\ \\
    y' & = x - \left( \dfrac{(t+c)^{-2}}{e^a+(t+c)^{-1}} \right) + \ln \left( e^a+(t+c)^{-1} \right) - x_0 ~,
\end{cases}    
\end{equation}
with fixed constants $c>0$, $x_0$, $y_0$ and parameter $a \in \mathbb{R}$. 
System (\ref{edo21label}) is well-defined for any $t>0$, and the analytical solution of the system (\ref{edo2label}) is given by:
\begin{equation}\label{Soledo21label}
\begin{cases}
    x(t) & = \ln \left( e^a+(t+c)^{-1} \right) + \cos(t) + x_0  ~,\\ \\
    y(t) & = \ln \left( e^a+(t+c)^{-1} \right) +\sin(t) + y_0 ~.
\end{cases}    
\end{equation}

When $x_0 = y_0 = 0$, and for $t>>0$ (that is, $t$ very large), we obtain:
$$\ln \left( e^a+(t+c)^{-1} \right) \approx a ~,$$
so that
\begin{equation}\label{Limitedo21label}
    (x(t), y(t)) \approx (a + \cos(t), a + \sin(t)) ~,
\end{equation}
which corresponds to the convergence for a periodic equilibrium (limit cycle). 

We fix $c>0$ and address the problem of calibrating the parameter $a$, which can assume any real value. We define an observational parameter $a_\text{Obs}$ and generate synthetic observations from the asymptotic expression~\eqref{Limitedo21label}, that is,
$$(x_\text{Obs}(t),y_\text{Obs}(t)) = (a_\text{Obs} + \cos t,\; a_\text{Obs} + \sin t) ~.$$

To control the sampling along the limit cycle, we introduce the rescaled time variable $s = \pi t$, so that $(x(s),y(s))$ has period~2 in $s$ instead of period $2\pi$ in $t$. All observation times are then chosen in a final time window corresponding to one or more periods after a sufficiently long integration time, ensuring that the numerical solution is close to the limit cycle.

To control the sampling along the limit cycle, we introduce the rescaled time variable $s = \pi t$, so that $(x(s),y(s))$ has period 2 in $s$ instead of period $2 \pi$ in $t$. All observation times are then chosen in a final time window corresponding to one or more periods after a \textit{spinup} (a sufficiently long integration period used to stabilize the model, see Box \ref{box:spinup}), ensuring that the numerical solution is close to the limit cycle.

\paragraph{Misfit residual definition:}

We consider $N+1$ sample points along the last cycle (after the spinup). For each observation time $t \in T_\text{Obs}$ we compute both the model output $(x(t),y(t))$ and the observational state $(x_\text{Obs}(t),y_\text{Obs}(t))$. The residual vector is defined by collecting pointwise Euclidean distances,
$$r_\text{Misfit}(a) =
\left(r_1(a),\dots,r_{N+1}(a)\right) ~,$$
with
$$r_i(a) = \left\| 
(x(t_i; a),y(t_i; a)) - (x_\text{Obs}(t_i),y_\text{Obs}(t_i))
\right\| ~.$$

The misfit function is
$$f_\text{Misfit}(a) = \sum_{i=1}^{N+1} r_i(a)^2 ~.$$

Figure~\ref{fig:018} illustrates the pointwise fitting on the limit cycle: the red curve is the observation set, the green curve is a model trajectory, and the points correspond to $N+1$ samples used in the residual. In the illustrated case, the time interval chosen for the model integration, $t \in [1,398]$, was not long enough to be considered a spinup of the model. This situation highlights the need for caution when defining the spinup time for the model, as it may lead to accuracy loss when attempting to recover the observational parameters.

\vspace{0.6cm}
\begin{figure}[htpb!]
\caption{Illustration of the pointwise fitting process on the limit cycle.} 
\begin{center}
\includegraphics[scale=0.7]{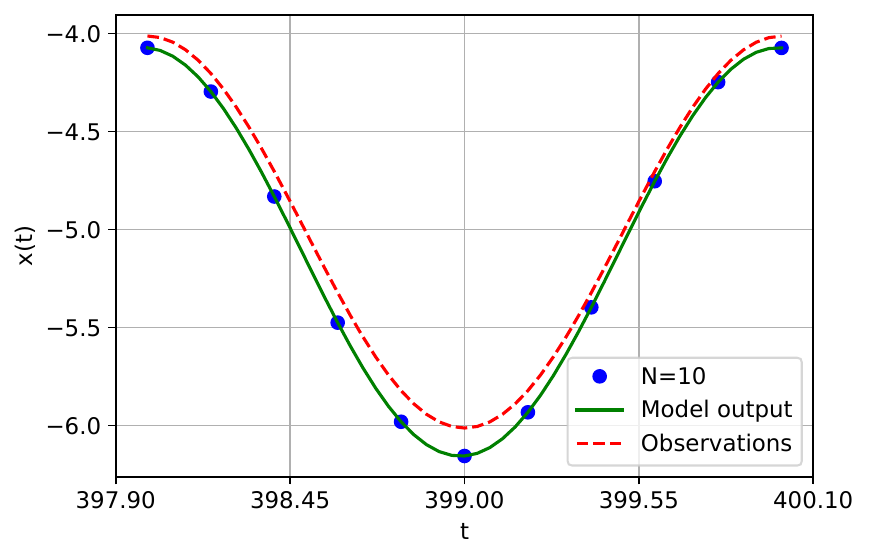}

\vspace{0.3cm}
\includegraphics[scale=0.7]{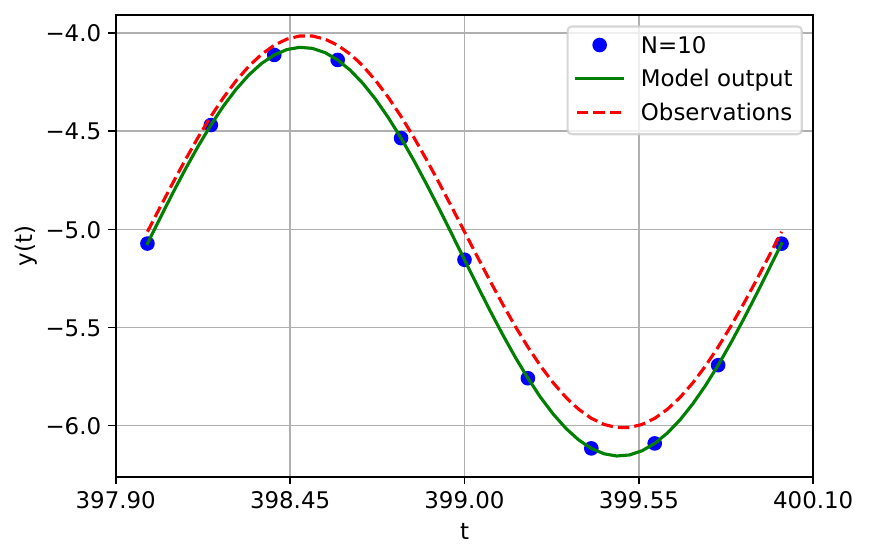}    
\end{center}
\label{fig:018}

\smallskip
The red dashed line represents the observational data, while the green solid line represents the model output. Dots represent a set of $N$ sampled points used in computing the misfit residual.

\smallskip 
\textbf{Source:} the author (2026).
\end{figure}


\paragraph{Numerical experiment (single-parameter limit-cycle calibration):}

\paragraph{Experiment 1.}
We now design a specific experiment based on this example. The goal is to recover the parameter $a_\text{Obs}$ from synthetic observations using DFO-LS and to study the influence of:
\begin{itemize}
    \item the integration time (length of spinup),
    \item the number of samples $N+1$,
    \item the need for a heuristic step when defining the initial guess for the parameter search.
\end{itemize}

We proceed as follows:
\begin{enumerate}
    \item Fix $c>0$ (in this experiment, we set $c=100$), choose $a_\text{Obs}$, and generate observations via~\eqref{Limitedo21label}.
    \item Integrate~\eqref{edo21label} with an initial condition away from the limit cycle, for a sufficiently long time.
    \item In a first set of runs, use a heuristic stage: evaluate $f_\text{Misfit}$ on a grid of $a$ values in a neighborhood of $a_\text{Obs}$ and select the best value as initial guess for DFO-LS.
    \item In a second set of runs, bypass the heuristic and start directly from a generic guess further away from $a_\text{Obs}$.
    \item In both cases, define bounds as a fixed-size neighborhood of $a_\text{Obs}$. Note that the bounds are informed to the solver without any information about the observational parameter.
\end{enumerate}

Figure~\ref{fig:Aug121} shows a typical heuristic evaluation for $a_\text{Obs}=5.01$. The misfit displays a nearly linear trend around the observational parameter, and the best point found by the heuristic is close to $a_\text{Obs}$.

\vspace{0.5cm}
\begin{figure}[htpb!]
\caption{Heuristic evaluation of $f_\text{Misfit}(a)$ around $a_\text{Obs}=4.77$.}
\begin{center}
\includegraphics[scale=0.7]{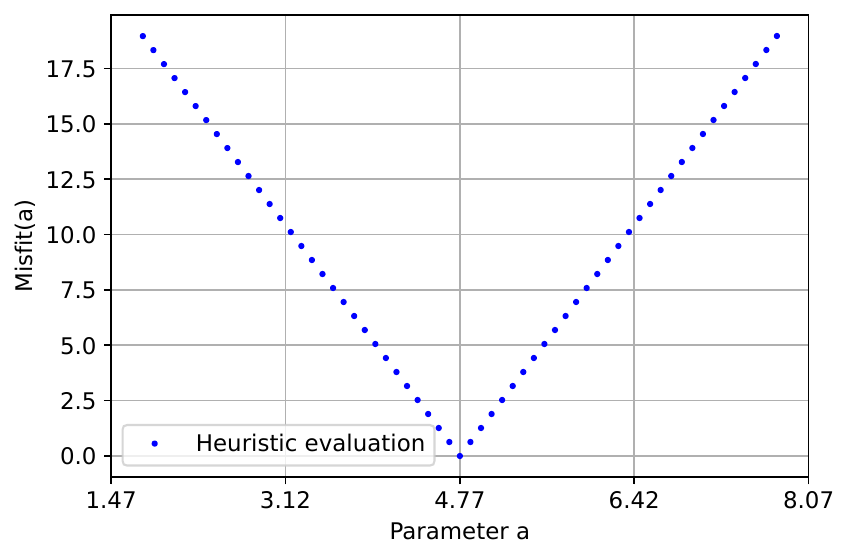}    
\end{center}
\label{fig:Aug121}

\smallskip
The seemingly linear behavior of the misfit function suggests that DFO-LS can converge to the observational parameter if  given an initial guess from a  relatively broad neighborhood of $a_\text{Obs}$, provided that the integration time is long enough.

\smallskip 
\textbf{Source:} the author (2026).
\end{figure}

\vspace{0.2cm}
Figure \ref{fig:Aug122} shows typical fittings in $x$ and $y$ for the observational parameter  $a_\text{Obs}=5.1$ and two different numbers of observations $N+1$. In all tested settings, the observational parameter is recovered with good accuracy when the integration time is sufficiently long.

\begin{figure}[htpb!]
\caption{Fitting of sampled model output values to the corresponding observations.}
\begin{center}
\includegraphics[scale=0.7]{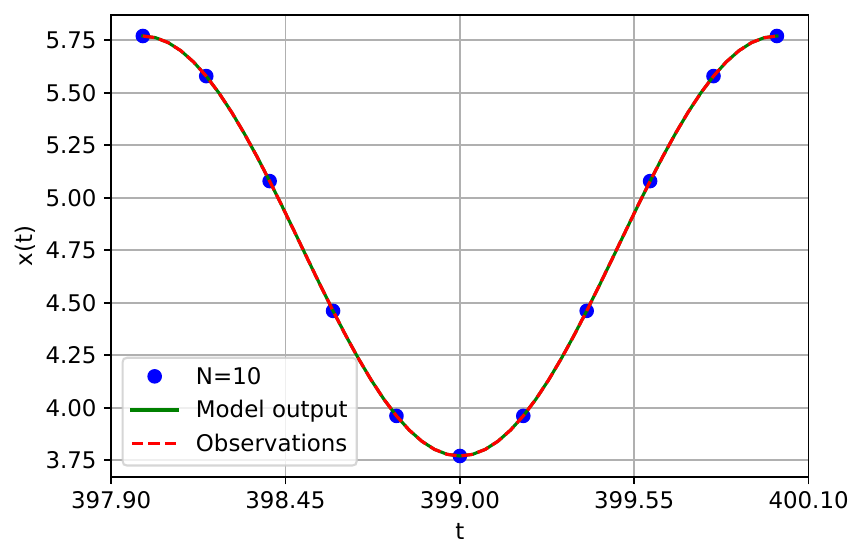}

\vspace{0.3cm}
\includegraphics[scale=0.7]{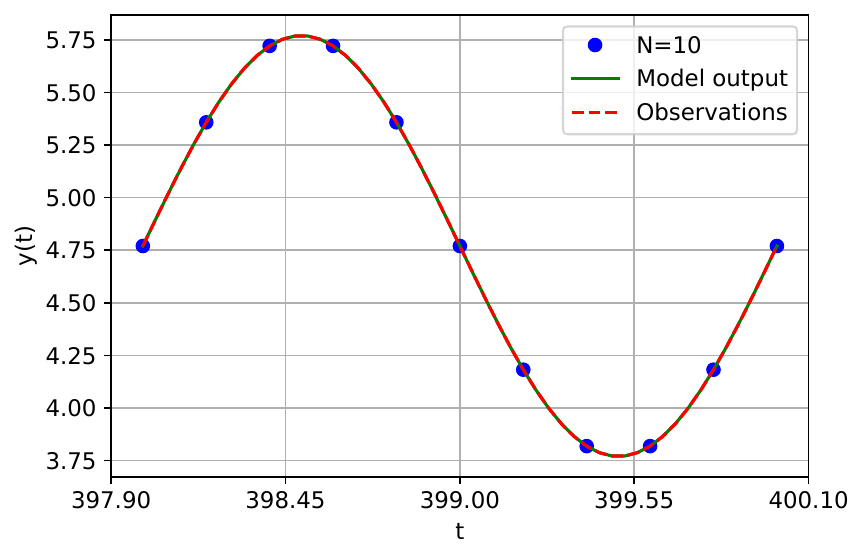}    
\end{center}
\label{fig:Aug122}

\smallskip
Considering the model defined by the equilibrium of equations (\ref{edo21label}), for $a_\text{Obs}=4.77$ and  $N+1 = 11$, following the structure in Figure~\ref{fig:018}.

\smallskip 
\textbf{Source:} the author (2026).
\end{figure}

Figures~\ref{fig:Aug123} and \ref{fig:Aug124} summarize the dependence of the calibration accuracy on the number of observations and on the integration time. Due to the convergence to the limit cycle being fast, relatively short spinup time definitions were sufficient. For the experiments performed in this subsection, convergence to the limit cycle was relatively fast, and accurate recovery is possible over a broad range of $N$ and integration windows.

\begin{figure}[htpb!]
\caption{Dependence on the number of observations $N+1$ for the calibration accuracy of the parameter $a$.}
\begin{center}
\includegraphics[scale=0.85]{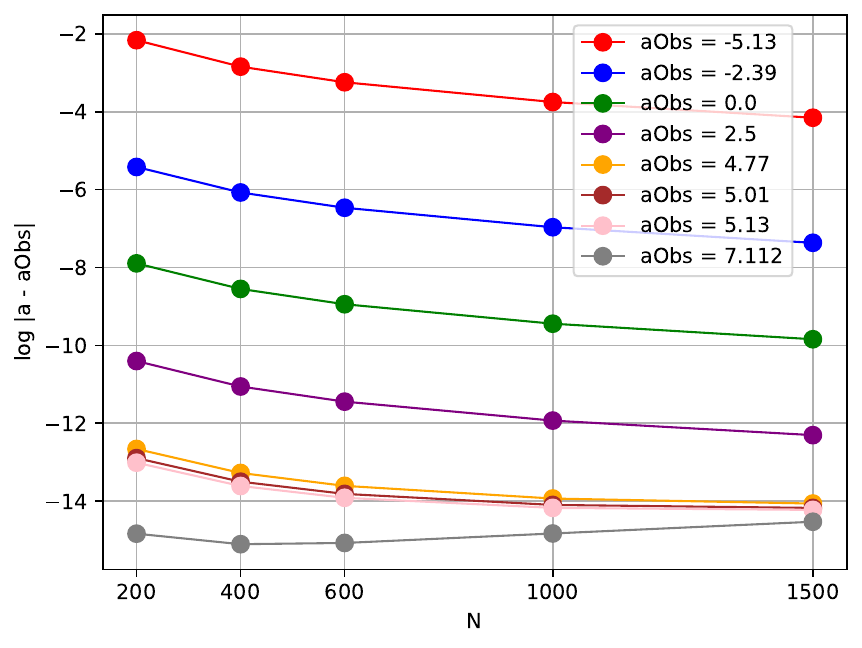}    
\end{center}
\label{fig:Aug123}

\smallskip
The dots represent the accuracy attained for the calibration experiment defined in Subsection \ref{subsec:ODE001} with the corresponding observational parameter  $a_\text{Obs}$.

\smallskip
\textbf{Source:} the author (2026).
\end{figure}

\begin{figure}[htpb!]
\caption{Dependence of calibration accuracy for $a$ on the integration time.}
\begin{center}
\includegraphics[scale=0.85]{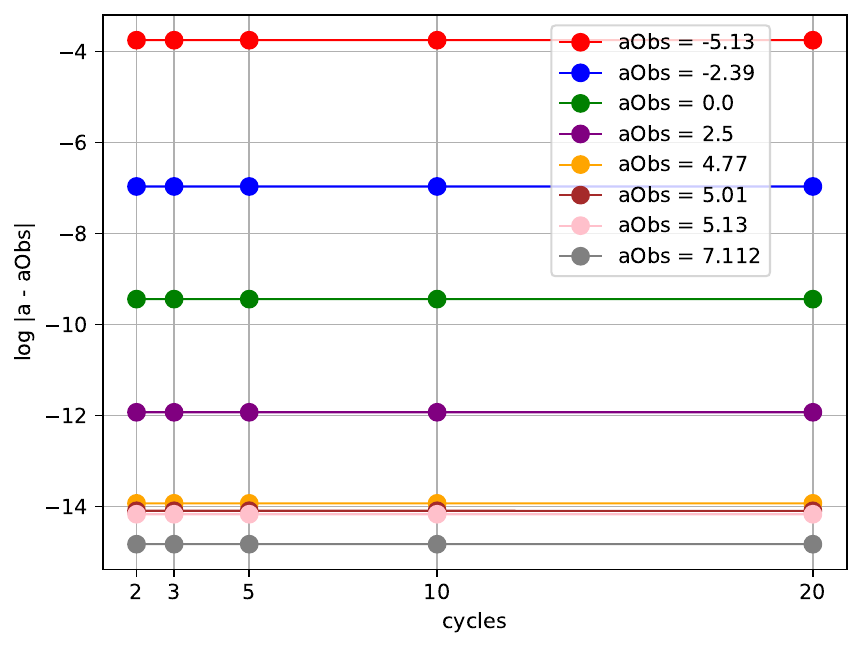}    
\end{center}
\label{fig:Aug124}

\smallskip
The length of spinup was considered as approximately $[1,T_{cycles} \cdot 2]$, where $T_{cycles}$ is the number of cycles considered, each cycle with length equal to 2.
The dots represent the accuracy attained for the calibration experiment defined in Subsection \ref{subsec:ODE001} with the corresponding observational parameter  $a_\text{Obs}$.

\smallskip
\textbf{Source:} the author (2026).
\end{figure}

\newpage\newpage
\subsection{An ODE system with cyclic equilibrium -- Part 2}\label{subsec:ODE002}

We now generalize the previous mathematical model in order to calibrate two parameters simultaneously.

\paragraph{Model and observations:}

Consider the ODE system:
\begin{equation}\label{edo2label}
\begin{cases}
    x' & = -y - \left( \dfrac{(t+c)^{-2}}{e^a+(t+c)^{-1}} \right) + \ln \left( e^b+(t+c)^{-1} \right) + y_0  ~,\\ \\
    y' & = x - \left( \dfrac{(t+c)^{-2}}{e^b+(t+c)^{-1}} \right) + \ln \left( e^a+(t+c)^{-1} \right) - x_0 ~,
\end{cases}    
\end{equation}
with fixed $c>0$, constants $x_0$, $y_0$ and parameters $a,b$.
System (\ref{edo2label}) can be seen as a generalization of system (\ref{edo21label}) and is well-defined for any $t > 0$. The analytical solution of the system (\ref{edo2label}) is:

\begin{equation}\label{Soledo2label}
\begin{cases}
    x(t) & = \ln \left( e^a+(t+c)^{-1} \right) + \cos(t) + x_0  ~,\\ \\
    y(t) & = \ln \left( e^b+(t+c)^{-1} \right) +\sin(t) + y_0 ~.
\end{cases}    
\end{equation}

For $x_0 = y_0 = 0$ and $t \gg 0$, we obtain:
\begin{align*}
    \ln \left( e^a+(t+c)^{-1} \right) & \approx a ~, \\
    \text{and ~~~} \ln \left( e^b+(t+c)^{-1} \right) & \approx b ~,
\end{align*}
so that 
\begin{equation}\label{Limitedo2label}
    (x(t), y(t)) \approx (a + \cos(t), b + \sin(t)) ~,
\end{equation}
which is again a limit cycle. 

We fix $c>0$ and address the problem of calibrating parameters $a,b>0$. Given observational parameters $a_\text{Obs}, b_\text{Obs}$, we generate synthetic observations as:
$$(x_\text{Obs}(t), y_\text{Obs}(t)) = (a_\text{Obs} + \cos t,\; b_\text{Obs} + \sin t) ~,$$
at fixed time instants chosen as in Part 1 (subsection \ref{subsec:ODE001}): after a spinup, over one or more cycles, with $s=\pi t$ for convenience.

\paragraph{Misfit residual definition:}

As in Part~1, we consider $N+1$ points along the last cycle and define the residual vector by pointwise distances between the model trajectory $(x(t),y(t))$ and the observational trajectory $(x_\text{Obs}(t),y_\text{Obs}(t))$. The misfit function is
$$  f_\text{Misfit}(a,b) = \sum_{i=1}^{N+1} r_i(a,b)^2 ~,$$
with $r_i$ defined in the same way as before.

\paragraph{Numerical experiment (two-parameter limit-cycle calibration):} 

\paragraph{Experiment 2.}
In this experiment, we apply DFO-LS to simultaneously recover the parameters $a_\text{Obs}$ and $b_\text{Obs}$, using a similar setup as in Part~1 except that here we set $c=10^{-4}$ and now we consider a two-dimensional parameter space:
\begin{enumerate}
    \item We set observational parameters $[a_\text{Obs}, b_\text{Obs}]$ and generate observations from~\eqref{Limitedo2label}.
    \item We integrate system~\eqref{edo2label} from an initial point away from the limit cycle and record samples after a spinup.
    \item We perform a heuristic step where we evaluate $f_\text{Misfit}(a,b)$ on a grid of parameter values in a rectangular neighborhood around $[a_\text{Obs},b_\text{Obs}]$ and select the best point as the initial guess for DFO-LS.
    \item We set bounds as a smaller rectangle centered on the initial guess and then run DFO-LS.
\end{enumerate}

Figure~\ref{fig:Aug1241} shows contour plots of the heuristic evaluation of the misfit function for observational parameters $[a_\text{Obs}, b_\text{Obs}] = [2.5, 3.5]$ and two choices for the number of observation, $N+1$. In both cases the observational parameters appear as the global minimizer among the sampled points.

\begin{figure}[htpb!]
\caption{Heuristic contour plot of $f_\text{Misfit}(a,b)$.} 
\begin{center}
\includegraphics[scale=0.7]{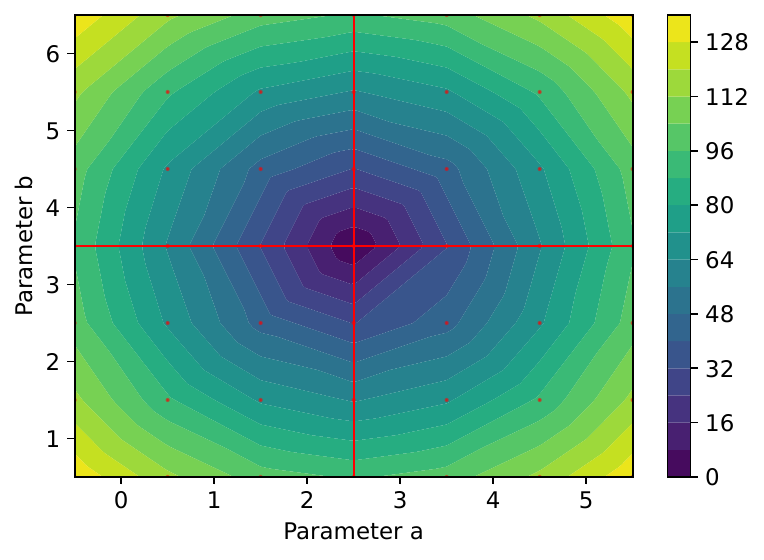}    
\end{center}
\label{fig:Aug1241}

\smallskip
Observational parameters are $[a_\text{Obs}, b_\text{Obs}] = [2.5, 3.5]$, considering $N=10$ observations, and function evaluation at 49 parameter pairs. The highlighted lines intersect at the location of the observational parameters.

\smallskip
\textbf{Source:} the author (2026).
\end{figure}

Figure \ref{fig:Aug125} shows the resulting fitting of the predicted values to the observations for $x$ and $y$, again using the sampling structure of Figure \ref{fig:018}. The behavior is qualitatively similar to the one-parameter case, with successful recovery of $a_\text{Obs},b_\text{Obs}$ over a range of settings.

\subsection{A difficult case: when systematic calibration may fail}

We now consider a model that illustrates intrinsic limitations of systematic calibration when the parameter is non-identifiable from the observations, that is, there is no way to ensure the recovering of observational parameters by fitting model outputs to observations.

\paragraph{Model and observations:} Consider the ODE system:
\begin{equation}\label{odeSinCos}
    \begin{cases}
        x' & = a y  ~,\\
        y' & = -a x ~.
    \end{cases}
\end{equation}

The non-trivial solution of system (\ref{odeSinCos}) is given by:
\begin{equation}\label{SolodeSinCos}
    \begin{cases}
        x(t) & = \sin(a t) ~, \\
        y(t) & = \cos(a t) ~.
    \end{cases}
\end{equation}

Let $a_\text{Obs}$ be the observational parameter. If suitable bounds were not set for the search of parameter $a$, it may become impossible to recover  $a_\text{Obs}$ via calibration. In this example, for any $a_\text{Obs} \in \R$, we see that:
\begin{equation}\label{equivSolodeSinCos}
    \begin{cases}
        x(t) & = \sin(\bar{a} t) ~, \\
        y(t) & = \cos(\bar{a} t) ~,
    \end{cases}
\end{equation}
whenever $\bar{a} = a_\text{Obs} + 2\pi \cdot k$, for any $k \in \mathbb{Z}$. Thus, in this case, it is important to ensure that the length of the search range for the parameter $a$ is less than $2\pi$.

To remove this ambiguity, we define a constrained optimization problem as
\begin{align*}
    \text{minimize} ~~& f_\text{Misfit}(a)  \\
    \text{subject to:} ~~& a \ \in  [0.001, 2\pi - 0.001]
\end{align*}
and set $a_\text{Obs}=1$ as the parameter to be recovered. Observations are defined as
$$(x_\text{Obs}(t),y_\text{Obs}(t)) = (\sin(a_\text{Obs} t), \cos(a_\text{Obs} t)) ~,$$
at a predefined set of time instants.

\vspace{0.7cm}
\begin{figure}[htpb!]
\caption{Fitting of the model output to the observational data.}
\begin{center}
\includegraphics[scale=0.7]{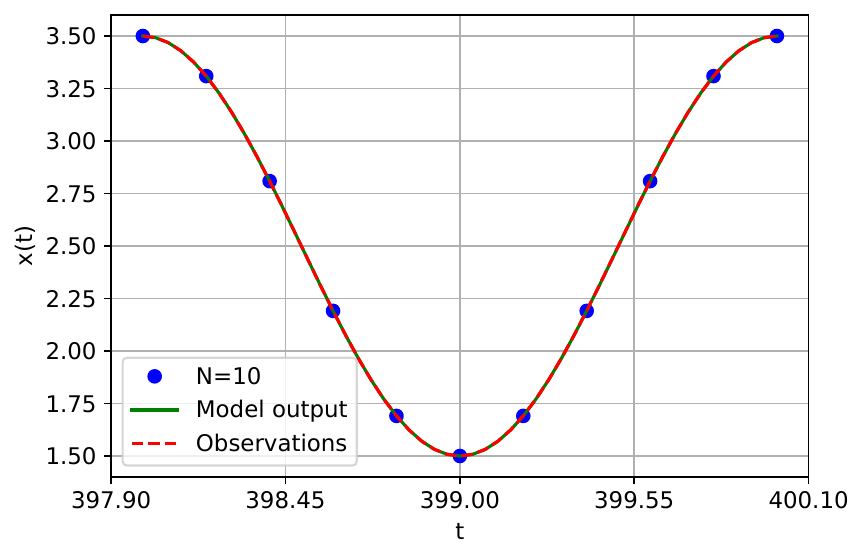}

\vspace{0.3cm}
\includegraphics[scale=0.7]{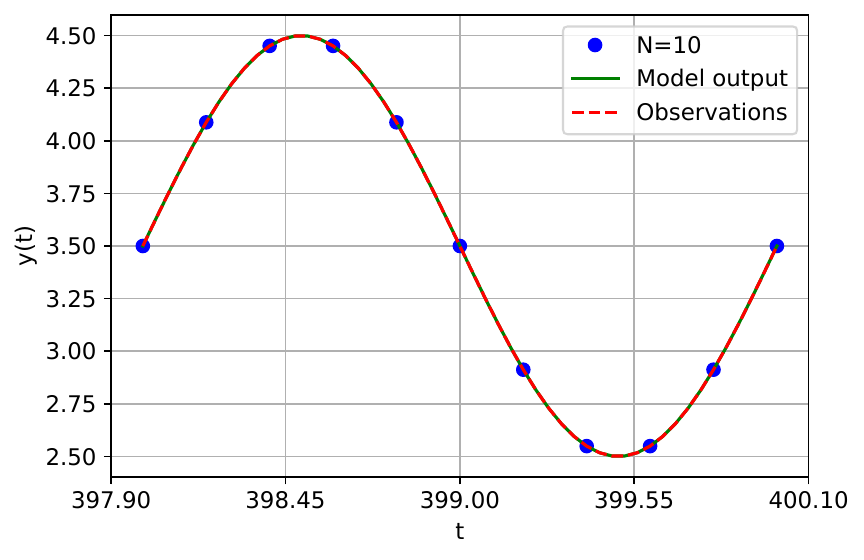}    
\end{center}
\label{fig:Aug125}

\smallskip
Considering the model defined by the equilibrium of equations (\ref{edo2label}), for $[a_\text{Obs}, b_\text{Obs}] = [2.5, 3.5]$ and  $N+1 = 11$ samples, following the structure of Figure~\ref{fig:018}.

\smallskip
\textbf{Source:} the author (2026).
\end{figure}

\paragraph{Misfit residual definition:} We consider two types of observations and corresponding residual entries:
\begin{enumerate}
    \item \textbf{Type-1 (pointwise)}: direct comparisons between $(x(t),y(t))$ and $(x_\text{Obs}(t),y_\text{Obs}(t))$ at selected time instants.
    \item \textbf{Type-2 (aggregated)}: groupings of type-1 observations into blocks and comparison of their averages.
\end{enumerate}

Let $m_1$ be the number of type-1 residuals and $m_2$ the number of type-2 residuals, with $m = m_1 + m_2$. We introduce a weighting constant $\alpha \in [0,1]$ and define:
\begin{equation*}
   r_\text{Misfit}(a) = \left[ \dfrac{\alpha^2}{m_1} \ r_1(a)\ , \ \dfrac{(1-\alpha)^2}{m_2} \ r_2(a) \right] 
\end{equation*}
where $r_1(a)\in\mathbb{R}^{m_1}$ contains type-1 residuals, and $r_2(a)\in\mathbb{R}^{m_2}$ contains type-2 residuals (for example, the averages over fixed groups of type-1 misfits). The misfit function is again the squared norm of $r_\text{Misfit}(a)$.

The choice of $(m_1,m_2,\alpha)$ controls the dimension and structure of the residual and hence the cost and conditioning of the least-squares problem.

\paragraph{Numerical experiment (non-identifiable sinusoidal system):}

\paragraph{Experiment 3.}
We consider the problem of recovering $a_\text{Obs}=1$ with feasible region $\Omega = [0.001, 2\pi-0.001]$. We test several configurations of $(m_1,m_2,N,\alpha)$, where $N+1$ is the number of time samples used to define type-1 residuals:
$$m_1 = 5,\quad N=20,\quad \alpha = 1 ~,$$
and variants with $m_2>0$ and $\alpha\in(0,1)$.

Figures~\ref{figsincos:001} and \ref{figsincos:002} show heuristic plots of $f_\text{Misfit}(a)$ in a neighborhood of $a_\text{Obs}$ for two different residual settings. Although the global minimum is close to $a_\text{Obs}$ and $f_\text{Misfit}\approx 0$ near it, the misfit function has many local minima. 
A zoom around $a_\text{Obs}$ (Figure~\ref{figsincos:003}) shows the highly oscillatory structure that can trap DFO-LS in a local minima away from $a_\text{Obs}$. 

\vspace{0.5cm}
\begin{figure}[htpb!]
\caption{Heuristic evaluation of $f_\text{Misfit}(a)$.} 

\vspace{-0.1cm}
\begin{center}
\includegraphics[scale=0.7]{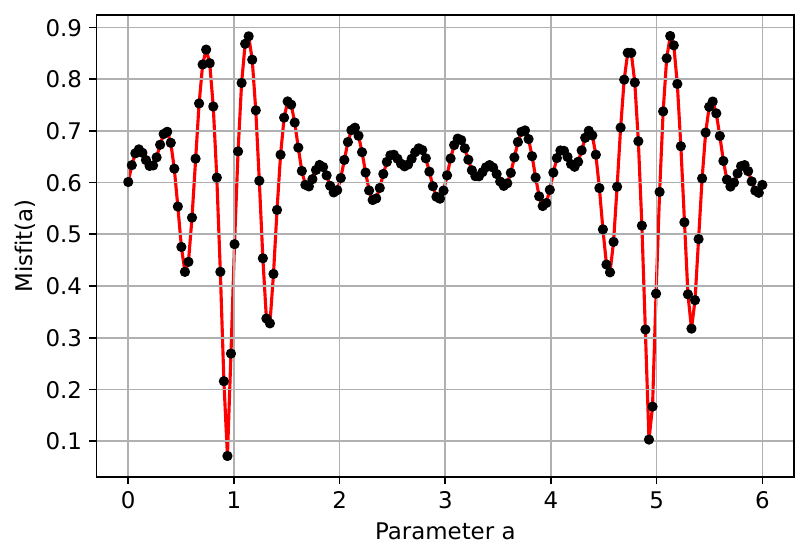}    
\end{center}
\label{figsincos:001}

\vspace{-0.1cm}
Evaluation at 180 points in a neighborhood of  $a_\text{Obs}=1$ (excluding $a_\text{Obs}$ itself) for $m_1=5$, $N=20$, $\alpha=1$ (type-1 residuals only). Black dots indicate function evaluations, and the red line is an interpolation.

\smallskip
\textbf{Source:} the author (2026).
\end{figure}

\begin{figure}[htbp!]
\caption{Heuristic evaluation of $f_\text{Misfit}(a)$.} 
\begin{center}
\includegraphics[scale=0.7]{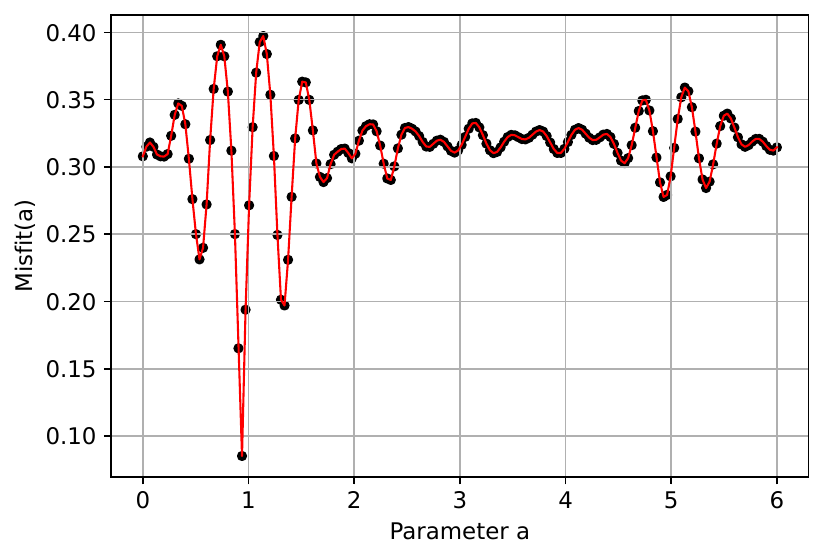}    
\end{center}
\label{figsincos:002}

\smallskip
Evaluation at 180 points in a neighborhood of  $a_\text{Obs}=1$ (excluding $a_\text{Obs}$ itself) for $m_1=5$, $m_2=3$, $N=20$, and $\alpha=0.5$ (combining type-1 and type-2 residuals). Black dots indicate function evaluations, and the red line is an interpolation.

\smallskip
\textbf{Source:} the author (2026).
\end{figure}

\begin{figure}[htbp!]
\caption{Heuristic evaluation of $f_\text{Misfit}(a)$.} 
\begin{center}
\includegraphics[scale=0.7]{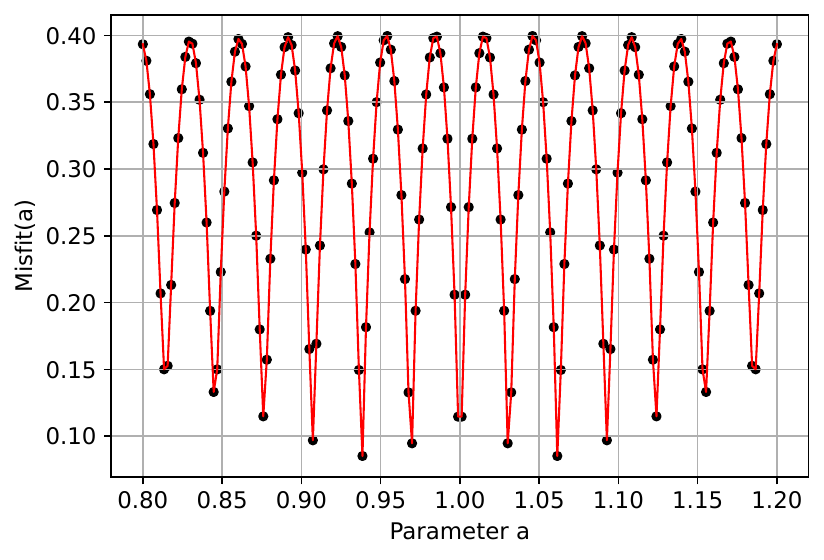}    
\end{center}
\label{figsincos:003}

\smallskip
Evaluation at 180 points in a close  neighborhood of  $a_\text{Obs}=1$ (excluding $a_\text{Obs}$ itself) for the same setting as Figure~\ref{figsincos:002}. Black dots indicate function evaluations, and the red line is an interpolation.

\smallskip
\textbf{Source:} the author (2026).    
\end{figure}

If we artificially restrict the search region to a very small interval around $a_\text{Obs}$ (Figure~\ref{figsincos:004}), DFO-LS can typically recover $a_\text{Obs}$. However, this requires prior knowledge of the true parameter and therefore does not constitute a practical calibration strategy.

\vspace{0.5cm}
\begin{figure}[htbp!]
\caption{Heuristic evaluation of $f_\text{Misfit}(a)$.} 

\vspace{-0.2cm}
\begin{center}
\includegraphics[scale=0.7]{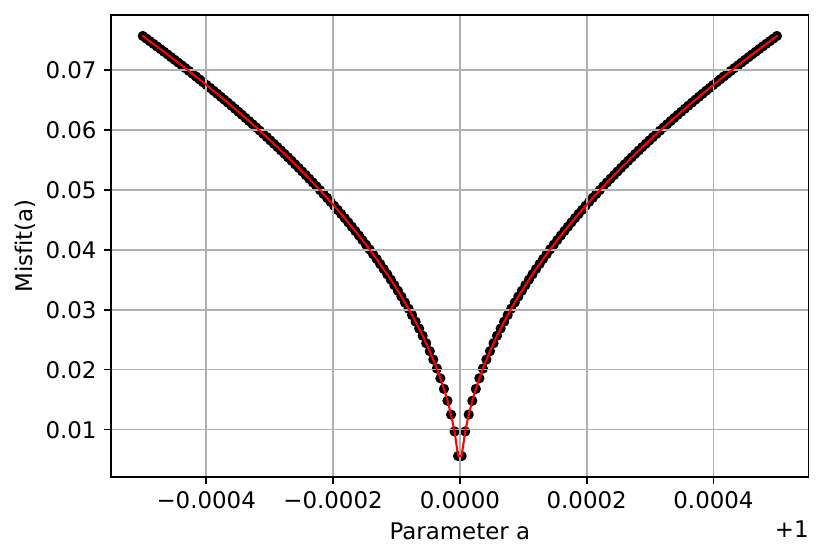}    
\end{center}
\label{figsincos:004}

Evaluation at 180 points in a very small neighborhood of $a_\text{Obs}=1$ (excluding $a_\text{Obs}$ itself) for the same setting as Figure~\ref{figsincos:002}. Black dots indicate function evaluations, and the red line is an interpolation.

\smallskip
\textbf{Source:} the author (2026).
\end{figure}

\paragraph{Last comments:}
\begin{itemize}
    \item For simplicity, we set the observational parameter as $a_\text{Obs}=1$ in all these experiments, since for other observational parameters it is possible to follow the same procedure.
    
    \item As we will see in the following subsection, using heuristics to obtain a good initial guess is possible, which helps the DFO-LS solver find the solution more quickly. However, it is also possible to use an arbitrary initial guess, such as approximately the midpoint of the feasible interval, $a_0 = 3.14$.

    \item Regarding the type-2 of observations, we initially considered taking the value of the definite integral of the half-period ODE system as an observation and comparing it to the Riemann sum, obtained by manipulating the samples. However, this is another limitation for this calibration: When the estimate produced by the misfit function is coarse, comparing it to a precise theoretical reference does not benefit the calibration's accuracy: on the contrary, it can worsen it. In this sense, an important feature in defining observations is taking into account the model's ability to replicate them.
\end{itemize}

This example shows that even with a carefully constructed residual and a sophisticated optimizer, successful calibration may be impossible when the parameter is effectively non-identifiable from the available observations.

\subsection{A nonlinear example: when the analytical solution is unknown}\label{subsec:Selkov}

We now turn to a nonlinear system with a limit cycle and no closed-form solution, allowing us to test the calibration strategy in a setting closer to real biogeochemical models.

\paragraph{Model and feasible region:} We consider the Sel'kov model (Strogatz \cite{strogatz2018nonlinear}, Example~7.3.3, p.~208):
\begin{equation}\label{eqStrogatz}
\begin{cases}
    \dfrac{dx}{dt} & = -x + y \cdot (a+x^2)  ~,\\ \\
    \dfrac{dy}{dt} & = b - y \cdot (a+x^2)  ~.
\end{cases}
\end{equation}
where $a>0$ and $b>0$ are parameters. For suitable choices of $(a,b)$, the system admits a stable fixed point or a stable limit cycle. In particular, for parameters in the region bounded by
$$b^2 = \frac{1}{2}\left(1-2a \pm \sqrt{1-8a}\right), \qquad a,b>0 ~,$$
which we define as the feasible region for the parameters in the following experiments (Figure \ref{StrogatzFeasibleRegion}), trajectories starting away from the equilibrium converge to a limit cycle, while trajectories starting sufficiently close to the equilibrium may converge to the fixed point.

\vspace{0.5cm}
\begin{figure}[htpb!]
\caption{Convergence to a limit cycle.}

\vspace{-0.2cm}
\begin{center}
\includegraphics[scale=0.7]{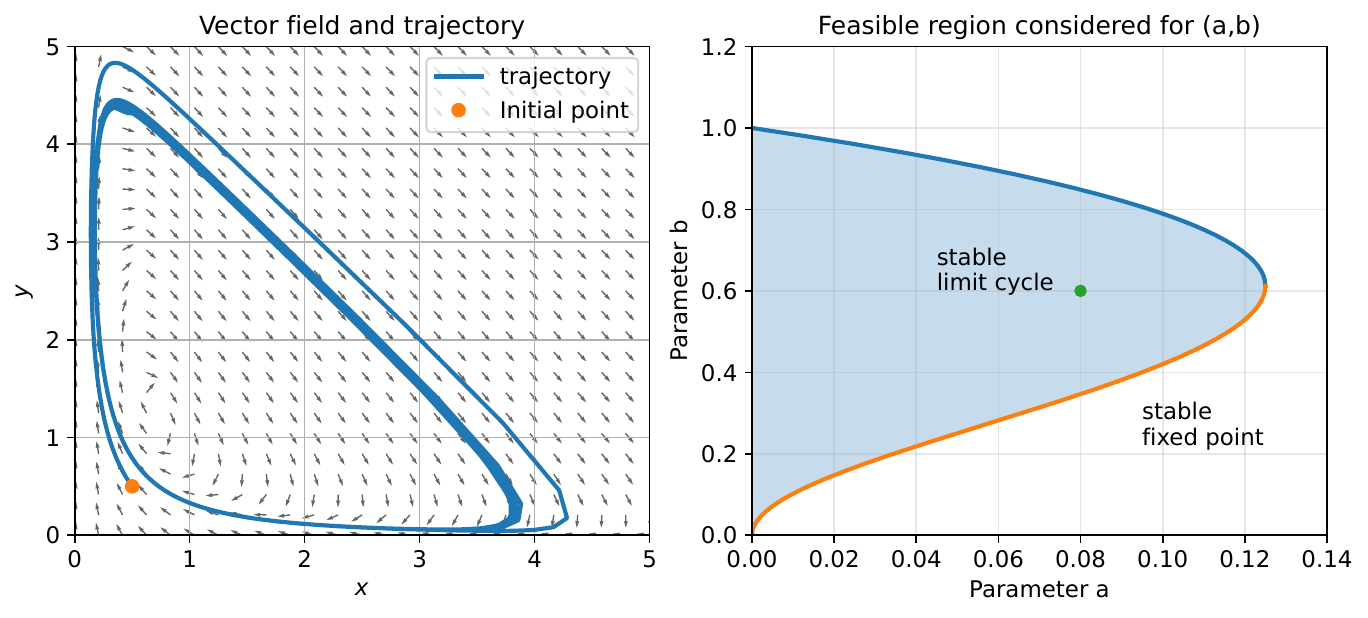}    
\end{center} 
\label{StrogatzFeasibleRegion}

Right-hand panel shows a point $(\Bar{a},\Bar{b})$ in the region defined by $b^2 = \frac{1}{2}\big(1-2a \pm \sqrt{1-8a}\big)$, $a,b>0$. Left-hand panel shows the trajectories obtained by integrating the system (\ref{eqStrogatz}) with parameters $(a,b) = (\Bar{a},\Bar{b})$ and an initial point away from the equilibrium.

\smallskip
\textbf{Source:} the author (2026).
\end{figure}

\paragraph{Producing the observations:}

Once we choose a pair of observational parameters inside the feasible region presented in Figure \ref{StrogatzFeasibleRegion}, the observations are constructed in two stages:
\begin{enumerate}
    \item \textbf{Finding a point on the limit cycle:} we integrate~\eqref{eqStrogatz} from a point slightly displaced from the fixed point and let the trajectory evolve for a long time, so that it converges to the limit cycle. The last point of this integration, denoted $y_0$, is taken as an approximation of a point on the observational limit cycle.
    \item \textbf{Sampling and interpolation:} starting from $y_0$, we integrate~\eqref{eqStrogatz} again over a time interval long enough to cover at least one full period of the limit cycle, and we save $N+1$ points along this integration. The resulting set of points constitutes the raw observational set. We then perform a linear interpolation to densify the observational curve so that the distance between consecutive points does not exceed a prescribed value $\Delta c$ (Figure~\ref{fig:001}).
\end{enumerate}

The parameter $\Delta c$ is initially chosen empirically and can later be scaled relative to the diameter of the rectangle containing the limit cycle. In this way, we balance two competing goals: 
filling in gaps along the observational limit cycle and avoiding an excessive number of interpolation points, which would increase computational cost. Figure~\ref{fig:001} compares the observational data, sampled from the approximated limit cycle obtained by integrating system (\ref{eqStrogatz}) with observational parameters $[a_\text{Obs}, b_\text{Obs}] = [0.02, 0.6]$ and $N+1=401$ observations, and the interpolated dataset.

\vspace{0.5cm}
\begin{figure}[htpb!]
\caption{Observational data sampled from the approximated limit cycle for $[a_\text{Obs}, b_\text{Obs}] = [0.001, 0.6]$, $N+1=401$ points (left panel), and the dataset obtained when adding interpolation points (right panel).}
\begin{center}
\includegraphics[scale=0.7]{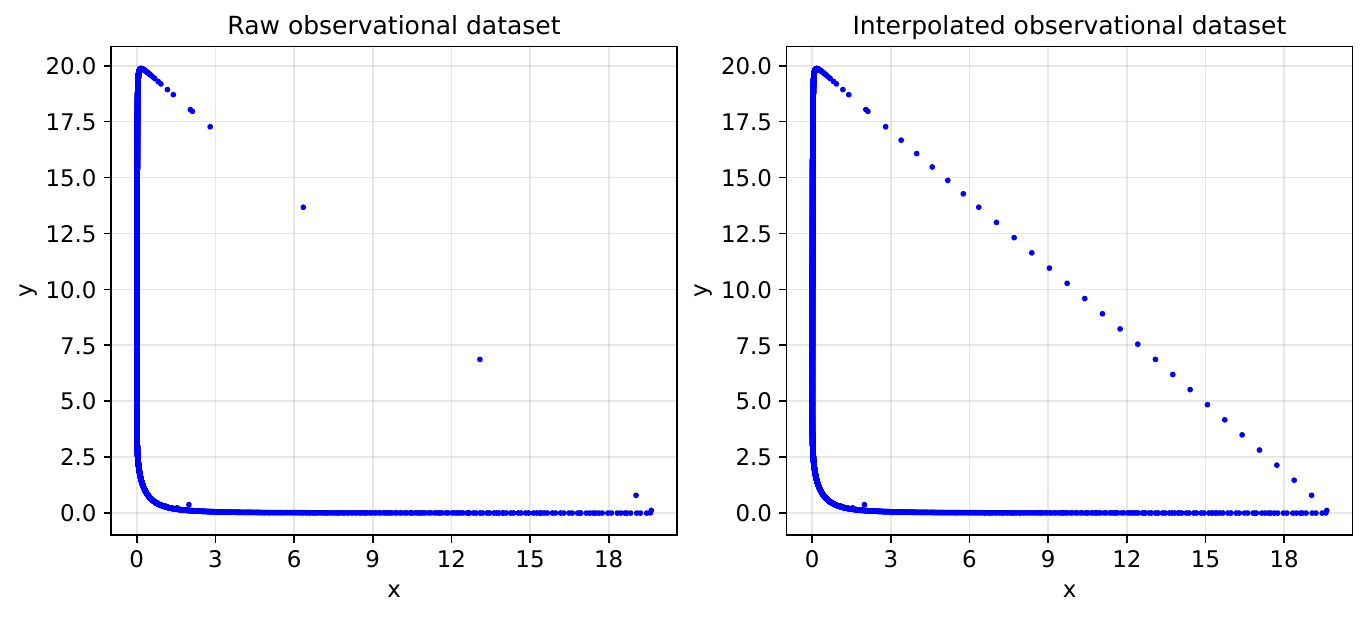}  
\end{center}
\label{fig:001}

\smallskip
\textbf{Source:} the author (2026).
\end{figure}

To implement the data interpolation, we first reorder the points obtained from the integration of the observational limit cycle. Then, the distance between consecutive points is calculated. If this distance is greater than $\Delta c$, points from the line segment connecting these two consecutive points are added to the observation set.
This construction assumes that the observational limit cycle is (approximately) convex, ensuring that the piecewise linear interpolation is well-defined and does not introduce large geometric distortions.

In the following, we considered two options for setting the integration interval from which we sample the observations:
\begin{enumerate}
\item The last 100 integration times;
\item A time interval empirically expected to contain at least one, and at most two loops of the observational limit cycle.
\end{enumerate}
In all the experiments presented here, the first option corresponded to a larger interval than the second one.

\paragraph{Optimization heuristic step:} For the heuristic step, we define a simplified misfit function $f_H(a,b)$ comparing the distances between a small number of model output points and the observational dataset. The ODE system~\eqref{eqStrogatz} is integrated from $y_0$ over a relatively short time interval, and $N+1$ samples are taken at pre-defined time instants (which also need to be included when generating the observational set). The misfit $f_H(a,b)$ is taken as the least-squares sum of a four-dimensional version of the residual function, $r_H(a,b)$, defined as the first four entries of the residual function $r_\text{Misfit}$, which will be better described in the following.

We then build a coarse grid of parameter values $(a,b)$ within the feasible region and evaluate $f_H$ on this grid (Figure~\ref{fig:006}, left panel). The set of grid points that appear as local minimizers of $f_H$ is retained for further refinement.

\begin{figure}[ht]
\caption{Heuristic evaluation for recovering the parameter $a_\text{Obs} = 0.08$, where the observational parameter $b_\text{Obs} = 0.6$ was fixed.}

\vspace{-0.2cm}
\begin{center}
\includegraphics[scale=0.7]{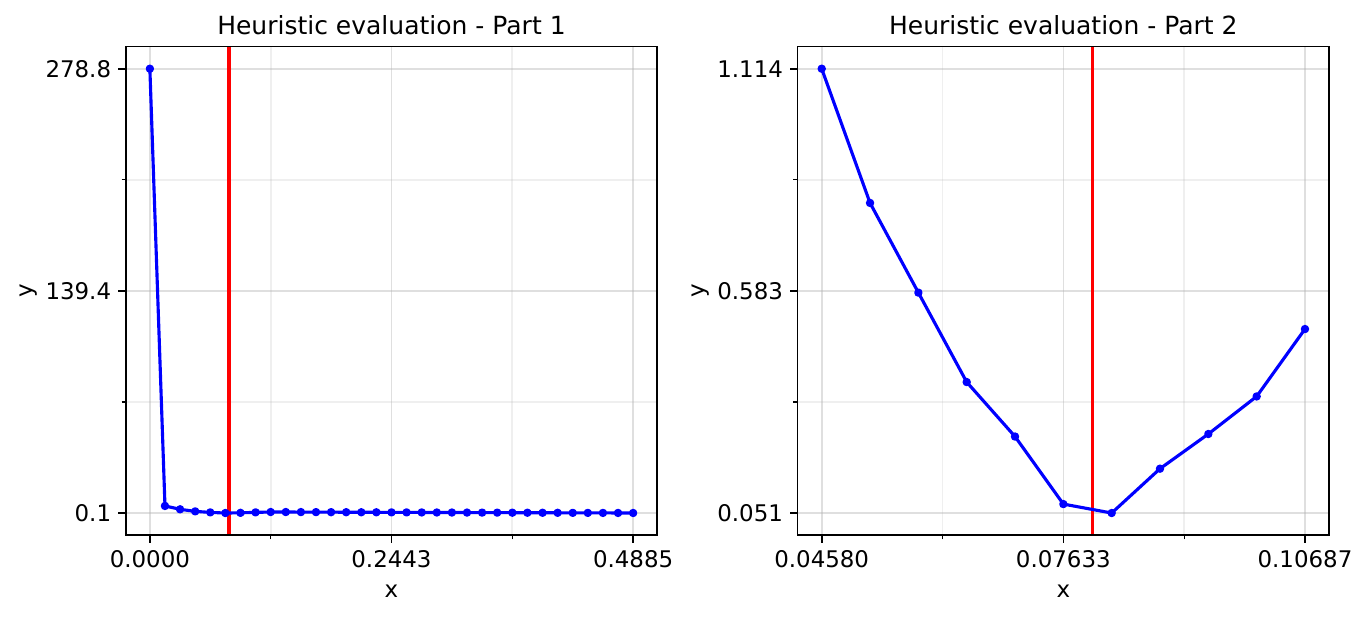}    
\end{center}
\label{fig:006}

Left panel: First part of the heuristic evaluation, where a light version of the misfit function, $f_H(a)$, was evaluated over a discretization of the feasible interval for $a$ into 33 points. 
Right panel: Heuristic refinement, with the evaluation of $f_H(a)$ at 11 points around each coarse local minimizer identified in the left panel, where a ceil is applied. In both plots, blue dots represent function evaluations, blue lines represent linear interpolation between these points, and the red vertical line identifies
the observational parameter $a_\text{Obs}$ on the x-axis.

\smallskip 
\textbf{Source:} the author (2026).

\vspace{0.4cm}
\end{figure}

\paragraph{Initial guess and bounds:}
After evaluating $f_H$ on the coarse grid, we select all local minimizers and, for each of them, we evaluate $f_H$ at a finer set of points in a small neighborhood. The median of the resulting misfit values is used to define $\rho_\text{max}$ (which we recall later), while the best local minimizer in this refined search is chosen as the initial guess for DFO-LS (Figure~\ref{fig:006}, right panel).

We then define personalized rectangular bounds around the initial guess, with width up to 20\% of the total feasible range in each parameter. Assuming the initial guess is reasonable, this restriction helps DFO-LS focus on the most relevant region and reduces the number of iterations.

\paragraph*{Misfit residual definition:}We now define the residual function used by DFO-LS to calibrate $(a,b)$.
After the heuristic step, we construct the residual function $r_\text{Misfit}(a,b)$, which has five components:
\begin{itemize}
    \item The first four components are the distance measures between the model output points and the observational dataset;
    \item A fifth component is an external penalty enforcing the feasibility constraints  for the parameters $(a,b)$.
\end{itemize}

The observational limit cycle is divided into four quadrants (Figure~\ref{fig:003}). For each of the $N+1$ points of the model output, we:
\begin{enumerate}
    \item find the closest observational point;
    \item compute the distance between this observation and the model output point;
    \item and assign this distance to the residual component associated with the quadrant containing that observational point.
\end{enumerate}  
To promote continuity between quadrants, we also use transition regions (Figure~\ref{fig:003}): points near the boundaries contribute to more than one component via a smooth transition.

\vspace{0.7cm}
\begin{figure}[htpb!]
\caption{Division of the observational limit cycle into four quadrants.}
\begin{center}
\includegraphics[scale=0.7]{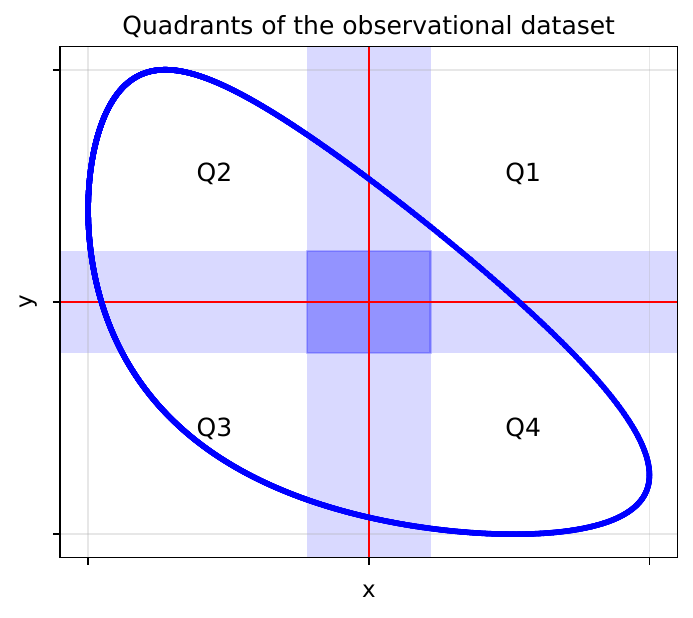}    
\end{center}
\label{fig:003}
\smallskip
The red lines demarcate the limits of each quadrant. Blue rectangles identify transition regions between quadrants.

\smallskip 
\textbf{Source:} the author (2026).

\vspace{0.4cm}
\end{figure}

\vspace{0.2cm}
The fifth residual component is a penalty term, denoted $r_\text{penalty}(a,b)$, that enforces the feasibility conditions
\[
8a \le 1,
\qquad
b^2 \le \frac{1}{2}\Big(1-2a+\sqrt{1-8a}\Big),
\qquad
b^2 \ge \frac{1}{2}\Big(1-2a-\sqrt{1-8a}\Big).
\]
We define
\[
r_\text{penalty}(a,b) =
\begin{cases}
  \rho_\text{max}, & \text{if } v_1(a,b) > 0,\\[0.5ex]
  \displaystyle
  \min\left\{\rho_\text{max},
  p_\text{scale}\,\max\big\{0,\, v_2(a,b),\, v_3(a,b)\big\}\right\},
  & \text{if } v_1(a,b)\le 0,
\end{cases}
\]
where
\begin{align*}
    v_1(a,b) & = 8a - 1,\\
    v_2(a,b) & = b^2 - \dfrac{1}{2} \Big(1-2a+\sqrt{1-8a}\Big),\\
    v_3(a,b) & = -b^2 + \dfrac{1}{2} \Big(1-2a-\sqrt{1-8a}\Big),
\end{align*}
$\rho_\text{max}>0$ is a ceiling value for the residual components, and $p_\text{scale}$ is a positive scaling constant (set to 1 in the experiments). The penalty ensures that points outside the feasible set incur a large residual.

The corresponding misfit function, $f_\text{Misfit}(a,b)$, is the squared norm of this five-dimensional residual vector.

\paragraph{Numerical experiments:} Once the five-dimensional residual function and the initial guess are defined, we run DFO-LS to calibrate $(a,b)$. The experiments are grouped into three tests, each corresponding to different calibration scenarios but all based on the same underlying example.

\paragraph{Experiment 4 (calibration of $a$ for fixed $b$).} We fix $b=b_\text{Obs}=0.6$ and consider
$$[a_\text{Obs}, b_\text{Obs}] \in
\left\{
[0.001,0.6], [0.02,0.6], [0.04,0.6], [0.06,0.6],
[0.08,0.6], [0.1,0.6], [0.12,0.6]
\right\} ~.$$
For each one of these seven cases, we run DFO-LS for $N+1\in\{26,101,401\}$ samples and compare the accuracy of the optimized parameter $\Bar{a}$ in the recovery of $a_\text{Obs}$.

Figure \ref{fig:test2_a_001} summarizes the results for two sampling strategies:
\begin{itemize}
    \item Samples taken over the last 100 units of integration time;
    \item Samples taken over the last 1-2 periods of the limit cycle (that is, a shorter time window but still covering at least one full cycle).
\end{itemize}

\vspace{0.4cm}
\begin{figure}[htpb!]
\caption{Experiment 4: accuracy in recovering $a_\text{Obs}$, with $b_\text{Obs} = 0.6$ fixed, depending on the number of samples $N+1$ $N+1$.}
\begin{center}
\includegraphics[scale = 0.7]{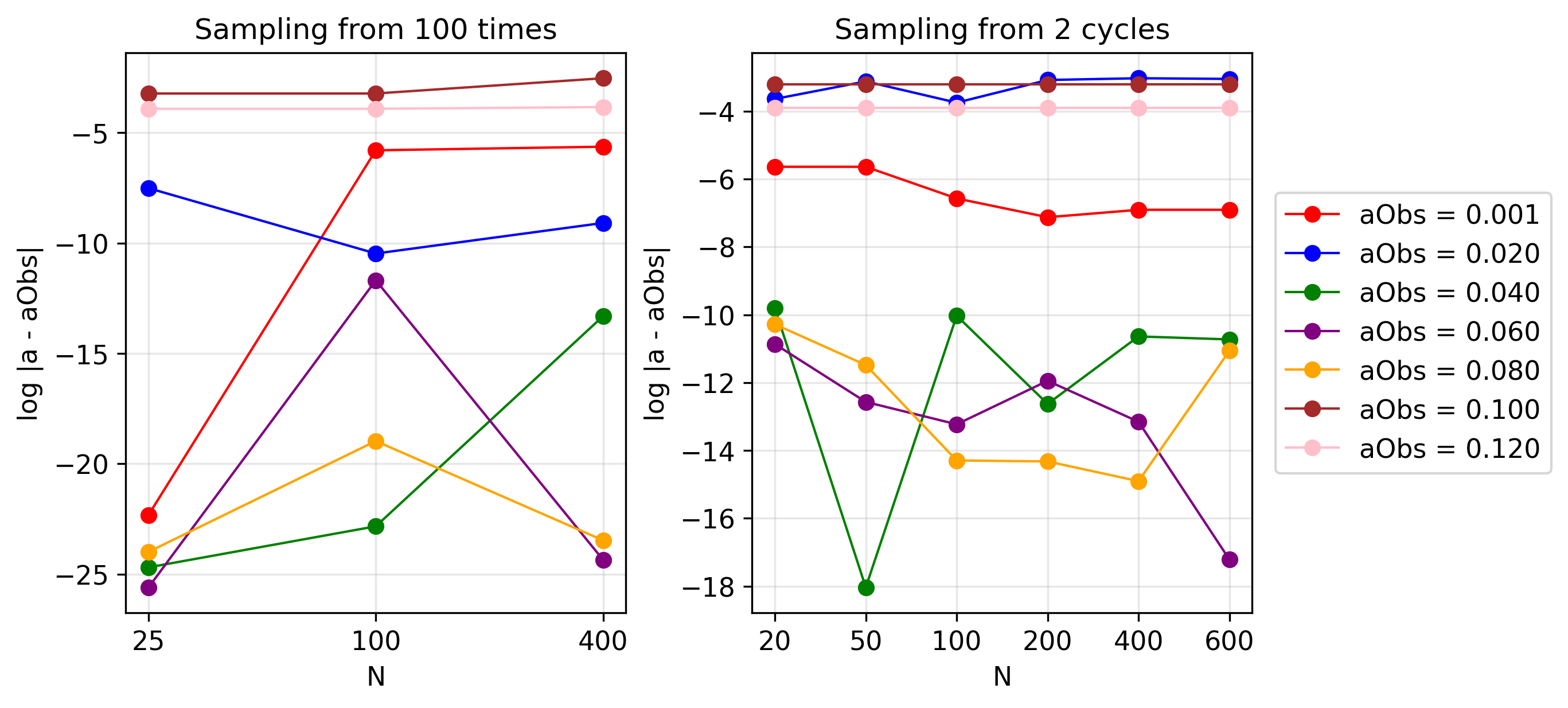}    
\end{center}
\label{fig:test2_a_001}

\smallskip
Results when samples were taken over the last 100 time units of the model integration (left), and results when samples were taken over approximately 1--2 periods of the limit cycle obtained by the model integration (right).

\smallskip
\textbf{Source:} the author.
\end{figure}

\paragraph{Experiment 5 (calibration of $b$ for fixed $a$).} Here we fix $a=a_\text{Obs}=0.04$ and vary $b_\text{Obs}$:
$$[a_\text{Obs}, b_\text{Obs}] \in
\left\{
[0.04,0.3], [0.04,0.4], [0.04,0.6],
[0.04,0.7], [0.04,0.9]
\right\} ~.$$

Again, we consider $N+1\in\{26,101,401\}$ and compare sampling over the last 100 time units versus sampling over the last 1--2 periods. The results are shown in Figure \ref{fig:test3_b_001}.

\begin{figure}[htpb!]
\caption{Experiment 5: accuracy in recovering $b_\text{Obs}$, with $a_\text{Obs} = 0.04$ fixed, depending on the number of samples $N+1$ $N+1$.}
\begin{center}
\includegraphics[scale=0.7]{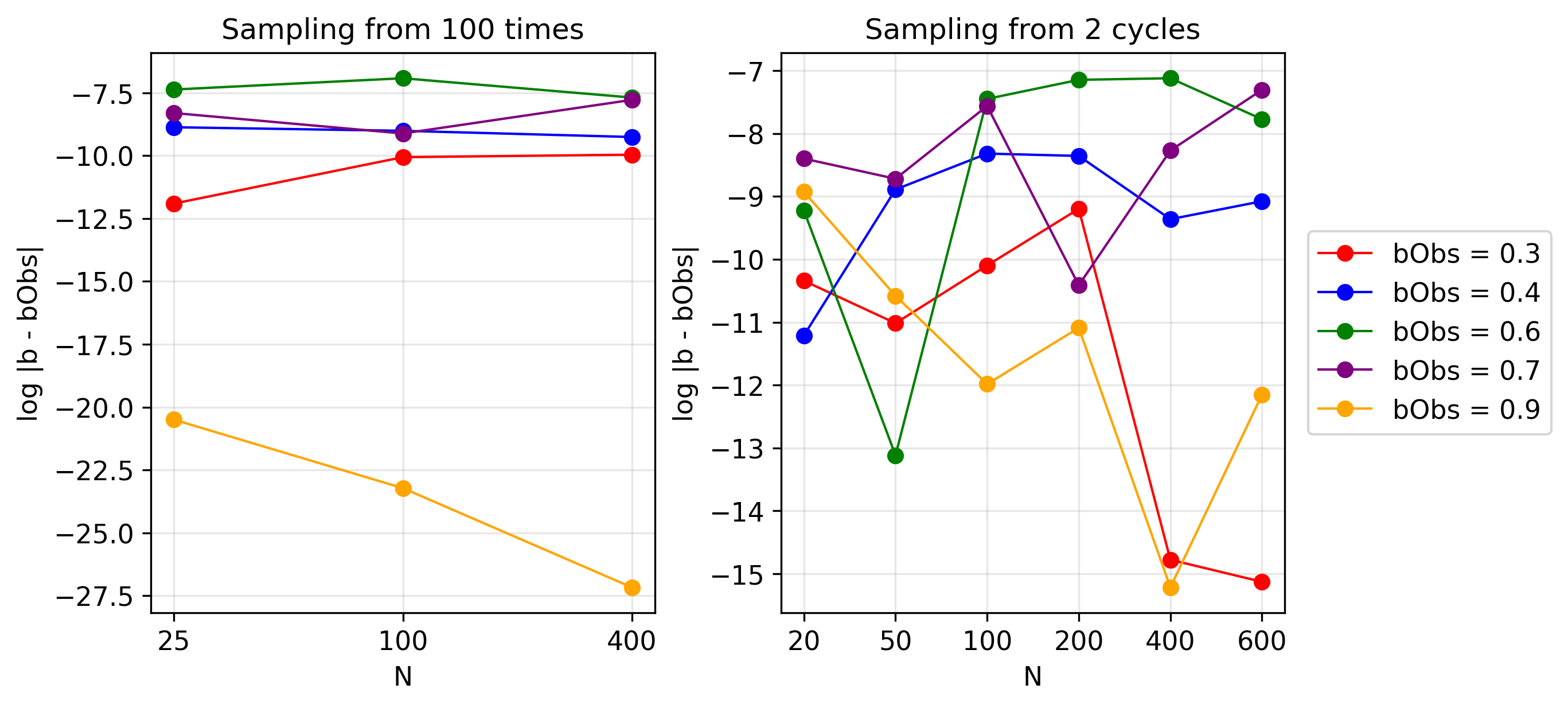}    
\end{center}

\smallskip
Results when samples were taken over the last 100 time units of the model integration (left), and results when samples were taken over approximately 1--2 periods of the limit cycle obtained by the model integration (right).
\label{fig:test3_b_001}

\smallskip
\textbf{Source:} the author.

\vspace{0.4cm}
\end{figure}

\paragraph{Experiment 6 (simultaneous calibration of $a$ and $b$).}
In the final set of experiments, we calibrate both parameters simultaneously for
$$[a_\text{Obs}, b_\text{Obs}] \in
\left\{
[0.02,0.4], [0.02,0.6], [0.02,0.8],
[0.06,0.6], [0.12,0.6]
\right\} ~,$$
and $N+1\in\{26,101,401\}$. Samples are taken over the last 100 time units. Figure \ref{fig:test4_ab_001} shows the calibration accuracy for different $(a_\text{Obs}, b_\text{Obs})$ and values of $N+1$.

\paragraph{Last comments:}
\begin{itemize}
    \item To obtain the first 4 entries of the residual array, we divide the set of observations into quadrants. This strategy enables us to assign different weights to each quadrant, for example, proportionally to the number of observations located in each.

    \item For very small values of $a_\text{Obs}$, the Sel'kov system becomes stiff, and the misfit function may become ill-conditioned in a neighborhood of the observational parameters. This behavior becomes clear when attempting to recover $a_\text{Obs}=0.001$ while $b_\text{Obs}=0.6$ is fixed. In this case, the heuristic misfit $f_H(a)$ is poorly behaved near $a_\text{Obs}$ (Figure~\ref{fig:013}), leading to a poor initial guess. As a consequence, the personalized bounds may exclude $a_\text{Obs}$, and even when the full feasible region is used, the DFO-LS search tends to move away from the observational parameter $a_\text{Obs}$. 
\end{itemize}

\vspace{0.7cm}
\begin{figure}[htpb!]
\caption{Experiment 6: accuracy in simultaneously recovering $a_\text{Obs}$ and $b_\text{Obs}$ depending on the number of samples $N+1$.}
\begin{center}
\includegraphics[scale=0.7]{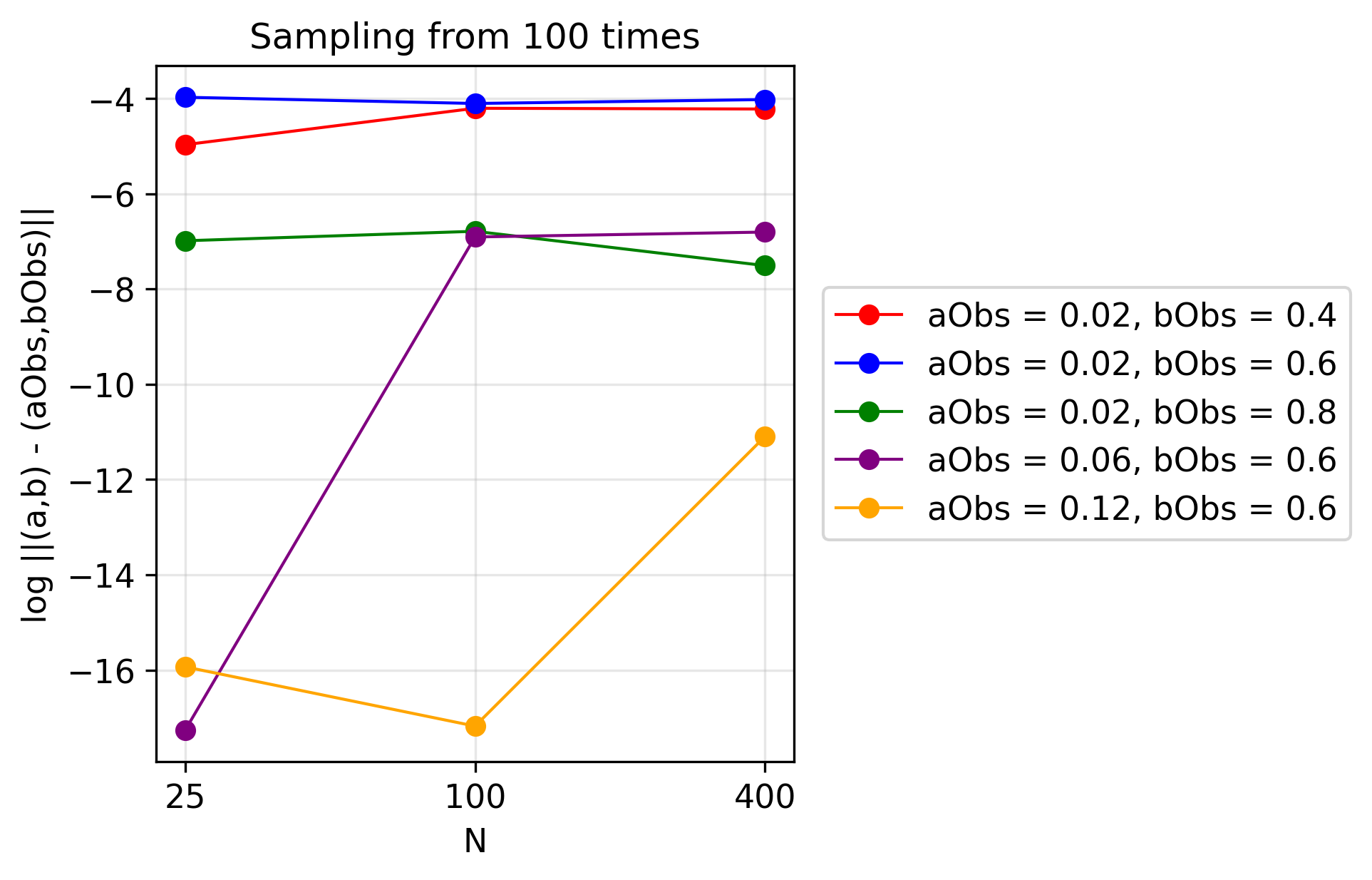}    
\end{center}
\label{fig:test4_ab_001}

\smallskip
Samples were taken over the last 100 time units of the model integration.

\smallskip
\textbf{Source:} the author.
\end{figure}

\begin{figure}[htpb!]
\caption{Heuristic evaluation for recovering the parameter $a_\text{Obs} = 0.001$, where the observational parameter $b_\text{Obs} = 0.6$ was fixed.} 
\begin{center}
\includegraphics[scale=0.7]{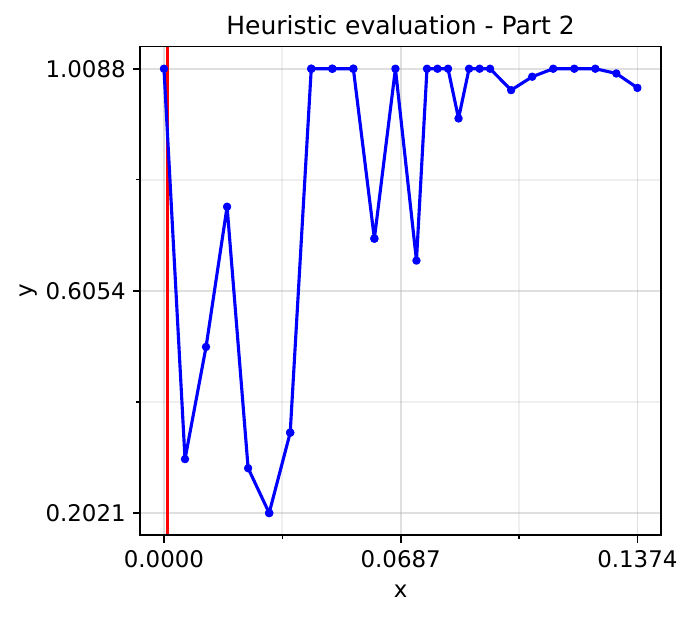}
\end{center}
\label{fig:013}

The function is ill-conditioned around $a_\text{Obs}$, leading to a bad choice of the initial guess. Points represent $f_H(a)$ evaluations with ceil, the red vertical line identifies the observational parameter $a_\text{Obs}$ on the x-axis, which is a global minimizer for $f_H(a)$, not included in the evaluations. See the right panel of Figure \ref{fig:006}.

\smallskip
\textbf{Source:} the author.

\vspace{0.3cm}
\end{figure}

\section{Results and Discussion}
The experiments in this chapter highlight several modeling choices that strongly influence the quality and robustness of parameter calibration:

\paragraph{Heuristic steps:} During the experiments in this chapter, we noted that not all calibration problems require a heuristic step. Well-behaved problems can even recover observational parameters using simpler strategies than those presented here.
However, when dealing with more complex problems, performing a heuristic search to obtain a good initial guess generally reduces the computational cost of the calibration process, especially when the function used for the heuristic is computationally lighter than the function used for the calibration parameter search. Still, note that we cover examples of calibration with only one or two parameters: In problems with several parameters to be calibrated, it may be impractical to reproduce the heuristic step that considers a discretization of the feasible space of each variable, and so we may restrict the heuristic step to sample a set of points with less coverage of the search space.

\paragraph{Misfit residual features:}
From the experiments in this chapter, we note that certain modeling and optimization choices can impact the quality of the results. One of them is the choice of an upper limit for the residual function entries (denoted as $\rho_\text{max}$ in Subsection \ref{subsec:Selkov}), which we will refer to as a \textit{ceil}. This feature ensures that the search algorithm considers regions above the ceil as unattractive while mitigating scale variations that could compromise the accuracy of the results. By setting the ceil after the heuristic step, we ensure that this value is not too small, which would overly restrict the search region and make it difficult for DFO-LS to find a descent path to reach the minimizer. First, we have set the ceil as the maximum value of the misfit function obtained during the heuristic. This has allowed us to preserve all possible descent directions of the misfit function within the feasible set, a strategy that prioritizes exploration. Another choice was setting the ceil as the median of the heuristic misfit function values sampled in a neighborhood of local minimizers, in a strategy that prioritizes optimization accuracy. We considered the second strategy for the experiments in Subsection \ref{subsec:Selkov}.

\paragraph{Effect of sample distribution and quadrant weights:}
In the Subsection \ref{subsec:Selkov}, we represented the parameter space into quadrants, and considered transition regions between quadrants to ensure that the first four entries of the residual function effectively serve as coordinates. These transition regions aim to serve as a continuity factor between the first four entries of the residual array, achieved through a covering of the plane by four cards, and transition homeomorphisms between them.

When sampling more observations, we noticed that the accuracy of recovering observational parameters does not necessarily increase. Nevertheless, the overall accuracy remains very good. This phenomenon is due to the uneven distribution of points along the observational limit cycle obtained from the integration, although evenly distributed over time.
Thus, in some regions of the limit cycle, the residual function compares the model output points only to the interpolated observational data. If there are model output points in a quadrant with these characteristics, there will be an additional error to the calibration. Even so, one precaution that should be taken is to avoid assigning weights too close to zero to quadrants containing few or no observations, as ignoring one of the quadrants leads to a calibration strategy similar to comparing model output and observational points one-on-one, which may only be effective in a small neighborhood of the optimal parameters.

In this context, an additional improvement is obtained by selecting the $N+1$ samples from a shorter interval that empirically contains one to three periods of the limit cycle, rather than from a very long time window. In this case, the observational points tend to be more evenly distributed among quadrants, and interpolation becomes less critical.

\paragraph{{Framework limitations and possible developments:}} 
There is a wide variety of optimization algorithms (for example, trust-region, direct search, model-based, evolutionary, surrogate, and Bayesian methods), and their performance depends strongly on characteristics of the problem, such as the amount of noise in the data, how smooth the function is, and the number of parameters \cite{Larson2019}. For computationally expensive models, optimization methods that construct simpler probabilistic or regression-type approximations of the misfit function are particularly useful, as they can reduce the number of costly ODE integrations and handle noisy data in a natural manner \cite{Muller2019, Ambrosio2017}. In problems involving the calibration of multiple parameters simultaneously, multi-objective versions of MADS and other derivative-free methods can be used to approximate Pareto fronts in a fully black-box setting \cite{custodio2011}.

From a modeling perspective, the framework proposed here can be extended in several directions. First, identifiability and sensitivity analyses (for example, variance-based methods) could be included as a preliminary step to help decide which parameters should be calibrated and how to design the residuals and observation sets. Second, the heuristic step, ceilings, and residual weights could be chosen adaptively, for example, using robust statistics of the misfit values or simple learning rules that update these choices during optimization \cite{Ambrosio2017,Andres2024}. Third, the current form of the residuals (based on quadrants, penalties, and aggregated distances) could be altered to reflect better realistic data features of the problem approached. Finally, the full framework could be tested on larger conceptual models and real datasets, where model error, structural mismatch, and computational features are evaluated in detail.

In this sense, the examples presented in this chapter serve as a proof of concept, revealing both the strengths and limitations of systematic calibration with derivative-free methods, and pointing to a possible path for combining problem-specific residual design with more advanced and diverse optimization strategies.

\paragraph{Computational implementation:} All models in this chapter were solved numerically using an implicit variable-step integration scheme based on the Backward Differentiation Formula (BDF) method. Relative and absolute tolerances were adaptively adjusted between $10^{-12}$ and $10^{-3}$ to ensure stable convergence and completion of the integration while keeping numerical errors as small as possible. As the examples show, the numerical integration settings can directly affect the smoothness and conditioning of the misfit function, and hence the performance of the calibration procedure. 
Further implementation details, including Python code and additional experiments, are provided in the \textit{Appendix}.

\newpage

\vspace{1cm}

\chapter{A data-constrained model for the PEC biogeochemistry dynamics}
\chaptermark{A data-constrained model for the PEC}
\label{chapterCalibratingPEC}

In this chapter, we apply the parameter-calibration framework developed in \textit{Chapter \ref{chapterframework}} to the conceptual model of the Paranaguá Estuarine Complex (PEC) developed in \textit{Chapter \ref{chapterPEC}}. To this end, we fit the model outputs to an observed dataset of nitrate (NO$_3$) and phytoplankton (measured in terms of Chlorophyll-A) over a one-year period. Finally, we use the resulting data-constrained model to test a few scenarios involving increased riverine nitrate loads. The Python code for these experiments is available in the \textit{Appendix}.

\section{General modeling setup}
We consider the conceptual two-box model for the PEC described in \textit{Chapter \ref{chapterPEC}}. Specific details on our modeling approach are presented below. The data were identified from plots in \cite{Machado1997ParanaguaBay} using \textit{Web Plot Digitalizer} \cite{WebPlotDigitizer}. For the PEC dynamics using literature data, we initially applied linear interpolation to the observations, but the model–data fit was poor. We then hypothesized the presence of measurement noise and applied a smoothing strategy, in which each monthly observation was replaced by the average of that month and the two previous months. These data was presented in\textit{ Chapter \ref{chapterPEC}}, as the upper box tracers' concentration in figures \ref{fig:dataNitrate} and \ref{fig:dataPhy}.

\paragraph{Forcing, initial and boundary conditions:}
Initial conditions for the model integration are obtained from linear interpolation of observational data reported in the literature \cite{Machado1997ParanaguaBay}, as presented in \textit{Chapter \ref{chapterPEC}}.


\paragraph{Target parameters:}
The parameters to be calibrated in this model are the daily maximum phytoplankton growth rate, $V_\text{max}$, and the mortality rate $\lambda$.

\paragraph{Heuristic:}
For the heuristic, we considered a version of the misfit function in which the last year of the model integration over a period of 51 years is compared against 13 observations distributed over one year. For the benchmark tests, we implemented only one heuristic step, where we searched for a good initial guess for the local search on the best-fitting parameters. During the computational implementation, this stage was parallelized.

\paragraph{Optimization approach:}
We used the DFO-LS solver for parameter calibration with an initial guess determined by the heuristic and bounds estimated from theoretical considerations in the literature \cite{Brandini1985, Cloern1978Skeletonema, Olmstead2011PilotStudy}. The misfit residual array compares the last year of the model integration over 201 years to $N{+}1$ observations distributed across one year, with $N \in \{12,120\}$. In all cases, we used the basic configuration of DFO-LS implementation \cite{roberts2025dfo-ls}, including bounds and the option \texttt{scaling\_within\_bounds}. The DFO-LS Python package also offers options to handle noise, which were not explored in this study.

\section{Benchmarking and validating the approach}
To validate the fitting capability of our framework, we conducted tests with the following objectives:
\begin{enumerate}
    \item Ideally, to recover the observational parameters through the optimization strategy.
    \item To reproduce the behavior of observational data using the calibrated model.
\end{enumerate}
In this section, we will set chosen observational parameters and artificially generated observational data as outputs obtained from our model, based on the observational parameters and the previously considered forcing data. To obtain the observational data, we integrated our model with the observational parameters set over the 1000-year integration period, which is considered sufficient time for the model to reach an equilibrium state. Then, we sample the observations from the last year of the integration period at the same time instants considered in the definitions of each misfit function version. 

For benchmark tests, the parameters $V_\text{max}$ and $\lambda$ were considered constant throughout the year, that is, invariant with respect to time. This simplification was related to the low level of complexity of the conceptual model, although it is not the only possibility. In the next section, we will see an example of calibration involving a time-varying parameter.

Considering the assumption that the behavior of the PEC phytoplankton population can be represented as a combination of the diatom species \textit{S. costatum} and \textit{A. glacialis} \cite{Brandini2022}, we set bounds for the parameter space as:
\begin{align}
    0 & \leq V_\text{max} \leq \max \left \lbrace V_{\text{max},1} , V_{\text{max},2} \right\rbrace , \\
    0 & \leq \lambda \leq 1 ,    
\end{align}
where $V_{\text{max},1} = 2.5$ and $V_{\text{max},2} = 1.5$ are set as the maximum daily growth parameter estimates for the \textit{S. costatum} phytoplankton population \cite{Khan1998}, and \textit{A. glacialis} phytoplankton population \cite{Olmstead2011PilotStudy}, respectively. As we did not perceive a significant enhancement in accuracy due to taking 120 samples instead of 12 in the previous simulations, here we considered only the second option. We now present the case studies used as benchmark tests.

\paragraph{Case 1:} Here, we choose the observational parameters:
\begin{align*}
    V_\text{max}^\text{Obs} & = 1.0 ,\\
    \lambda^\text{Obs} & = 0.05 .
\end{align*}
That means the observations were generated under the assumptions of low daily growth and moderate daily mortality of the phytoplankton population. The artificial observations generated by the model outputs when considering these observational parameters and the fitting attained after the calibration process are presented in figures \ref{fig:fittingBench01N} and \ref{fig:fittingBench01P}. 

\vspace{0.7cm}
\begin{figure}[htbp!]
\caption{Fitting of average nitrate concentrations in the upper part of the PEC over a one-year period, for the Case 1 experiment.}
\begin{center}
\includegraphics[scale=0.8]{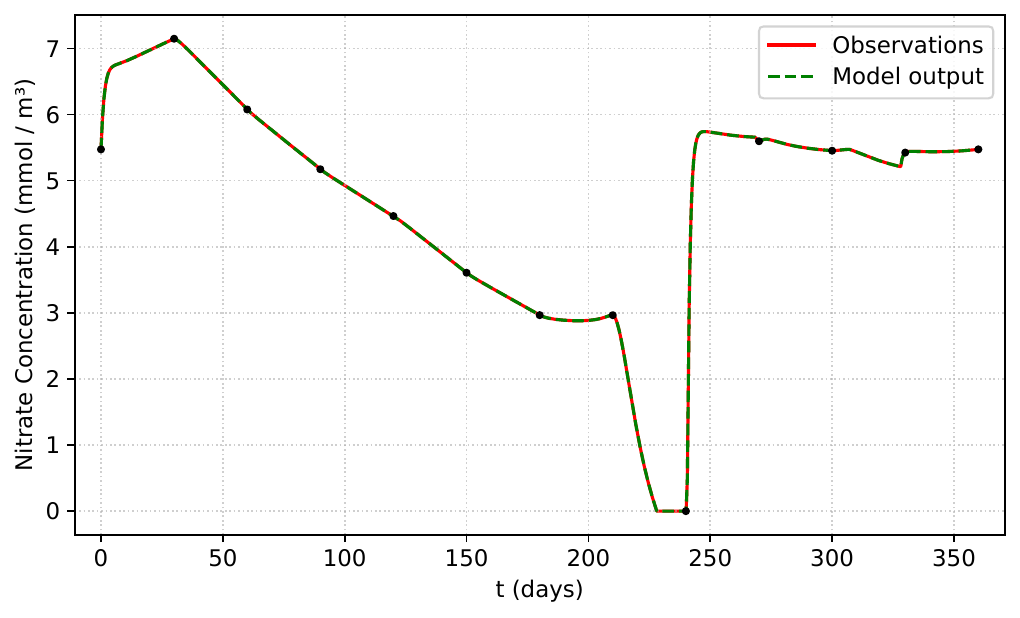}
\end{center}
\label{fig:fittingBench01N}

\smallskip
Observations were artificially generated by integrating the model with observational parameters $V_\text{max}^\text{Obs} = 1.0$, and $\lambda^\text{Obs} = 0.05$.
The model output plot was obtained by integrating the model described by equations (\ref{EstuarioN}) - (\ref{estacionario}) with the calibrated parameters $\Bar{V}_\text{max} \approx V_\text{max}^\text{Obs} + 8 \cdot 10^{-8}$, and $\Bar{\lambda} \approx \lambda^\text{Obs} + 2 \cdot 10^{-9}$.
The points represent the $N+1$ fitting points used to calculate the misfit residual array; here, $N=12$.
    
\smallskip
\textbf{Source:} the author.
\end{figure}

\begin{figure}[htbp!]
\caption{Fitting of average phytoplankton concentrations in the upper part of the PEC over a one-year period, for the Case 1 experiment.}
\begin{center}
\includegraphics[scale=0.8]{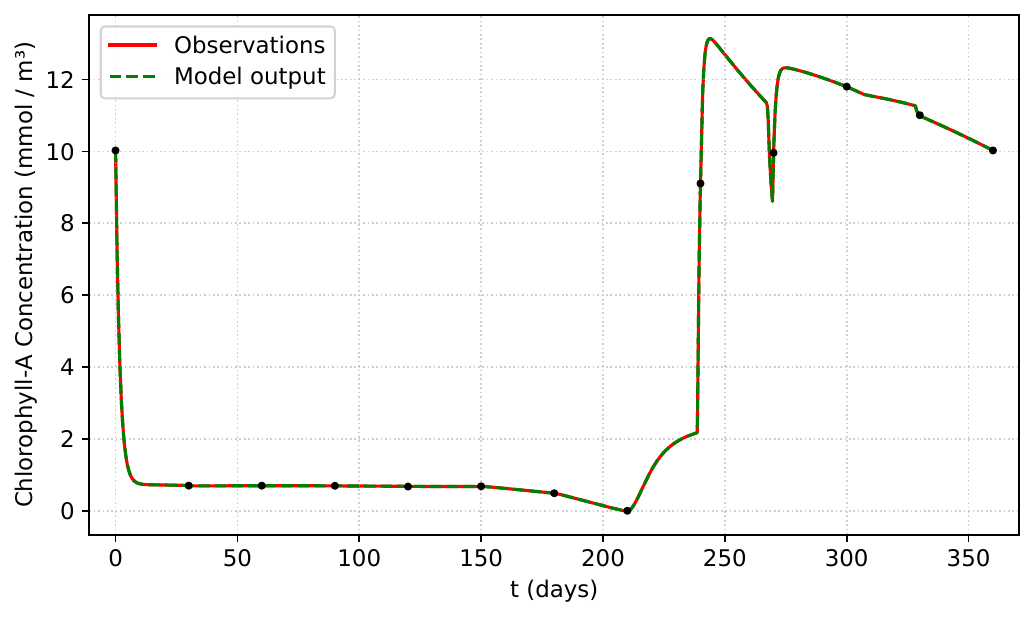}  
\end{center}

\smallskip
See the full description of the experiment settings in Figure \ref{fig:fittingBench01N}.
\label{fig:fittingBench01P}

\smallskip
\textbf{Source:} the author.

\vspace{0.2cm}
\end{figure}

\newpage

The parameter calibration is presented in figures \ref{fig:searchBench01.1} and \ref{fig:searchBench01.2}, as the one-step heuristic search for an adequate initial guess, which was found as:
\begin{align*}
    V_\text{max}^\text{Initial} & = 2.1700000000000004 ,\\
    \lambda^\text{Initial} & = 0.089 ,
\end{align*}
and the DFO-LS search, which converged to the optimized parameters:
\begin{align*}
    \Bar{V}_\text{max} & = 1.000000084741517 ,\\
    \Bar{\lambda} & = 0.05000000269829062 .
\end{align*}

\begin{figure}[htpb!]
\caption{Heuristic step on the calibration process of the parameters $V_\text{max}$ and $\lambda$, in the Case 1 experiment.}
\begin{center}
\includegraphics[scale=0.8]{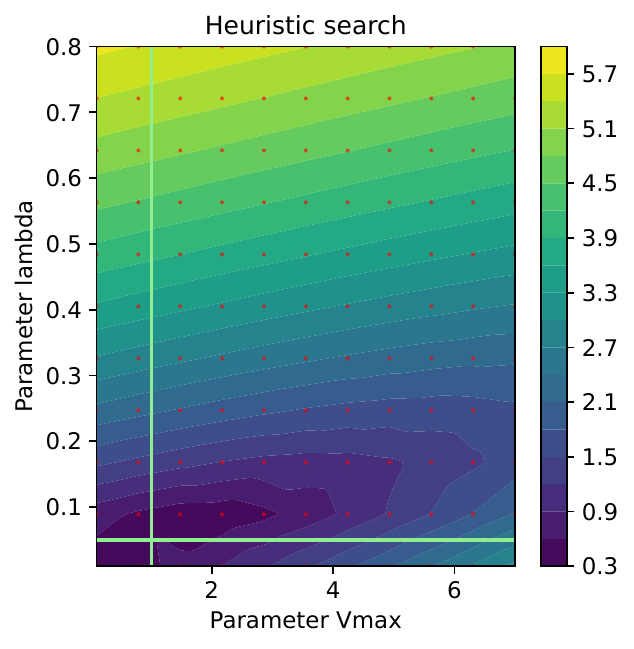} 
\end{center}
\label{fig:searchBench01.1} 

\smallskip
The figure shows a contour plot representing the heuristic step used to determine the initial guess. The sidebar associates colors with the orders of magnitude of the heuristic misfit function evaluations, and the grid of points in the background identifies the parameters evaluated during the heuristic.
The intersection of the highlighted horizontal and vertical lines marks the location of the observational parameters.

\smallskip
\textbf{Source:} the author.
\end{figure}

\begin{figure}[htbp!]
\caption{DFO-LS search on the calibration process of the parameters $V_\text{max}$ and $\lambda$, in the Case 1 experiment.}
\begin{center}
\includegraphics[scale=0.8]{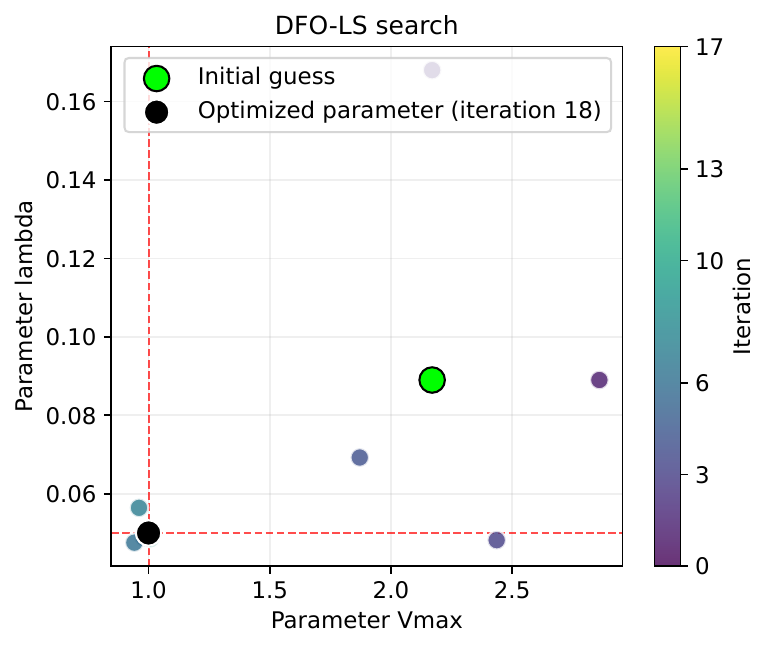}    
\end{center}
\label{fig:searchBench01.2}

\smallskip
Each point corresponds to the parameters obtained in an iteration of the DFO-LS optimization algorithm, being the initial guess obtained from the heuristic step (Figure \ref{fig:searchBench01.1}). The sidebar assigns a color scale to the iterations, allowing identification of the convergence pattern. The intersection of the highlighted horizontal and vertical lines marks the location of the observational parameters.
    
\smallskip
\textbf{Source:} the author.
\end{figure}

\newpage
\paragraph{Case 2:} Here, we choose the observational parameters:
\begin{align*}
    V_\text{max}^\text{Obs} & = 1.4 ,\\
    \lambda^\text{Obs} & = 0.05 .
\end{align*}
That means the observations were generated under the assumptions of moderate daily growth and moderate daily mortality of the phytoplankton population. The artificial observations generated by the model outputs when considering these observational parameters and the fitting attained after the calibration process are presented in figures \ref{fig:fittingBench02N} and \ref{fig:fittingBench02P}. 

\vspace{0.7cm}
\begin{figure}[htbp!]
\caption{Fitting of average nitrate concentrations in the upper part of the PEC over a one-year period, for the Case 2 experiment.}
\begin{center}
\includegraphics[scale=0.8]{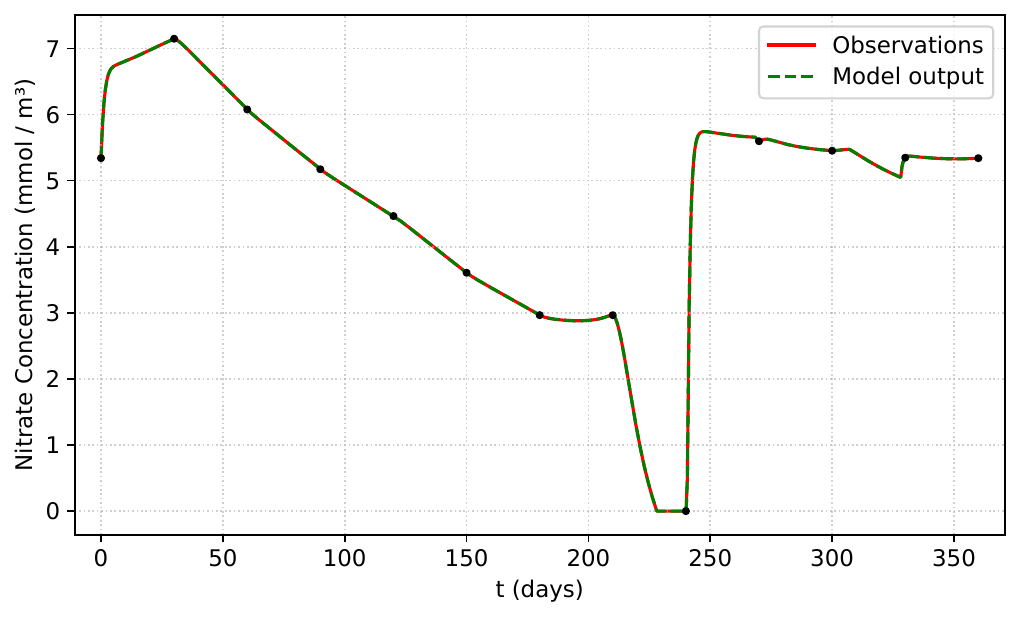}  
\end{center}
\label{fig:fittingBench02N}

\smallskip
Observations were artificially generated by integrating the model with observational parameters $V_\text{max}^\text{Obs} = 1.4$, and $\lambda^\text{Obs} = 0.05$.
The model output plot was obtained by integrating the model described by equations (\ref{EstuarioN}) - (\ref{estacionario}) with the calibrated parameters $\Bar{V}_\text{max} \approx V_\text{max}^\text{Obs} + 9 \cdot 10^{-8}$, and $\Bar{\lambda} \approx \lambda^\text{Obs} - 1 \cdot 10^{-9}$.
The points represent the $N+1$ fitting points used to calculate the misfit residual array; here, $N=12$.
    
\smallskip
\textbf{Source:} the author.
\end{figure}

\begin{figure}[htbp!]
\caption{Fitting of average phytoplankton concentrations in the upper part of the PEC over a one-year period, for the Case 2 experiment.}
\begin{center}
\includegraphics[scale=0.8]{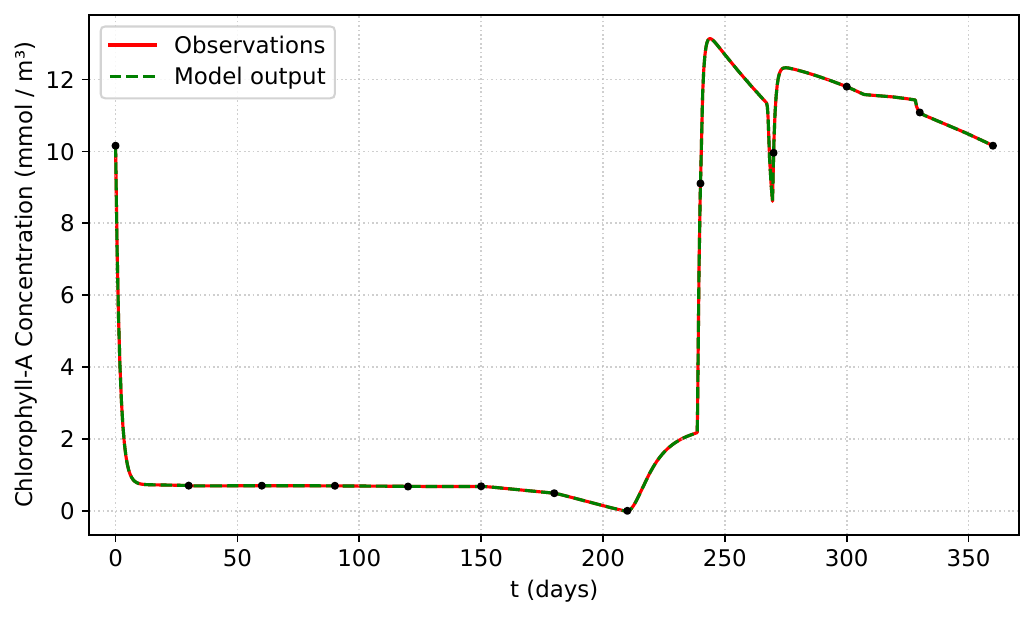}  
\end{center}
\label{fig:fittingBench02P}

\smallskip
See the full description of the experiment settings in Figure \ref{fig:fittingBench02N}.
    
\smallskip
\textbf{Source:} the author.

\vspace{0.2cm}
\end{figure}

\newpage
The parameter calibration is presented in figures \ref{fig:searchBench02.1} and \ref{fig:searchBench02.2}, as the one-step heuristic search for an adequate initial guess, which was found as:
\begin{align*}
    V_\text{max}^\text{Initial} & = 2.1700000000000004 ,\\
    \lambda^\text{Initial} & = 0.089 ,
\end{align*}
and the DFO-LS search, which converged to the optimized parameters:
\begin{align*}
    \Bar{V}_\text{max} & = 1.4000000989350427 ,\\
    \Bar{\lambda} & = 0.04999999880062813 .
\end{align*}

\begin{figure}[htbp!]    
\caption{Heuristic step on the calibration process of the parameters $V_\text{max}$ and $\lambda$, in the Case 2 experiment.}
\begin{center}
\includegraphics[scale=0.8]{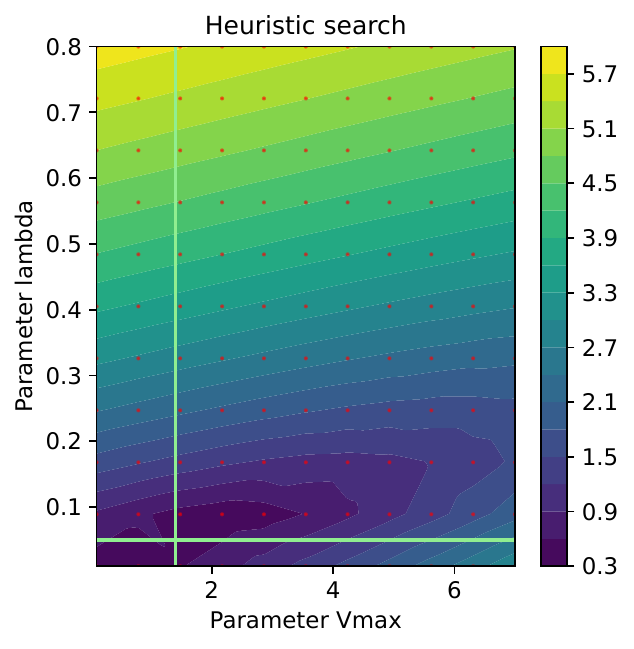}
\end{center}
\label{fig:searchBench02.1}

\smallskip
The figure shows a contour plot representing the heuristic step used to determine the initial guess. The sidebar associates colors with the orders of magnitude of the heuristic misfit function evaluations, and the background grid of points identifies the parameters evaluated during the heuristic.
The intersection of the highlighted horizontal and vertical lines marks the location of the observational parameters.
    
\smallskip
\textbf{Source:} the author.
\end{figure}

\begin{figure}[htbp!]
\caption{DFO-LS search on the calibration process of the parameters $V_\text{max}$ and $\lambda$, in the Case 2 experiment.}
\begin{center}
\includegraphics[scale=0.8]{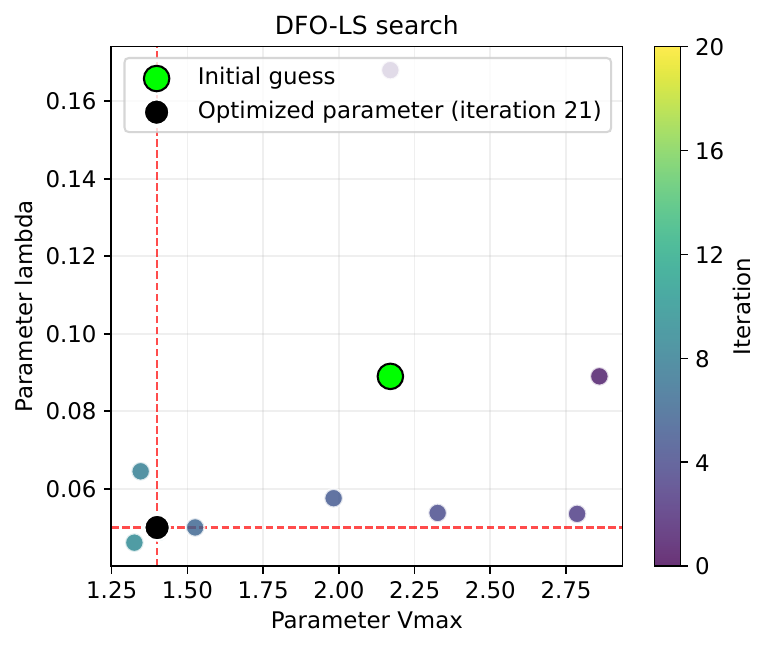}    
\end{center}
\label{fig:searchBench02.2}

\smallskip
Each point corresponds to the parameters obtained in an iteration of the DFO-LS optimization algorithm, being the initial guess obtained from the heuristic step (Figure \ref{fig:searchBench02.1}). The sidebar assigns a color scale to the iterations, allowing identification of the convergence pattern. The intersection of the highlighted horizontal and vertical lines marks the location of the observational parameters.
    
\smallskip
\textbf{Source:} the author.
\end{figure}

\newpage
\paragraph{Case 3:} Here, we choose the observational parameters:
\begin{align*}
    V_\text{max}^\text{Obs} & = 1.4 ,\\
    \lambda^\text{Obs} & = 0.3 .
\end{align*}
That means the observations were generated under the assumptions of moderate daily growth and high daily mortality of the phytoplankton population. The artificial observations generated by the model outputs when considering these observational parameters and the fitting attained after the calibration process are presented in figures \ref{fig:fittingBench03N} and \ref{fig:fittingBench03P}. 

\vspace{0.7cm}
\begin{figure}[htbp!]
\caption{Fitting of average nitrate concentrations in the upper part of the PEC over a one-year period, for the Case 3 experiment.}
\begin{center}
\includegraphics[scale=0.8]{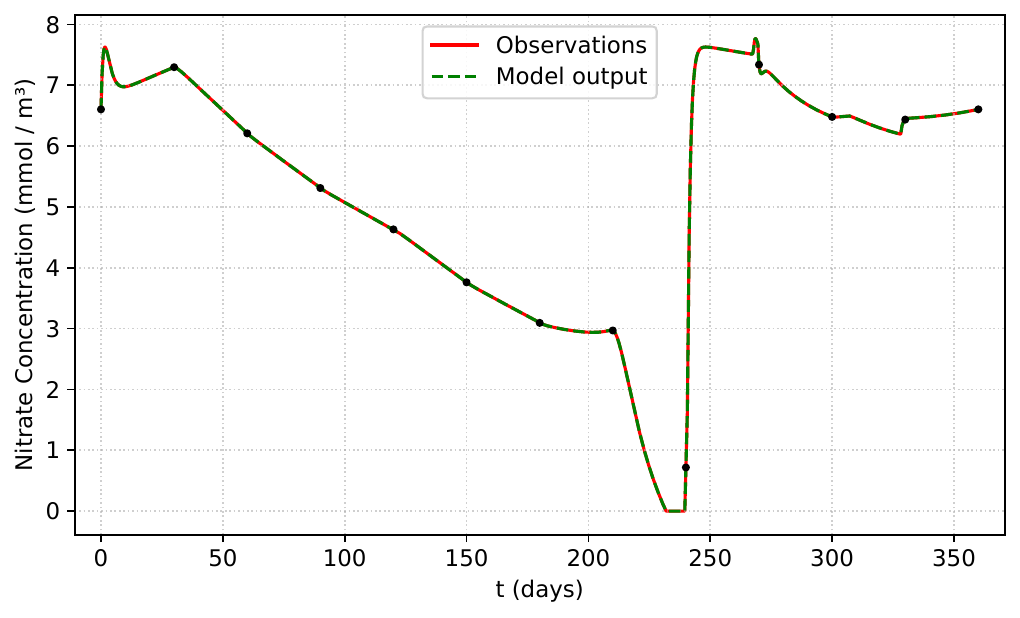}  
\end{center}
\label{fig:fittingBench03N} 

\smallskip
Observations were artificially generated by integrating the model with observational parameters $V_\text{max}^\text{Obs} = 1.4$, and $\lambda^\text{Obs} = 0.3$.
The model output plot was obtained by integrating the model described by equations (\ref{EstuarioN}) - (\ref{estacionario}) with the calibrated parameters $\Bar{V}_\text{max} \approx V_\text{max}^\text{Obs} + 3 \cdot 10^{-8}$, and $\Bar{\lambda} \approx \lambda^\text{Obs} + 8 \cdot 10^{-10}$.
The points represent the $N+1$ fitting points used to calculate the misfit residual array; here, $N=12$.
    
\smallskip
\textbf{Source:} the author.
\end{figure}

\begin{figure}[htbp!]
\caption{Fitting of average phytoplankton concentrations in the upper part of the PEC over a one-year period, for the Case 3 experiment.}
\begin{center}
\includegraphics[scale=0.8]{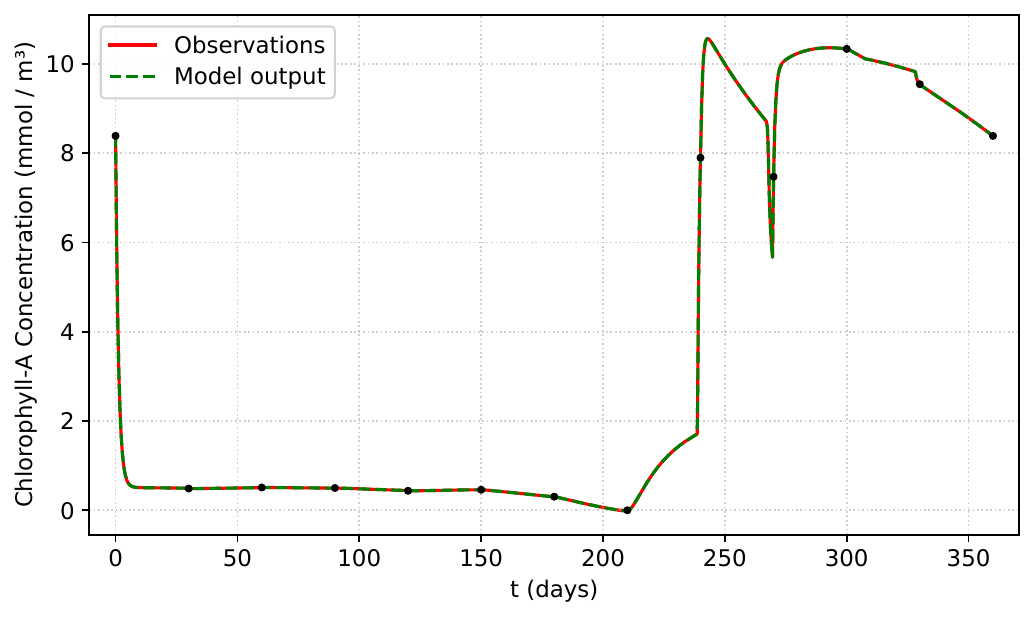}  
\end{center}
\label{fig:fittingBench03P} 

\smallskip
See the full description of the experiment settings in Figure \ref{fig:fittingBench03N}.

\smallskip
\textbf{Source:} the author.

\vspace{0.2cm}
\end{figure}

\newpage
The parameter calibration is presented in figures \ref{fig:searchBench03.1} and \ref{fig:searchBench03.2}, as the one-step heuristic search for an adequate initial guess, which was found as:
\begin{align*}
    V_\text{max}^\text{Initial} & = 2.1700000000000004 ,\\
    \lambda^\text{Initial} & = 0.326 ,
\end{align*}
and the DFO-LS search, which converged to the optimized parameters:
\begin{align*}
    \Bar{V}_\text{max} & = 1.400000032860641 ,\\
    \Bar{\lambda} & = 0.3000000008494046 .
\end{align*}

\begin{figure}[htbp!]
\caption{Heuristic step on the calibration process of the parameters $V_\text{max}$ and $\lambda$, in the Case 3 experiment.}
\begin{center}
\includegraphics[scale=0.8]{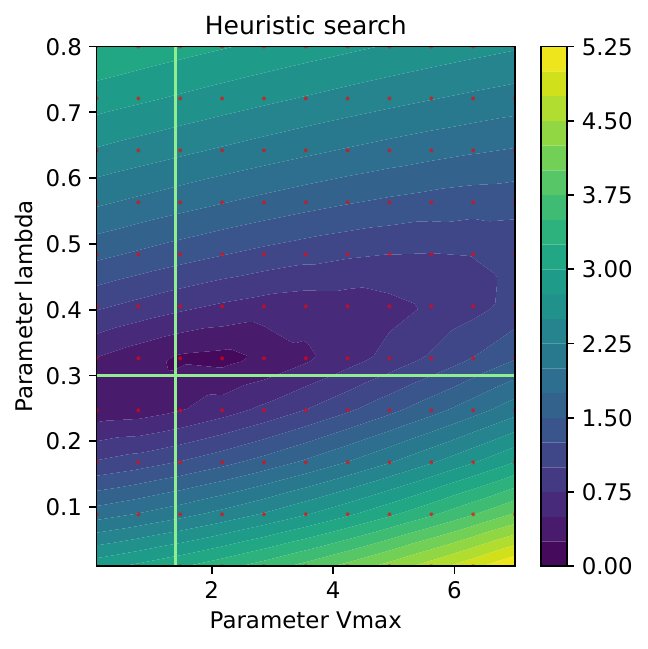} 
\end{center}
\label{fig:searchBench03.1}

\smallskip
The figure shows a contour plot representing the heuristic step used to determine the initial guess. The sidebar associates colors with the orders of magnitude of the heuristic misfit function evaluations, and the background grid of points identifies the parameters evaluated during the heuristic.
The intersection of the highlighted horizontal and vertical lines marks the location of the observational parameters.

\smallskip
\textbf{Source:} the author.
\end{figure}

\begin{figure}[htbp!]
\caption{DFO-LS search on the calibration process of the parameters $V_\text{max}$ and $\lambda$, in the Case 3 experiment.}
\begin{center}
\includegraphics[scale=0.8]{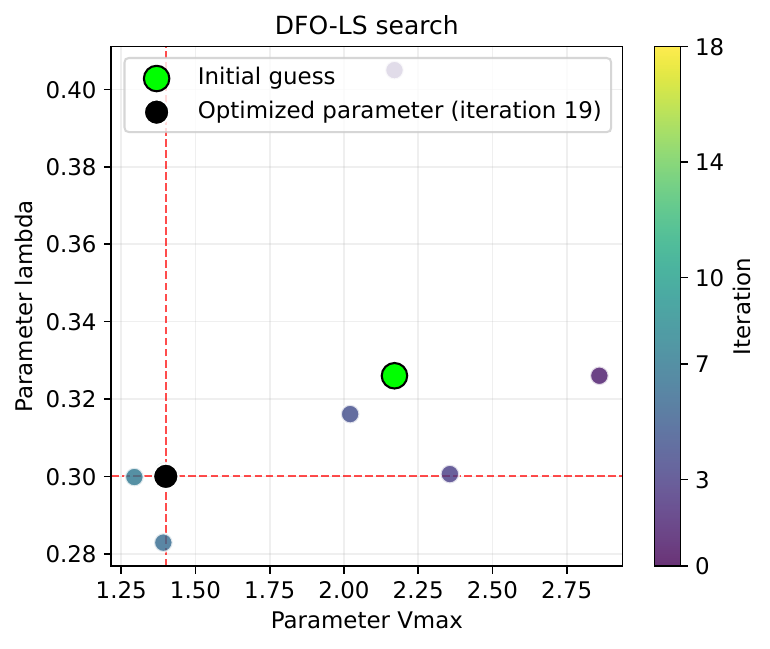}
\end{center}
\label{fig:searchBench03.2}

\smallskip
Each point corresponds to the parameters obtained in an iteration of the DFO-LS optimization algorithm, being the initial guess obtained from the heuristic step (Figure \ref{fig:searchBench03.1}). The sidebar assigns a color scale to the iterations, allowing identification of the convergence pattern. The intersection of the highlighted horizontal and vertical lines marks the location of the observational parameters.

\smallskip
\textbf{Source:} the author.
\end{figure}

\newpage
\paragraph{Case 4:} Here, we choose the observational parameters:
\begin{align*}
    V_\text{max}^\text{Obs} & = 2.0 ,\\
    \lambda^\text{Obs} & = 0.05 .
\end{align*}
That means the observations were generated under the assumptions of high daily growth and moderate daily mortality of the phytoplankton population. The artificial observations generated by the model outputs when considering these observational parameters and the fitting attained after the calibration process are presented in figures \ref{fig:fittingBench04N} and \ref{fig:fittingBench04P}. 

\vspace{0.7cm}
\begin{figure}[htbp!]
\caption{Fitting of average nitrate concentrations in the upper part of the PEC over a one-year period, for the Case 4 experiment.}
\begin{center}
\includegraphics[scale=0.8]{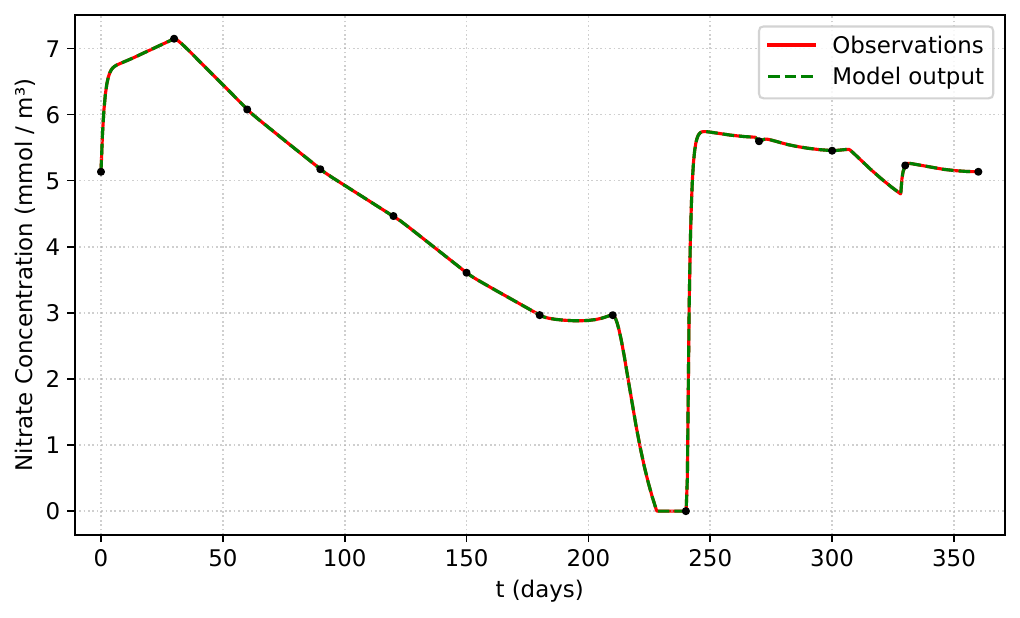}
\end{center}
\label{fig:fittingBench04N}

\smallskip
Observations were artificially generated by integrating the model with observational parameters $V_\text{max}^\text{Obs} = 2.0$, and $\lambda^\text{Obs} = 0.05$.
The model output plot was obtained by integrating the model described by equations (\ref{EstuarioN}) - (\ref{estacionario}) with the calibrated parameters $\Bar{V}_\text{max} \approx V_\text{max}^\text{Obs} + 1 \cdot 10^{-6}$, and $\Bar{\lambda} \approx \lambda^\text{Obs} + 9 \cdot 10^{-9}$.
The points represent the $N+1$ fitting points considered for the calculation of the misfit residual array, here we consider $N=12$.

\smallskip
\textbf{Source:} the author.
\end{figure}

\begin{figure}[htbp!]
\caption{Fitting of average phytoplankton concentrations in the upper part of the PEC over a one-year period, for the Case 4 experiment.}
\begin{center}
\includegraphics[scale=0.8]{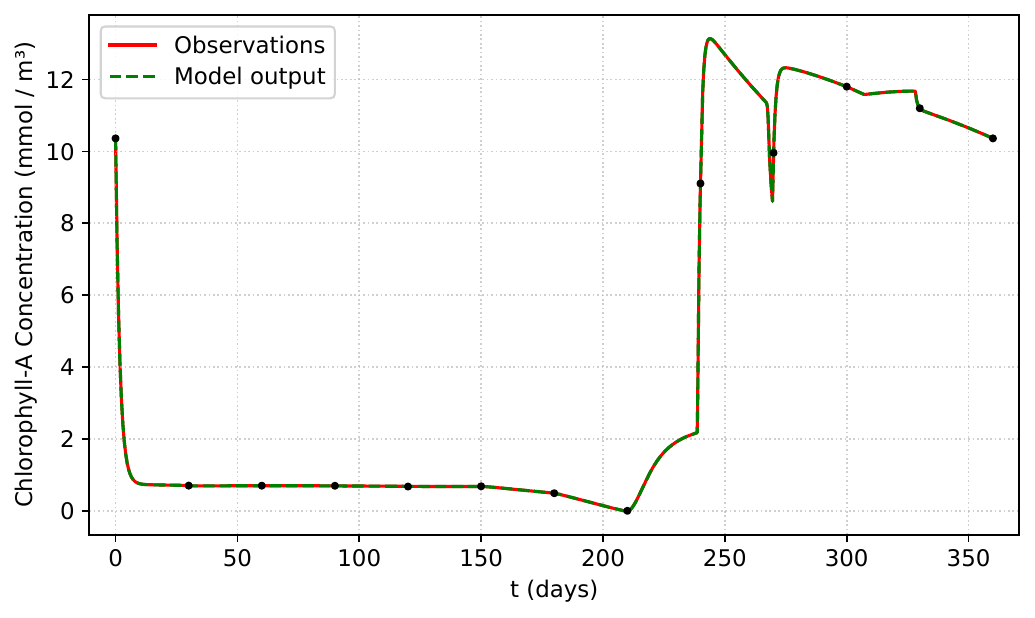}  
\end{center}
\label{fig:fittingBench04P}

\smallskip
See the full description of the experiment settings in Figure \ref{fig:fittingBench04N}.

\smallskip
\textbf{Source:} the author.

\vspace{0.2cm}
\end{figure}

\newpage
The parameter calibration is presented in figures \ref{fig:searchBench04.1} and \ref{fig:searchBench04.2}, as the one-step heuristic search for an adequate initial guess, which was found as:
\begin{align*}
    V_\text{max}^\text{Initial} & = 2.8600000000000003 ,\\
    \lambda^\text{Initial} & = 0.089 ,
\end{align*}
and the DFO-LS search, which converged to the optimized parameters:
\begin{align*}
    \Bar{V}_\text{max} & = 2.0000014823868333 ,\\
    \Bar{\lambda} & = 0.05000000913800312 .
\end{align*}

\begin{figure}[htbp!]
\caption{Heuristic step on the calibration process of the parameters $V_\text{max}$ and $\lambda$, in the Case 4 experiment.}
\begin{center}
\includegraphics[scale=0.8]{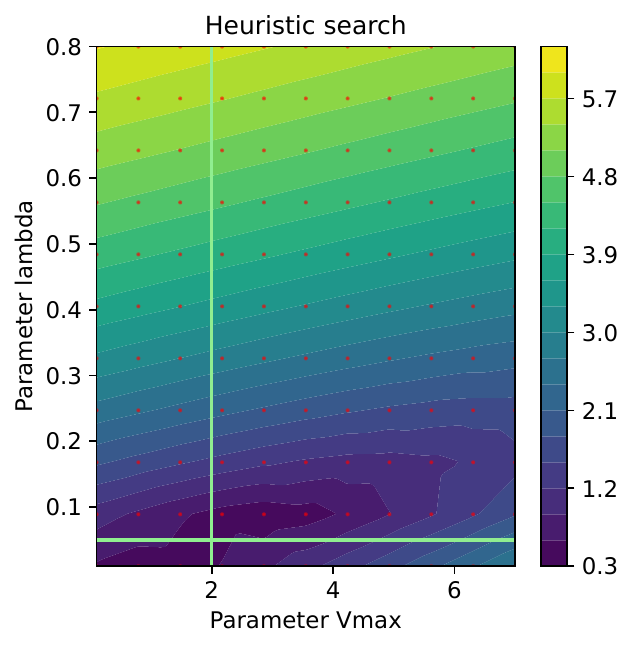}
\end{center}
\label{fig:searchBench04.1}

\smallskip
The figure shows a contour plot representing the heuristic step used to determine the initial guess. The sidebar associates colors with the orders of magnitude of the heuristic misfit function evaluations, and the background grid of points identifies the parameters evaluated during the heuristic.
The intersection of the highlighted horizontal and vertical lines marks the location of the observational parameters.

\smallskip
\textbf{Source:} the author.
\end{figure}

\begin{figure}[htbp!]
\caption{DFO-LS search on the calibration process of the parameters $V_\text{max}$ and $\lambda$, in the Case 4 experiment.}
\begin{center}
\includegraphics[scale=0.8]{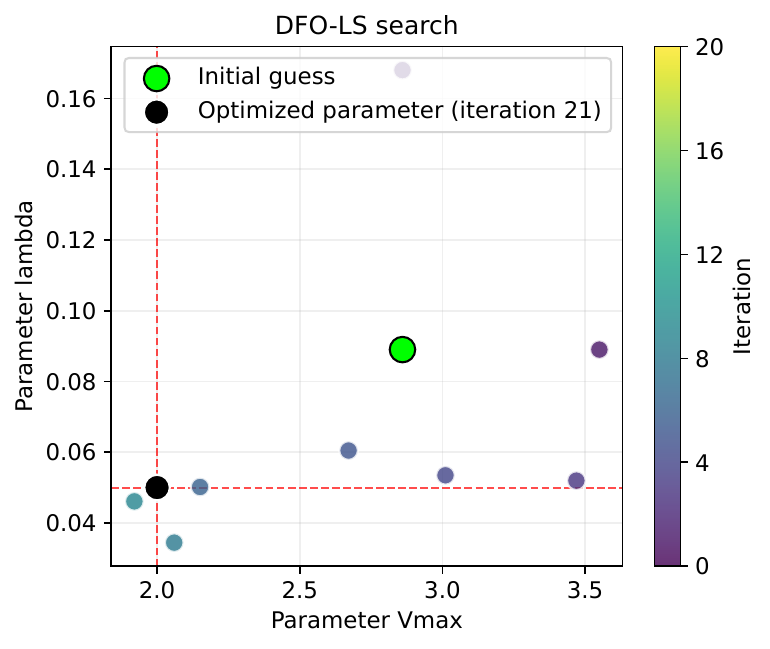}    
\end{center}
\label{fig:searchBench04.2}

\smallskip
Each point corresponds to the parameters obtained in an iteration of the DFO-LS optimization algorithm, being the initial guess obtained from the heuristic step (Figure \ref{fig:searchBench04.1}). The sidebar assigns a color scale to the iterations, allowing identification of the convergence pattern. The intersection of the highlighted horizontal and vertical lines marks the location of the observational parameters.

\smallskip
\textbf{Source:} the author.
\end{figure}

\newpage
\paragraph{Case 5:} Here, we choose the observational parameters:
\begin{align*}
    V_\text{max}^\text{Obs} & = 2.0 ,\\
    \lambda^\text{Obs} & = 0.3 .
\end{align*}
That means the observations were generated under the assumptions of high daily growth and high daily mortality of the phytoplankton population. The artificial observations generated by the model outputs when considering these observational parameters and the fitting attained after the calibration process are presented in figures \ref{fig:fittingBench05N} and \ref{fig:fittingBench05P}. 

\vspace{0.7cm}
\begin{figure}[htbp!]
\caption{Fitting of average nitrate concentrations in the upper part of the PEC over a one-year period, for the Case 5 experiment.}
\begin{center}
\includegraphics[scale=0.8]{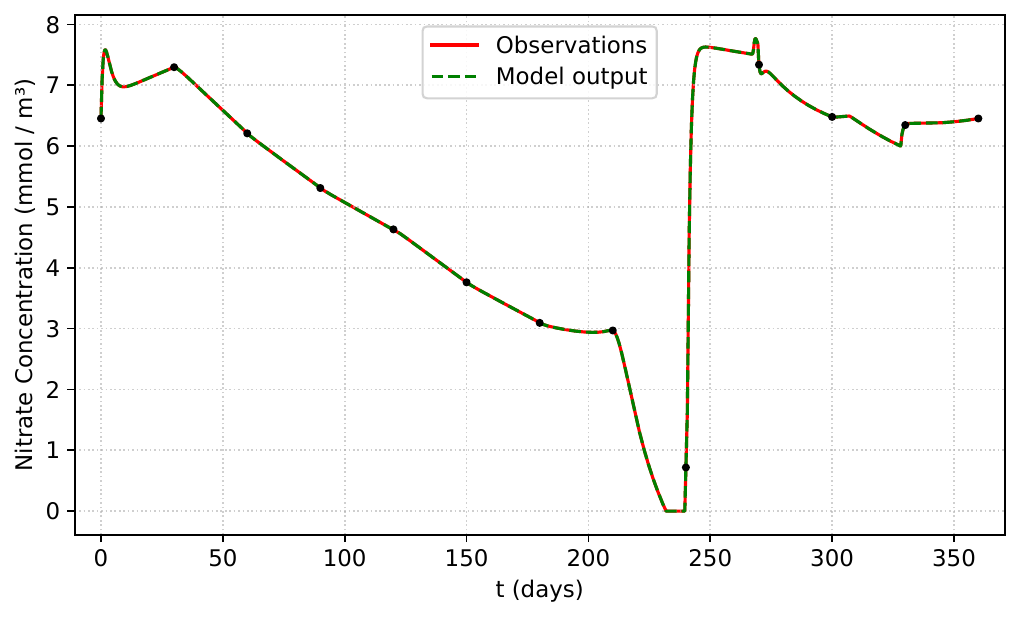}
\end{center}
\label{fig:fittingBench05N}

\smallskip
Observations were artificially generated by integrating the model with observational parameters $V_\text{max}^\text{Obs} = 2.0$, and $\lambda^\text{Obs} = 0.3$.
The model output plot was obtained by integrating the model described by equations (\ref{EstuarioN}) - (\ref{estacionario}) with the calibrated parameters $\Bar{V}_\text{max} \approx V_\text{max}^\text{Obs} + 1 \cdot 10^{-6}$, and $\Bar{\lambda} \approx \lambda^\text{Obs} - 1 \cdot 10^{-8}$.
The points represent the $N+1$ fitting points considered for the calculation of the misfit residual array, here we consider $N=12$.

\smallskip
\textbf{Source:} the author.
\end{figure}

\begin{figure}[htbp!]
\caption{Fitting of average phytoplankton concentrations in the upper part of the PEC over a one-year period, for the Case 5 experiment.}
\begin{center}
\includegraphics[scale=0.8]{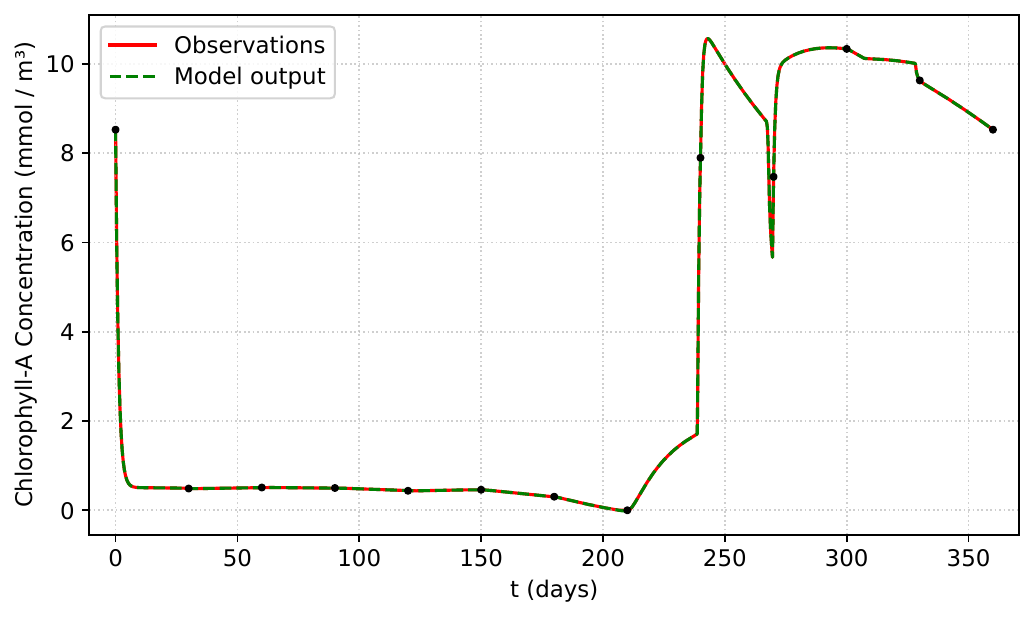}  
\end{center}
\label{fig:fittingBench05P}

\smallskip
See the full description of the experiment settings in Figure \ref{fig:fittingBench05N}.

\smallskip
\textbf{Source:} the author.

\vspace{0.2cm}
\end{figure}

\newpage
The parameter calibration is presented in figures \ref{fig:searchBench05.1} and \ref{fig:searchBench05.2}, as the one-step heuristic search for an adequate initial guess, which was found as:
\begin{align*}
    V_\text{max}^\text{Initial} & = 2.8600000000000003 ,\\
    \lambda^\text{Initial} & = 0.326 ,
\end{align*}
and the DFO-LS search, which converged to the optimized parameters:
\begin{align*}
    \Bar{V}_\text{max} & = 2.000001216474024 ,\\
    \Bar{\lambda} & = 0.2999999872607704 .
\end{align*}

\begin{figure}[htbp!]
\caption{Heuristic step on the calibration process of the parameters $V_\text{max}$ and $\lambda$, in the Case 5 experiment.}
\begin{center}
\includegraphics[scale=0.8]{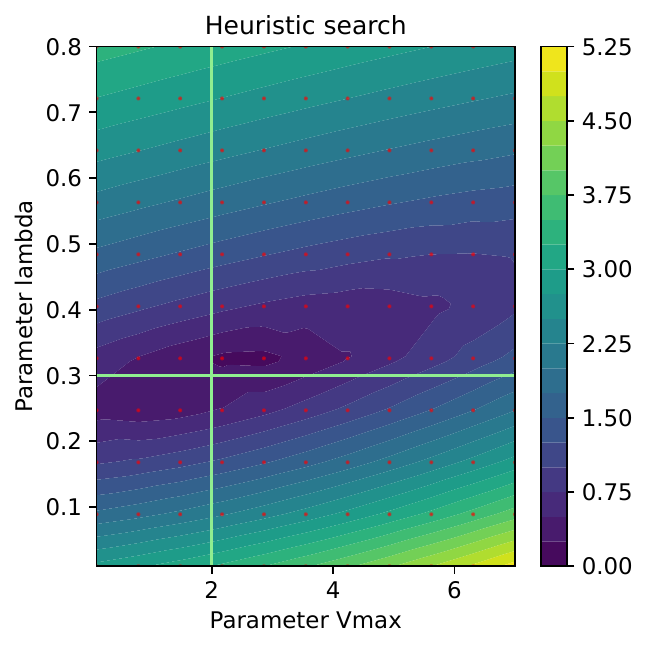}
\end{center}
\label{fig:searchBench05.1}

\smallskip
The figure shows a contour plot representing the heuristic step used to determine the initial guess. The sidebar associates colors with the orders of magnitude of the heuristic misfit function evaluations, and the background grid of points identifies the parameters evaluated during the heuristic.
The intersection of the highlighted horizontal and vertical lines marks the location of the observational parameters.

\smallskip
\textbf{Source:} the author.
\end{figure}

\begin{figure}[htbp!]
\caption{DFO-LS search on the calibration process of the parameters $V_\text{max}$ and $\lambda$, in the Case 5 experiment.}
\begin{center}
\includegraphics[scale=0.8]{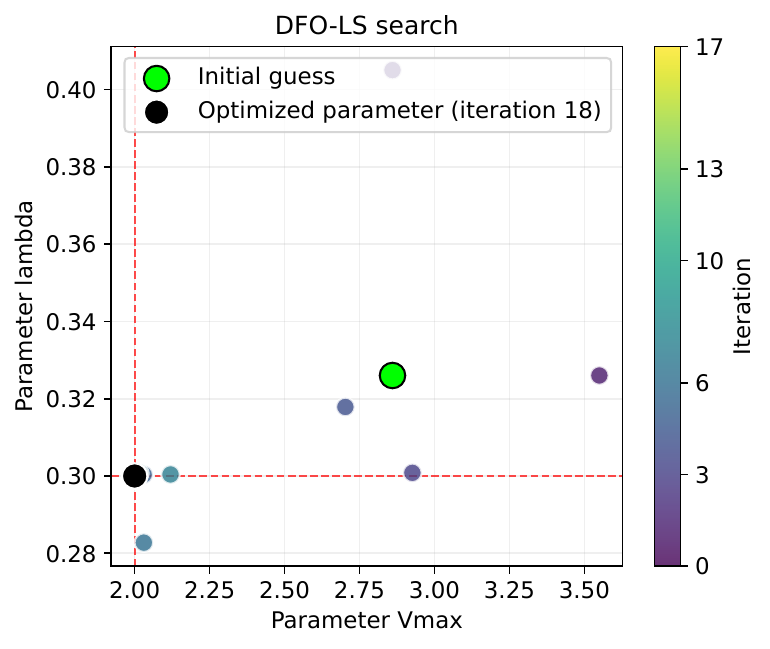}    
\end{center}
\label{fig:searchBench05.2}

\smallskip
Each point corresponds to the parameters obtained in an iteration of the DFO-LS optimization algorithm, being the initial guess obtained from the heuristic step (Figure \ref{fig:searchBench05.1}). The sidebar assigns a color scale to the iterations, allowing identification of the convergence pattern. The intersection of the highlighted horizontal and vertical lines marks the location of the observational parameters.

\smallskip
\textbf{Source:} the author.
\end{figure}

\newpage
\section{Fitting the model to observations}

Here we present two experiments on fitting the model to the data available, without knowing the values of observational parameters. In our first experiment, we calibrate both parameters ${V}_\text{max}$ and $\lambda$ simultaneously. In our second experiment, we set time-dependent values for ${V}_\text{max}$ and calibrate only the parameter $\lambda$.

\paragraph{Case 6:} Calibration of two parameters simultaneously.
The fitting attained after the calibration process is presented in figures \ref{fig:fittingPEC2paramN} and \ref{fig:fittingPEC2paramP}. 

\vspace{0.7cm}
\begin{figure}[htbp!]
\caption{Fitting of average nitrate concentrations in the upper part of the PEC over a one-year period, for the Case 6 experiment.}
\begin{center}
\includegraphics[scale=0.8]{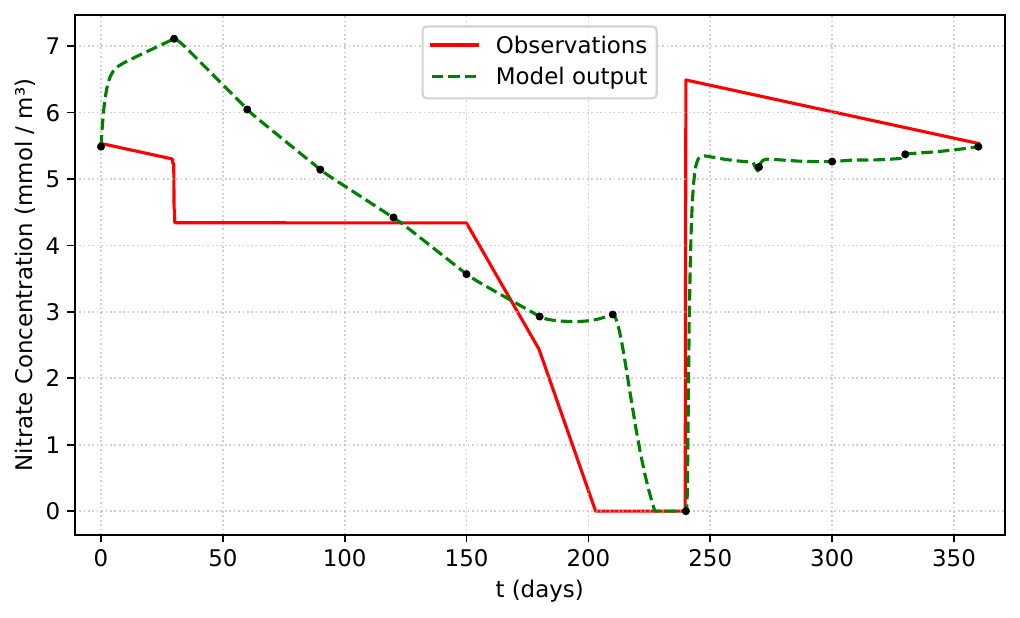}   
\end{center}
\label{fig:fittingPEC2paramN}

\smallskip
Observations were generated by interpolation of tracers' concentration data (figures \ref{fig:dataNitrate} and \ref{fig:dataPhy}).
The model output plot was obtained by integrating the model described by equations (\ref{EstuarioN}) - (\ref{estacionario}) with the calibrated parameters $\Bar{V}_\text{max} \approx 0.24$, and $\Bar{\lambda} \approx 0.01$.
The points represent the $N+1$ fitting points used to calculate the misfit residual array; here, $N=12$.

\smallskip
\textbf{Source:} the author.
\end{figure}

\begin{figure}[htbp!]
\caption{Fitting of average phytoplankton concentrations in the upper part of the PEC over a one-year period, for the Case 6 experiment.}
\begin{center}
\includegraphics[scale=0.8]{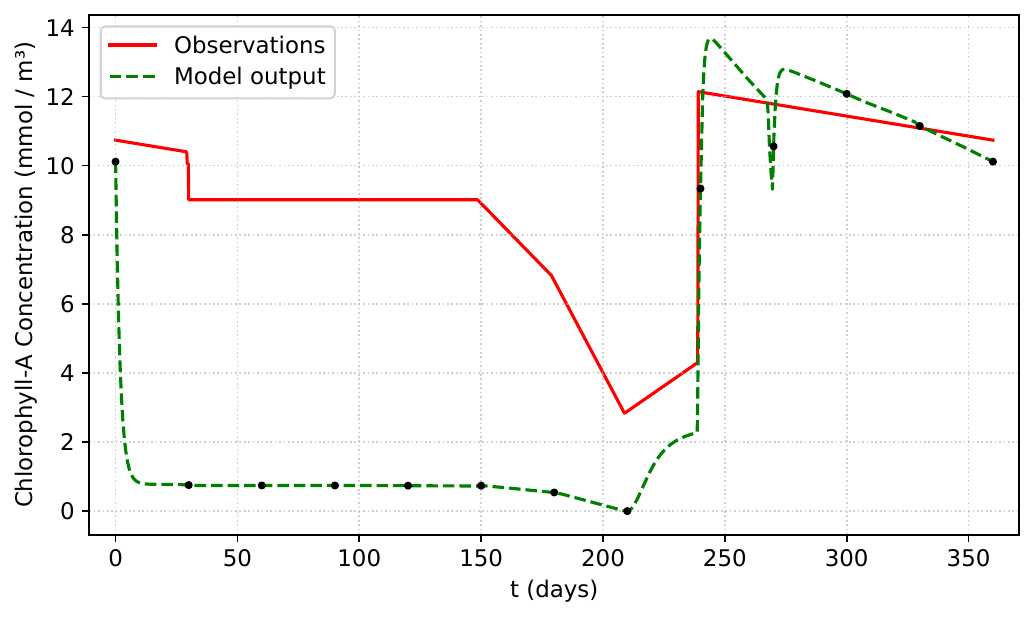}    
\end{center}
\label{fig:fittingPEC2paramP}

\smallskip
See the full description of the experiment settings in Figure \ref{fig:fittingPEC2paramN}.

\smallskip
\textbf{Source:} the author.

\vspace{0.2cm}
\end{figure}

\newpage
The parameter calibration is presented in Figures \ref{fig:searchPEC2param} and \ref{fig:searchPEC2paramDFOLS}, as the one-step heuristic search for an adequate initial guess, which was found as:
\begin{align*}
    V_\text{max}^\text{Initial} & = 0.7899999999999999 ,\\
    \lambda^\text{Initial} & = 0.010000000000000009 ,
\end{align*}
and the DFO-LS search, which converged to the optimized parameters:
\begin{align*}
    \Bar{V}_\text{max} & = 0.24001317934207675 ,\\
    \Bar{\lambda} & = 0.01 .
\end{align*}
This represents both a very low daily growth of the phytoplankton population and a very low daily mortality rate. This may result from a correlation between the two parameters or even from the water circulation flux of the PEC. To evaluate the first of these possibilities, we carried out the following experiment.

\begin{figure}[htbp!]
\caption{Heuristic step on the calibration process of the parameters $V_\text{max}$ and $\lambda$, in the Case 6 experiment.}
\begin{center}
\includegraphics[scale=0.8]{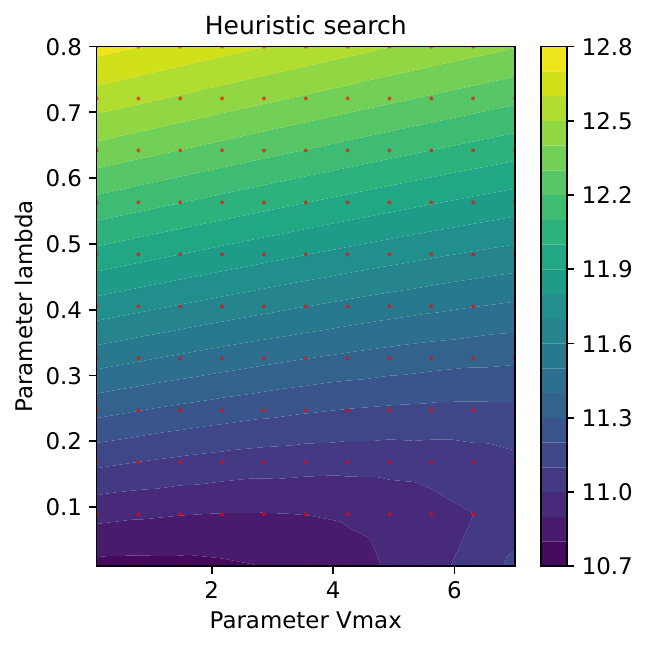}    
\end{center}
\label{fig:searchPEC2param}

\smallskip 
The figure shows a contour plot representing the heuristic step used to determine the initial guess. The sidebar associates colors with the orders of magnitude of the heuristic misfit function evaluations, and the background grid of points identifies the parameters evaluated during the heuristic.

\smallskip
\textbf{Source:} the author.
\end{figure}

\begin{figure}[htbp!]
\caption{DFO-LS search on the calibration process of the parameters $V_\text{max}$ and $\lambda$, in the Case 6 experiment.}
\begin{center}
\includegraphics[scale=0.8]{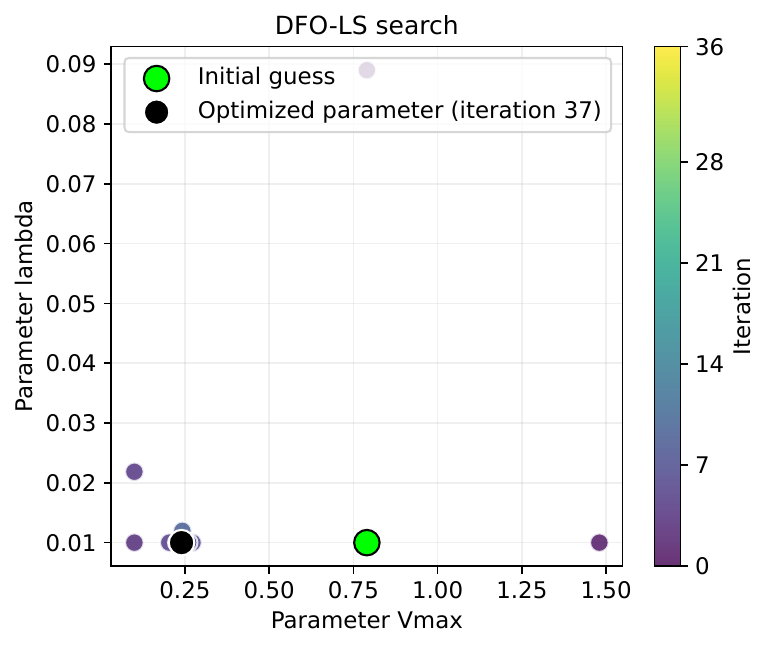}    
\end{center}
\label{fig:searchPEC2paramDFOLS}

\smallskip
Each point corresponds to the parameters obtained in an iteration of the DFO-LS optimization algorithm, being the initial guess obtained from the heuristic step (Figure \ref{fig:searchPEC2param}). The sidebar assigns a color scale to the iterations, allowing identification of their convergence pattern.

\smallskip
\textbf{Source:} the author.
\end{figure}

\newpage
\paragraph{Case 7:} Calibration of one parameter only. In this experiment, we calibrated only the parameter $\lambda$, corresponding to the average daily mortality rate of the entire phytoplankton population in the PEC, as a constant. To do so, we set the parameter  $V_\text{max}$, the maximum daily growth rate of the phytoplankton population in PEC, as a function of time, $V_\text{max}(t)$. This setting was based on sampled data found in the literature, in a manner analogous to the temperature and salinity forcing described previously. More specifically, we assumed the hypothesis that the phytoplankton population in the PEC is dominated by the diatom species \textit{Skeletonema costatum} and \textit{Asterionellopsis glacialis}, which is supported by the data presented in \cite{Brandini2022}. For both of these species, laboratory data on daily growth rates are available \cite{Khan1998,Olmstead2011PilotStudy}, which were considered in setting $V_\text{max}(t)$. We consider the equation:
\begin{equation}\label{assumption1Case7}
    C_\mathrm{PHY}(t) = C_{\mathrm{PHY},1}(t) + C_{\mathrm{PHY},2}(t) ~,
\end{equation}
where $C_{\mathrm{PHY},1}(t)$ and $C_{\mathrm{PHY},2}(t)$ corresponds to the Clorophill-A concentration associated to \textit{S. costatum} and \textit{A. glacialis}, respectively, in the upper region of PEC at a given time instant $t$. Considering the data represented in Figure \ref{fig:phyProportions}, we defined a weighting function $w_\mathrm{PHY}(t)$ such that:
\begin{align}
    C_{\mathrm{PHY},1}(t) & = w_\mathrm{PHY} \cdot C_\mathrm{PHY}(t) ~, \label{assumption2Case7}\\
    \text{and} ~~~~~~  C_{\mathrm{PHY},2}(t) & = \left( 1- w_\mathrm{PHY}(t) \right) \cdot C_\mathrm{PHY}(t) ~, \label{assumption3Case7}
\end{align}
where $0 \leq w_\mathrm{PHY}(t) \leq 1$, at any time instant $t$. In Figure \ref{fig:phyProportions}, $w_\mathrm{PHY}(t)$ is represented as the blue line plot.

\vspace{0.7cm}
\begin{figure}[htbp!]
\caption{Percentage distribution of diatom cell counts in the PEC, when considering only the species \textit{S. costatum} and \textit{A. glacialis}.}
\begin{center}
\includegraphics[scale=0.8]{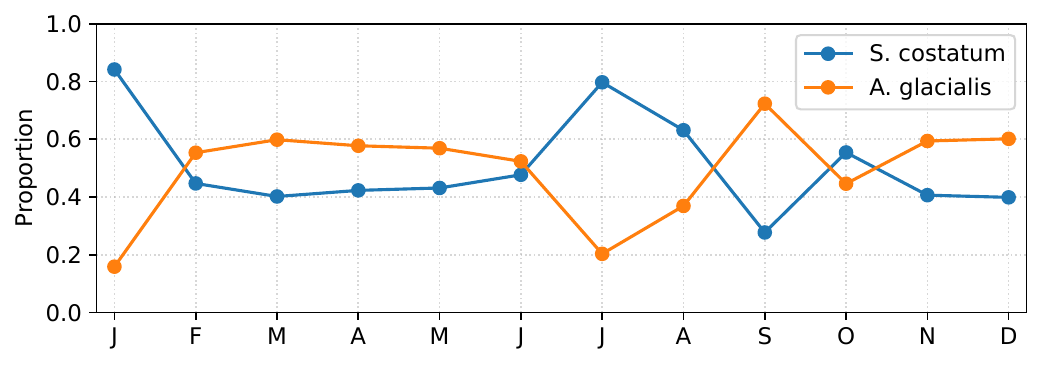}
\end{center}
\label{fig:phyProportions}

\smallskip
The points correspond to monthly averages, and the line represents daily values computed by linear interpolation of the monthly data.

\smallskip
\textbf{Source:} the author.

\smallskip
\textbf{Data:} \citeauthor{Brandini2022} (\citeyear{Brandini2022}).

\vspace{0.2cm}
\end{figure}

\newpage
We can summarize the assumptions presented in equations (\ref{assumption1Case7}) - (\ref{assumption3Case7}) by setting $V_\text{max}(t)$ (Figure \ref{fig:estimatesVmax}) as:
\begin{equation}
    V_\text{max}(t) = w_\mathrm{PHY}(t) \cdot  V_\text{{max},1} + \left( 1- w_\mathrm{PHY}(t) \right) \ V_\text{{max},2} ~,
\end{equation}
recalling that $V_\text{{max},1}$ and $V_\text{{max},2}$ are constant maximum daily growth rates for the \textit{S. costatum} and \textit{A. glacialis} phytoplankton populations, respectively, and considering a slightly modified version of our model, described by the equations
(\ref{ODEcase7calibration01}) - (\ref{ODEcase7calibration03}):

\begin{align}
\dfrac{d C_\mathrm{N}(t)}{dt} ~=~~ & \dfrac{0.588 \ C_\mathrm{N}^\text{river}(t) \cdot Q_\text{river}(t)}{\text{Vol}_\text{box}}+\dfrac{C_\mathrm{N}^\text{low}(t) \cdot Q_\text{ocean}(t)}{\text{Vol}_\text{box}}+ r \cdot \lambda \cdot C_\mathrm{PHY}(t) \nonumber\\ \nonumber \\
& - \alpha(T(t)) \ \beta(S_\text{up}(t)) \  \dfrac{V_{\max}(t) \cdot C_\mathrm{PHY}(t) \ C_\mathrm{N}(t)}{C_\mathrm{N}(t)+K} -  \dfrac{C_\mathrm{N}(t) \cdot Q_\text{ebm}(t)}{\text{Vol}_\text{box}} ~. \label{ODEcase7calibration01}   
\end{align}

\begin{align}
\dfrac{d C_\mathrm{PHY}(t)}{dt} ~=~~ & C_\mathrm{PHY}(t) \cdot \left(  \alpha(T(t)) \ \beta(S_\text{up}(t)) \ \dfrac{V_{\max}(t) \cdot 
 C_\mathrm{N}(t)}{C_\mathrm{N}(t)+K} - \lambda - \dfrac{Q_\text{ebm}(t)}{\text{Vol}_\text{box}} \right) + \nonumber \\ \nonumber \\
 & +\dfrac{C_\mathrm{PHY}^\text{low}(t) \cdot Q_\text{ocean}(t)}{\text{Vol}_\text{box}}~. \label{ODEcase7calibration02}
\end{align}

\begin{align}
\dfrac{d C_\mathrm{N}(t)}{dt} ~=~~ & 0 ~, ~ \dfrac{d C_\mathrm{PHY}(t)}{dt} ~=~ 0 ~. \label{ODEcase7calibration03}
\end{align}

\begin{figure}
\caption{Daily estimates on the parameter $V_\text{max}(t)$.} 
\begin{center}
\includegraphics[scale=0.8]{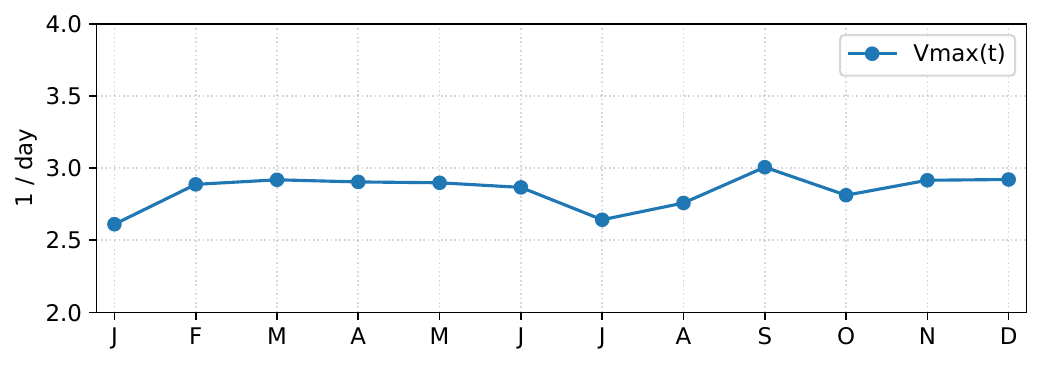}    
\end{center}
\label{fig:estimatesVmax}

\smallskip
Based on data from Figure \ref{fig:phyProportions}.

\smallskip
\textbf{Source:} the author.
\end{figure}

The value obtained by calibrating $\lambda$ was $\Bar{\lambda}=0.01$. This information may be representative of the average mortality rate of the phytoplankton population in the PEC. However, since the lower limit established for the parameter is reached, it also reinforces the possibility that the estuary's circulation flow is so intense that phytoplankton mortality does not have time to make a relevant impact on phytoplankton concentration in the PEC. For the calibrated value of $\lambda$, the fitting obtained from the comparison between observational data and data produced by the model, in relation to concentrations of biochemical tracers in the upper part of the PEC, is presented in figures \ref{fig:fittingPEC1paramN} and \ref{fig:fittingPEC1paramP}.

\vspace{0.7cm}
\begin{figure}[htbp!]
\caption{Fitting of average nitrate concentrations in the upper part of the PEC over a one-year period, for the Case 7 experiment.}
\vspace{-0.3cm}
\begin{center}
\includegraphics[scale=0.8]{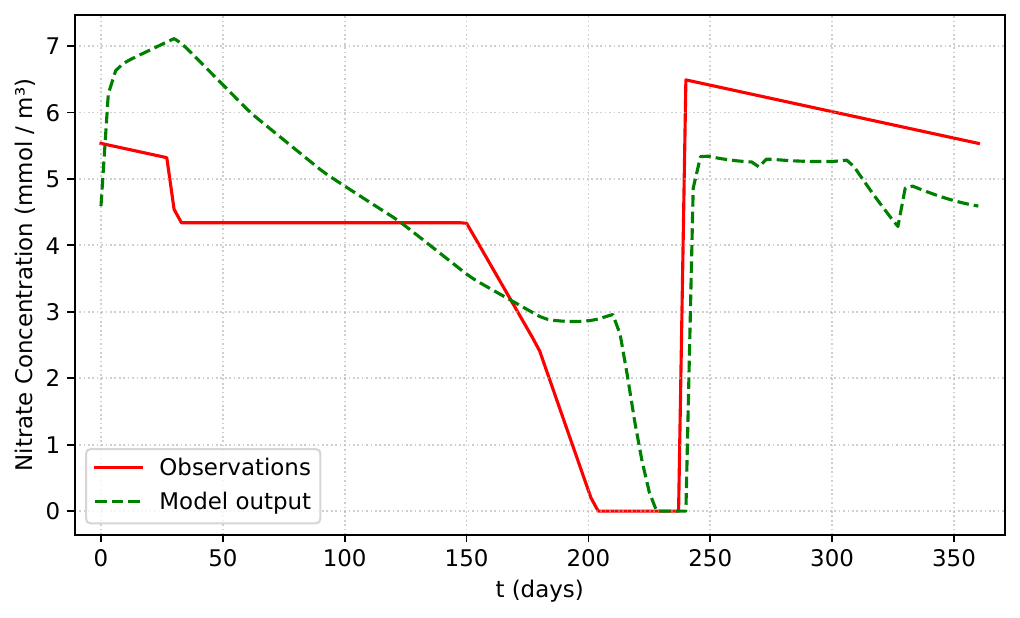}    
\end{center}
\label{fig:fittingPEC1paramN}

\vspace{-0.2cm}
Observations were generated by interpolation of tracers' concentration data (figures \ref{fig:dataNitrate} and \ref{fig:dataPhy}).
The model output plot was obtained by integrating the model described by equations (\ref{ODEcase7calibration01}) - (\ref{ODEcase7calibration03}) with the calibrated parameter $\Bar{\lambda} = 0.01$. Here we consider $N=120$ fitting points.

\smallskip
\textbf{Source:} the author.

\vspace{0.3cm}
\end{figure}

\begin{figure}[htbp!]
\caption{Fitting of average phytoplankton concentrations in the upper part of the PEC over a one-year period, for the Case 7 experiment.}
\begin{center}
\includegraphics[scale=0.8]{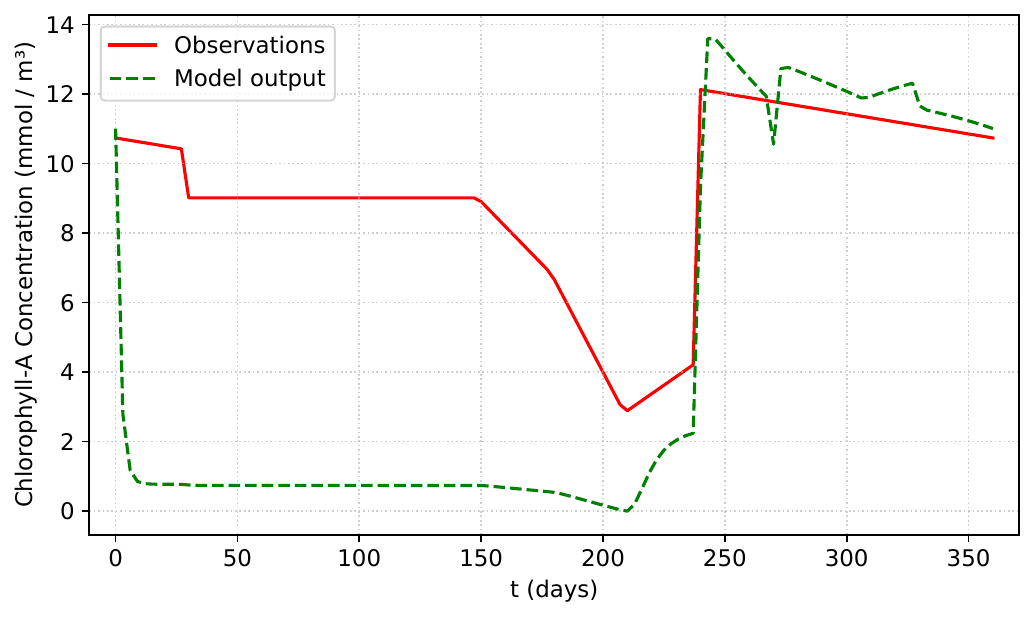}    
\end{center}
\label{fig:fittingPEC1paramP}

\smallskip
See the full description of the experiment settings in Figure \ref{fig:fittingPEC1paramN}.

\smallskip
\textbf{Source:} the author.
\end{figure}

Unlike the previous experiments in this chapter, in an attempt to obtain maximum precision in this calibration, we considered $N=120$ observations and two heuristic steps to find a good initial guess, in addition to imposing a ceiling for the misfit function that was calibrated (figures \ref{fig:searchPEC1param1} and \ref{fig:searchPEC1param2}). The use of increased complexity in the calibration strategy did not guarantee a significant performance improvement, suggesting that in this case the simpler version of the framework, with only one heuristic step, without limiting the ceiling of the misfit function and considering only $N=12$ observations, was much more advantageous to use, presenting lower computational cost, greater ease of implementation, and similar performance to the more complex version.

\begin{figure}[htbp!]
\caption{Heuristic steps on the calibration process of the parameter $\lambda$, in the Case 7 experiment.}
\begin{center}
\includegraphics[scale=0.8]{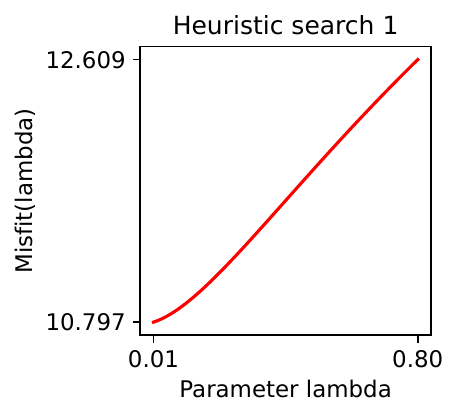} \hspace{0.3cm} 
\includegraphics[scale=0.8]{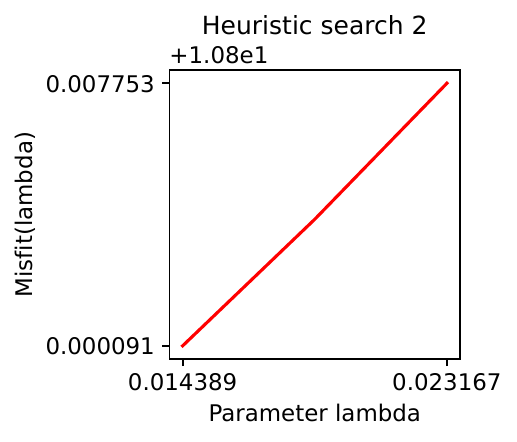}    
\end{center}
\label{fig:searchPEC1param1}

\smallskip
The left-hand and the right-hand plots represent the first and second heuristic steps used to determine the initial guess. The vertical axis shows the values attained by the evaluations of the respective heuristic misfit functions.

\smallskip
\textbf{Source:} the author.
\end{figure}

\begin{figure}[htbp!]
\caption{Optimization step on the calibration process of the parameter $\lambda$, in the Case 7 experiment.}
\begin{center}
\includegraphics[scale=0.8]{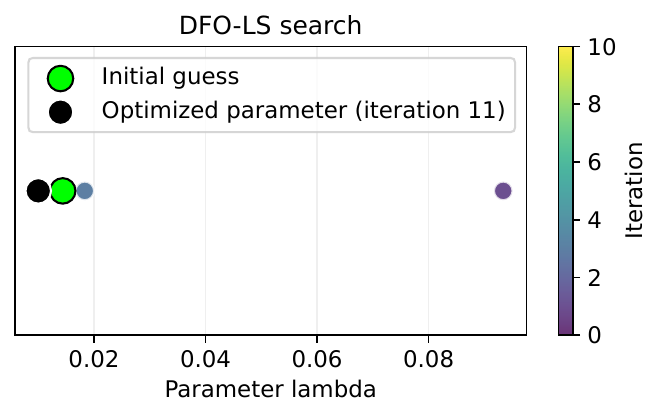}    
\end{center}
\label{fig:searchPEC1param2}

\smallskip
The figure shows the iterations of the parameter search to optimize the fit between model output and observational data for the average tracer concentrations in the upper region of PEC. Each point corresponds to the parameters obtained in an iteration of the DFO-LS optimization algorithm, being the initial guess obtained from the heuristic steps (Figure \ref{fig:searchPEC1param1}). The sidebar assigns a color scale to the iterations, allowing identification of their convergence pattern.

\smallskip
\textbf{Source:} the author.

\vspace{0.2cm}
\end{figure}


The data produced by the model, when using the calibrated parameter, presents behavior similar to that of the observational data. This result reinforces the argument that even a conceptual model, when well-calibrated and based on appropriate assumptions, can reproduce the modeled data with some degree of accuracy.
It would be valuable to explore the ideal horizon for the integration period used in the heuristic misfit and the residual misfit definitions in each case, as using a smaller interval reduces computational cost, although in some cases it may also impair calibration accuracy.
In the benchmark experiments, we observed that the target parameters were accurately recovered, while the fit between the outputs produced by the optimized parameters and those produced by the observational parameters evidenced the applicability of the calibration framework.

\section{Estimating the nitrate input from rivers, rain, and human activities}

Initially, we have set the nitrate concentration related to the flow $Q_\text{river}(t)$ to be 1.47 times the nitrate concentration in the Nhundiaquara River, $C_\mathrm{N}^\text{river}(t)$, which was an assumption partially supported by the literature. However, the model's output did not seem to fit the PEC data adequately. To solve this problem, we considered a change in the model and calibrated a positive constant $w_\text{river}$ that multiplies $C_N^\text{river}(t)$. We considered only the observational data for phytoplankton (Figure \ref{fig:dataPhy}), and denoted the concentration observations for phytoplankton in the upper part of PEC at a time instant $t$ by $C_\mathrm{PHY}^\text{up}(t)$. The mathematical model considered for this experiment is defined by equations (\ref{ODEnitratecalibration01}) - (\ref{ODEnitratecalibration03}).

\begin{align}
\dfrac{d C_\mathrm{N}(t)}{dt} ~=~~ & \dfrac{w_\text{river} \ C_\mathrm{N}^\text{river}(t) \cdot Q_\text{river}(t)}{\text{Vol}_\text{box}}+\dfrac{C_\mathrm{N}^\text{low}(t) \cdot Q_\text{ocean}(t)}{\text{Vol}_\text{box}}+ r \cdot \lambda \cdot C_\mathrm{PHY}^\text{up}(t) \nonumber\\ \nonumber \\
& - \alpha(T(t)) \ \beta(S_\text{up}(t)) \  \dfrac{V_{\max}(t) \cdot C_\mathrm{PHY}^\text{up}(t) \ C_\mathrm{N}(t)}{C_\mathrm{N}(t)+K} -  \dfrac{C_\mathrm{N}(t) \cdot Q_\text{ebm}(t)}{\text{Vol}_\text{box}} ~. \label{ODEnitratecalibration01}   
\end{align}

\begin{align}
\dfrac{d C_\mathrm{N}(t)}{dt} ~=~~ & 0 ~. \label{ODEnitratecalibration03}
\end{align}

During calibration, we noticed a downward trend in the optimized value for parameter $w_\text{river}$, but since our idea was not to reduce it too much in order to keep conservative changes, we considered search bounds of 40\% to 60\% of the previously set value for $w_\text{river}$, which was 1.47. The fitting attained after the calibration process is presented in Figure \ref{fig:fittingNitrate}.

\begin{figure}[htbp!]
\caption{Fitting of average nitrate concentrations in the upper part of the PEC over a one-year period.}
\begin{center}
\includegraphics[scale=0.8]{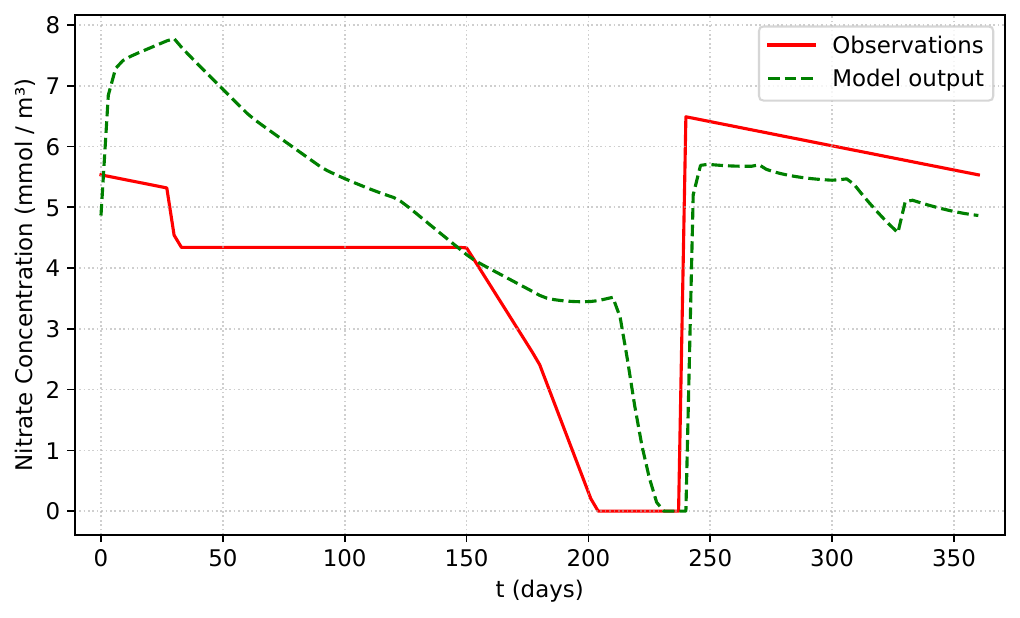}    
\end{center}
\label{fig:fittingNitrate}

\smallskip
Fitting for the experiment on calibrating the nitrate input on PEC from riverine, pluvial, and human sources. Observations were generated by interpolation of tracers' concentration data (figures \ref{fig:dataNitrate} and \ref{fig:dataPhy}).
The model output plot was obtained by integrating the model described by equations (\ref{ODEnitratecalibration01}) - (\ref{ODEnitratecalibration03}) with the calibrated parameter $\Bar{w}_\text{river} = 0.588$.
Here we consider $N=120$ fitting points.

\smallskip
\textbf{Source:} the author.
\end{figure}

We set the parameters $V_\text{max}(t)$, based on the \textbf{Case 7} data, and $\lambda = 0.05$, constant. The parameter calibration is presented in figures \ref{fig:searchnitrate1} and \ref{fig:searchnitrate2}, as the two-step heuristic search for an adequate initial guess, which was found as the lowest permitted by the searching bounds.
The DFO-LS search converged this same value:
\begin{align*}
    \Bar{w}_\text{river} & = 0.588 .
\end{align*}

\begin{figure}[htbp!]
\caption{Calibration process of the parameter $w_\text{river}$.}
\begin{center}
\includegraphics[scale=0.8]{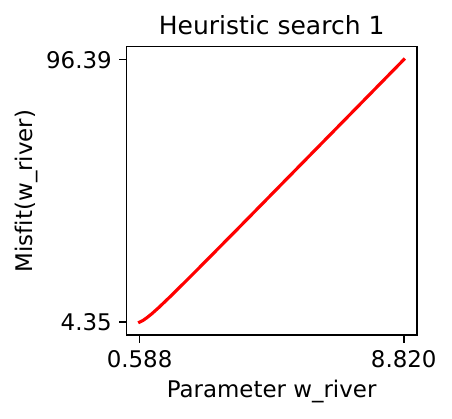} \hspace{0.3cm} 
\includegraphics[scale=0.8]{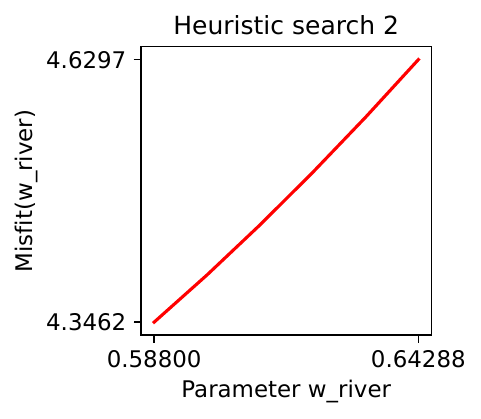}    
\end{center}
\label{fig:searchnitrate1}

\smallskip
Heuristics for the experiment to calibrate the nitrate input to the PEC from riverine, pluvial, and human sources. The left-hand and the right-hand plots represent the first and second heuristic steps used to determine the initial guess. The vertical axis shows the values attained by the evaluations of the respective heuristic misfit functions.

\smallskip
\textbf{Source:} the author.
\end{figure}

\begin{figure}[htbp!]
\caption{Optimization step on the calibration process of the parameter $w_\text{river}$.}
\begin{center}
\includegraphics[scale=0.8]{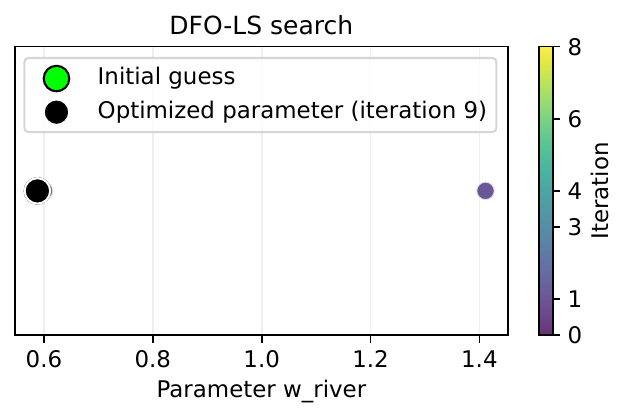}    
\end{center}
\label{fig:searchnitrate2}

\smallskip
Parameter search for the experiment on calibrating the nitrate input on PEC from riverine, pluvial, and human sources. 
The figure shows the iterations of the parameter search to optimize the fit between model output and observational data for the average tracer concentrations in the upper region of PEC. Each point corresponds to the parameters obtained in an iteration of the DFO-LS optimization algorithm, being the initial guess obtained from the heuristic steps (Figure \ref{fig:searchnitrate1}). The sidebar assigns a color scale to the iterations, allowing identification of their convergence pattern.

\smallskip
\textbf{Source:} the author.
\end{figure}

Since we considered the fit obtained to be satisfactory, we considered the result of this calibration for our model presented in \textit{Chapter \ref{chapterPEC}}. This modeling strategy proved useful for extrapolating data in the model based on correlated data available.
\chapter{Conclusions}

\section*{Summary}
This thesis developed a bespoke conceptual biogeochemical model for the PEC, and introduced a data-constrained framework for calibrating its parameters based on a derivative-free optimization method. We explored the tracer-conservation problem and the formulation of a least-squares problem to accommodate the problem features in different scenarios. The framework was exercised in synthetic (twin) experiments and with one year of NO$_3$ and phytoplankton observations from the PEC, followed by a scenario analysis of increased riverine nitrate load, exemplifying how to apply the same framework for a marine region when data is available.

\section*{Implications, Relevance, and Limitations}
This work delivers a reproducible formulation of parameter calibration for marine biogeochemical tracers as a weighted least-squares problem, as well as a demonstration on how to use a modern derivative-free optimizer for calibrating a conceptual estuarine model using both synthetic and real data, and a small, documented pipeline -- from problem specification to optimization runs and diagnostics -- that can be adapted to other systems.

Within the Brazilian context, the study is, to our knowledge, among the first to combine data-constrained calibration of an estuarine biogeochemical model with state-of-the-art derivative-free optimization, with several possibilities for future development. Even though a conceptual model was used, the exercise opens a path for more ambitious national applications as observational datasets and computational resources expand.

The work also clarifies where the approach falls short. The biogeochemical state can be reproduced reasonably well even when individual parameters are not precisely recovered (equifinality), calling for performance metrics and uncertainty analyses that prioritize state trajectories. Temporal and spatial data sparsity -- and measurement noise -- limit identifiability; simple smoothing helped, but more explicit noise modeling (for example, within DFO-LS) remains to be explored. By design, the conceptual model does not capture all seasonal and spatial heterogeneity of the PEC present in the bibliography, which caps attainable fit and suggests structural extensions (for example, seasonal terms in $V_{\text{max}}$ and $\lambda$). Finally, time/CPU constraints limited restarts and methodological variants, which matter in nonconvex settings.
Practically, this study lays groundwork for conceptual modeling studies that seek to integrate local data, marine biogeochemical formulations, and modern optimization in a reproducible way.

\section*{Future avenues for this work}
Although this thesis focused on experimenting with conceptual models, the methodology developed here can be applied, with minor adaptations, to complex, multidimensional settings. Below we exemplify two possible avenues for further work.

 A natural next step is to extend the basic framework developed here to a medium-complexity model of the PEC by coupling a (potentially more complex) biogeochemical module to a circulation model that better represents the region’s geography, introducing seasonal structure in $V_{\text{max}}$ and $\lambda$ (harmonics or splines), add additional tracers when available (PO$_4$, O$_2$), and exploring the optimization options available, such as noise-aware DFO-LS with multi-start restarts and CMA-ES. Scenario analysis can be expanded (for example, joint changes in riverine NO$_3$ loads and freshwater discharge) with uncertainty quantification via ensembles of calibrated parameter sets. The anticipated availability of updated data for the PEC in the near future will enable insights into contemporary environmental issues and the impacts of human activities.
 In fact, historical reconstructions show that human interference on PEC has profoundly altered biogeochemical functioning over the last decades \cite{marines2023}. Together with evidence of eutrophication and shifts in organic matter sources \cite{Martins2010,Martins2015}, this scenario underscores the need for predictive tools that can explore the effects of various stressors. Incorporating new observational datasets into the current modeling framework is a future aim that will enable more robust scenario testing, such as assessing thresholds for harmful algal blooms \cite{Brandini1985, Procopiak} or evaluating strategies for mitigating anthropogenic nutrient loads. The combination of updated empirical data with model-based experiments represents a promising approach to enhancing environmental management and informing policy interventions in estuarine environments.

\cleardoublepage
\thispagestyle{empty}
\centerline{\Large\bfseries BIBLIOGRAPHY}

\vspace{1cm}

\phantomsection
\addcontentsline{toc}{chapter}{BIBLIOGRAPHY}

\printbibliography[heading=none]


\cleardoublepage
\thispagestyle{empty}

\begin{center}
\begin{minipage}{0.9\textwidth}
\centering
\Large\bfseries
APPENDIX -- Supplementary Materials for Computational Implementation and Experiments
\end{minipage}
\end{center}
\vspace{1cm}

\phantomsection
\addcontentsline{toc}{chapter}{APPENDIX -- Supplementary Materials for Computational Implementation and Experiments}

This appendix complements the references made throughout the thesis, gathering all supplementary materials related to the computational experiments conducted here. It includes: (i) the organization of the source code, (ii) implementation details of the mathematical models, (iii) instructions for executing the experiments, (iv) the DFO-LS optimization method used for parameter calibration, and (v) practical examples of applying the DFO-LS method to the models.

\section*{A.1 Repository structure and code organization}
To ensure reproducibility and transparency, all source code used in this thesis is publicly available. The repositories are logically organized as follows:

\begin{itemize}
    \item \textbf{/Optimization/} – Functions associated with the DFO-LS method presented in \textit{Chapter \ref{chapterdfols}}.
    \item \textbf{/Experiments/} – Scripts used in the computational experiments discussed in \textit{Chapter \ref{chapterframework}}.
    \item \textbf{/PECmodel/} – Implementation of the mathematical models described in \textit{Chapter \ref{chapterModelref}}, including differential equations and auxiliary functions, files related to the PEC model introduced in \textit{Chapter \ref{chapterPEC}}, and calibration experiments for the estuary model used in \textit{Chapter \ref{chapterCalibratingPEC}}.
\end{itemize}

All repositories are available at:

\begin{center}
\url{https://github.com/leticiabecher/ThesisExperiments}
\end{center}

\section*{A.2 Computational environment and execution procedures}
The experiments described in chapters \ref{chapterModelref} and \ref{chapterframework} were conducted under the following computational environment:

\begin{itemize}
    \item \textbf{Python}: version 3.10 or newer.
    \item \textbf{Libraries}: NumPy, SciPy, Matplotlib, Pandas.
    \item \textbf{Operating system}: Linux or Windows with Python support.
\end{itemize}

To run any experiment, clone the repository and execute:

\begin{verbatim}
python experiment_name.py
\end{verbatim}

Each folder includes further explanations inside its respective \texttt{README.md} file.

\section*{A.3 Implementation of the mathematical models}
The mathematical models described in the main chapters were fully implemented in Python. The differential equations from \textit{Chapter \ref{chapterModelref}} were solved using SciPy’s \texttt{solve\_ivp} integrator.

Example implementation snippet:

\begin{verbatim}
def model_equations(t, y, params):
    k1, k2, vmax = params
    dydt = [
        k1 * y[0] - k2 * y[1],
        vmax * y[1] / (0.1 + y[1])
    ]
    return dydt
\end{verbatim}

\section*{A.4 The DFO-LS optimization method}
The DFO-LS method described in \textit{Chapter \ref{chapterdfols}} is a derivative-free optimization algorithm for least-squares problems. The adapted pseudocode is shown in Algorithm~\ref{algDFOLS}.

\begin{algorithm}[H]
\caption{DFO-LS Method adapted from \cite{DFOLSmops2022}}
\label{algDFOLS}
\begin{algorithmic}[1]
\STATE \textbf{Input:} number of parameters $n$, initial point $\theta_0$, initial trust region radius $\Delta_0$, parameter bounds, and maximum number of function evaluations.
\STATE Construct the initial interpolation set $Y_0$ by evaluating the misfit function at $n+1$ points.
\FOR{each iteration}
    \STATE Build a local quadratic model of the misfit.
    \STATE Solve the trust-region subproblem.
    \STATE Evaluate the candidate point.
    \STATE Update the interpolation set.
    \STATE Adjust the trust-region radius.
\ENDFOR
\STATE \textbf{Output:} estimated calibrated parameters.
\end{algorithmic}
\end{algorithm}

Examples and explanations of the DFO-LS Python solver usage, including basic and advanced settings, can be found in \cite{roberts2025dfo-ls}.

\section*{A.5 Calibration framework tests}
In \textit{Chapter \ref{chapterframework}}, several calibration examples are presented. 
The complete scripts are available at:
\begin{center}
\url{https://github.com/leticiabecher/ThesisExperiments/tree/main/Framework_testing}
\end{center}

\section*{A.6 PEC model experiments}
The PEC model introduced in \textit{Chapter \ref{chapterPEC}} was calibrated using the scripts contained in the PEC directory. The materials include:

\begin{itemize}
    \item differential equations of the PEC model;
    \item parameter sets used in calibration experiments;
    \item optimization scripts;
    \item result validation codes and visualizations.
\end{itemize}

Full implementations are located at:

\begin{center}
\url{https://github.com/leticiabecher/ThesisExperiments/tree/main/PEC_experiments}
\end{center}

\section*{A.7 Estuary calibration experiments}
The experiments described in \textit{Chapter \ref{chapterCalibratingPEC}} use real environmental data and apply the DFO-LS method for parameter estimation.

The supplementary files include:

\begin{itemize}
    \item data preprocessing routines;
    \item the estuary dynamical model;
    \item the misfit function for calibration;
    \item experiment execution scripts and visualization tools.
\end{itemize}

Complete code available at:

\begin{center}
\url{https://github.com/leticiabecher/ThesisExperiments/tree/main/PEC_experiments}
\end{center}

\bigskip
This appendix supports and documents all computational methods and experiments performed in the thesis, ensuring full reproducibility of the presented results.

\end{document}